\documentclass[superscriptaddress,onecolumn,nofootinbib,longbibliography,a4paper,accepted=2026-09-29]{quantumarticle}
\pdfoutput=1

\usepackage{adjustbox,amsfonts,amsmath,amssymb,amsthm,array,bbm,bm,braket,centernot,color,enumerate,enumitem,fancyhdr,graphicx,hyperref,mathdots,mathtools,mdwlist,multirow,pgfplots,soul,subfigure,thmtools,tikz,tikz-qtree,titlesec,xcolor}

\usepackage[toc]{appendix}
\usepackage[breakable]{tcolorbox}

\tcbuselibrary{theorems}
\newcounter{questioncounter}
\newcommand{\questiondisplay}[1]{\arabic{questioncounter}} 
\newtcbtheorem[use counter=questioncounter,number format=\questiondisplay]{questionbox}{Central Question}
{
	colback=quantumviolet!15,
	colframe=quantumviolet,
	fonttitle=\bfseries,
        breakable = true,
}
{questioncounter}

\definecolor{mygray}{gray}{0.85}
\definecolor{quantumviolet}{RGB}{62, 37, 113}
\usetikzlibrary{chains,quantikz2}

\hypersetup{
    colorlinks=true,
    linkcolor=blue,
    citecolor=blue,
    urlcolor=blue,
    hypertexnames=false
}

\pgfplotsset{compat=1.18} 

\allowdisplaybreaks 

\usepackage{titlesec}

\theoremstyle{plain}
\newtheorem{thm}{Theorem}
\newtheorem{assumption}[thm]{Assumption}
\newtheorem{lemma}[thm]{Lemma}
\newtheorem{prop}[thm]{Proposition}
\newtheorem{cor}[thm]{Corollary}

\theoremstyle{definition}
\newtheorem{definition}[thm]{Definition}
\newtheorem{rmk}[thm]{Remark}

\numberwithin{thm}{section}
\numberwithin{assumption}{section}
\numberwithin{lemma}{section}
\numberwithin{prop}{section}
\numberwithin{cor}{section}
\numberwithin{definition}{section}
\numberwithin{rmk}{section}
\numberwithin{equation}{section}

\definecolor{mygreen}{RGB}{0, 100, 0}
\newcommand{\ketbra}[2]   {\left|#1\middle\rangle\!\middle\langle#2\right|}

\newcommand{\abs}[1]      {\left| #1 \right|}
\newcommand{\norm}[1]     {\left\| #1 \right\|}

\newcommand{\ot}{\otimes}
\newcommand\F{ {\mathcal F} }

\newcommand\D{ {\mathcal D} }
\newcommand\Q{ {\mathcal Q} }

\newcommand\T{ {\mathcal T} }

\newcommand{\mx}{\mathrm{max}}

\renewcommand{\v}[1]{\ensuremath{\boldsymbol #1}}
\newenvironment{aligns}{\subequations \align} {\endalign \endsubequations}

\newcommand{\hh}[1]{\ensuremath{\mathfrak{H}\left[#1\right]}}

\newcommand{\R}{\mathbb{R}}
\newcommand{\myC}{\mathbb{C}}

\newcommand{\myN}{\mathbb{N}}

\newcommand{\lrb}[1]{\left ( #1 \right )}
\newcommand{\lrsb}[1]{\left [ #1 \right ]}
\newcommand{\lrcb}[1]{\left \{ #1 \right \}}
\newcommand{\defeq}{\coloneqq}
\renewcommand{\emptyset}{\varnothing}

\let\dotlessi\i
\renewcommand{\i}{\ensuremath{\mathbf{i}}}

\newcommand{\realpart}[1]{\operatorname{Re}\lrb{#1}}
\newcommand{\imaginarypart}[1]{\operatorname{Im}\lrb{#1}}
\newcommand{\myspan}[1]{\operatorname{span}\lrcb{#1}}

\newcommand{\mypoly}[1]{\operatorname{poly} \lrb{#1}}
\newcommand{\mypolylog}[1]{\operatorname{polylog} \lrb{#1}}
\newcommand{\mylog}[1]{\operatorname{log} \lrb{#1}}
\newcommand{\myln}[1]{\operatorname{ln} \lrb{#1}}

\newcommand{\mymax}[1]{\operatorname{max}\lrb{#1}}
\newcommand{\mymin}[1]{\operatorname{min}\lrb{#1}}
\newcommand{\myO}[1]{O\lrb{#1}}
\newcommand{\mytO}[1]{\tilde{O}\lrb{#1}}
\newcommand{\myOmega}[1]{\Omega\lrb{#1}}
\newcommand{\myTheta}[1]{\Theta\lrb{#1}}

\newcommand{\diag}[1]{\operatorname{diag}\lrb{#1}}

\newcommand{\mbt}{\mathbf{t}}
\newcommand{\mbf}{\mathbf{T}}

\newcommand{\tdeg}[1]{\widetilde{\operatorname{deg}}\lrb{#1}}

\definecolor{edgecolor}{RGB}{72, 37, 113}
\definecolor{nodecolor}{RGB}{72, 37, 113}
\newlength{\eqTree}
\newlength{\treeTree}
\newcommand{\bt}{\bullet}
\newcommand{\btx}{\textbf{x}}

\begin{document}


\title{Quantum algorithms for general nonlinear dynamics based on the Carleman embedding}

\author{David Jennings}
\affiliation{PsiQuantum, 700 Hansen Way, Palo Alto, CA 94304, USA}

\author{Kamil Korzekwa}
\affiliation{PsiQuantum, 700 Hansen Way, Palo Alto, CA 94304, USA}

\author{Matteo Lostaglio}
\affiliation{PsiQuantum, 700 Hansen Way, Palo Alto, CA 94304, USA}

\author{Andrew T Sornborger}
\affiliation{Computer, Computational, and Statistical Sciences Division, Los Alamos National Laboratory, Los Alamos, New Mexico 87545, USA}

\author{Yi\u{g}it Suba\c{s}\dotlessi}
\affiliation{Computer, Computational, and Statistical Sciences Division, Los Alamos National Laboratory, Los Alamos, New Mexico 87545, USA}

\author{Guoming Wang\textsuperscript{*}}
\affiliation{PsiQuantum, 700 Hansen Way, Palo Alto, CA 94304, USA}

\renewcommand{\thefootnote}{\fnsymbol{footnote}} 

\footnotetext[1]{Corresponding author. Email: gwang@psiquantum.com}
\renewcommand{\thefootnote}{\arabic{footnote}}

\begin{abstract}
    
Important nonlinear dynamics, such as those found in plasma and fluid systems, are typically hard to simulate on classical computers. Thus, if fault-tolerant quantum computers could efficiently solve such nonlinear problems, it would be a transformative change for many industries. In a recent breakthrough [Liu et al., PNAS 2021], the first efficient quantum algorithm for solving nonlinear differential equations was constructed, based on a single condition $R<1$, where $R$ characterizes the ratio of nonlinearity to dissipation. This result, however, is limited to the class of purely dissipative systems with negative log-norm, which excludes application to many important problems. In this work, we correct technical issues with this and other prior analysis, and substantially extend the scope of nonlinear dynamical systems that can be efficiently simulated on a quantum computer in a number of ways. Firstly, we extend the existing results from purely dissipative systems to a much broader class of stable systems, and show that every quadratic Lyapunov function for the linearized system corresponds to an independent $R$-number criterion for the convergence of the Carleman scheme. Secondly, we extend our stable system results to physically relevant settings where conserved polynomial quantities exist. Finally, we provide extensive results for the class of non-resonant systems. With this, we are able to show that efficient quantum algorithms exist for a much wider class of nonlinear systems than previously known, and prove the BQP-completeness of nonlinear oscillator problems of exponential size. In our analysis, we also obtain several results related to the Poincar\'{e}-Dulac theorem and diagonalization of the Carleman matrix, which could be of independent interest.
\end{abstract}

\maketitle
\newpage
\tableofcontents
\newpage


\section{Overview}


\subsection{Introduction}

The simulation of dynamical systems plays a central role in science and engineering, underpinning models across physics, chemistry, biology, and finance. While linear differential equations are well understood and efficiently solvable with classical methods, nonlinear differential equations present a significantly greater challenge. Their nonlinearity often leads to complex phenomena, and solutions typically require costly numerical integration schemes, occupying a substantial part of modern high-performance computing resources.

To overcome these challenges, one may need to completely change the computational paradigm, and quantum computers are a compelling candidate for that. As the field of quantum computing matures, it broadens its scope for scientific applications~\cite{lin2022lecture}, with recent advancements in quantum algorithms presenting efficient methods for solving linear ordinary differential equations~\cite{berry2014high,berry2017quantum,berry2024quantum,an2023linear,an2023quantum,jennings2023cost}. The seminal paper of Liu~\textit{et~al.}~\cite{liu2021efficient} provided the first convincing argument that quantum algorithms could indeed provide a computational advantage also for solving nonlinear differential equations. However, the nonlinear regime is far less explored, and it is still an active area of research~\cite{surana2022efficient,costa2023further,krovi2023improved,liu2023efficient, javier2025quantum}.

One of the central challenges in bringing the power of quantum computing to solving nonlinear problems is that quantum computers, by their very construction, operate linearly on quantum states -- evolving them via unitary transformations governed by the linear Schr\"{o}dinger equation. This fundamental mismatch creates a barrier to directly implementing nonlinear dynamics on quantum hardware. A promising strategy to overcome this limitation is to embed the nonlinear system into a higher-dimensional \emph{linear representation}, such that its evolution can be captured through the simulation of an associated linear system~\cite{engel2021linear,lin2022challenges}. This embedding -- achieved, for example, through Carleman linearization~\cite{Carleman32}, Liouville representations~\cite{lin2022challenges}, Koopman-von Neumann methods~\cite{joseph2020koopman}, or other lifting techniques  -- makes it possible to harness the rich toolbox of quantum linear algebra and Hamiltonian simulation~\cite{lin2022lecture}. By doing so, one can leverage quantum speedups for linear systems to indirectly simulate nonlinear behavior, provided that the embedding preserves essential properties of the dynamics and the associated error converges under suitable conditions.

In this work, we establish rigorous convergence criteria for Carleman linearization in the context of quantum simulation. We derive explicit conditions under which the Carleman embedding error converges for a wide range of nonlinear dynamical systems. First, we extend existing results from purely dissipative systems to a much broader class of stable systems, and show that each quadratic Lyapunov function for the linearized system yields an independent $R$-number criterion for the convergence of the Carleman scheme. We further extend the analysis to physically relevant systems that possess conserved polynomial observables, which fall outside the stable setting treated above. Finally, we present a collection of results for non-resonant systems and leverage them to establish the BQP-completeness of certain nonlinear oscillator problems. Together, these results significantly broaden the class of nonlinear dynamical systems for which quantum simulation can offer provable computational advantages.

In essence, we significantly extend the setting where quantum algorithms can tackle nonlinear dynamics via Carleman embedding from 
\begin{center}
    \begin{tcolorbox}[
      colback=quantumviolet!15,
      colframe=quantumviolet,
      boxrule=0.9pt,            
      arc=2pt,                  
      left=8pt,right=8pt,       
      top=6pt,bottom=6pt,       
      width=0.5\linewidth,         
    ]
        purely dissipative systems under  $R<1$,
    \end{tcolorbox}
\end{center}
\noindent where $R$ is a specific measure of nonlinearity to dissipation (the main result of Ref.~\cite{liu2021efficient}), to  
\begin{center}
    \begin{tcolorbox}[
      colback=quantumviolet!15,
    colframe=quantumviolet,
      boxrule=0.9pt,            
      arc=2pt,                  
      left=8pt,right=8pt,       
      top=6pt,bottom=6pt,       
      width=0.95\linewidth,         
    ]
        stable/conservative/nonresonant systems under a family of generalized $R<1$ conditions,
    \end{tcolorbox}
\end{center}
where the generalized conditions integrate information about Lyapunov functions, conserved quantities and resonance properties of the linearized dynamics to provide families of convergence guarantees.


\subsection{Aims of this work}

We shall aim at providing a broad framework for approximately embedding a general quadratic nonlinear ordinary differential equation (ODE) of the form
\begin{equation}
\label{eq:nonlinear}
    \dot{x}(t) = F_0 + F_1 x(t) + F_2 x(t)^{\otimes 2}, \quad x(t) \in \mathbb{C}^N,
\end{equation}
into a linear problem in a much larger space, which is then to be `solved' on a fault-tolerant quantum computer. Here, $F_1 \in \mathbb{C}^{N\times N}$ is a square matrix encoding linear dynamics of an $N$-dimensional system of interest; $F_2 \in \mathbb{C}^{N\times N^2}$ is a rectangular matrix encoding nonlinearities that maps two copies of the vector $x$ into one copy; and $F_0 \in \mathbb{C}^{N\times 1}$ is a rectangular matrix (a vector) encoding the driving. In what follows, we shall typically assume that these matrices are time-independent, so that Eq.~\eqref{eq:nonlinear} is an \emph{autonomous} ODE system. However, we shall also give techniques to deal with some time-dependent scenarios, including a time-dependent linear term~$F_1$. Importantly, note that quadratic nonlinear ODEs considered here in principle encompass all autonomous polynomial differential equations, as they can always be transformed into the quadratic form through the introduction of auxiliary variables (see, for example~\cite{kerner1981universal}). 

It is commonplace to obtain systems of the form~\eqref{eq:nonlinear} by spatial discretization of time-dependent nonlinear partial differential equations (PDEs). If the PDE is first order in time and has quadratic nonlinearity, one obtains directly \eqref{eq:nonlinear} by spatial discretization. In principle, however, higher order time-derivatives and higher order polynomial nonlinearities can also be reduced to the form~\eqref{eq:nonlinear} by introducing additional variables. 

Consider a scheme where we track the tensor powers $x(t)^{\otimes j}$ for $j \in \mathbb{N}$, a technique known in the classical literature as \emph{Carleman embedding} (or \emph{Carleman linearization})~\cite{Carleman32, kowalski1991nonlinear}, to be distinguished from standard `linearization', as the Carleman embedding captures nonlinear effects. For example, for $N=3$, $x^{\otimes j}$ is a vector involving all monomials of degree $j$, meaning $x^{j_1}_1 x^{j_2}_2 x^{j_3}_3$ with $j_1 + j_2 + j_3 = j$.  As we describe in more detail in Section~\ref{sec:carlemanlinearizationscheme}, we can write a linear differential equation for $x^{\ot j}$ that recursively involves $x^{\ot (j+1)}$. This defines an infinite-dimensional linear differential equation for the vector of all powers $\{ x^{\ot j} \}_{j=1}^{\infty}$. In practice, however, this linear system needs to be truncated at some finite power $k$. We can then write a \emph{Carleman ODE} system for the truncated vector:
\begin{equation}
\begin{aligned}
\label{eq:linearODE}
    \dot{y}(t)& = Ay(t) + a,  \\
     y(t) &=  [y^{[1]}(t), y^{[2]}(t), \dots, y^{[k]}(t)]^T, \quad  \quad y^{[j]}(t) \in \mathbb{C}^{N^j},  \\
    y(0) &= [ x(0),x(0)^{\otimes 2},x(0)^{\otimes 3}, \dots, x(0)^{\otimes k} ]^T,
\end{aligned}
\end{equation}
where $A$ is an $N+N^2 + \dots + N^{k}$-dimensional square matrix -- the \emph{Carleman matrix} -- and $a$ is a $N+N^2 + \dots + N^{k}$-dimensional vector given by $[F_0,0,\dots,0]^T$. While this suggests that we can trade nonlinearity for higher dimensionality, doing so on a classical computer is generally impractical due to the prohibitive memory requirements of simulating the system~\eqref{eq:linearODE}. In contrast, fault-tolerant quantum computers can efficiently perform certain linear operations on vector spaces whose dimensions scale exponentially with the number of qubits. This opens the door for linear embedding methods to serve as a potentially powerful approach to tackling nonlinear problems on a quantum computer.

However, to provide rigorous evidence that such methods can be effectively used to construct quantum algorithms for nonlinear problems, we must first precisely delineate their regimes of validity:

\begin{questionbox}{Carleman convergence}{}
    \label{question1}
    Under what conditions is the solution $y(t)$ of the Carleman ODE system \eqref{eq:linearODE} $\epsilon$-close to the vector $[x(t), x(t)^{\otimes 2}, \dots, x(t)^{\otimes k}]^T$, where $x(t)$ is the solution of the nonlinear system described by Eq.~\eqref{eq:nonlinear}? Can we set $k = \myO{\mylog{t/\epsilon}}$?  
\end{questionbox}

An affirmative answer to the above central question \emph{provably} reduces the problem of solving the nonlinear ODE in Eq.~\eqref{eq:nonlinear} for an $N$-dimensional vector $x$ to the linear ODE in Eq.~\eqref{eq:linearODE} for the $\sim N^{k}$-dimensional vector~$y$. The Carleman ODE system~\eqref{eq:linearODE} can be encoded in $\myO{k \log N} = \myO{\mylog{t/\epsilon}\log N}$ qubits, making it amenable to quantum simulation. For a range of explicit a priori results on this problem see Ref.~\cite{forets2017explicit, amini2022carleman, surana2022efficient, liu2021efficient, krovi2023improved, costa2023further, wu2024quantum, chen2024carleman} and references therein.

However, we emphasize that answering Central Question 1 is not the end of the story. Constructing an efficient quantum algorithm for the nonlinear differential equation~\eqref{eq:nonlinear} also requires analyzing the complexity of preparing a coherent quantum encoding of the Carleman ODE system. We formulate this as our second central question, presented as follows:

\begin{questionbox}{Algorithmic efficiency}{}
    Under the same conditions introduced to ensure an affirmative answer to the Central Question~1, and given access to unitaries block-encoding the matrices $F_0$, $F_1$ and $F_2$, what is the asymptotic query complexity of a quantum algorithm outputting a quantum state encoding a vector $[\tilde{x}(0), \tilde{x}(h), \tilde{x}(2h), \dots, \tilde{x}(t)]^T$ 
    such that each  
    $\tilde{x}(j h)$ is $\epsilon$-close to the $x(jh)$ 
    solving the nonliner ODE system \eqref{eq:nonlinear} at the corresponding time? And what is the complexity of outputting just $\tilde{x}(t)$?    
\end{questionbox}

While the detailed scalings will depend on the specific algorithmic choices and setups, we will see that the results from Central Question 1 allow us to derive a range of new efficiency guarantees for quantum algorithms for nonlinear differential equations in regimes where such guarantees were not previously available.

Note that, even when we can guarantee efficiency of the quantum algorithm in the sense of Central Question~2,
one still needs to analyze on a case-by-case basis how  to input data into the quantum computer (encoding the matrices $F_0$, $F_1$, $F_2$, and the initial condition $x(0)$), as well as extracting information from the `Carleman vector' $y(t)$ solving Eq.~\eqref{eq:linearODE}. This question necessarily requires us to focus on a specific application, which is beyond the scope of the present work. 

    The class of systems~\eqref{eq:nonlinear} under consideration here contains quantum dynamics as a special case ($x(t) = \ket{\psi(t)}, F_1 = -\i H$, for $H$ Hermitian, $F_2=0$, $F_0 =0$), for which BQP-complete problems are the most promising early-generation candidates for quantum advantage. However, it is relevant to ask whether there are decision problems involving classical dynamical systems admitting exponential advantage. This was recently answered in the positive for a system of exponentially many sparsely coupled linear oscillators~\cite{babbush2023exponential} (which is a case with $F_2=0$, $F_0 =0$). The result was then extended to other classes of similar problems, such as exponentially many sparsely coupled linear oscillators with damping inverse polynomial in the number of qubits~\cite{krovi2024quantum} and Gaussian bosonic circuits on exponentially many modes~\cite{barthe2025gate}. Given the established exponential quantum advantage in the linear regime, our ultimate goal is to determine to which degree this advantage can be extended into the nonlinear setting, ideally achieving exponential or substantial polynomial speedups over classical algorithms. So we ask and answer the following question in the affirmative \footnote{In the process of completing this work, Ref.~\cite{bravyi2025quantum} appeared, where BQP completeness of a class of nonlinear problems was proved. However, there are two key differences between their results and ours. First, we focus on deterministic dynamics governed by standard differential equations, while they consider noisy systems modeled by stochastic differential equations and reduce the problem to the Kolmogorov backward equation, a linear PDE, which they solve using a spectral method. Their BQP-completeness result applies only when the noise is sufficiently strong. Second, they study the evolution of a fixed observable over time and aim to estimate its expected value at a given time, averaged over noise. In contrast, we are concerned with the evolution of the system's state itself, and the output of our quantum algorithm can be used flexibly depending on the user's needs.}:

\begin{questionbox}{Exponential advantage in the nonlinear regime}{}
    Can we establish an exponential quantum advantage over classical methods in the simulation of genuinely nonlinear classical dynamical problems?
\end{questionbox}

Finally, in the process of providing a range of sufficient conditions for a positive answer to the Central Question~1, a key technical contribution of this work is the development of a detailed framework for when a nonlinear dynamical system near an equilibrium can be brought into a normal form via nonlinear transformations of the dynamical variable. At the level of the Carleman matrix $A$, this becomes a question of identifying the structure of the similarity matrix $V$ that brings $A$ into Jordan normal form, and how the particular dynamical properties -- such as the spectral properties of linear terms, the strength of nonlinearities -- impact this similarity transformation and the convergence of the Carleman algorithm. These technical results are particularly important, for instance, in the context of Central Question~2, where we apply quantum algorithms such as \cite{berry2017quantum} whose runtime depends on the condition number of the matrix diagonalizing the generator of the dynamics.


\subsection{Summary of the main results}

We present four classes of main results, giving answers to each of the three Central Questions listed above and explaining the main features of the technical framework that we developed for diagonalizing the Carleman matrix.


\subsubsection{Conditions for the convergence of the Carleman  embedding}
\label{sec:summary_carleman}

We consider a range of nonlinear systems in this work, under different criteria. Table~\ref{tablesummary} offers a guide to the results.

\begin{table}[h!]
    \centering
    \begin{tabular}{|l|l|}
    \hline
    \textbf{Type of nonlinear system}  &  \textbf{Covered in} \\
    \hline
    Stable, autonomous, $\mu(F_1) <0$  & Section~\ref{sec:solutionnormbounds} \\
    Stable, autonomous, Lyapunov observables  & Theorem~\ref{thm:Carlemanstable} \\
    Conserved and oscillating linear observables, no driving, autonomous & Theorem~\ref{thm:Carleman_cons} \\
    Conserved and oscillating linear observables, with driving & Corollary~\ref{cor:Carleman_cons_extension} \\
    Conserved polynomial observables, no driving, autonomous & Theorem~\ref{thm:Carleman_poly_cons} \\
    Dissipative, non-autonomous & Section~\ref{sec:carleman_time_dependent} \\
    Nonresonant, Poincar\'{e} domain, autonomous & Theorem~\ref{thm:non-resonant_Poincare} \\
    Nonresonant, Siegel domain, $F_2$ constrained, autonomous & Theorem~\ref{thm:error_bound_nonresonant_siegel_domain_decompose_f2} \\
    Possibly resonant, oscillating $F_2$ & Theorem~\ref{thm:resonant_oscillating_F2} \\
    Subset of nonlinear oscillator systems & Section~\ref{subsubsec:coupled_nonlinear_oscillators} \\
    Subset of nonlinear Schr\"{o}dinger equations & Section~\ref{subsec:nonlinear_schrodinger_equations} \\
    \hline
    \end{tabular}
    \caption{Summary of new Carleman convergence criteria for nonlinear systems.}
    \label{tablesummary}
\end{table}

At the highest level, our answers to Central Question 1 are summarized in the right column of Table~\ref{tab:summaryofresults} for three broad categories of systems: \emph{stable}, \emph{conservative} and \emph{nonresonant}, defined in the left column of that table.

\begin{table}[h!]
\begin{adjustbox}{width=\textwidth}
\begin{tabular}{|>{\centering\arraybackslash}p{6.2cm}|>{\centering\arraybackslash}p{8cm}|}
    \hline 
    \textbf{Stable systems}  & \textbf{Convergence condition} 
    \tabularnewline
    \hline 
    \[
        \begin{array}{c}
            \alpha:= \max\limits_i \mathrm{Re}(\lambda_i) < 0 \\ \\
            \includegraphics[width=0.38\columnwidth]{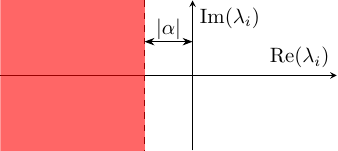}
        \end{array}
    \]
    & \[
        \begin{array}{c}
            R_P <1 \\
        \end{array}
        \]
        \[
            R_P:=  \frac{  \|F_{2}\|_P \|x(0)\|_P + \frac{\|F_0\|_P}{ \|x(0)\|_P}}{|\mu_P(F_{1})|} \]
        \[
        \begin{array}{c}\\
            P F_1 + F_1^\dag P < 0 \textrm{~for~any~$P>0$ } \\ 
            \textrm{(such $P$ is guaranteed to exist)}
        \end{array}
        \]
    \tabularnewline
    \hline 
    \hline 
    \textbf{Conservative \& oscillating systems} & \textbf{Convergence condition}
    \tabularnewline
    \hline 
    \[
        \begin{array}{c}
            \delta:= -\!\!\!\max\limits_{\substack{i\\\mathrm{Re}(\lambda_i)\neq 0}} \mathrm{Re}(\lambda_i )> 0\\ \\
            F_1\ket{\xi}=i\omega\ket{\xi} \quad\Rightarrow\quad \bra{\xi}F_2=0\\ \\
            \includegraphics[width=0.38\columnwidth]{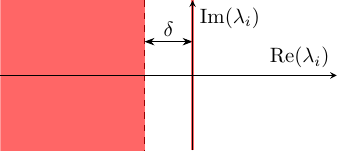}
        \end{array}
    \]
    & 
    \[
        \begin{array}{c} 
            F_1 \textrm{~diagonalized~by~}Q \\ \\
            \forall t:~\|Q^{-1}x(t)\| \leq \|\tilde{x}_{\mathrm{max}}\|\\ \\
            R_\delta <1 \\ 
        \end{array}
    \]
    \[
        R_\delta:=  \frac{2e\|Q^{-1}F_2Q^{\otimes 2}\| \sqrt{\|\tilde{x}_{\mathrm{max}}\|^2+\frac{\|Q^{-1} F_0\|}{\|Q^{-1}F_2Q^{\otimes 2}\|}}}{\delta(F_1)}
    \]
\tabularnewline
\hline 
\hline 
\textbf{Nonresonant systems} & \textbf{Convergence condition}
\tabularnewline
\hline 
    \[
        \begin{array}{c}
            \Delta \defeq \min\limits_i \!\!\! \inf\limits_{\substack{\alpha_j \in \myN_0 \\ \sum_{j=1}^N \alpha_j \ge 2}}  \frac{\abs{\lambda_i-\sum_{j=1}^N \alpha_j\lambda_j}}{\sum_{j=1}^N \alpha_j-1} > 0 \\ \\
            \includegraphics[width=0.38\columnwidth]{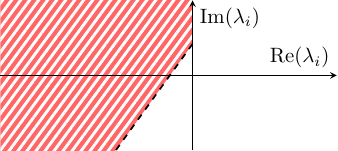}    
            \\
            \textrm{The~other~scenario~is~obtained~by}\\ 
            \textrm{reflecting~across~the~real~axis.}
        \end{array}
    \]
    & 
    \[
        \begin{array}{c} 
            F_0 = 0\\ \\
            F_1 \textrm{~diagonalized~by~}Q \\ \\
            Q^{-1}F_2Q^{\otimes 2} \textrm{~is~}s\textrm{-column-sparse} \\ \\
            \forall t:~\|Q^{-1}x(t)\| \leq \|\tilde{x}_{\mathrm{max}}\|\\ \\
            R_\Delta <1 \\ 
        \end{array}
    \]
    \[
        R_\Delta:=  \frac{8s \|Q^{-1}F_2Q^{\otimes 2}\| \|\tilde{x}_{\mathrm{max}}\|}{\Delta(F_1)}
    \]
\tabularnewline
\hline 
\end{tabular}
\end{adjustbox}
\caption{\textbf{A family of $R$-numbers.} Shown here are sufficient conditions under which the error incurred in solving the Carleman linear ODE problem truncated at power $k$ (see Eq.~\eqref{eq:linearODE}) rather than the original nonlinear system (see Eq.~\eqref{eq:nonlinear}) decreases exponentially fast with $k$. We use the following notation: $\lambda_i$ denote the eigenvalues of $F_1$, $\| \cdot \|_P$ indicates vector or matrix norms induced by the scalar product associated with $P$ (see Eq.~\eqref{eq:scalar_P}), and  $\mu_P(F_1)$ is a generalized log-norm of $F_1$ (see Eq.~\eqref{eq:Plognorm}), which is guaranteed to be negative as a consequence of the Lyapunov inequality. The solid red region indicates the allowed eigenvalues of $F_1$, whereas the striped red region indicates a  necessary condition for $F_1$ eigenvalues (i.e., eigenvalues must lie in this region, but not all choices of eigenvalues within that region are allowed by the condition $\Delta>0$).}
\label{tab:summaryofresults} 
\end{table}

Our work relies on both the existence of monotonic observables of the dynamics and whether the dynamics has resonances or not. We first extend purely dissipative dynamics, which has been the dominant focus of quantum Carleman algorithms to date, to the more general setting of \emph{stable} systems, where a range of Lyapunov observables exist (monotonically decreasing observables). We then extend our scope further to systems that have one or more \emph{conserved} quantities. In both cases, there is a key dependence on the spectral magnitudes of the linear term $F_1$.  Both cases also depend on the non-normality of $F_1$, which is a key property that has a non-trivial impact on the resulting quantum algorithms. However, there is another important property relevant to a nonlinear dynamical system near an equilibrium, namely whether the system is \emph{resonant} or \emph{nonresonant}. When a dynamical system is nonresonant, it is `topologically simple' and can in principle be smoothly deformed into a purely linear system. For such systems, the stability analysis of the full nonlinear system is simplified, as it reduces to linear stability analysis. In contrast, resonant systems cannot be smoothly deformed to a linear system and thus their stability analysis is much more involved. Given this overview, we now summarize our results. The detailed proofs can be found in Sections~\ref{sec:stable},~\ref{sec:conservative},~and~\ref{sec:nonresonant_systems}. 

\paragraph{Stable systems.}

Stable autonomous systems are all those for which the \emph{spectral abscissa} of $F_1$ satisfies
\begin{align}
    \label{eq:spectralabscissa}
    \alpha(F_1) := \max_i \mathrm{Re}(\lambda_i(F_1))<0,
\end{align}
where here and in what follows $\lambda_i(F_1)$ denotes the eigenvalues of $F_1$ (see Table~\ref{tab:summaryofresults}, first block). These are systems that are stable if we set $F_2 = 0$, and so, in this sense, they are linearly stable. Here we use the term `stable' to distinguish from some of the previous quantum algorithm literature, and use the term `dissipative' to indicate the restricted class of systems with negative log-norm~\cite{costa2023further}, i.e.,
\begin{align}
\label{eq:lognorm}
    \mu(F_1):= \frac{1}{2} \max_i \lambda_i\left( F_1+ F_1^\dag \right)<0,
\end{align} 
which is a stronger condition compared to \eqref{eq:spectralabscissa}, since $\mu(F_1) \geq \alpha(F_1)$.

Stable systems are the first class of nonlinear systems for which provably efficient quantum algorithms were introduced, based on a convergence proof for Carleman embedding~\cite{liu2021efficient}. However, as we discuss later there is non-trivial fine-print with previous analysis.
The available convergence results and a comparison with the contributions of this work are summarized in Table~\ref{tab:stableconvergence}.

\begin{table}
    \begin{adjustbox}{width=\textwidth,center}
    \begin{tabular}{|c|c|c|c|}
        \hline 
        \textbf{Class} & \textbf{Toy example} & \textbf{Convergence condition} & \textbf{Comments} 
        \tabularnewline
        \hline 
        \hline 
        $\begin{array}{c}
       \textrm{Negative log-norm} \\
    \textrm{No driving}
        \end{array}$ & $\dot{x} = - \mu x$ & $R_\mu <1$ \cite{forets2017explicit, costa2023further} & $\begin{array}{c} - \end{array}$
        \tabularnewline
        \hline 
       $\begin{array}{c}
       \textrm{Negative log-norm}\\
    F_1 \, \textrm{normal}
        \end{array}$ & $\dot{x} = nx^2 - \mu x + c $ & $R_\mu <1$ \cite{liu2021efficient} & Issue with proof
        \tabularnewline
        \hline 
       \textrm{Negative log-norm}  & $\dot{x} = nx^2 - \mu x + c $ & $R_\mu <1$ \cite{krovi2023improved} & Issue with proof
        \tabularnewline
        \hline 
        \textrm{Stable}  & $\ddot{x} = nx^2 -x - \mu \dot{x}$ &$\begin{array}{c}
        R_P < 1 \\ \textrm{(for any $P$ satisfying Eq.~\eqref{eq:lyapunovinequality})}
        \end{array}$  & $\begin{array}{c}
        \textrm{Corrects and/or extends}  \\
        \textrm{all previous cases}
        \end{array}$
            \tabularnewline
        \hline
    \end{tabular}
    \end{adjustbox}
    \caption{\textbf{Related works for stable systems.} The above table shows convergence conditions for stable systems, i.e, systems satisfying~\eqref{eq:spectralabscissa}. Stability ensures the existence of matrices $P>0$ such that the Lyapunov inequality~\eqref{eq:lyapunovinequality}  is satisfied. Note that for systems with negative log-norm, $P= I$ satisfies~\eqref{eq:lyapunovinequality}  and we have $R_P = R_\mu$. For $F_1$ diagonalizable by $Q$, $P=Q^\dag Q$ satisfies~\eqref{eq:lyapunovinequality} and, if we write the system in the basis where $F_1$ is diagonal, we have $R_P = R_\alpha$. $R_\alpha$, $R_\mu$ and $R_P$ are defined, respectively, in \eqref{eq:Ralpha}, \eqref{eq:Rmu}, \eqref{eq:RPfirst}.}
    \label{tab:stableconvergence}
\end{table}

In Ref.~\cite{liu2021efficient}, the authors claim that Eq.~\eqref{eq:nonlinear} is well-approximated by the linear Carleman ODE~\eqref{eq:linearODE} for stable systems if the following condition is satisfied:
    \begin{align}
        F_1 \textrm{ diagonalizable,}\quad  R_{\alpha} < 1 \quad\Rightarrow\quad \textrm{convergence},
\end{align}
where
\begin{equation}
    \label{eq:Ralpha}R_{\alpha}:= \frac{1}{-\alpha(F_1)} \left (  \|F_2\|\|x(0)\| + \frac{\|F_0\|}{ \|x(0)\|} \right ).
\end{equation}
However, as we discuss in more detail in Section~\ref{sec:stable}, the central proof of Ref.~\cite{liu2021efficient} implicitly assumes that $F_1$ is normal (which excludes systems where non-normality plays a key role in e.g. transient growth effects). 
For this special class of dynamics, $\alpha(F_1)=\mu(F_1)$, and so this restricts us to systems with negative log-norm (which, for example, excludes linear oscillator systems). Therefore, the condition $R_\alpha<1$ can be rewritten as $R_\mu<1$, where
\begin{align}
    \label{eq:Rmu}
        R_\mu:= \frac{1}{-\mu(F_1)} \left (  \|F_2\|\|x(0)\| + \frac{\|F_0\|}{ \|x(0)\|} \right ).
    \end{align} 
Introducing these extra assumptions, the convergence claim in Ref.~\cite{liu2021efficient} is restricted to
 \begin{align}
\label{eq:Carlemanconvergencelognormlimited}
     F_1 \textrm{ normal,}\quad \mu(F_1)<0, \quad R_\mu < 1 \quad\Rightarrow \quad\textrm{convergence}.
 \end{align}
However, in this restricted setting
we also flag an issue with the proof of Lemma~2 in Ref.~\cite{liu2021efficient} (see Appendix~\ref{app:issueliu} for details). Without the convergence proof, the efficiency claim of the quantum algorithm in Ref.~\cite{liu2021efficient} is upended.
    
Another proof of the convergence of Carleman embedding for stable systems with negative log-norm was given in Ref.~\cite{krovi2023improved}, in a form that extends Eq.~\eqref{eq:Carlemanconvergencelognormlimited} to
 \begin{align}
\label{eq:Carlemanconvergencelognorm}
     \mu(F_1) <0,\quad R_\mu < 1\quad \Rightarrow \quad\textrm{convergence}.
 \end{align}
Also here, we flag an issue with the convergence proof, which we describe in Appendix~\ref{app:issuekrovi}. We also show, by constructing an explicit counterexample, that the main technical claim upon which the convergence proofs in Refs.~\cite{liu2021efficient, krovi2023improved} are based does not hold, meaning that the issue cannot be trivially corrected. Our counterexample has $F_0 \neq 0$, and in fact a proof of convergence exists for stable systems with negative log-norm and no driving~\cite{costa2023further}. 

This leaves us only with the proofs from Refs.~\cite{forets2017explicit, costa2023further}, ensuring convergence for dissipative systems in the form\footnote{Note that in this work we focus on the case of quadratic nonlinearities, whereas Ref.~\cite{costa2023further} more broadly considers degree $M$ nonlinearity of the form $\dot{x} = F_1 x + F_M x^{\otimes M}$.}
 \begin{align}
    \label{eq:convergence_lognorm_noF0}
     F_0 =0,\quad \mu(F_1) <0,\quad R_\mu < 1 \quad\Rightarrow\quad \textrm{convergence}.
 \end{align}

In this work, we correct these previous results and build beyond them. Our main result is that for every second-order Lyapunov function $f_P(x)$ for the linearized dynamical system we have an associated quantity $R_P$ for which $R_P <1$ implies Carleman convergence of the full nonlinear dynamical system. In more detail, for every stable autonomous system, one can prove that there exist positive-definite matrices $P>0$\footnote{Recall that $P>0$ means that $P$ is Hermitian and its eigenvalues are positive.} satisfying the Lyapunov inequality~\cite{plischke2005transient}
 \begin{align}
\label{eq:lyapunovinequality}
    P F_1 + F_1^\dag P < 0.
\end{align}
Each matrix $P$ is associated to a Lyapunov function of the linear system $\dot{x} = F_1 x$, meaning that $f_P(x) = x^\dag P x$ is non-increasing along the trajectories of that linear system. We leverage such stability analysis of the linearized system to prove convergence of the Carleman embedding of the nonlinear system, for small enough nonlinearities. In particular, we prove that every such $P$ provides a sufficient condition for convergence:
 \begin{align}
    R_P < 1\quad\Rightarrow\quad \textrm{convergence},
 \end{align}
where
\begin{align}
\label{eq:RPfirst}
     R_P :=  \frac{1}{-\mu_P(F_{1})} \left (  \|F_{2}\|_P \|x(0)\|_P + \frac{\|F_0\|_P}{ \|x(0)\|_P} \right ),
\end{align}
and the subscript $P$ for norms and the log-norm means that the Euclidean inner product in their definitions is replaced by $\langle x, y \rangle_P := x^\dag P y$, see Section~\ref{sec:stable} for details. For systems with negative log-norm, $\mu(F_1)<0$, we have that $P=I$ satisfies Eq.~\eqref{eq:lyapunovinequality} and $R_{P=I} = R_\mu$, so we recover Eq.~\eqref{eq:Carlemanconvergencelognorm} as a sufficient condition, while circumventing the issue with the proof presented in Ref.~\cite{krovi2023improved}. For $F_1$ diagonalizable via $Q$, $P= Q^\dag Q$ satisfies Eq.~\eqref{eq:lyapunovinequality} and, writing all operators in the basis where $F_1$ is diagonal, $R_{P= Q^\dag Q} = R_\alpha$. This essentially recovers the claim from Eq.~\eqref{eq:Ralpha}, while circumventing the issue with the proof in Ref.~\cite{liu2021efficient}. Crucially, beyond these two special cases, we prove that every~$P$ satisfying Eq.~\eqref{eq:lyapunovinequality} gives a sufficient convergence condition $R_P<1$, extending the set of stable systems for which Carleman linearization converges.

For a simple example, consider a single non-linear damped `oscillator',
\begin{align}
\label{eq:simpleoscillator}
    \ddot{q} = -q -r \dot{q} - n q^2,  \quad r>0,
\end{align}
which we can rewrite in standard form from Eq.~\eqref{eq:nonlinear} by introducing $v = \dot{q}$ and setting $x = (q, \dot{q})$. This system is linearly stable, $\alpha(F_1)<0$, but its logarithmic norm is zero, $\mu(F_1)=0$. Hence, the sufficient condition result from Eq.~\eqref{eq:convergence_lognorm_noF0} does not guarantee convergence for any nontrivial choice of parameters, see the left panel of Fig.~\ref{fig:oscillatorexample}. Our general result $R_P<1$ returns the sufficient condition $R_{\alpha} <1$ (if we write the equation in the basis where $F_1$ is diagonal), see the center panel of Fig.~\ref{fig:oscillatorexample}. Finally, more broadly, our convergence result gives an infinite number of sufficient conditions $R_P<1$, ensuring convergence at least for all parameters shown in right panel of Fig.~\ref{fig:oscillatorexample}.

\begin{figure}[t]
    \centering
    \includegraphics[width=0.3\linewidth]{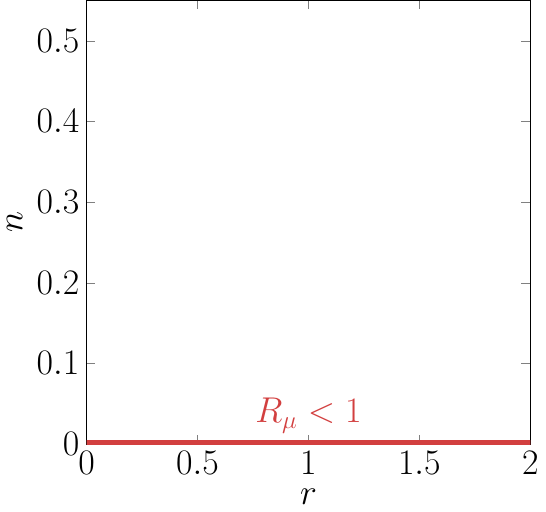}\hspace{0.5cm}
    \includegraphics[width=0.3\linewidth]{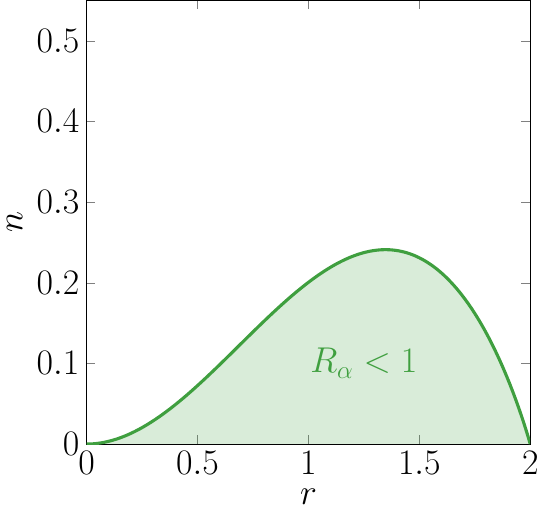}\hspace{0.5cm}
    \includegraphics[width=0.3\linewidth]{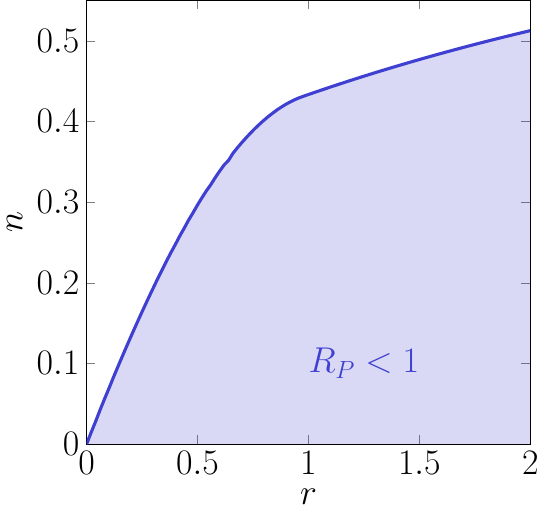}
    \caption{\textbf{A continuum of $R_P<1$ conditions.} Shown here is the region of guaranteed Carleman convergence for the toy nonlinear oscillator problem in Eq.~\eqref{eq:simpleoscillator}, with $q(0) = v(0) = 1/2$, as a function of the dissipation parameter $r \geq 0$ and the nonlinear parameter $n \geq 0$, for the intermediate damping regime $r\leq 2$. Left: convergence regions with the previously known~\cite{forets2017explicit, costa2023further} sufficient condition $R_\mu <1$ (no convergence guarantees for any $n>0$, since $R_\mu = \infty$). Middle: convergence region with our corrected version of Ref.~\cite{liu2021efficient}, $R_\alpha <1$. Right: convergence region using a numerically optimized $R$-number sufficient conditions $\{R_P<1\}$ from Theorem~\ref{thm:Carlemanstable}. 
    }
    \label{fig:oscillatorexample}
\end{figure}


\paragraph{Conservative and oscillating systems.}

    Although Carleman convergence results for stable systems open a path for developing quantum algorithms for simulating nonlinear dynamics, many physically relevant systems are not captured by this formalism. The reason is that systems with conserved quantities are \emph{marginally stable}, i.e., 
    \begin{align}
        \alpha(F_1)=0,
    \end{align}
    which also implies $\mu(F_1)\geq 0$. So for such systems the existing results -- including our extended results -- do not guarantee convergence, even in the presence of arbitrarily weak nonlinearities. We remedy this situation by proving non-trivial Carleman convergence conditions for systems with general polynomial conserved quantities~$\Q(x)$. These are polynomial functions $\Q$ of the system state $x$ satisfying $d\Q/dt=0$, examples of which include conserved charge or fluid density (when $x$ describes the distribution of charge or fluid), and conserved linear momentum (when $x$ describes such distributions in phase space). As we explain in Section~\ref{sec:conservative}, the existence of such conserved quantities at the linear level implies the existence of zero eigenvalues of~$F_1$. We then focus on systems for which all zero eigenvalues of~$F_1$ arise from conservation laws (formally that means that whenever $F_1\xi=0$ we have $\xi^\dagger F_2=0$, see Section~\ref{sec:conservative_setting} for details), and the remaining eigenvalues have negative real parts with absolute values upper bounded by $\delta(F_1)$ (see Table~\ref{tab:summaryofresults}, second row). In such cases, we prove that the nonlinear system from Eq.~\eqref{eq:nonlinear} is well-approximated by the linear Carleman ODE equation~\eqref{eq:linearODE} under the following conditions:
        \begin{equation}
            F_1\text{~diagonalized~by~} Q,\quad R_\delta<1 \quad\Rightarrow\quad \textrm{convergence},
        \end{equation}
        where
        \begin{equation}
            R_\delta:=  \frac{2e\|Q^{-1}F_2Q^{\otimes 2}\| }{\delta(F_1)}\sqrt{\|\tilde{x}_{\mathrm{max}}\|^2+\frac{\|Q^{-1} F_0\|}{\|Q^{-1}F_2Q^{\otimes 2}\|}},
        \end{equation}
    and $\|\tilde{x}_{\mx}\|$ is an upper bound on the solution norm in the eigenbasis of $F_1$ that holds for all $t \ge 0$:
    \begin{equation}
        \forall t \ge 0: \|Q^{-1} x(t)\|\leq \|\tilde{x}_{\mx}\|.
    \end{equation}
     One can interpret the condition $R_\delta<1$ as the requirement for the relative strength of nonlinearity to dissipation of the non-conserved degrees of freedom to be weak enough. Finally, we prove that any higher order polynomial observable that is conserved by the dynamics can be reformulated as a conserved linear observable for an embedding higher dimensional system. This implies the analysis and convergence criterion for linear systems can be applied to the more general polynomial case.
     
     While we are not aware of any previous results directly dealing with Carleman convergence for marginally stable systems that we could compare our result with, in certain cases and after some manipulations it is possible to use the stable convergence results for systems with linear conserved quantities. More precisely, assuming we know the form of the conserved quantities $\Q$, it is possible to reduce the system dimension by removing the number of degrees of freedom equal to the number of conserved quantities. Such dimensional reduction process, the example of which we present in Section~\ref{sec:cons_discussion}, modifies matrices $F_0$, $F_1$ and $F_2$. The modified system may become stable and then one can use the condition $R_P<1$ to prove Carleman convergence. The crucial point we want to emphasize is that neither $R_P<1$ for such a reduced system implies $R_\delta<1$ for the original system, nor the other way round. In this sense, our conservative convergence results and the modification of dissipative convergence results are \emph{complementary criteria} for systems with linear conserved quantities. This can be illustrated for a two-dimensional toy problem (see Section~\ref{sec:cons_discussion} for details) described by
     \begin{equation}
     \label{eq:toy_cons}
         \dot{x}_1(t) = 0,\qquad \dot{x}_2(t) = -x_2(t) - b x_2(t)^2 + a x_1(t)^2,
     \end{equation}
     where $x_1(t)$ is a conserved quantity, and the reduced problem is the evolution of $x_2(t)$ with $x_1(t)$ replaced by a constant $x_1(0)$. In Fig.~\ref{fig:R_regions}, we present convergence regions for which the considered $R$-numbers are smaller than 1 for an exemplary choice of parameters, illustrating the complementarity of the results.
     
    \begin{figure}[t]
        \centering
        \includegraphics[width=0.35\linewidth]{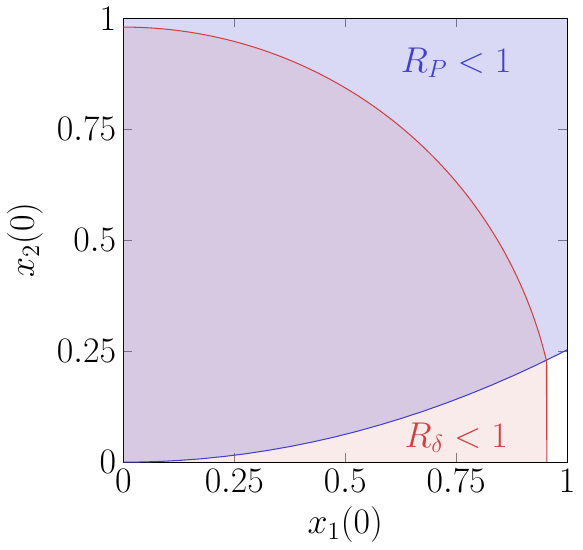}\hspace{1cm}
        \includegraphics[width=0.35\linewidth]{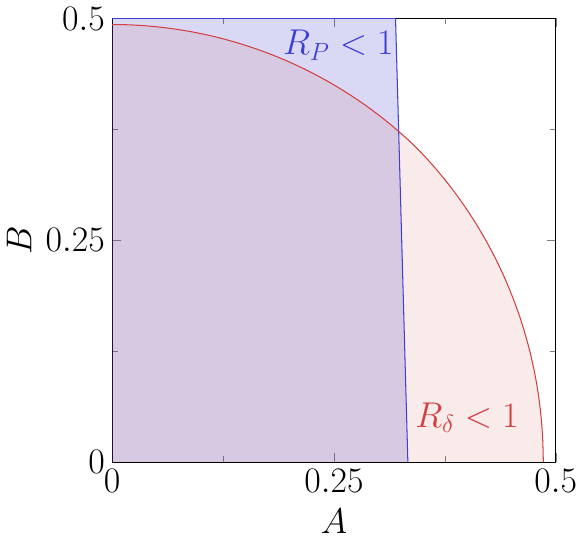}
        \caption{\textbf{Inequivalence of stability and conservative criteria.} Shown here are regions of guaranteed Carleman convergence via the dissipative result for the reduced system, $R_P<1$, and conservative result for the original system, $R_\delta<1$, for the toy nonlinear problem from Eq.~\eqref{eq:toy_cons}. Left: Convergence regions as a function of initial conditions $x_1(0)$ and $x_2(0)$ for parameters $a=1/5$, $b=1/20$. Right: Convergence regions as a function of parameters $a$ and $b$ for initial conditions $x_1(0)=1/2$, $x_2(0)=1/12$.}
        \label{fig:R_regions}
    \end{figure}

     Stability analysis for non-autonomous systems is more complex, and so Carleman convergence analysis is in turn affected. In Section~\ref{sec:conservative}, we also explain how our results on Carleman convergence for systems with linear conserved quantities can be generalized to include systems with linear oscillating quantities, i.e., quantities $\Q$ that evolve according to $\Q(t)=\Q(0)e^{\i \omega t}$ (with $\omega=0$ reproducing the conservative result). This then extends the allowed set of eigenvalues of $F_1$ not only to zero, but to all purely imaginary numbers $\i\omega$, as long as whenever $F_1\xi =\i\omega \xi$, we have $\xi^\dagger F_2=0$ (which is a formal definition of a linear oscillating quantity described by $\xi$). Using this together with Fourier methods we prove Carleman convergence for \emph{non-autonomous} nonlinear systems that have time-dependent linear matrix $F_1(t)$ and driving $F_0(t)$.


\paragraph{Nonresonant systems.}

Most prior work on quantum algorithms for nonlinear dynamics has focused on dissipative systems, where all eigenvalues of $F_1$ have negative real parts. While such systems are prevalent across science and engineering, they fail to capture the oscillatory behavior intrinsic to many physical models. These dynamics -- often associated with purely imaginary spectra -- lie beyond the scope of results such as Ref.~\cite{liu2021efficient} that rely on dissipation. Recent work~\cite{wu2024quantum} attempted to address the challenge of simulating non-dissipative, nonlinear dynamics on a quantum computer. The work of Ref.~\cite{wu2024quantum} highlighted the importance of the spectrum of $F_1$, and then provided a convergence error analysis for non-dissipative systems to obtain a convergence criterion. However, as we describe later, this analysis contains errors in key parts, and the obstacles involved are non-trivial to circumvent\footnote{The work~\cite{wu2024quantum} seeks to establish the convergence of Carleman embedding for nonresonant systems. However, the proof contains fundamental errors that resist straightforward correction. 
For instance, Corollary 2.10.1 is incorrect -- even one-dimensional counterexamples exist, as demonstrated in Appendix~\ref{app:issuenonresonant}. Therefore, the subsequent arguments that rely on it, such as Lemma 2.11, are also invalid. In addition, the overall proof strategy overlooks a critical structural aspect: it does not account for the inverse of the linear operator used to diagonalize the Carleman matrix, which we show to be essential for a correct convergence argument. See Appendix~\ref{app:issuenonresonant} for further details.}. Therefore, the question of efficient quantum algorithms for non-dissipative systems is re-opened as a problem. 

To address this gap and broaden the applicability of Carleman convergence, we study nonresonant systems whose spectra are not confined to the dissipative regime, thereby encompassing models that exhibit persistent oscillations driven by purely imaginary eigenvalues. Informally, a system is said to be nonresonant if no eigenvalue of $F_1$ equals a nontrivial integer combination of the others. Under reasonable assumptions, we establish the convergence of Carleman linearization for such systems, significantly extending the reach of quantum algorithms to a broader class of nonlinear dynamical systems -- including those characterized by sustained, non-decaying motion, such as systems governed by nonlinear Schr\"{o}dinger equations.

Formally, suppose $F_1=Q \Lambda Q^{-1}$, where $\Lambda=\diag{\lambda_1, \lambda_2, \dots, \lambda_N}$. We say that the system is \emph{nonresonant} if for all $i \in [N]$, 
\begin{align}
\lambda_i \neq \sum_{j=1}^N \alpha_j \lambda_j,~~~\forall \alpha_j \in \myN_0,~~~\sum_{j=1}^N \alpha_j \ge 2.    
\end{align}
Otherwise, the system is said to be \emph{resonant}. 

We analyze in detail the explicit diagonalization of Carleman matrices, and derive upper bounds on the approximation errors of truncated Carleman linearizations for nonresonant systems. These bounds depend critically on  the geometric arrangement of the eigenvalues of $F_1$. Specifically, the spectrum $(\lambda_1, \lambda_2, \dots, \lambda_N)$ is said to lie in
the \emph{Poincar\'{e} domain} if the origin of the complex plane is not contained in the convex hull of $\lambda_1, \lambda_2, \dots, \lambda_N$; otherwise, it lies in the \emph{Siegel domain}. Interestingly, under the condition $\realpart{\lambda_j}\le 0$ for each $j$, the spectrum $(\lambda_1, \lambda_2, \dots, \lambda_N)$ belongs to the Poincar\'{e} domain if and only if the eigenvalues do not simultaneously span both the positive and negative segments of the imaginary axis, as illustrated in Table~\ref{tab:summaryofresults}.

When the spectrum of $F_1$ is nonresonant and lies in the Poincar\'{e} domain, we define a positive, normalized no-resonance gap $\Delta(F_1)$ as:
\begin{align}
\Delta(F_1)\defeq \min_{i \in [N]} \inf_{\substack{\alpha_j \in \myN_0 \\ \sum_{j=1}^N \alpha_j \ge 2}} \frac{\abs{\lambda_i-\sum_{j=1}^N \alpha_j\lambda_j}}{\sum_{j=1}^N \alpha_j-1} > 0. 
\label{eq:def_noramlized_no_resonance_gap_intro}
\end{align}
This gap is strictly positive only in the Poincar\'{e} domain and plays a central role in our error analysis. In this setting, we prove that the truncation error of the Carleman scheme converges to zero as the truncation order increases, provided the following condition holds:   
\begin{align}
F_0=0, \quad F_1\textrm{~diagonalized~by~} Q,\quad R_\Delta<1 \quad\Rightarrow\quad \textrm{convergence},    
\end{align}
where
\begin{align}
    R_\Delta:=  \frac{8s \|Q^{-1}F_2Q^{\otimes 2}\| \|\tilde{x}_{\mathrm{max}}\|}{\Delta(F_1)}
\end{align}
with $s$ denoting the column sparsity of $Q^{-1}F_2Q^{\otimes 2}$ and $\norm{\tilde{x}_{\rm max}}$ being an upper bound on the norm of $\tilde{x}(t) \defeq Q^{-1}x(t)$ for all $t \ge 0$:
\begin{align}
    \forall t \ge 0: \norm{Q^{-1} x(t)}\leq \norm{\tilde{x}_{\mx}}.    
\end{align}

In contrast, when the spectrum of $F_1$ is nonresonant but lies in the Siegel domain, the normalized no-resonance gap defined in Eq.~\eqref{eq:def_noramlized_no_resonance_gap_intro} becomes zero. In this case, an alternative definition of $\Delta(F_1)$ is required. Based on this revised definition, we can still  derive an upper bound on the truncation error of the Carleman embedding in terms of $\Delta(F_1)$ and other system parameters.
However, unlike in the case where the spectrum lies in the  Poincar\'{e} domain, this bound does not necessarily vanish as the truncation order increases. Nevertheless, if the quadratic term $F_2$ satisfies additional structural conditions or becomes time-dependent, convergence may still be recovered. A complete resolution of the convergence behavior of the Carleman scheme for nonresonant systems in the Siegel domain remains an open problem for future investigation.

 
\subsubsection{Quantum algorithms for nonlinear systems}

Building upon the error analysis of the Carleman scheme for stable, conservative, and nonresonant systems, we develop quantum algorithms for preparing quantum states that encode the solutions of these differential equations in Section~\ref{sec:quantum_algorithms_nonlinear_systems}. These algorithms are constructed by applying the linear ODE solvers of Ref.~\cite{jennings2023cost, berry2024quantum} to the Carleman-embedded ODE system, and, for dissipative and conservative systems, by incorporating a newly introduced quantum Lyapunov transformation.

While the Carleman embedding provides a pathway to linearize the dynamics, realizing efficient quantum algorithms requires addressing several technical challenges. In particular, Ref.~\cite{jennings2023cost} requires \emph{a priori} bounds on the norms of the exponentials of the Carleman matrices, which are essential for controlling the condition numbers of the linear systems that embed the Carleman ODEs. In Section~\ref{subsec:quantum_algorithm_nonresonant_systems}, we establish such bounds for nonresonant systems
 by drawing on the new structural insights into the diagonalization of the Carleman matrices developed in Section \ref{subsec:carleman_matrix_diagonalization}. Moreover, care must be taken to ensure that the quantum state output by the algorithm accurately reflects the solution of the original nonlinear system, which involves a careful rescaling step. The efficiency of these algorithms depends critically on the condition number of the matrix $Q$ that diagonalizes the linear term $F_1$, or put another way: on the non-normality of $F_1$. Under mild assumptions, the runtime scales polynomially in the evolution time $T$ and inverse accuracy $1/\epsilon$ when $Q$ has condition number $\myO{1}$. This condition is satisfied, for instance, when $F_1$ is a normal matrix, in which case $Q$ is unitary. Consequently, we obtain efficient quantum algorithms for preparing solution states of a broad class of nonlinear Schr\"{o}dinger equations. When $F_1$ is far from normal -- that is, when $Q$ has a large condition number -- our algorithms may exhibit unfavorable scaling in $T$ and $1/\epsilon$. 

Even in the favorable case where $Q$ has condition number $\myO{1}$, these algorithms may still incur a polynomial scaling in $T/\epsilon$. For dissipative systems, in Section~\ref{subsec:quantum_algorithm_stable_conservative_systems}, we address this limitation by introducing a quantum Lyapunov transform. This requires access to an oracle block-encoding the Lyapunov matrix $P$ introduced earlier, and performs a change of coordinates in which the system becomes purely dissipative and the Carleman matrix remains well-conditioned. After a filtering stage, where the $j=1$ component of the Carleman vector is isolated, the problem is mapped back to the original coordinates. For history-state preparation, this transformation reduces the polynomial dependence on $T/\epsilon$ in the complexity to $\log(T/\epsilon)$ when $\kappa_{Q}=\myO{1}$ (where $P = Q^\dag Q$). However, even this approach does not entirely eliminate sensitivity to the condition number of $Q$, and hence the non-normality of $F_1$. Whether this dependence on $\kappa_Q$ is intrinsic to the problem or merely an artifact of our approach remains an open question, which we leave for future investigation. For conservative systems, much of the same holds if we take $P = Q^\dag Q$ with $Q$ being the matrix diagonalizing $F_1$ (assuming $F_1$ is diagonalizable).

These developments collectively enable a rigorous quantum treatment of nonlinear dynamics and lay the groundwork for our BQP-completeness results in Section~\ref{sec:exponential_quantum_advantage}.


\subsubsection{Exponential quantum advantage for the simulation of nonlinear systems}

Beyond convergence guarantees and algorithmic design, it is crucial to understand the computational complexity of simulating nonlinear dynamical systems -- particularly in identifying regimes where exponential quantum advantage can be rigorously established. To this end, we first analyze the complexity of a class of nonlinear oscillator problems that fit within the formalism we have developed. These dynamics can be embedded into a nonlinear Schr\"{o}dinger equation and solved efficiently using the algorithm described above. Our result generalizes those of Refs.~\cite{babbush2023exponential} and~\cite{krovi2024quantum}, which consider only linear oscillators, and it inherits BQP-completeness in the regime of sufficiently weak nonlinearities.

We further show that BQP-hardness persists even for strictly non-zero nonlinearities. By analyzing small nonlinear perturbations of the BQP-hard linear system from Ref.~\cite{babbush2023exponential}, we demonstrate that the solution remains sufficiently close to preserve computational hardness. This establishes BQP-completeness for a class of inherently nonlinear problems, which can still be efficiently solved using our Carleman-based approach in combination with existing quantum algorithms for linear differential equations. These results suggest that quantum computers can provide exponential quantum advantage that extends into the regime of nonlinear classical dynamics. It remains an open question as to whether one can find such advantage in high-impact problems of interest.


\subsubsection{General structure of nonresonant Carleman matrices}
\label{subsubsec:general_structure_carleman_matrix}
As part of our Carleman error analysis, we derive a range of results on the structure of the similarity transformation that diagonalizes the associated Carleman matrices, which may be of independent interest. In particular, we constructively prove that when the system is nonresonant and undriven, the matrix $A$ arising from the Carleman embedding is diagonalizable via a similarity transformation. That is, there exists an invertible matrix $V$ such that $A = V D V^{-1}$, where $D$ is diagonal with entries given by integer combinations of the eigenvalues of $F_1$. The matrices $V$ and $V^{-1}$ are upper triangular and can be partitioned into blocks in a way that respects the Carleman embedding structure: $V = (V_{i,j})_{1 \le i \le j \le k}$, $V^{-1} = ((V^{-1})_{i,j})_{1 \le i \le j \le k}$, where $V_{i,j}, (V^{-1})_{i,j} \in \mathbb{C}^{N^i \times N^j}$. 

We propose a method for explicitly constructing these blocks. Specifically, each $V_{i,j}$ is expressed as a sum of linear operators, each associated with a distinct binary forest consisting of $i$ trees and a total of $j$ leaves. These forests capture the combinatorial structure of the Carleman embedding and determine how contributions to $V_{i,j}$ are assembled. A similar representation applies to the blocks $(V^{-1})_{i,j}$ as well. As an example, in the case $(i,j)=(2,4)$ we have that
\begingroup
\setlength{\jot}{12pt} 
\tikzset{
	level distance=16pt,
	every tree node/.style = {align=center, anchor=north,nodecolor},
	edge from parent/.style = {draw,thick,edgecolor},
	edge from parent path = {
		([yshift=4.5pt]\tikzparentnode.south) -- 
		([yshift=-3.5pt]\tikzchildnode.north)
	}
}
\begin{align}
    \lrb{V^{-1}}_{2,4}=&
	g\left(
	\raisebox{-0.6cm}{
		\begin{tikzpicture}[scale=1]
			\Tree [.$\bt$ 
			[.$\bt$ [.$\bt$ ] [.$\bt$ ] ] 
			[.$\bt$ ] 
			];		
			\node at (0,0) {$a$};
			\node at (-0.6,-0.75) {$b$};
			\node at (0.6,-0.75) {$c$};
			\node at (-0.9,-1.35) {$d$};
			\node at (0.15,-1.35) {$e$};
		\end{tikzpicture}
		\hspace{\treeTree}
		\begin{tikzpicture}[scale=1]
			\Tree [.$\bt$  ];
			\node at (0,0.05) {$f$};
			\node at (0,-0.2) {\phantom{$f$}};
		\end{tikzpicture}
	}
	\right)
	+
	g\left(
	\raisebox{-0.6cm}{
		\begin{tikzpicture}[scale=1]
			\Tree [.$\bt$ 
			[.$\bt$ ] 
			[.$\bt$ [.$\bt$ ] [.$\bt$ ] ] 
			];
			\node at (0,0) {$a$};
			\node at (-0.6,-0.75) {$b$};
			\node at (0.6,-0.75) {$c$};
			\node at (-0.15,-1.35) {$d$};
			\node at (0.9,-1.35) {$e$};
		\end{tikzpicture}
		\hspace{\treeTree}
		\begin{tikzpicture}[scale=1]
			\Tree [.$\bt$  ];
			\node at (0,0.05) {$f$};
			\node at (0,-0.2) {\phantom{$f$}};
		\end{tikzpicture}
	}
	\right)
	+
	g\left(
	\raisebox{-0.6cm}{
		\begin{tikzpicture}[scale=1]
			\Tree [.$\bt$  ];
			\node at (0,0.05) {$d$};
			\node at (0,-0.2) {\phantom{$f$}};
		\end{tikzpicture}
		\hspace{\treeTree}
		\begin{tikzpicture}[scale=1]
			\Tree [.$\bt$ 
			[.$\bt$ [.$\bt$ ] [.$\bt$ ] ] 
			[.$\bt$ ] ];
			\node at (0,0) {$a$};
			\node at (-0.6,-0.75) {$b$};
			\node at (0.6,-0.75) {$c$};
			\node at (-0.9,-1.35) {$e$};
			\node at (0.15,-1.35) {$f$};
		\end{tikzpicture}
	}
	\right)
	\nonumber
	\\
	&
	+
	g\left(
	\raisebox{-0.6cm}{
		\begin{tikzpicture}[scale=1]
			\Tree [.$\bt$  ];
			\node at (0,0.05) {$d$};
			\node at (0,-0.2) {\phantom{$f$}};
		\end{tikzpicture}
		\hspace{\treeTree}
		\begin{tikzpicture}[scale=1]
			\Tree [.$\bt$ 
			[.$\bt$ ] 
			[.$\bt$ [.$\bt$ ] [.$\bt$ ] ] 
			];
			\node at (0,0) {$a$};
			\node at (-0.6,-0.75) {$b$};
			\node at (0.6,-0.75) {$c$};
			\node at (-0.15,-1.35) {$e$};
			\node at (0.9,-1.35) {$f$};
		\end{tikzpicture}
	}
	\right)
	+
	g\left(
	\raisebox{-0.35cm}{
		\begin{tikzpicture}[scale=1]
			\Tree [.$\bt$ [.$\bt$ ] [.$\bt$ ] ];
			\node at (0,0.05) {$a$};
			\node at (-0.5,-0.75) {$c$};
			\node at (0.5,-0.75) {$d$};
		\end{tikzpicture}
		\hspace{\treeTree}
		\begin{tikzpicture}[scale=1]
			\Tree [.$\bullet$ [.$\bt$ ] [.$\bt$ ] ];
			\node at (0,0.05) {$b$};
			\node at (-0.5,-0.75) {$e$};
			\node at (0.5,-0.75) {$f$};
		\end{tikzpicture}
	}
	\right),
\end{align}
\endgroup
where $g$ is a function that encodes both the linear $F_1$ and nonlinear $F_2$ terms in the defining ODE. This binary forest formalism enables us to derive upper bounds on the norms of $V_{i,j}$ and $(V^{-1})_{i,j}$, which play a crucial role in bounding the truncation error of Carleman scheme in the nonresonant regime, as well as the complexities of the quantum algorithms in Section \ref{subsec:quantum_algorithm_nonresonant_systems}. Furthermore, these results offer an alternative perspective on the Poincar\'{e}-Dulac normal forms of differential equations and may find applications beyond the present work.

\newpage

\section{Background material and notation }
\label{sec:carlemanlinearizationscheme}


\subsection{The rescaling degree of freedom}
\label{sec:rescaling}

A key aspect of nonlinear equations that does not arise for linear equations is that terms transform nontrivially under the change of units. In particular, we may multiply both sides of Eq.~\eqref{eq:nonlinear} by a positive constant $\gamma >0$ to obtain
\begin{equation}
    \gamma \dot{x}(t) = \gamma F_0 + F_1 \gamma x(t)+ \frac{F_2}{\gamma} \left (\gamma x(t) \right )^{\otimes 2}.
\end{equation}
Thus, the considered system from Eq.~\eqref{eq:nonlinear} is physically equivalent to the following ODE system:
\begin{align}
\label{eq:gammarescaledsystem} 
\dot{\bar{x}}(t) &= \bar{F}_0+ \bar{F}_1 \bar{x}(t)+ \bar{F}_2 \bar{x}(t)^{\otimes 2},
\end{align}
if we have the relations
\begin{align}
    \bar{F}_0  = \gamma F_0, \qquad
    \bar{F}_1 = F_1, \qquad
    \bar{F}_2 =  \frac{1}{\gamma}F_2, \qquad
    \bar{x}(t) = \gamma x(t).
    \label{eq:rescalingrelations}
\end{align}
At this stage, $\gamma$ is a free parameter that does not change the physics or the exactness of the nonlinear ODE. However, when we come to approximating this nonlinear system, the choice of $\gamma$ does make a difference. From now on we shall drop the bars, but keep in mind that all the relevant quantities are defined modulo the transformation in Eq.~\eqref{eq:rescalingrelations}.


\subsection{Infinite and truncated Carleman systems}\label{sec:Core-Carleman}

In this work, we make use of the Carleman linear representation scheme. The Carleman linearization scheme is a classical method for transforming a finite-dimensional nonlinear dynamical system into an infinite-dimensional linear system. It operates by expanding the nonlinear dynamics in terms of monomials of the system's variables and then expressing the evolution of these monomials as a linear system of ordinary differential equations. Although the resulting system is infinite, truncating the expansion at a finite order yields a high-dimensional linear ODE that approximates the original nonlinear dynamics. This representation provides a bridge between nonlinear and linear dynamics, making it especially well-suited for quantum algorithms that rely on linear evolution.

Let $\gamma >0$ be fixed, but arbitrary. The linear representation of the nonlinear dynamics is obtained by considering tensor powers $x^{\ot j}$ as a basis of linearly transforming variables. For any integer $j=2,3\dots \infty$, we have that
\begin{equation}
    \frac{d}{dt}{x}^{\ot j} = \dot{x} \otimes x^{\otimes (j-1)} + x\otimes \dot{x}\otimes x^{\otimes (j-2)} + \cdots +x^{\otimes (j-1)}\otimes \dot{x}.
\end{equation}
However, to simplify the presentation, we will use the following ``shift" notation, defined for any matrices $X,Z$ of the same output dimension by
\begin{equation}
   Z\otimes X^{\otimes (j-1)} + \mbox{ shifts } :=  Z\otimes X^{\otimes (j-1)} + X \otimes Z \otimes X^{\otimes (j-2)} + X^{\otimes 2} \otimes Z \otimes X^{\otimes (j-3)} + \cdots + X^{\otimes (j-1)} \otimes Z,
\end{equation}
for any integer $j\ge 2$. In other words, ``$\mbox{shifts}$" denotes other terms that are obtained by shifting the system the $Z$ matrix acts on. We now have that
\begin{equation}
\begin{aligned}
    \frac{d}{dt}{x}^{\ot j} &= \dot{x} \otimes x^{\otimes (j-1)} + \mbox{shifts}  \\
    &= (F_0+ F_1 x + F_2 x^{\otimes 2} ) \otimes x^{\otimes (j-1)} + \mbox{shifts}  \\
        &= (F_0 \otimes I^{\otimes (j-1)}) x^{\otimes (j-1)} + (F_1 \otimes I^{\otimes (j-1)} )x^{\otimes j} + (F_2 \otimes I^{\otimes (j-1)}) x^{\otimes (j+1)} + \mbox{shifts}  \\
        &= A_{j,j-1} x^{\otimes (j-1)} + A_{j,j}x^{\otimes j}+ A_{j,j+1} x^{\otimes (j+1)},
\end{aligned}
\end{equation}
where $I$ is the $N$-dimensional identity operator and where we introduced the following $A_{i,j}$ matrices:
\begin{subequations}  
\begin{align}
    A_{j,j-1}  &:= F_0\otimes I^{\otimes (j-1)} + \mbox{shifts}, \label{eq:A_F0} \\
    A_{j,j}  &:= F_1 \otimes I^{\otimes (j-1)}  + \mbox{shifts},  \label{eq:A_F1}\\
    A_{j,j+1}  &:=F_2\otimes I^{\otimes (j-1)}+ \mbox{shifts}.\label{eq:A_F2}
\end{align}
\end{subequations}
Note that we have $N^j$-dimensional vector spaces and maps
\begin{align}
    A_{j,j-1}: \mathbb{C}^{N^{j-1}} \rightarrow \mathbb{C}^{N^{j}}, \qquad A_{j,j}: \mathbb{C}^{N^{j}} \rightarrow \mathbb{C}^{N^{j}}, \qquad  A_{j,j+1}: \mathbb{C}^{N^{j+1}} \rightarrow \mathbb{C}^{N^{j}}, \label{eq:Amatricesdomains} 
\end{align}
where $\mathbb{C}^{N^{j}}$ denotes the $j$-copy vector space. 

We thus obtain an infinite set of coupled ODEs, which is then truncated at some level $j=k$ by setting $A_{k,k+1} =0$ to enforce closure. Also, note the other boundary case of $j=1$ requires us to interpret $x^{\ot 0} =1$. We then end up with the following two ODE systems.
\bigskip

\emph{Infinite Carleman system}: $\forall j\in \mathbb{N}$,
\begin{equation}
\begin{aligned}
    \frac{d}{dt} x^{\ot j} &= A_{j,j-1} x^{\ot (j-1)} + A_{j,j} x^{\ot j} + A_{j,j+1} x^{\ot (j+1)},  \\
    x^{\ot j} (0) &= x(0)^{\otimes j},\\
    x^{\ot 0}(t) &:=1.
    \label{eq:infiniteCarleman}
\end{aligned}
\end{equation}

\bigskip

\emph{Truncated Carleman system}:  \mbox{$j\in\{1,2,\dots, k-1\}$}, 
\begin{equation}
\begin{aligned}
    \frac{d}{dt}y^{[j]} &= A_{j,j-1} y^{[j-1]} + A_{j,j}y^{[j]} + A_{j,j+1} y^{[j+1]}, \\
    \frac{d}{dt} y^{[k]} &= A_{k,k-1} y^{[k-1]} + A_{k,k}y^{[k]},  \\
    y^{[j]} (0) &= x(0)^{\otimes j},\quad y^{[k]} (0) = x(0)^{\otimes k},\\
    y^{[0]}(t)&:=1.
\end{aligned}
\end{equation}

By construction, the solution to the infinite Carleman system takes the form $x(t)^{\otimes j}$ for all $j \in \mathbb{N}$, where $x(t)$ is the exact solution of Eq.~\eqref{eq:nonlinear} at time $t \ge 0$. Note that the solution of the truncated system does not need to take a simple tensor product form, and so we denote it as $y^{[j]}(t)$ for all $j=1,2,\dots, k$. We expect that only for sufficiently large $k$ we may get $y^{[j]}(t) \approx x(t)^{\otimes j}$ under appropriate conditions. By writing $y = [y^{[1]}, y^{[2]}, \dots, y^{[k]}]^T$, we therefore get that 
\begin{equation}
    y(t) \in \mathbb{C}^{\sum_{j=1}^{k} N^j},
\end{equation}
and so the dimension of $y(t)$ is
\begin{equation}
    \sum_{j=1}^{k} N^j = \frac{N^{k+1} - N}{N-1} = O(N^{k}).
\end{equation}
This vector is subject to the system of equations on $\bigoplus_{j=1}^{k} \mathbb{C}^{N^j}$:
\begin{equation}
\begin{aligned}
\label{eq:CarlemanODE}
    \dot{y} &= A y + a,  \\
    a &= [F_0, 0, 0, \dots , 0]^T,  \\
    y(0) &= [ x(0),x(0)^{\otimes 2}, \dots, x(0)^{\otimes k} ]^T,
\end{aligned}
\end{equation}
where we defined the \emph{Carleman matrix} $A$:
\begin{equation}
    \begin{aligned}
        A &= H_0 + H_1 +H_2,  \\
        H_0 = \sum_{j=2}^{k} \ketbra{j}{j-1}\otimes A_{j,j-1},
         \qquad
        H_1 &= \sum_{j=1}^{k} \ketbra{j}{j}\otimes A_{j,j},
         \qquad
        H_2 = \sum_{j=1}^{k-1} \ketbra{j}{j+1} \otimes A_{j,j+1}.
    \end{aligned}
    \label{eq:CarlemannMatrix2}
\end{equation}


\subsection{Carleman truncation error}

We now consider the error that arises in using the Carleman truncation relative to the exact, infinite-dimensional system. 
We have an infinite vector $[x,x^{\ot 2},\dots]^T$ representing the exact dynamics with $\gamma>0$ fixed but arbitrary, and $y = [y^{[1]}, y^{[2]}, \dots, y^{[k]}]^T$ to be the solution of the truncated Carleman system. We now define
\begin{equation}
\label{eq:errorvectorcomponents}
    \eta_j(t) := x(t)^{\ot j} - y^{[j]}(t)
\end{equation}
for $j=1,2, \dots k$ and are interested in $\|\eta_j(t)\|$.

Note that the truncated Carleman system can also be viewed as an infinite linear system, by replacing $A_{j,j \pm 1}$ in Eq.~\eqref{eq:infiniteCarleman} with matrices $B_{j,j\pm 1 }$, defined as $B_{j,j-1} = A_{j,j-1}$, $B_{j,j} = A_{j,j}$  for $j=1, 2, \dots, k$, $B_{j,j+1} = A_{j,j+1}$ for  $j=1, 2, \dots, k-1$, and $B_{j,j\pm 1} =0$ otherwise. Correspondingly, we embed $y(t)$ in the infinite dimensional setting via
\begin{equation}
    y(t) = [y^{[1]}, y^{[2]}, \dots , y^{[k]}, 0, 0, \dots]^T.
\end{equation}
In other words, we set $y^{[j]}(t) = 0$ for all $j> k$. 

We have
\begin{equation}
    \frac{d}{dt}y^{[j]} = B_{j,j-1} y^{[j-1]} + B_{j,j} y^{[j]} + B_{j,j+1} y^{[j+1]} 
\end{equation}
for any $j\in \mathbb{N}$. Therefore, for any $j \in \mathbb{N}$, we have
\begin{align}    
    \dot{\eta}_j &= (A_{j,j-1} x^{\ot (j-1)} - B_{j,j-1} y^{[j-1]}) +(A_{j,j} x^{\ot j} - B_{j,j} y^{[j]}) +(A_{j,j+1} x^{\ot (j+1)} - B_{j,j+1} y^{[j+1]}). 
\end{align}
Hence, for $j=1,2,\dots k-1$,
\begin{equation}
    \dot{\eta}_j = A_{j,j-1} \eta_{j-1} + A_{j,j} \eta_j + A_{j,j+1} \eta_{j+1},
\end{equation}
where $\eta_0 \equiv 0$. For $j=k$, we have
\begin{equation}
    \dot{\eta}_k = A_{k,k-1} \eta_{k-1} + A_{k,k} \eta_k + A_{k,k+1} x^{\ot (k+1)},
\end{equation}
whereas for $j \ge k+1$
\begin{equation}
    \dot{\eta}_{j} = A_{j,j-1} x^{\ot (j-1)} + A_{j,j} x^{\ot j} + A_{j,j+1} x^{\ot (j+1)} .
\end{equation}
The above gives the dynamics for the error vector in the full infinite dimensional space.

For $j=1,2,\dots,k$, we can view the above equations for the error vector as an ODE system in $\eta_j$ plus an external time-dependent driving source term $\zeta(t):=A_{k,k+1} x(t)^{\ot (k+1)}$. This can be viewed as the errors coming from the bounding truncation scale $j=k$, and it acts as a source for the linear ODE error vector.

We define $\eta \in \mathbb{C}^{\sum^{k}_{j=1} N^j}$ the error vector
\begin{equation}
    \eta := [\eta_1, \eta_2, \dots, \eta_{k}]^T,
\end{equation}
which satisfies the set of linear coupled equations
\begin{equation}
\begin{aligned}
    \dot{\eta}(t) &= A \eta(t) + \zeta(t), \nonumber \\
    \zeta(t) &:= [0, 0, \dots , 0, A_{k,k+1} x^{\ot (k+1)}(t)]^T, \\
    \eta(0) &= [0,0,\dots, 0]^T. \label{eq:error_evolution}
\end{aligned}
\end{equation}

Central Question 1 concerns identifying sufficient conditions under which the error $\|\eta(t)\|$ or $\|\eta_1(t)\|$ can be bounded to decay rapidly -- ideally exponentially -- with respect to $k$, while growing slowly -- ideally polynomially -- with other relevant parameters such as the total simulation time.


\subsection{Solving the truncated Carleman system on a quantum computer}

Once the nonlinear dynamics are embedded into a high-dimensional linear system via Carleman linearization, the problem reduces to solving a linear ordinary differential equation (ODE) of the form 
\begin{align}
\dot{y} &= A y + a, \\    
a &= [F_0, 0, 0, \dots , 0]^T,  \\
y(0) &= [ x(0),x(0)^{\otimes 2}, \dots, x(0)^{\otimes k} ]^T,
\end{align}
This reformulation makes it possible to leverage a variety of quantum algorithms developed for linear systems. Under suitable assumptions on sparsity, condition number, and time discretization, quantum linear ODE solvers can offer exponential or polynomial speedups over the best known classical methods, making them particularly attractive for simulating the truncated Carleman system. Below we briefly review the main classes of such algorithms.

A widely used strategy is to further recast the problem as a large system of linear equations and then apply a quantum linear system algorithm (QLSA). The dynamics can be encoded into such a system in several ways. One approach is to quantize established numerical discretization schemes, such as finite difference methods \cite{berry2014high} and finite element methods \cite{montanaro2016quantum}. Another is to represent the Taylor (or Dyson) series expansion of the evolution operator as a linear system \cite{berry2017quantum, berry2024quantum, krovi2023improved, jennings2023cost}. A third approach employs spectral methods \cite{childs2020quantum, childs2021highprecision}, expanding the solution in a suitable basis and solving for the expansion coefficients. Regardless of the encoding strategy, the efficiency of any QLSA-based method ultimately depends on the condition number of the resulting coefficient matrix, which is typically determined by intrinsic properties of $A$ and the total evolution time $T$. 

In addition to QLSA-based approaches, there are quantum algorithms for linear differential equations that avoid the reduction to a linear system. Certain problem classes -- such as coupled harmonic oscillators \cite{babbush2023exponential} and wave equations \cite{costa2019quantum} -- can be naturally embedded into Hamiltonian simulation problems and solved directly using quantum simulation algorithms. Another approach, known as Schr\"{o}dingerization \cite{jin2024quantum, jin2023quantum}, reformulates linear partial differential equations as systems of Schr\"{o}dinger equations, enabling the application of Hamiltonian simulation methods to a broader class of problems. Alternatively, some quantum algorithms for linear differential equations are based on directly implementing the matrix exponential $e^{At}$. This can be achieved using recent techniques such as linear combination of Hamiltonian simulations (LCHS) \cite{an2023linear, an2023quantum}, quantum eigenvalue transformation (QEVT) \cite{low2024quantum}, or, in certain cases, quantum singular value transformation (QSVT) \cite{fang2023timemarchingbased, an2025quantum}.

In this work, we primarily employ the quantum algorithm of Ref. \cite{jennings2023cost} to solve the truncated Carleman system. The efficiency of this algorithm can be summarized as follows. Consider a linear ODE system
\begin{align}
    \dot{z} = B z + b.
\end{align}
Let $U_B$ be an $(\omega, a, 0)$-block-encoding of the matrix $B$, and let $U_0$ be a unitary that prepares the normalized initial state $\ket{z(0)}$. Then, the algorithm of Ref.~\cite{jennings2023cost} prepares a quantum state that is $\epsilon$-close to $\ket{z(T)}$ using
$$\mytO{\omega C_{\rm max}\hat{g}T\mylog{1/\epsilon}}$$
queries to (controlled-) $U_B$, $U_0$, and their inverses, where 
\begin{align}
C_{\rm max} = \max_{t \in [0,T]} \norm{e^{Bt}},~~~\hat{g}=\frac{\max_{t \in [0,T]} \norm{z(t)}}{\norm{z(T)}},
\end{align}
and additional
$$\mytO{\omega C_{\rm max}\hat{g}T\mylog{1/\epsilon}\cdot \mypoly{\mylog{N}}}$$
elementary quantum gates. In situations where $B$ generates stable dynamics, the algorithm in Ref.~\cite{jennings2023cost} gives fast-forwarding, with sublinear in $T$ scaling. This fast-forwarding has been further developed in Ref.~\cite{an2024fast} for autonomous and non-autonomous linear ODEs, and more recently in Ref.~\cite{yang2025quantum}.  For dissipative dynamics, in Section~\ref{subsec:quantum_algorithm_stable_conservative_systems} we shall also introduce a method based on a \emph{quantum Lyapunov transform}. This is achieved given oracle access to a Lyapunov matrix of the linearized dynamics, and allows one to transform to a setting where $C_{\max}=1$.

\newpage 

\section{Stable systems}
\label{sec:stable}


Dynamical systems typically have solutions with bounded norms that do not blow up to infinity. A subset of these are \emph{stable} systems for which the dynamics in the late-time will decay to a fixed-point with some characteristic exponent. Moreover, this stability property is typically robust against small perturbations in the dynamical equation. The ability to have reliable descriptions of late-time dynamics is crucial for the analysis of Carleman convergence. Indeed, the seminal work~\cite{liu2021efficient} crucially relied on the dynamics being dissipative with a known exponential decay rate~$\mu$. Here, we provide a broad framework for stable autonomous systems that exploits Lyapunov stability theory~\cite{plischke2005transient,jennings2023cost}. This leads us to a new continuous family of criteria for the Carleman convergence of stable systems.

\subsection{Setting}
\label{sec:solutionnormbounds}

We set out to construct a Carleman embedding theory for general stable autonomous systems, i.e., systems satisfying 
\begin{align}
    \alpha(F_1)<0,
\end{align} 
where $\alpha$ is the spectral abscissa of $F_1$ defined in Eq.~\eqref{eq:spectralabscissa}. The consequences of stability for quantum algorithms for \emph{linear} ODE equations have been recently analyzed in Ref.~\cite{jennings2023cost}. The present analysis can be seen as an extension of those considerations to quantum algorithms for \emph{nonlinear} differential equations. For autonomous systems, stability is equivalent to the existence of a positive-definite matrix $P>0$ such that the Lyapunov inequality from Eq.~\eqref{eq:lyapunovinequality} is satisfied~\cite{plischke2005transient}. This in turn means that the \emph{generalized log-norm} is negative:
\begin{align}
\label{eq:Plognorm}
    \mu_P(F_1) : = \max_{x \neq 0 } \mathrm{Re} \frac{\langle F_1 x, x \rangle_P}{\langle x, x \rangle_P} < 0,
\end{align}
where 
\begin{equation}
    \label{eq:scalar_P}
    \langle x, y \rangle_P := x^\dag P y    
\end{equation}
is the scalar product induced by $P$.
The generalized log-norm recovers the log-norm from Eq.~\eqref{eq:lognorm} (which plays a central role in quantum algorithm analysis presented in Refs.~\cite{krovi2023improved, costa2023further}) when $P=I$: 
\begin{align}
    \mu_{P=I}(F_1) = \mu(F_1).
\end{align} 
Also, if $F_1$ is diagonalized by a matrix $Q$, as assumed for example in Ref.~\cite{liu2021efficient}, then $P=Q^\dag Q$ satisfies the Lyapunov inequality and $\mu_{Q^\dag Q}(F_1) = \alpha(F_1)$. This shows that the generalized log-norm of Eq.~\eqref{eq:Plognorm} gives a unified description of previously considered stability properties.

We shall prove that, if a system is stable, for every $P$ satisfying a Lyapunov inequality~\eqref{eq:lyapunovinequality} we can associate a \emph{Lyapunov $R$-number}:
\begin{align}
\label{eq:RP}
     R_P :=  \frac{1}{-\mu_P(F_{1})} \left (  \|F_{2}\|_P \|x(0)\|_P + \frac{\|F_0\|_P}{ \|x(0)\|_P} \right ),
\end{align}
giving a sufficient convergence condition for the Carleman embedding method as $R_P<1$. Here
 $\| x\|_P := \sqrt{x^\dag P x}$ is the norm induced by $P$, and similarly with the operator norms:
\begin{align}
    \| F_2\|_P = \left\| P^{1/2} F_2\left( P^{-1/2} \otimes P^{-1/2}\right)\right\|.
\end{align}
Note that when $P=I$, the above recovers the $R$-number $R_\mu$ in Eq.~\eqref{eq:Rmu} from Refs.~\cite{krovi2023improved, costa2023further}. When $P=Q^\dag Q$, we recover the $R$ number $R_\alpha$ from Eq.~\eqref{eq:Ralpha} introduced in Ref.~\cite{liu2021efficient}, if norms are computed in the basis where $F_1$ is diagonal. 

\begin{rmk}
\label{rmk:knowingPmatrix}
Note that we do not necessarily need to know $P$ explicitly to certify convergence. For example, if we have a bound on the condition number of $P$, $\kappa_P := \| P\| \| P^{-1}\|$, and knowledge of the value of $\mu_P(F_1)$, then we can replace the $R_P <1$ condition with a weaker one only involving $\kappa_P$ or its upper bound. In particular, we can instead impose that
\begin{equation}
    R_P \le \frac{1}{-\mu_P(F_1)} \left( \kappa_P \|F_2\| \|x(0)\| + \sqrt{\kappa_P} \frac{\|F_0\|}{\|x(0)\|} \right ) \le  1.
\end{equation}
This provides simpler sufficient criteria for convergence that do not need detailed knowledge of $P$.
\end{rmk}


\subsubsection*{Solution norm bounds}

An important prerequisite in bounding the Carleman embedding error is control over how the solution norm of the nonlinear ODE~\eqref{eq:nonlinear} changes over time. For systems with negative log-norm the central result is 
\begin{lemma}[Ref.~\cite{krovi2023improved}, Lemma 15]
\label{lem:normdecrease}
Let $\mu(F_1)$ be defined as in Eq.~\eqref{eq:lognorm} and $R_\mu$ as in Eq.~\eqref{eq:Rmu}. Then,
\begin{align}
  \mu(F_1)<0,  \quad  R_\mu<1 \quad \quad \Rightarrow \quad  \quad \|x(t)\| \le \|x(0)\|.
\end{align}
\end{lemma}
The above implies that for systems with $\mu(F_1)<0$, $R_\mu<1$ the norm can be controlled. However, Lemma~\ref{lem:normdecrease} cannot be applied to general stable system, since it is limited to the subclass of systems with negative log-norm. To go beyond this, we give the following generalization of Lemma~\ref{lem:normdecrease}, the proof of which can be found in Appendix~\ref{sec:proofCarlemanstable}.
\begin{lemma}
\label{lem:normdecreasegeneral}
Given $P>0$, let $\mu_P(F_1)$ be defined as in Eq.~\eqref{eq:Plognorm} and $R_P$ as in Eq.~\eqref{eq:RP}. Then
\begin{align}
\label{eq:solutionormbound}
  \mu_P(F_1)<0,  R_P<1 \quad  \Rightarrow \quad \|x(t)\|_P \le \|x(0)\|_P.
\end{align}
\end{lemma}

\begin{rmk}
\label{issue1liu}
    Note that Ref.~\cite{liu2021efficient} (Lemma 1) claims the alternative result:
    \begin{align*}
        F_1  \textrm{ diagonalizable,}\quad
    \alpha(F_1)<0,\quad R_{\alpha} <1 \qquad\Rightarrow\qquad \|x(t)\| \le \|x(0)\|.
    \end{align*}
    However, the proof is incorrect. As a counterexample, consider the system of linear equations
    \begin{align}
      \frac{d}{dt}  \begin{bmatrix}
            x_1(t) \\ x_2(t) 
        \end{bmatrix} = 
        \begin{bmatrix}
            1/4 & 1 \\
            -1 & -1/2
        \end{bmatrix}  \begin{bmatrix}
            x_1(t) \\ x_2(t) 
        \end{bmatrix},
    \end{align}
    with $x_1(0) = 1$, $x_2(0) = 1$. The matrix $F_1$ is diagonalizable and $\alpha(F_1)=-1/8$. Furthermore, $R_{\alpha} = 0 <1$. However, we have
    \begin{align}
        \|x(0)\| = \sqrt{2}, \quad \|x(1)\| \geq 1.65 > \sqrt{2}
    \end{align}
    and so $\|x(1)\|> \|x(0)\|$, which is consistent with the ability of stable systems to have transient growth. 
    The issue with the proof in Ref.~\cite{liu2021efficient} is that the following inequality is claimed [Eq.~(4.2) ~\cite{liu2021efficient}]
    \begin{align}
    \label{eq:error}
        x^\dag (F_1 + F^\dag_1) x \leq 2 \alpha(F_1) \| x\|^2,
    \end{align} 
    which does not hold in general. One simple fix is to limit the scope of the result in Ref.~\cite{liu2021efficient} and add the assumption that $F_1$ has negative log-norm, in which case the result follows from Lemma~\ref{lem:normdecrease}. But then the result in Lemma~\ref{lem:normdecrease} is more general, as it does not require that $F_1$ is diagonalizable.
\end{rmk}

\begin{rmk}[Improving solution norm bound]
For simplicity, we give the bound Eq.~\eqref{eq:solutionormbound} in a time-independent form, but this is in fact derived from a stronger bound where one can show $\|x(t)\|_P$ will decrease below $\|x(0)\|_P$. For example, for times $t \gg \Delta :=  \sqrt{\mu_P(F_1)^2 - 4 \| F_0\|_P \| F_2\|_P}$ we have $\| x(t)\|_P \leq c \approx \frac{\mu_P(F_1) - \Delta}{2 \|F_2\|_P}$, a bound that can be much stronger than $\|x(t)\|_P \leq \|x(0)\|_P$ and hence give exponential improvements in the Carleman error bounds we will present in the next section.
\end{rmk}


\subsection{Convergence of the Carleman embedding}
\label{sec:stableconvergence}

We start by presenting a statement of our main result. It gives us a generally infinite family of conditions, each of them sufficient for Carleman convergence:
\begin{align}
    \left\{R_P<1 : P >0 \textrm{ satisfies the Lyapunov inequality } P F_1 + F_1^\dag P < 0\right\}.
\end{align}
More formally, we have the following theorem, the proof of which can be found in  Appendix~\ref{sec:proofCarlemanstable}.

\begin{thm}[Carleman error bound for stable systems]
\label{thm:Carlemanstable}
Let $x(t)$ be the solution to the quadratic ODE system
\begin{align}
\dot{x}(t) = F_0 + F_1 x(t)+ F_2x(t)^{\otimes 2},
\end{align}
and the corresponding Carleman ODE system truncated at level $k$ as in Eq.~\eqref{eq:CarlemanODE}:
\begin{align}
    \dot{y}(t) = Ay(t) +a, \quad \quad y(0) = [x(0), x(0)^{\otimes 2}, \dots, x(0)^{\otimes k}], \quad   a = [F_0, 0, 0, \dots , 0].
\end{align}
 Assume the system is stable, i.e., $\alpha(F_1)<0$, with $\alpha(F_1)$ defined in Eq.~\eqref{eq:spectralabscissa}. If for $P>0$ satisfying the Lyapunov inequality~\eqref{eq:lyapunovinequality} the generalized $R$-number condition  
\begin{align}
    R_P :=  \frac{1}{-\mu_P(F_{1})} \left (  \|F_{2}\|_P \|x(0)\|_P + \frac{\|F_0\|_P}{ \|x(0)\|_P} \right )<1
\end{align}
is satisfied, then there is a rescaling as in Eq.~\eqref{eq:rescalingrelations} such that $\| x(0) \|_P<1$ and the Carleman embedding error vector $\eta_j(t) = x(t)^{\otimes j} - y^{[j]}(t)$ satisfies for $j=1, \dots, k$:
  \begin{align}
  \label{eq:Carlemanerrorbound}
  \|\eta_j(t)\|  \leq \frac{1}{-\xi_P } k \|F_2\|_P \| P^{-1}\|^{j/2} \| x(0)\|_P^{k+1},
\end{align}
 with
\begin{align}
\xi_P \leq  4 \mu_P(F_1) + 5 \| F_0\|_P + 3 \| F_2\|_P <0.
\end{align}
\end{thm}

\begin{rmk}[Carleman errors for higher order components]
    The result for $j=1$ is what matters for the convergence of the Carleman embedding. However, as we shall see in more detail when we analyze quantum algorithms based on the Carleman method, to obtain an efficient quantum algorithm one also needs to bound the errors that one makes in the high-order components $j=2,\dots, k$. Hence we also give bounds for those components.
\end{rmk}


\subsection{Discussion and example}
\label{sec:comparisons}


\subsubsection*{Comparison with prior results}

When $P = I$ satisfies the Lyapunov inequality -- i.e., for systems for negative log-norm -- Theorem~\ref{thm:Carlemanstable} gives a proof of convergence of Carleman linearization that extends~\cite{forets2017explicit, costa2023further} to systems with driving, and coincides with the condition claimed in Ref.~\cite{krovi2023improved}. We have also seen that the stronger result claimed in Ref.~\cite{liu2021efficient} also ends up being subsumed in the claim of Ref.~\cite{krovi2023improved}, once we account to extra assumptions required to prove the solution norm bounds (see Section~\ref{sec:solutionnormbounds}).

Hence, we see that Theorem~\ref{thm:Carlemanstable} extends results claimed in the literature from systems with negative log-norm to general stable systems. However, even in the special case of systems with negative log-norm, we cannot rely on the available convergence proofs. In fact, at a high level the convergence proofs in Ref.~\cite{liu2021efficient, krovi2023improved} are based based on two steps:
\begin{enumerate}
    \item[(a)] Show that $R_{P=I}= R_\mu <1$ implies there is a rescaling $\gamma>0$ for which
\begin{align}
\label{eq:previousconditions}
    \mu(F_1) + \| F_0\| + \| F_2\| < 0, \quad \|x(0)\|<1,
\end{align}
see Eq.~(B.6) in Ref.~\cite{krovi2023improved}, Eq.~(2.3) in Ref.~\cite{liu2021efficient};
\item[(b)] Show that Eq.~\eqref{eq:previousconditions} implies that the Carleman ODE system has negative log-norm, \mbox{$\mu(A) < 0$},  which then implies the exponential decrease of Carleman errors with $k$ (Eq.~7.18-7.19 in Ref.~\cite{krovi2023improved}, or proof of Lemma 2 in Ref.~\cite{liu2021efficient} after adding the assumption of negative log-norm of $F_1$).
\end{enumerate}  
However, (b) does not hold, as we can see via a counterexample: 
\begin{rmk}[Eq.~\eqref{eq:previousconditions} $\not\Rightarrow  \mu(A) <0$]
    Consider the rescaled equation $\dot{x} = f_0 + f_1 x + f_2 x^2$, with $f_0 = 0.97$, $f_1 = -1$, $f_2 = 0.02$ and $x(0)=0.99$. We have $R_\mu= 0.999598 <1$ and clearly with this rescaling the conditions from Eq.~\eqref{eq:previousconditions} are satisfied. The Carleman matrix at $k =4$ reads
\begin{align}
  A =  \left[\begin{matrix}
  -1 & 0.02 & 0 & 0\\
  1.94 & -2 & 0.04 & 0 \\
  0 & 2.91 & -3 & 0.06 \\
  0 & 0 & 3.88 & -4
\end{matrix}\right].
\end{align}
However, $\mu(A) > 0.012 >0$. 
\end{rmk}
Recall that in Appendix~\ref{app:issueliu} and Appendix~\ref{app:issuekrovi} we also pinpoint an error in the convergence proofs in Ref.~\cite{liu2023efficient, krovi2023improved}. Without these proofs, one cannot claim efficiency of the corresponding quantum algorithms. Theorem~\ref{thm:Carlemanstable} restores the convergence proofs and extends them to a much larger class of systems.


\subsubsection*{Elementary example: damped nonlinear oscillator}

Consider a toy one-dimensional nonlinear oscillator problem described by
\begin{align}
    \label{eq:toyoscillatorbelow}
    \ddot{q} = - q - r \dot{q} - n q^2, \quad r>0,
\end{align}
where $r$ is the damping strength and $n$ is the strength of a nonlinear force towards the left. We will analyze the parameter regime $(r,n)$ for which the Carleman embedding procedure convergences according to Theorem~\ref{thm:Carlemanstable}.

Introducing $v = \dot{q}$ and $x = (q, v)^T$ we can transform the system into a standard form:
\begin{align}
    \dot{x} = F_1 x + F_2 (x \otimes x),
\end{align}
where
\begin{align}
    F_1 = \begin{bmatrix}
        0 & 1 \\ -1 & -r
    \end{bmatrix}, \quad F_2 = \begin{bmatrix}
        0 & 0 & 0 & 0 \\
        -n & 0 & 0 & 0
    \end{bmatrix}.
\end{align}
We have that
\begin{align}
    \alpha(F_1) = \frac{1}{2} \mathrm{Re} \left[-r + \sqrt{r^2 - 4}\right]< 0,
\end{align}
so the system is stable. However,
\begin{align}
\mu(F_1) = 0,
\end{align}
and so previous convergence results do not apply.

We will now find the form of the most general $P>0$ satisfying the Lyapunov inequality from Eq.~\eqref{eq:lyapunovinequality} for the considered two-dimensional problem. First, note that, since $P>0$, its trace is positive, but its value has no effect on satisfying the inequality. Thus, without loss of generality, we can restrict to matrices $P$ with unit trace. Then, note that positive matrices of size 2 with unit trace are equivalent to mixed (not pure) density matrices of a qubit system, and so can be parametrized by the interior of the Bloch sphere:
\begin{equation}
    P = \begin{pmatrix}
        \frac{1+r_z}{2} & \frac{r_x-i r_y}{2}\\[1ex]
        \frac{r_x+i r_y}{2} & \frac{1-r_z}{2}
    \end{pmatrix},\qquad r_x^2+r_y^2+r_z^2 < 1.
\end{equation}
Since there are just three parameters $(r_x,r_y,r_z)$ in a bounded region, for a given initial condition $x(0)$ and parameter values $(r,n)$, one can efficiently minimize $R_P$ numerically under the Lyapunov inequality constraint (which can be simplified to the constraint $\mu_P(F_1)<0$). The right panel of Fig.~\ref{fig:oscillatorexample} was generated exactly in this way, with the blue region corresponding to points where the described minimization procedure gave the value of $R_P$ smaller than one.

\newpage 

\section{Conservative and oscillating systems}
\label{sec:conservative}

In the previous section, we have developed a formalism for stable nonlinear systems using Lyapunov theory. This applies whenever the autonomous ODE has the spectrum of the linear contribution $F_1$ lying strictly in the left-half of the complex plane. However, many dynamical systems have part of their linear spectrum lying on the imaginary axis, and so such systems are not strictly stable (a small perturbation in $F_1$ can lead to unbounded growth). One notable example is that of linear quantum dynamics, where the Schr\"{o}dinger equation takes the form $\dot{x} = F_1x$, and $F_1 = -i H$ for some Hermitian operator $H$. This implies that the spectrum of $F_1$ is \emph{entirely} on the imaginary axis, and the bounded dynamics is robust due to the constraint of unitarity.

Another such \emph{marginally stable} scenario arises when the system has at least one \emph{conservation law} at play that constrains the dynamics to a sub-manifold defined by the conserved quantities. For this situation the linear term $F_1$ has one or more zero eigenvalues, and moreover the nonlinear term $F_2$ preserves the nullspace of $F_1$. The goal of this section is to extend our stability analysis to include the case of dynamical systems that have one or more conserved quantities. In the first instance, we consider purely linear conserved quantities, and then describe the case of more general polynomial conserved quantities. Moreover, we explain how to extend the allowed setting from conserved to oscillating quantities (when $F_1$ has purely imaginary eigenvalues and $F_2$ preserves that imaginary subspace), and use the obtained result to obtain Carleman convergence for non-autonomous systems.


\subsection{Setting}
\label{sec:conservative_setting}

A dynamical system with $M$ linear conserved quantities $\Q^{(m)}$ is defined as a system for which there exist $M$ linear functionals $\Q^{(m)}$ of the system state, 
\begin{equation}
    \label{eq:conserved_q}
    \Q^{(m)} := q^{(m)\dagger} x,\quad q^{(m)} = \sum_{j=1}^N q_j^{(m)} \ket{j}, \quad m\in\{1,\dots,M\},
\end{equation}
such that 
\begin{equation}
    \label{eq:cons_gneral}
    \forall m, \forall t \ge 0:~\frac{d\Q^{(m)}}{dt} = 0.
\end{equation}
For example, the set of such systems includes cases when the variable $x$ describes the distribution of charge or fluid density in space, and the total charge or fluid density is conserved by the dynamics. Similarly, one can also capture dynamics with conserved linear momentum when $x$ describes such distributions in phase space. More generally, the above setting can be applied whenever $x$ describes a distribution over arbitrary states to which one can ascribe a value of some additive conserved quantity. This physical setup with conserved linear quantities can be naturally generalized to a setting where the linear quantities $\Q^{(m)}$ given by Eq.~\eqref{eq:conserved_q} are oscillating and satisfying:
    \begin{equation}
        \label{eq:oscillating}
        \forall m, \forall t \ge 0:~\frac{d\Q^{(m)}}{dt} = \i \omega_m \Q^{(m)},
    \end{equation}
where $\omega_m\in\mathbb{R}$. The scenario with conserved linear quantities is then simply a special case with $\omega_m=0$, so that in what follows we will focus on oscillating quantities and keep in mind that the obtained results apply also to the scenario with conserved quantities.

Using the dynamical equation,
\begin{equation}
    \dot{x}(t) = F_0 + F_1 x(t) + F_2 x(t)^{\ot 2},
\end{equation}
the oscillating quantities condition from Eq.~\eqref{eq:oscillating} can be rewritten as
\begin{equation}
    \label{eq:cons_dynamic}
    \forall m:~q^{(m)\dagger} F_1 x+q^{(m)\dagger} F_2 x^{\otimes 2}+q^{(m)\dagger}F_0 = {\i \omega_m q^{(m)\dagger} x}.
\end{equation}
For this to hold for all states $x$, one needs: 
\begin{equation}
    \label{eq:conservation}
    \forall m:~ q^{(m)\dagger} F_1 x=\i \omega_m q^{(m)\dagger}x\quad \text{and}\quad 
    q^{(m)\dagger} F_2 x^{\otimes 2}=q^{(m)\dagger}F_0 = 0.
\end{equation}
Indeed, if Eq.~\eqref{eq:cons_dynamic} holds for $x$, it should also hold for $ax$, where $a$ is any nonzero real number, and then it is clear that the only way to satisfy this is by enforcing Eq.~\eqref{eq:conservation}. Then, we clearly see that in such a setting $q^{(m)\dagger}$ are the left eigenvectors of $F_1$ corresponding to the eigenvalue $\i \omega_m$. 

We assume that $F_1$ is diagonalizable,
\begin{equation}
    F_1 = Q\Lambda Q^{-1},
\end{equation}
and that all its eigenvalues, $\lambda_j(F_1)=\Lambda_{jj}$, have negative real parts except for the eigenvalues $\i \omega_m$ corresponding to the oscillating quantities. Consequently, the considered system is linearly \emph{marginally} stable and not stable:
\begin{align}
    \alpha(F_1)=0.
\end{align}
As such, the existing results, whether presented in Section~\ref{sec:stable} or in Refs.~\cite{forets2017explicit,liu2021efficient,krovi2023improved,costa2023further}, cannot be used to guarantee Carleman error convergence under any conditions (except the trivial one with vanishing nonlinearity, $F_2=0$).

Define the \emph{real spectral gap} of $F_1$ by 
\begin{equation}
    \label{eq:real_spectral}
    \delta(F_1):= -\!\!\!\!\!\!\!\max\limits_{\substack{i\\\mathrm{Re}(\lambda_i(F_1))\neq 0}} \mathrm{Re}(\lambda_i(F_1) )> 0.
\end{equation}
Introducing the transformed vectors and matrices,
\begin{align}
    \tilde{x} = Q^{-1}x,\qquad \tilde{F}_0 = Q^{-1} F_0,\qquad \tilde{F}_2 = Q^{-1} F_2 Q^{\otimes 2},\qquad \tilde{q}^{(m)^\dagger} = {q}^{(m)^\dagger} Q,
\end{align}
Eq.~\eqref{eq:conservation} implies that
\begin{equation}
    \forall \tilde{x},m:~ \tilde{q}^{(m)\dagger}\Lambda \tilde{x} = {\i \omega_m \tilde{q}^{(m)\dagger} \tilde{x}}.
\end{equation}
This means that $\tilde{q}^{(m)}$ are eigenstates of a diagonal matrix corresponding to eigenvalues $\i\omega_m$, so they can be identified with the first $M$ computational basis states:
\begin{equation}
    \label{eq:q_tilde}
    \forall m:~\tilde{q}^{(m)} = \ket{m},
\end{equation}
while the $\i \omega_m$ eigenstates of $F_1$ are simply $Q\ket{m}$. Moreover, combining Eq.~\eqref{eq:conservation} with Eq.~\eqref{eq:q_tilde}, one gets
\begin{equation}
    \label{eq:F2_cons_zero}
    \forall \tilde{x},m:~ \bra{m}\tilde{F}_2 \tilde{x}^{\otimes 2} =0.
\end{equation}
This means that the image of $\tilde{F}_2$ is orthogonal to the first $M$ computational basis states,
\begin{equation}
    \label{eq:f2_cons}
    \forall \tilde{x}:~\tilde{F}_2 \tilde{x}^{\otimes 2} = \sum_{j=M+1}^N c_j \ket{j},
\end{equation}
or that, for any $x$, the decomposition of $F_2 x^{\otimes 2}$ into the eigenstates of $F_1$ has no components corresponding to eigenstates of $F_1$ with eigenvalues $\i\omega_m$.


\subsection{Convergence of the Carleman embedding}

The central result of this section states that, even if $\alpha(F_1)=0$, we can still ensure convergence of the Carleman embedding for systems with linear conservation laws or with linearly oscillating quantities, under appropriate conditions on the real spectral gap $\delta(F_1)$. We first present the result for a system without driving, the proof of which can be found in Appendix~\ref{app:thm_cons}.
\begin{thm}[Carleman error convergence for conservative \& oscillating systems without driving]
    \label{thm:Carleman_cons}
    Let $x(t)$ be the solution to the $N$-dimensional quadratic ODE system
    \begin{align}
    \dot{x}(t) = F_1 x(t)+ F_2x(t)^{\otimes 2},
    \end{align}
    with $M$ linear oscillating (or conserved) quantities $\Q^{(m)}$ defined in Eq.~\eqref{eq:oscillating}, and 
    \begin{align}
        \dot{y}(t) = Ay(t), \quad \quad y(0) = [x(0), x(0)^{\otimes 2}, \dots, x(0)^{\otimes k}]
    \end{align}
    the corresponding Carleman ODE system truncated at level $k$, see Eq.~\eqref{eq:CarlemanODE}. Assume $F_1$ is diagonalized by $Q$, with the only eigenvalues with non-negative real parts being the purely imaginary eigenvalues related to oscillating (or conserved) quantities. Moreover, suppose that the norm of the transformed solution vector is bounded for all $t \ge 0$:
    \begin{equation}
        \forall t \ge 0: \|Q^{-1} x(t)\| \leq \|\tilde{x}_{\mathrm{max}}\|.
    \end{equation}
    Then, for every rescaling parameter $\gamma$ as in Eq.~\eqref{eq:rescalingrelations} and satisfying
    \begin{equation}
        \gamma = p \gamma_0,\qquad \gamma_0=\frac{e\|Q^{-1}F_2Q^{\otimes 2}\|}{\delta(F_1)\|Q\|},\qquad p\in(0,1],
    \end{equation}
    the Carleman embedding error vectors $\eta_j(t) = x(t)^{\otimes j} - y^{[j]}(t)$ satisfy
    \begin{align}
        \|\eta_j(t)\| &\leq \frac{j}{2(k+1)}p^j R_\delta^{k+1},
    \end{align}
    where
    \begin{equation}\label{eq:R-number-conservative}
        R_\delta := \frac{2e\|\tilde{x}_{\mathrm{max}}\|\|Q^{-1}F_2Q^{\otimes 2}\|}{\delta(F_1)},
    \end{equation}
    with $\delta(F_1)$ denoting the real spectral gap of $F_1$, see Eq.~\eqref{eq:real_spectral}. As a result, $\|\eta(t)\|\to 0$ as $k\to\infty$ when $R_\delta<1$. The result can be improved for the special case of $\|\eta_1\|$ by replacing $2e$ by $4$ in the expression for $R_\delta$ and $e$ by $16$ in the expression for $\gamma_0$. 
\end{thm}

\begin{rmk}
    Similar to the conditions for Lyapunov stability and the matrix $P$, as discussed in Remark~\ref{rmk:knowingPmatrix}, in some situations we may not know $Q$, but have knowledge of, e.g.,~its condition number~$\kappa_Q$. In such situations we can replace the above with weaker conditions, for example
    \begin{equation}
        R_\delta \le \frac{2e \kappa_Q^2 \max_t \|x(t)\|\|F_2\|}{\delta(F_1)} <1,
    \end{equation}
    which is a sufficient condition to give convergence of the Carleman expansion.
\end{rmk}

Next, in Appendix~\ref{app:thm_cons_ext}, we explain how to use the above result to bound the Carleman error for conservative systems with driving, as captured by the following corollary.

\begin{cor}[Carleman error convergence for conservative \& oscillating systems with driving]
\label{cor:Carleman_cons_extension}
    Let $x(t)$ be the solution to the $N$-dimensional quadratic ODE system
    \begin{align}
    \dot{x}(t) = F_0 + F_1 x(t)+ F_2x(t)^{\otimes 2},
    \end{align}
    with $M$ linear oscillating (or conserved) quantities $\Q^{(m)}$. Assume $F_1$ is diagonalized by $Q$, with only eigenvalues with non-negative real parts being the purely imaginary eigenvalues related to oscillating (or conserved) quantities. Moreover, suppose that the norm of the transformed solution vector is bounded for all $t \ge 0$:
    \begin{equation}
        \forall t \ge 0: \|Q^{-1} x(t)\| \leq \|\tilde{x}_{\mathrm{max}}\|.
    \end{equation}
    Then, let $\upsilon \in \mathbb{R}$ and consider an $(N+1)$-dimensional quadratic ODE system
    \begin{align}
    \label{eq:G_system}
    \dot{z}(t) = G_1 z(t)+ G_2 z(t)^{\otimes 2}, \quad \quad
            z(0) = [x(0),\gamma \upsilon],
    \end{align}
    with 
    \begin{equation}
        G_1 = \sum_{i,j=1}^N (F_1)_{i,j} \ketbra{i}{j},\qquad
        G_2 = \sum_{i,j,l=1}^N (F_2)_{i,jl} \ketbra{i}{jl} + \sum_{i=1}^N \frac{(F_0)_i}{\gamma^2 \upsilon^2} \ketbra{i}{N+1,N+1},
    \end{equation}
    and $\gamma$ is a rescaling parameter as in Eq.~\eqref{eq:rescalingrelations}.
    The solution to Eq.~\eqref{eq:G_system} is $z(t)=[x(t),\gamma \upsilon]$. Denote by $y(t)$ the solution to the Carleman ODE system for Eq.~\eqref{eq:G_system} and by $\D_j$ the operation of discarding from $y^{[j]}$ all entries involving the $N+1$ label.\footnote{E.g., if $x=(x_1,x_2)$, the system is extended to $(x_1,x_2,x_3)$. Then $y^{[2]} = (y_{11}, y_{12}, y_{13}, y_{21}, y_{22}, y_{23}, y_{31}, y_{32}, y_{33})$ and $\mathcal{D}_2 y^{[2]} = (y_{11}, y_{12}, y_{21}, y_{22})$.} 
    Then, for every \begin{equation}
        \gamma = p \gamma_0,\qquad \gamma_0=\frac{e\|Q^{-1}F_2Q^{\otimes 2}\|}{\delta(F_1)\|Q\|},\qquad p\in(0,1],
    \end{equation}
    we have
    \begin{equation}
        \|x^{\ot j}(t) -\D_j y^{[j]}(t)\| \leq  \frac{j}{2(k+1)} p^j R_{\delta,a}^{k+1},
    \end{equation}
    where
    \begin{equation}
        R_{\delta,\upsilon}=\frac{2e\sqrt{\|\tilde{x}_{\mathrm{max}}\|^2+\upsilon^2} \max\left\{ \|Q^{-1}{F}_2 Q^{\ot 2}\|,\|Q^{-1}{F}_0\|/\upsilon^2\right\}}{\delta(F_1)}.
    \end{equation}
    The optimal choice of $\upsilon$ that minimizes $R_{\delta,\upsilon}$ is
    \begin{equation}
    \upsilon_*=\sqrt{\|Q^{-1}F_0\|/\|Q^{-1}F_2 Q^{\ot 2}\|}
    \end{equation}
    leading to
    \begin{equation}
        R_{\delta}=  \frac{2e\|Q^{-1}F_2Q^{\otimes 2}\| \sqrt{\|\tilde{x}_{\mathrm{max}}\|^2+\frac{\|Q^{-1} F_0\|}{\|Q^{-1}F_2Q^{\otimes 2}\|}}}{\delta(F_1)}<1
    \end{equation}
    as a sufficient condition for convergence of the Carleman error bound.
\end{cor}

Finally, the underlying analysis for linear conserved quantities can be extended to general polynomial quantities that are conserved under the dynamics. This is captured by the following theorem, which implies a Carleman convergence criterion for systems with quite general conservation laws, and whose proof can be found in Appendix~\ref{apps:Carleman_poly_cons}.
\begin{thm}[Carleman convergence for systems with polynomial conserved quantities]
    \label{thm:Carleman_poly_cons}
        Let $x(t)$ be the solution to the $N$-dimensional quadratic ODE system
    \begin{align}
    \dot{x}(t) = F_1 x(t)+ F_2x(t)^{\otimes 2}, 
    \end{align}
    with $M$ conserved quantities $\Q_r^{(m)}(x)$ that are order--$r$ polynomial functions of the variable $x(t)$. Then, there exists an embedding of the $N$--dimensional nonlinear dynamics into an $O(N^r)$--dimensional quadratic ODE system $\check{x}(t)$ 
    \begin{align}
    \dot{\check{x}}(t) = \check{F}_1 \check{x}(t)+ \check{F}_2\check{x}(t)^{\otimes 2},
    \end{align}
    from which $x(t)$ can be obtained, and for which $\Q_r^{(m)}$ are linearly conserved observables in the embedding system. Hence, an application of Theorem~\ref{thm:Carleman_cons} implies convergence if $\check{R}_\delta<1$, where $\check{R}_\delta$ is the quantity defined in Eq.~\eqref{eq:R-number-conservative} for the embedding system $\check{x}(t)$.
\end{thm}


\subsection{Discussion and examples}
\label{sec:cons_discussion}

Let us consider a toy two-dimensional system described by a state vector $x=(x_1,x_2)$ evolving according to
\begin{equation}
    \label{eq:toy1}
    \dot{x}(t) = F_1 x(t) + F_2 x(t)^{\otimes 2},
\end{equation}
with
\begin{equation}
    \label{eq:toy2}
    F_1 = \begin{bmatrix}
        0 & 0\\
        0 & -1
    \end{bmatrix},
    \qquad
    F_2 = \begin{bmatrix}
        0 & 0 & 0 & 0 \\
        a & 0 & 0 & -b
    \end{bmatrix},
\end{equation}
where we assume $a,b >0$. Clearly, one cannot directly apply the result for linearly stable systems here, since $\alpha(F_1)=0$. One might, however, use the fact that $x_1$ is conserved, and reduce the problem dimension, looking only at the evolution of $x_2$, which reads
\begin{equation}
    \label{eq:reduced}
    \dot{x}_2(t) = -x_2(t) -bx_2(t)^2+ a x_1(0)^2,
\end{equation}
where we used the fact that $x_1(t)=x_1(0)$. For simplicity, we will also restrict to $x_1(0) >0$ and $x_2(0)\geq 0$, which then enforces $x_2(t) \geq 0$ for all $t$. Now, Eq.~\eqref{eq:reduced} describes a one-dimensional  stable system with $F_1=-1$, $F_2=-b$ and $F_0=a x_1(0)^2$. Thus, we can apply the dissipative result noting that for one-dimensional systems, up to scaling, there is just one positive-definite matrix, $P=1$, so that the $R_P$-number reads
\begin{equation}
    R_P = b x_2(0) + a \frac{x_1(0)^2}{x_2(0)}.
\end{equation}
It is then a straightforward calculation to show that the condition $R_P<1$ is equivalent to
\begin{equation}
    \label{eq:ellipse}
     \frac{x_1(0)^2}{\left(\frac{1}{2\sqrt{ab}}\right)^2} + \frac{\big(x_2(0) - \frac{1}{2b} \big)^2}{\Big(\frac{1}{2b}\Big)^2} \leq 1,
\end{equation}
which is an ellipse in the $(x_1(0),x_2(0))$ plane.

On the other hand, if we use the conservative result with $Q=I$ and $F_0 =0$, we get the $R_\delta$-number given by
\begin{equation}
    R_\delta =  2e \sqrt{a^2+b^2}  \sqrt{x_1(0)^2+\max_t x_2(t)^2}.
\end{equation}
From Eq.~\eqref{eq:reduced} and our assumptions on non-negativity of parameters involved, one can infer that
\begin{equation}
    \max_t x_2(t) = \left\{\begin{array}{cc}
         x_2(0) & \mathrm{for~~~} x_2(0)+bx_2(0)^2-ax_1(0)^2 \geq 0,\\
         \frac{\sqrt{1+4abx_1(0)^2}-1}{2b}   & \mathrm{otherwise.}
    \end{array}\right.
\end{equation}
Now, the point we want to make is that there are cases when $R_\delta<1$ and $R_P\geq 1$. The easiest way to see this is to choose $x_2(0)=0$ and a small enough $x_1(0)>0$. Since this point lies outside the ellipse from Eq.~\eqref{eq:ellipse}, we know that $R_P\geq 1$. At the same time, $R_\delta=0$ for $x_1(0)=0$, so by continuity there exists a range of positive values of $x_1(0)$ for which $R_\delta<1$. We have also seen in Fig.~\ref{fig:R_regions} that neither of the convergence regions is a subregion of the other. We conclude that the Carleman convergence region arising from our conservative result cannot be simply obtained from the result for stable systems. Moreover, the conservative result can be applied directly to a conservative system, without the need of modifying it by decreasing its dimensionality to effectively render it dissipative, something that can bring extra complications for higher-dimensional systems.

To illustrate the case with linear oscillating quantities, let us slightly modify our two-dimensional toy example by changing $F_1$ and $F_2$ matrices as follows:
\begin{equation}
    F_1 = \begin{bmatrix}
        \i\omega & 0\\
        0 & -1
    \end{bmatrix},
    \qquad
    F_2 = \begin{bmatrix}
        0 & 0 & 0 & 0 \\
        a & \frac{b}{2} & \frac{b}{2} & c
    \end{bmatrix},
\end{equation}
As before, we cannot directly apply the dissipative result here, since $\alpha(F_1)=0$, but we may reduce the problem dimensionality. We first observe that
\begin{equation}
    x_1(t) = e^{\i \omega t} x_1(0),
\end{equation}
and so the evolution of $x_2$ reads
\begin{equation}
    \dot{x}_2(t) = (b x_1(0) e^{\i \omega t} -1) x_2(t) + c x_2(t)^2  + a x_1(0)^2 e^{2 \i\omega t}.
\end{equation}
We thus have a one-dimensional system with time-dependent matrices $F_0$ and $F_1$. Hence, even if the parameter $b$ was chosen such that the resulting $F_1$ has $\alpha(F_1)<0$, one would have to extend the existing results for linearly stable systems to be able to deal with the extra complication coming from time-dependence of the evolution matrices.\footnote{For time-dependent systems with negative log-norm one should be able to leverage the results we presented here in combination with Ref.~\cite{an2024fast}. But for general linearly stable systems more work is needed.} Instead, using the result from this section, we could simply calculate the $R_\delta$-number to be
\begin{equation}
    R_\delta = 4\sqrt{a^2 +b^2/2 +c^2}\sqrt{x_1(0)^2 + \max_t x_2(t)^2},
\end{equation}
for which there clearly exists a parameter regime guaranteeing $R_\delta<1$. This simple example also illustrates how one could obtain a Carleman convergence result for general systems with explicit time-dependence. We discuss this in more detail in the next section.

\subsection{Handling time-dependent driving and linear couplings}
\label{sec:carleman_time_dependent}

It turns out that the analysis of the oscillating systems can be used to prove convergence for nonlinear systems with time-dependent terms. We now explain how to use the results from this section to derive sufficient convergence conditions for the Carleman embedding for systems with time-dependent driving and linear terms. Consider the following nonlinear system:
\begin{equation}
     \dot{x}(t) = F_0(t) + F_1(t) x(t) + F_2 x(t)^{\ot 2},
\end{equation}
with $F_0(t)$ and $F_1(t)$ decomposed into Fourier modes as
\begin{equation}
    F_0(t) = F_0^{(0)} + \sum_{j=1}^J e^{\i \omega_j t} F_0^{(j)} ,\qquad F_1(t) = F_1^{(0)} +  \sum_{j=1}^J e^{\i \omega_j t} F_1^{(j)},
\end{equation}
where we explicitly separated the time-independent components $F_0^{(0)}$ and $F_1^{(0)}$. Now, the idea is to introduce an ancillary oscillating degree of freedom $z_j$, whose evolution reads $\dot{z}_j=\i\omega_j z_j$, for each of the $J$ Fourier modes. Then, one can build the time-independent $G_1$ and $G_2$ matrices of the enlarged system $[x,z]$, so that the ancillary degrees of freedom $z$ simply oscillate with their eigenfrequencies, but affect the original degrees via appropriate couplings corresponding to Fourier coefficients, effectively acting as $F_0(t)$ and $F_1(t)$. Note that this idea of simulating time-dependence via embedding in a higher-dimensional space was recently proposed and discussed in Ref.~\cite{muraleedharan2025simulating} (see also Ref.~\cite{chen2024carleman}).

More precisely, let us construct the following dynamical system
\begin{equation}
    \frac{d}{dt} \begin{bmatrix}
        x\\
        z
    \end{bmatrix} 
    = 
    G_0 +
    G_1 \begin{bmatrix}
        x\\
        z
    \end{bmatrix}
    +
    G_2
    \begin{bmatrix}
        x\ot x\\
        x\ot z\\
        z \ot x\\
        z \ot z
    \end{bmatrix}
\end{equation}
with $x(0)$ given by the original initial condition, $z(0)$ yet unspecified, and
\begin{subequations}
    \begin{align}
        G_0 &= 
        \begin{bmatrix}
            F_0^{(0)}\\
            0
        \end{bmatrix},\\
        G_1 &= 
        \begin{bmatrix}
            F_1^{(0)} & \frac{F_0^{(1)}}{z_1(0)} \dots \frac{F_0^{(J)}}{z_J(0)} \\
            0 & \i\Omega 
        \end{bmatrix},
        \qquad
        \Omega = \mathrm{diag}(\omega_1,\dots,\omega_J),\\
        G_2 &= 
        \begin{bmatrix}
            F_2 & C & 0 & 0 \\
            0 & 0 &0 & 0
        \end{bmatrix},
        \qquad
        C =  \sum_{j=1}^{J} F_1^{(j)}\ot\frac{\ketbra{j}{j}}{z_j(0)}.
    \end{align}
\end{subequations}
The solutions for the ancillary degrees of freedom are then clearly given by
\begin{equation}
    z_j(t) = z_j(0) e^{\i \omega_j t},
\end{equation}
and so they are linear oscillating quantities described in this section. Using these explicit solutions, the evolution of the original degrees of freedom is given by
\begin{equation}
    \dot{x}(t) = F_0^{(0)} + \sum_{j=1}^{J} e^{\i \omega_j t} F_0^{(j)}  + F_1^{(0)} x(t) + \sum_{j=1}^{J}  e^{\i \omega_j t} F_1^{(j)} x(t) + F_2 x(t)^{\ot 2},
\end{equation}
which is precisely the time-dependent nonlinear system we started with. Thus, we have simulated time-dependent driving and linear couplings of the original system using an enlarged system that is time-independent and satisfies our conditions for the result with linear oscillating quantities. We can thus get a converging Carleman error whenever
\begin{equation}
        \label{eq:R_time_dep}
        R_{\delta}=  \frac{2e\|Q^{-1}G_2Q^{\otimes 2}\| \sqrt{\max_t \|Q^{-1} [x(t),z(t)]^T\|^2+\frac{\|Q^{-1} G_0\|}{\|Q^{-1}G_2Q^{\otimes 2}\|}}}{\delta(G_1)}<1,
    \end{equation}
where $Q$ is a matrix diagonalizing $G_1$ and we can use the fact that $\delta(G_1)=\delta(F_1^{(0)})$ and $|z_j(t)|\leq |z_j(0)|$. Moreover, we can choose $z(0)$ to minimize $R_\delta$.

As a simple illustrative example, let us consider a one-dimensional quadratic system evolving according to
\begin{equation}
    \label{eq:time_dep_toy}
    \dot{x}(t)=\sum_{j=1}^J b_j e^{\i \omega_j t}+\left(\sum_{j=1}^J c_j e^{\i \omega_j t} -{\delta}\right) x(t) + a x(t)^2.
\end{equation}
We then construct a $(J+1)$-dimensional system described by $\xi=(x,z_1,\dots,z_J)$ and evolving according to 
\begin{equation}
    \label{eq:xi_evol}
     \dot{\xi}(t)= G_1 \xi(t) + G_2 \xi(t)^{\ot 2}
\end{equation}
with
\begin{equation}
\!\!\!\!\! G_1 = \begin{bmatrix}
        -{\delta} &\frac{b_1}{z_1(0)} & \frac{b_2}{z_2(0)} & \dots &\frac{b_J}{z_J(0)}\\
        0 & \i\omega_1 & 0 & \dots & 0\\
        0 & 0 & \i\omega_2 & \dots & 0\\
        \vdots &\vdots &\vdots &\ddots & 0\\
        0 & 0& 0 & \dots & \i\omega_J
    \end{bmatrix}\! ,
    ~
    G_2 = \begin{bmatrix}
        a & \frac{c_1}{z_1(0)} & \frac{c_2}{z_2(0)} &\dots  & \frac{c_J}{z_J(0)} & 0 &\dots & 0\\
        0 & 0 & \dots & \dots & \dots & \dots & \dots &0 \\
        \vdots & \vdots & \vdots & \vdots & \vdots & \vdots & \vdots &\vdots \\
        0 & 0 & \dots & \dots & \dots & \dots & \dots &0 
    \end{bmatrix}\! .\!
\end{equation}
Clearly, we have $z_j(t)=z_j(0)e^{\i \omega_j t}$ and with that it is easy to see that the first component of $\xi$, i.e., $x$, evolves according to Eq.~\eqref{eq:time_dep_toy}. Moreover, the matrix $Q$ diagonalizing $G_1$ and its inverse $Q^{-1}$ have the following simple form:
\begin{equation}
Q = \begin{bmatrix}
        1 & \beta_1 & \beta_2 & \dots &\beta_J\\
        0 & 1 & 0 & \dots & 0\\
        0 & 0 & 1 & \dots & 0\\
        \vdots &\vdots &\vdots &\ddots & 0\\
        0 & 0& 0 & \dots &1
    \end{bmatrix} ,
    \quad
    Q^{-1} = \begin{bmatrix}
        1 & -\beta_1 & -\beta_2 & \dots &-\beta_J\\
        0 & 1 & 0 & \dots & 0\\
        0 & 0 & 1 & \dots & 0\\
        \vdots &\vdots &\vdots &\ddots & 0\\
        0 & 0& 0 & \dots &1
    \end{bmatrix},
\end{equation}
where $\beta_j=\frac{b_j}{z_j(0)(\delta+i\omega_j)}$. Thus, all terms in the expression for $R_\delta$ in Eq.~\eqref{eq:R_time_dep}, except for $x(t)$, can be expressed explicitly as a function of model parameters. In the particularly simple case with no driving ($b_j=0$ for all $j$), we have
\begin{equation}
    R_\delta = 2e \sqrt{\sum_{j=1}^J \frac{|c_j|^2}{|z_j(0)|^2}+a^2} \sqrt{\sum_{j=1}^J |z_j(0)|^2 + \max_t |x(t)|^2}
\end{equation}
and, employing the freedom in choosing $z_j(0)$ to minimize $R_\delta$, we can set $z_j(0)=\sqrt{c_j}$ to get
\begin{equation}
    R_\delta \leq 2e \sqrt{\sum_{j=1}^J |c_j|+a^2} \sqrt{\sum_{j=1}^J |c_j| + \max_t |x(t)|^2}.
\end{equation}

In Fig.~\ref{fig:convergence}, we present the numerical study of error convergence in the case of $J=2$, $b_1=b_2$, $c_1=c_2$, and $\omega_1=-\omega_2$, so that the problem with time-dependent driving and linear coupling is given by
\begin{equation}
    \label{eq:time_dep_numerical}
    \dot{x}(t) = 2b\cos(\omega t) +\left(2c\cos(\omega t)-\delta\right) x(t) + a x(t)^2,
\end{equation}
whereas its time-independent version is given by Eq.~\eqref{eq:xi_evol} with 
\begin{equation}
    G_1 = \begin{bmatrix}
        -\delta &\frac{b}{\sqrt{c}} & \frac{b}{\sqrt{c}}\\
        0 & \i\omega & 0 \\
        0 & 0 & -\i\omega 
        \end{bmatrix}\! ,
    \quad
    G_2 = \begin{bmatrix}
        a & \sqrt{c} & \sqrt{c} & 0 &\dots  &  0\\
        0 & 0  & \dots & \dots & \dots &0 \\
        \vdots  & \vdots & \vdots & \vdots & \vdots &\vdots \\
        0 & 0 &  \dots & \dots & \dots &0 
    \end{bmatrix},
\end{equation}
where we made a choice $z_j(0)=\sqrt{c_j}$. For each choice of parameters $\{a,b,c,\delta,\omega\}$, we first numerically solved for $x(t)$ in the time interval $t\in[0,T]$, using the original time-dependent nonlinear formulation from Eq.~\eqref{eq:time_dep_numerical}. This allowed us to compute the value of $R_\delta$ using Eq.~\eqref{eq:R_time_dep} for each set of parameters. We then performed Carleman embedding of the corresponding time-independent formulation, and numerically solved the resulting system of linear differential equations in the same time interval. Note that we did not store up to $k$ copies of the system (as it would require an infeasible number of variables scaling as $3^k$), but instead tracked the evolution of distinct monomials of order at most $k$ (resulting in a much more tractable scaling with $k^2$). Since the first component $y_1(t)$ of the Carleman system is the approximation of $x(t)$, the Carleman truncation error in this simple case takes the form $|x(t)-y_1(t)|$, and we focused on its maximal value in the considered time interval, i.e.,
\begin{equation}
    \epsilon_C := \max_{t\in[0,T]} |x(t)-y_1(t)|.
\end{equation}
In Fig.~\ref{fig:convergence}, we can see the behavior of $\epsilon_C$ as a function of truncation order $k$ for several values of~$R_\delta$. When $R_\delta<1$, we clearly observe exponential convergence of $\epsilon_C$ with growing $k$, illustrating the capability of Carleman embedding to handle time-dependent driving and linear couplings. Since $R_\delta<1$ is only a sufficient condition for convergence, it is not surprising that the error still converges (even though slower) for values of $R_\delta$ slightly above 1 (on the order of 10). However, for large~$R_\delta$, we observe that after the initial decrease of $\epsilon_C$, the error starts increasing with growing $k$, indicating the limits of the Carleman embedding method.

\begin{figure}[t!]
\centering
\includegraphics[width=0.66\textwidth]{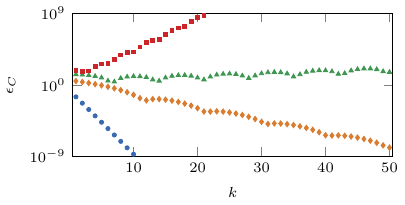}
\caption{\textbf{Carleman error convergence for systems with time-dependent driving and linear coupling.} Carleman truncation error $\epsilon_C$ as function of the truncation order $k$ for a system from Eq.~\eqref{eq:time_dep_numerical} and the Carleman embedding method tailored to handle time-dependent driving and linear couplings, as described in Section~\ref{sec:carleman_time_dependent}. The model parameters are $a=1/20$, $b=-2$, $c=2$, $\omega=1/2$, $T=50$, and the values of $\delta$, bottom to top, are $\delta\in\{25,4.618,2.663,0.56\}$. The resulting values of $R_\delta$, bottom to top, are $R_\delta\approx8.74 \cdot 10^j$ for $j\in\{-1,0,1,2\}$. Note that for $\delta=25$ and $\delta=0.56$, the Carleman error keeps exponentially decreasing and increasing, respectively, so the corresponding curves fall outside the visible region of the plot. }
\label{fig:convergence}
\end{figure}

As a final comment, observe that using a construction analogous to the one presented in this section, one can model time-dependent nonlinearities of order $l$, described by a matrix $F_l(t)$, using an enlarged system with time-independent nonlinearity of order $l+1$.

\newpage 


\section{Nonresonant systems}
\label{sec:nonresonant_systems}

In this section, we study nonresonant differential equations and analyze the convergence behavior of Carleman linearization applied to them. We begin by developing a framework that allows the associated Carleman matrices to be explicitly diagonalized through similarity transformations. Building on this framework, we derive upper bounds on the error introduced by truncating the Carleman embedding. These bounds reveal that the convergence of the Carleman scheme depends crucially on the geometric configuration of the eigenvalues of $F_1$—specifically, on whether the origin of the complex plane lies within their convex hull. Equivalently, this condition distinguishes whether the eigenvalues fall within the \emph{Poincar\'{e}} or \emph{Siegel domain}. While we prove convergence of the Carleman scheme in the Poincar\'{e} domain, its behavior in the Siegel domain is more subtle and remains an open problem.


\subsection{Setting}
\label{subsec:standard_form_nonresonant_system}

We consider a quadratic ODE system without driving:
\begin{align}
\dot{x}(t) = F_1 x(t) + F_2 x(t)^{\otimes 2},    
\label{eq:general_ode_system}
\end{align}
where $F_2 \in \myC^{N \times N^2}$, $F_1 \in \myC^{N \times N}$, $x(t) \in \myC^N$. We shall assume that $F_1$ is diagonalizable with eigenvalues $\lambda_1,\lambda_2, \dots, \lambda_N$ such that $\realpart{\lambda_i} \le 0$ for each $i \in [N]$. That is, there exists an invertible $N \times N$ matrix $Q$ such that
\begin{align}
F_1=Q\Lambda Q^{-1}
\end{align}
where $\Lambda = \diag{\lambda_1, \lambda_2,\dots, \lambda_N}$. We focus on the case where the system is \emph{nonresonant}, i.e., for all $i \in [N]$, 
\begin{align}
\lambda_i \neq \sum_{j=1}^N \alpha_j \lambda_j,~~~\forall \alpha_j \in \myN_0,~~~\sum_{j=1}^N \alpha_j \ge 2.    
\end{align}
The system is said to be \emph{resonant} if this condition is violated. The eigenvalue $\lambda_i$ is a called a \emph{resonant eigenvalue} of order $m \ge 2$ if there exist nonnegative integers $\alpha_1, \alpha_2, \dots, \alpha_N$ such that  $\lambda_i=\sum_{j=1}^N \alpha_j \lambda_j$ and $\sum_{j=1}^N \alpha_j=m$.

We can now rewrite $F_2$ in the eigenbasis of $F_1$ as 
\begin{align}
    \tilde{F}_2 \defeq Q^{-1} F_2 Q^{\otimes 2}.  
\end{align}
Defining $\tilde{x}(t)=Q^{-1}x(t)$ and substituting this into the system, we obtain
\begin{align}
 \dot{\tilde{x}}(t) = \Lambda \tilde{x}(t) + \tilde{F}_2 \tilde{x}(t)^{\otimes 2}.   
 \label{eq:simple_ode_system}
\end{align}
The solution $x(t)$ of Eq.~\eqref{eq:general_ode_system} and the solution $\tilde{x}(t)$ of Eq.~\eqref{eq:simple_ode_system} are related by the linear transformation $x(t)=Q \tilde{x}(t)$. 

Furthermore, the Carleman linearization of Eq.~\eqref{eq:general_ode_system} is also related to that of Eq.~\eqref{eq:simple_ode_system} in a similar way. More precisely,  the order-$k$ Carleman linearization of the original system~\eqref{eq:general_ode_system} takes the form
\begin{align}
    \frac{d}{dt}\begin{pmatrix}
    y^{[1]} \\
    y^{[2]} \\
    y^{[3]} \\
    \vdots \\
    y^{[k-1]} \\
    y^{[k]}
    \end{pmatrix}=\begin{pmatrix}
    A_{1,1} & A_{1,2} & & & &\\
     & A_{2,2} & A_{2,3} & & &\\
            &  & A_{3,3} & A_{3,4} &       & \\
            &         & \ddots   & \ddots   & \ddots & \\
            &         &         &  &A_{k-1,k-1}  & A_{k-1,k}\\      
            &         &         &         &   & A_{k,k}
    \end{pmatrix}\begin{pmatrix}
    y^{[1]} \\
    y^{[2]} \\
    y^{[3]} \\
    \vdots \\
    y^{[k-1]} \\
    y^{[k]}
    \end{pmatrix},      
    \label{eq:carleman_linearization_def}
\end{align}
where $y(t) = [y^{[1]}(t), y^{[2]}(t), \dots, y^{[k]}(t)]^T$
with $y^{[j]}(t) \in \myC^{N^j}$ for $j \in [k]$, and 
\begin{align}
A_{j,j}&=\sum_{l=0}^{j-1} I^{\otimes l} \otimes F_1 \otimes I^{\otimes j-l-1} \in \myC^{N^j\times N^j},
\label{eq:carleman_linearization_def_ajj}\\ 
A_{j,j+1}&=\sum_{l=0}^{j-1} I^{\otimes l} \otimes F_2 \otimes I^{\otimes j-l-1} \in \myC^{N^{j}\times N^{j+1}}.
\label{eq:carleman_linearization_def_ajj1}    
\end{align}
On the other hand, the order-$k$ Carleman linearization of  Eq.~\eqref{eq:simple_ode_system} is 
\begin{align}
    \frac{d}{dt}\begin{pmatrix}
    \tilde{y}^{[1]} \\
    \tilde{y}^{[2]} \\
    \tilde{y}^{[3]} \\
    \vdots \\
    \tilde{y}^{[k-1]} \\
    \tilde{y}^{[k]}
    \end{pmatrix}=\begin{pmatrix}
    \tilde{A}_{1,1} & \tilde{A}_{1,2} & & & &\\
     & \tilde{A}_{2,2} & \tilde{A}_{2,3} & & &\\
            &  & \tilde{A}_{3,3} & \tilde{A}_{3,4} &       & \\
            &         & \ddots   & \ddots   & \ddots & \\
            &         &         &  &\tilde{A}_{k-1,k-1}  & \tilde{A}_{k-1,k}\\      
            &         &         &         &   & \tilde{A}_{k,k}
    \end{pmatrix}\begin{pmatrix}
    \tilde{y}^{[1]} \\
    \tilde{y}^{[2]} \\
    \tilde{y}^{[3]} \\
    \vdots \\
    \tilde{y}^{[k-1]} \\
    \tilde{y}^{[k]}
    \end{pmatrix},      
    \label{eq:carleman_linearization_simple_system_def}
\end{align}
where $\tilde{y}(t) = [\tilde{y}^{[1]}(t), \tilde{y}^{[2]}(t), \dots, \tilde{y}^{[k]}(t)]^T$
with $\tilde{y}^{[j]}(t) \in \myC^{N^j}$ for $j \in [k]$, and
\begin{align}
\tilde{A}_{j,j}&=\sum_{l=0}^{j-1} I^{\otimes l} \otimes \Lambda \otimes I^{\otimes j-l-1} \in \myC^{N^j\times N^j}, \\    
\tilde{A}_{j,j+1}&=\sum_{l=0}^{j-1} I^{\otimes l} \otimes \tilde{F}_2 \otimes I^{\otimes j-l-1} \in \myC^{N^{j}\times N^{j+1}}.
\end{align}
Observe now that the blocks of $A$ and $\tilde{A}$ are related by the following transformations:
\begin{align}
A_{i,j}&=Q^{\otimes i} \tilde{A}_{i,j} (Q^{-1})^{\otimes j},~~~\forall 1\le i,j \le k,
\label{eq:aij_tildeaij_relation}
\end{align}
where $A_{i,j}=\tilde{A}_{i,j}=0$ for all $(i,j)\in [k] \times [k]$ satisfying $j-i \not\in \{0,1\}$. This enables us to reduce the diagonalization of $A$ to that of $\tilde{A}$. Moreover, the solution ${y}(t) = [{y}^{[1]}(t), {y}^{[1]}(t), \dots, {y}^{[k]}(t)]^T$ of Eq.~\eqref{eq:carleman_linearization_def} and
the solution $\tilde{y}(t) \defeq [\tilde{y}^{[1]}(t), \tilde{y}^{[1]}(t), \dots, \tilde{y}^{[k]}(t)]^T$ of Eq.~\eqref{eq:carleman_linearization_simple_system_def} are related as follows:
\begin{align}    
y^{[i]}(t) = Q^{\otimes i} \tilde{y}^{[i]}(t), ~~~\forall i \in [k]. 
\end{align}
This allows us to convert an error bound for  Eq.~\eqref{eq:simple_ode_system} into one for Eq.~\eqref{eq:general_ode_system}, using the inequality:
\begin{align}
\norm{{y}^{[i]}(t)-{x}(t)^{\otimes i}} \le  \norm{Q}^i \cdot  \norm{\tilde{y}^{[i]}(t)-\tilde{x}(t)^{\otimes i}}, ~~~\forall i \in [k].
\end{align}


\subsection{Poincar\'{e}-Dulac normal forms and Carleman linearization}
\label{subsec:poincare_dulac_carleman_relation}

Understanding the local behavior of nonlinear differential equations near fixed points often involves reducing the system to a simpler form. The Poincar\'{e}–Dulac theorem \cite{arnold2012geometrical} facilitates this by using formal series transformations to bring the system into a normal form that highlights the role of resonant terms. In contrast, Carleman linearization embeds the nonlinear system into an infinite-dimensional linear system, making it amenable to linear techniques. In this section, we explain how these two approaches are connected, particularly through the influence of resonances and the structure of the Carleman matrix in its Jordan normal form.

\begin{thm}[Poincar\'{e}-Dulac theorem]
    Consider an ODE system of the form
    \begin{align}
        \dot{x}=\Lambda x + g(x),   
    \end{align}
    where $x\in\myC^N$, $\Lambda=\diag{\lambda_1, \lambda_2, \dots, \lambda_N}$, 
    and $g(x)$ consists of  higher-degree terms in $x_1,x_2,\dots,x_N$. If this system is nonresonant, then it can be transformed into a canonical form  
    \begin{align}
        \dot{y}=\Lambda y + \myO{\norm{y}^{k+1}}    
    \end{align}
    via a change of variables $y_i=x_i+p_i(x)$, where $p_i$ is a polynomial of degree $\le k$, for $i \in [N]$. \\
    If the system is resonant, the canonical form becomes 
    \begin{align}
        \dot{y}=\Lambda y + w(y) + \myO{\norm{y}^{k+1}}    ,
    \end{align}
    where $w(y)$ consists only of the resonant monomials.
    \label{thm:poincare-dulac}
\end{thm}
Here a monomial $x_1^{m_1}x_2^{m_2}\dots x_N^{m_N} e_s$ is \emph{resonant} if and only if $\lambda_s=\sum_{j=1}^N m_j \lambda_j$, where $s \in [N]$, $m_j \in \myN_0$, $\sum_{j=1}^N m_j \ge 2$, and $e_s \in \mathbb{R}^N$ is the vector with 
$1$ in the $s$-th entry and $0$ elsewhere.

As detailed in Ref.~\cite{tsiligiannis1989normal}, the Poincar\'{e}-Dulac theorem is closely linked to the Jordan normal form of the Carleman matrix associated with the system. Specifically, by the results of Section~\ref{subsec:carleman_matrix_diagonalization}, 
when the system in Theorem \ref{thm:poincare-dulac} is nonresonant, its associated order-$k$ Carleman matrix $A$ can be diagonalized via a similarity transformation:
\begin{align}
    A=V D V^{-1},
\end{align}
where
\begin{align}
    D=\begin{pmatrix}
        \Lambda_1  &           &        & \\
                 & \Lambda_2 &        & \\
                 &           & \ddots & \\
                 &           &      &\Lambda_k \end{pmatrix},~~~
        V=\begin{pmatrix}
        I  & {V}_{1,2}   &  {V}_{1,3} & \dots  & {V}_{1,k} \\
                 & I^{\otimes 2}   &  {V}_{2,3} & \dots & {V}_{2,k} \\
                 &           & \ddots &   \ddots       & \vdots \\
                 &           &          & I^{\otimes (k-1)} & {V}_{k-1, k} \\
                &           &          &  & I^{\otimes k} \\\end{pmatrix},     
\end{align}
with
\begin{align}
    \Lambda_j=\diag{\sum_{l=1}^j \lambda_{i_l}:~i_1, i_2,\dots,i_j\in [N]}, ~~~\forall j \in [k].
\end{align}

Let $u=[x, x^{\otimes 2}, \dots, x^{\otimes k}]^T$ and define $z=V^{-1}u$. Then we have
\begin{align}
    \dot{z}=Dz+\myO{\norm{x}^{k+1}},
\end{align}
which follows from
\begin{equation}
\begin{aligned}
    \dot{z}&=V^{-1}\dot{u}\\
    &= V^{-1}\lrb{Au+\myO{\norm{x}^{k+1}}} \\
    &=DV^{-1}u+\myO{\norm{x}^{k+1}}\\
    &=Dz+\myO{\norm{x}^{k+1}}.    
\end{aligned}
\end{equation}
Writing $z=[z^{[1]},z^{[2]},\dots,z^{[k]}]^T$, where $z^{[j]} \in \myC^{N^j}$, we obtain
\begin{align}
\dot{z}^{[1]}=\Lambda z^{[1]}+\myO{\norm{z^{[1]}}^{k+1}},    
\end{align}
where each component $z^{[1]}_i=x_i+p_i(x)$, with  $\operatorname{deg}(p_i) \le k$, for $i \in [N]$.

Therefore, the coordinate transformation $y_i = x_i + p_i(x)$ in the Poincar\'{e}–Dulac theorem corresponds to (part of) the change of basis defined by $V^{-1}$ which  diagonalizes the Carleman matrix~$A$. Also, note that the Carleman matrix remains diagonalizable for certain resonant systems. In other cases, resonance leads to defective Jordan blocks, which manifest as resonant monomials in the canonical form of the system. This perspective provides an alternative interpretation of the Poincar\'{e}–Dulac theorem through the lens of Carleman linearization.


\subsection{The diagonalizing structure for Carleman matrices}
\label{subsec:carleman_matrix_diagonalization}

We now develop a novel framework that compactly characterizes each block matrix of $V$ and $V^{-1}$, which together diagonalize the Carleman system. This framework serves as the foundation for the error analysis of the Carleman scheme applied to nonresonant systems presented in Section~\ref{subsec:error_analysis_nonresonant_systems}. This framework allows us to constructively prove that the Carleman matrices associated with nonresonant differential equations are diagonalizable via similarity transformations, and specify conditions under which resonant systems are also diagonalizable. Every diagonalizable case then provides a criterion for which exponential convergence holds for the Carleman method.

The central result of this framework is given, at the high-level, by the following concise expressions for diagonalizable Carleman systems:
\begin{align}
    V_{i,j} & = \sum_{\mathbf{T}\in \mathcal{T}^i_j} f(\mathbf{T}), 
    \label{eq:V_structure}\\
    (V^{-1})_{i,j} & = (-1)^{j-i}\sum_{\mathbf{T}\in \mathcal{T}^i_j} g(\mathbf{T}),
    \label{eq:V_Vinverse_structure}
\end{align}
for every pair of indices $(i,j$) for which $1\le i\le j \le k$. All other blocks are zero. Here the sum is taken over all \emph{binary forests} $\mathbf{T}$ consisting of  $i$ trees and a total of $j$ leaves. The functions $f$ and $g$ are operator-valued functions that depend on both $F_1$ and $F_2$. 

The binary-forest formalism can be understood as a bookkeeping device for the different ways in which the quadratic nonlinearity propagates through the Carleman hierarchy. In the eigenbasis of $F_1$, each off-diagonal block $\tilde A_{m,m+1}$ arises from the quadratic map $\tilde F_2$, and maps $m+1$ tensor factors to $m$ tensor factors by replacing two adjacent factors with one. We refer to this operation as a \emph{fusion}. Thus, starting from $j$ input tensor factors and ending with $i$ output tensor factors requires exactly $j-i$ fusions. The different possible locations and nesting patterns of these fusions are naturally indexed by binary forests with $j$ leaves and $i$ roots. The leaves represent the input tensor factors, the internal vertices represent applications of $\tilde F_2$, and the roots represent the output tensor factors. In this way, a binary tree describes one possible history by which several input factors are successively fused into one output factor, while a binary forest describes several such histories occurring in parallel. The corresponding tree functions are first defined for the blocks of $\tilde V$ and $\tilde V^{-1}$ in the eigenbasis of $F_1$; they combine the vertex factors coming from $\tilde F_2$ with denominator factors determined by the eigenvalues of $F_1$. The functions $f$ and $g$ appearing in Eqs.~\eqref{eq:V_structure} and~\eqref{eq:V_Vinverse_structure} are then obtained by converting these expressions back to the original coordinates using $Q$ and $Q^{-1}$. 

For example, the matrix $(V^{-1})_{2,4}$ is given as follows. First, note that there are 5~distinct binary forests with 2~trees and a total of 4~leaves. By summing over these forests, we obtain
\begingroup
\setlength{\jot}{12pt} 
\tikzset{
	level distance=16pt,
	every tree node/.style = {align=center, anchor=north,nodecolor},
	edge from parent/.style = {draw,thick,edgecolor},
	edge from parent path = {
		([yshift=4.5pt]\tikzparentnode.south) -- 
		([yshift=-3.5pt]\tikzchildnode.north)
	}
}
\begin{align}
    \lrb{V^{-1}}_{2,4}=&
	g\left(
	\raisebox{-0.6cm}{
		\begin{tikzpicture}[scale=1]
			\Tree [.$\bt$ 
			[.$\bt$ [.$\bt$ ] [.$\bt$ ] ] 
			[.$\bt$ ] 
			];		
			\node at (0,0) {$a$};
			\node at (-0.6,-0.75) {$b$};
			\node at (0.6,-0.75) {$c$};
			\node at (-0.9,-1.35) {$d$};
			\node at (0.15,-1.35) {$e$};
		\end{tikzpicture}
		\hspace{\treeTree}
		\begin{tikzpicture}[scale=1]
			\Tree [.$\bt$  ];
			\node at (0,0.05) {$f$};
			\node at (0,-0.2) {\phantom{$f$}};
		\end{tikzpicture}
	}
	\right)
	+
	g\left(
	\raisebox{-0.6cm}{
		\begin{tikzpicture}[scale=1]
			\Tree [.$\bt$ 
			[.$\bt$ ] 
			[.$\bt$ [.$\bt$ ] [.$\bt$ ] ] 
			];
			\node at (0,0) {$a$};
			\node at (-0.6,-0.75) {$b$};
			\node at (0.6,-0.75) {$c$};
			\node at (-0.15,-1.35) {$d$};
			\node at (0.9,-1.35) {$e$};
		\end{tikzpicture}
		\hspace{\treeTree}
		\begin{tikzpicture}[scale=1]
			\Tree [.$\bt$  ];
			\node at (0,0.05) {$f$};
			\node at (0,-0.2) {\phantom{$f$}};
		\end{tikzpicture}
	}
	\right)
	+
	g\left(
	\raisebox{-0.6cm}{
		\begin{tikzpicture}[scale=1]
			\Tree [.$\bt$  ];
			\node at (0,0.05) {$d$};
			\node at (0,-0.2) {\phantom{$f$}};
		\end{tikzpicture}
		\hspace{\treeTree}
		\begin{tikzpicture}[scale=1]
			\Tree [.$\bt$ 
			[.$\bt$ [.$\bt$ ] [.$\bt$ ] ] 
			[.$\bt$ ] ];
			\node at (0,0) {$a$};
			\node at (-0.6,-0.75) {$b$};
			\node at (0.6,-0.75) {$c$};
			\node at (-0.9,-1.35) {$e$};
			\node at (0.15,-1.35) {$f$};
		\end{tikzpicture}
	}
	\right)
	\nonumber
	\\
	&
	+
	g\left(
	\raisebox{-0.6cm}{
		\begin{tikzpicture}[scale=1]
			\Tree [.$\bt$  ];
			\node at (0,0.05) {$d$};
			\node at (0,-0.2) {\phantom{$f$}};
		\end{tikzpicture}
		\hspace{\treeTree}
		\begin{tikzpicture}[scale=1]
			\Tree [.$\bt$ 
			[.$\bt$ ] 
			[.$\bt$ [.$\bt$ ] [.$\bt$ ] ] 
			];
			\node at (0,0) {$a$};
			\node at (-0.6,-0.75) {$b$};
			\node at (0.6,-0.75) {$c$};
			\node at (-0.15,-1.35) {$e$};
			\node at (0.9,-1.35) {$f$};
		\end{tikzpicture}
	}
	\right)
	+
	g\left(
	\raisebox{-0.35cm}{
		\begin{tikzpicture}[scale=1]
			\Tree [.$\bt$ [.$\bt$ ] [.$\bt$ ] ];
			\node at (0,0.05) {$a$};
			\node at (-0.5,-0.75) {$c$};
			\node at (0.5,-0.75) {$d$};
		\end{tikzpicture}
		\hspace{\treeTree}
		\begin{tikzpicture}[scale=1]
			\Tree [.$\bullet$ [.$\bt$ ] [.$\bt$ ] ];
			\node at (0,0.05) {$b$};
			\node at (-0.5,-0.75) {$e$};
			\node at (0.5,-0.75) {$f$};
		\end{tikzpicture}
	}
	\right).
\end{align}
\endgroup

From this graphical specification of $V_{i,j}$ and $(V^{-1})_{i,j}$, we derive corresponding norm bounds for these operators by leveraging structural properties of the spectrum of $F_1$. In particular, for nonresonant systems, the scaling behavior of $\|(V^{-1})_{i,j}\|$, as determined through the binary forest formalism, plays an essential role in establishing the Carleman convergence criterion presented later. 
The fact that binary trees/forests arise is due to the quadratic term $F_2$, which fuses two copies of the system into one. However, crucial to our analysis is the structure of the spectrum of $F_1$. More precisely, we find that the following two properties play important roles: (a) whether $F_1$ has resonant or nonresonant eigenvalues, and (b) whether the spectrum of $F_1$ lies in the Poincar\'e Domain (i.e. the origin is outside the convex hull of the eigenvalues) or Siegel Domain (i.e. the origin lies inside the convex hull of the eigenvalues). 

The discussion above provides a high-level overview of the theory we develop. We now begin the analysis by first specifying the recursive structure that the block matrices must satisfy, and then gradually build up the necessary framework to prove Eqs.~\eqref{eq:V_structure} and \eqref{eq:V_Vinverse_structure}.

Let $\tilde{A}$ denote the order-$k$ Carleman matrix associated with the ODE system~\eqref{eq:simple_ode_system}, i.e. the coefficient matrix of the Carleman system~\eqref{eq:carleman_linearization_simple_system_def}. We will show that when the system is nonresonant, $\tilde{A}$ is diagonalizable via a similarity transformation. We first note that $\tilde{A}$ is block upper triangular. So its eigenvalues are the union of the eigenvalues of $\tilde{A}_{1,1}, \tilde{A}_{2,2}, \dots, \tilde{A}_{k,k}$. The no-resonance condition ensures that $\tilde{A}_{1,1}$ and $\tilde{A}_{j,j}$ do not share common eigenvalues for each $j\ge 2$. But it is possible that $\tilde{A}_{j_1, j_1}$ and $\tilde{A}_{j_2, j_2}$ share common eigenvalues for $j_1, j_2 \ge 2$. Despite this, $\tilde{A}$ remains diagonalizable under the no-resonance condition. 

In fact, we will prove this statement by explicitly diagonalizing the matrix $\tilde{A}$. Specifically, we will construct an invertible matrix $\tilde{V}$ such that 
\begin{align}
    \tilde{A}=\tilde{V} D \tilde{V}^{-1},    
\end{align}
where 
\begin{align}
    D=\begin{pmatrix}
        \Lambda_1  &           &        & \\
                 & \Lambda_2 &        & \\
                 &           & \ddots & \\
                 &           &      &\Lambda_k \end{pmatrix},~~~
        \tilde{V}=\begin{pmatrix}
        \tilde{V}_{1,1}  & \tilde{V}_{1,2}   &  \tilde{V}_{1,3} & \dots  & \tilde{V}_{1,k} \\
                 & \tilde{V}_{2,2}   &  \tilde{V}_{2,3} & \dots & \tilde{V}_{2,k} \\
                 &           & \ddots &   \ddots       & \vdots \\
                 &           &          & \tilde{V}_{k-1,k-1} & \tilde{V}_{k-1, k} \\
                &           &          &  & \tilde{V}_{k, k} \\\end{pmatrix}, 
                \label{eq:def_D_tildeV}
\end{align}
with
\begin{align}
    \Lambda_j=\diag{\sum_{l=1}^j \lambda_{i_l}:~i_1, i_2,\dots,i_j\in [N]}, ~~~\forall j \in [k],    
\end{align}
and $\tilde{V}_{i,j} \in \myC^{N^i \times N^j}$ for $1 \le i \le j \le k$.

An explicit diagonalization of $\tilde{A}$ can be converted into one for $A$ using Eq.~\eqref{eq:aij_tildeaij_relation}. Specifically, we have
\begin{align}
    A=V D V^{-1},    
\end{align}
where $D$ is as defined earlier, and
\begin{align}
        {V}=\begin{pmatrix}
        {V}_{1,1}  & {V}_{1,2}   &  {V}_{1,3} & \dots  & {V}_{1,k} \\
                 & {V}_{2,2}   &  {V}_{2,3} & \dots & {V}_{2,k} \\
                 &           & \ddots &   \ddots       & \vdots \\
                 &           &          & {V}_{k-1,k-1} & {V}_{k-1, k} \\
                &           &          &  & {V}_{k, k} \\\end{pmatrix},     
\end{align}
with 
\begin{align}
    V_{i,j}=Q^{\otimes i} \tilde{V}_{i,j},\quad \forall 1\le i\le j\le k.
\end{align}

Observe that $\tilde{A}=\tilde{V}D\tilde{V}^{-1}$ is equivalent to 
\begin{align}
    \tilde{A}\tilde{V}=\tilde{V}D,    
\end{align}
provided that $\tilde{V}$ is invertible. Using the block-structures of $\tilde{A}$, $D$ and $\tilde{V}$, we can reduce this condition to
\begin{align}
\label{eq:V_recurrence}
    \tilde{V}_{i,j}\Lambda_j-\Lambda_i \tilde{V}_{i,j} =\tilde{A}_{i,i+1}\tilde{V}_{i+1,j},~~~\forall 1\le i<j \le k,
\end{align}
as well as
\begin{align}
    \Lambda_i \tilde{V}_{i,i}=\tilde{V}_{i,i}\Lambda_i,~~~\forall 1\le i\le k.
\end{align}
We will set $\tilde{V}_{i,i}=I^{\otimes i}$ for each $i$, and then the second condition is automatically fulfilled. Furthermore, since $\tilde{V}$ is an upper-triangular matrix with ones on the diagonal, all its eigenvalues are equal to $1$, and therefore $\tilde{V}$ is invertible.
Thus, it remains to find the matrices $\tilde{V}_{i,j}$ for $i<j$ which satisfy the first condition.

\subsubsection{A guiding example}
\label{subsubsec:guiding_example}

Before presenting the expressions for the matrices $\tilde{V}_{i,j}$, it is helpful to illustrate the underlying idea behind their generation through a simple example. Let us consider the case $k=3$. We need to find $\tilde{V}_{1,2}$, $\tilde{V}_{2,3}$ and $\tilde{V}_{1,3}$ such that
\begin{align}
        \begin{pmatrix}
        \Lambda_1  & \tilde{A}_{1,2}   &    \\
                 & \Lambda_2   &  \tilde{A}_{2,3}  \\
                 &           & \Lambda_3  
                 \end{pmatrix}
        \begin{pmatrix}
        I  & \tilde{V}_{1,2}   &  \tilde{V}_{1,3}  \\
                 & I   &  \tilde{V}_{2,3}  \\
                 &           &  I  
                 \end{pmatrix}
                 =\begin{pmatrix}
        I  & \tilde{V}_{1,2}   &  \tilde{V}_{1,3}  \\
                 & I   &  \tilde{V}_{2,3}  \\
                 &           &  I  
                 \end{pmatrix}
                 \begin{pmatrix}
        \Lambda_1  &     &    \\
                 & \Lambda_2   &    \\
                 &           & \Lambda_3  
                 \end{pmatrix},
\end{align}                 
which is equivalent to
\begin{align}
    -\Lambda_1 \tilde{V}_{1,2} + \tilde{V}_{1,2} \Lambda_2 &= \tilde{A}_{1,2}, \label{eq:diag_ex_1}\\    
    -\Lambda_1 \tilde{V}_{1,3} + \tilde{V}_{1,3} \Lambda_3 &= \tilde{A}_{1,2}\tilde{V}_{2,3},\label{eq:diag_ex_2}\\
    -\Lambda_2 \tilde{V}_{2,3} + \tilde{V}_{2,3} \Lambda_3 &= \tilde{A}_{2,3}.\label{eq:diag_ex_3}
\end{align}

We first derive the expression for $\tilde{V}_{1,2}$ using Eq.~\eqref{eq:diag_ex_1}. Note that 
\begin{align}
    -\Lambda_1 \tilde{V}_{1,2} + \tilde{V}_{1,2} \Lambda_2 = M_2 \odot \tilde{V}_{1,2},    
\end{align}
where 
\begin{equation}
    M_2=\sum_{i=1}^N \sum_{i_1,i_2=1}^N \lrb{\lambda_{i_1}+\lambda_{i_2}-\lambda_{i}} \ketbra{i}{i_1,i_2}    
\end{equation}
has non-zero entries only by the no-resonance condition, and $\odot$ is the Hadamard (entrywise) product of matrices. Then, since $\tilde{A}_{1,2}=\tilde{F}_2$, we get
\begin{align}
    \tilde{V}_{1,2}=N_2 \odot \tilde{F}_2,    
\end{align}
where
\begin{align}
    N_2=\sum_{i=1}^N \sum_{i_1,i_2=1}^N \frac{1}{\lambda_{i_1}+\lambda_{i_2}-\lambda_{i}} \ketbra{i}{i_1,i_2}
\end{align}
is well-defined.

Next, we derive the expression for $\tilde{V}_{2,3}$ using Eq.~\eqref{eq:diag_ex_3}. This is slightly more complicated than the procedure for $\tilde{V}_{1,2}$. Note that
\begin{align}
    -\Lambda_2 \tilde{V}_{2,3} + \tilde{V}_{2,3} \Lambda_3 = M_{2,3} \odot \tilde{V}_{2,3},    
\end{align}
where 
\begin{align}
M_{2,3}=\sum_{i_1,i_2=1}^N \sum_{j_1,j_2,j_3=1}^N \lrb{\lambda_{j_1}+\lambda_{j_2}+\lambda_{j_3}-\lambda_{i_1}-\lambda_{i_2}} \ketbra{i_1,i_2}{j_1,j_2,j_3}.    
\end{align}
The entries of $M_{2,3}$ might be zero, and it seems impossible to apply the previous strategy to this scenario. However, we can perform a similar procedure by exploiting the structure of $\tilde{A}_{2,3}$, i.e.
\begin{align}
\tilde{A}_{2,3}=\tilde{F}_2 \otimes I + I \otimes \tilde{F}_2.    
\end{align}
We aim to find two matrices $\tilde{V}_{2,3}^{(1)}$ and $\tilde{V}_{2,3}^{(2)}$ such that
\begin{align}
M_{2,3}\odot \tilde{V}_{2,3}^{(1)} &= \tilde{F}_2 \otimes I, \\   
M_{2,3}\odot \tilde{V}_{2,3}^{(2)} &= I \otimes \tilde{F}_2. 
\end{align}
Then, it suffices to set $\tilde{V}_{2,3}=\tilde{V}_{2,3}^{(1)}+\tilde{V}_{2,3}^{(2)}$. To find the desired $\tilde{V}_{2,3}^{(1)}$, we note that
\begin{align}
\tilde{F}_2 \otimes I=\sum_{i=1}^N \sum_{j_1,j_2=1}^N \sum_{l=1}^N \bra{i}\tilde{F}_2\ket{j_1,j_2} \ketbra{i,l}{j_1,j_2,l},    
\end{align}
and
\begin{align}
M_{2,3}\odot \tilde{V}_{2,3}^{(1)} = \sum_{i_1,i_2=1}^N \sum_{j_1,j_2,j_3=1}^N \lrb{\lambda_{j_1}+\lambda_{j_2}+\lambda_{j_3}-\lambda_{i_1}-\lambda_{i_2}}
\bra{i_1,i_2}\tilde{V}_{2,3}^{(1)}\ket{j_1,j_2,j_3} \ketbra{i_1,i_2}{j_1,j_2,j_3}.
\end{align}
To make these expressions identical, we need to set $i_2=j_3$ in the second equation. Then, it follows that 
\begin{align}
M_{2,3}\odot \tilde{V}_{2,3}^{(1)} = \sum_{i=1}^N \sum_{j_1,j_2,l=1}^N \lrb{\lambda_{j_1}+\lambda_{j_2}-\lambda_{i}}
\bra{i,l}\tilde{V}_{2,3}^{(1)}\ket{j_1,j_2,l} \ketbra{i,l}{j_1,j_2,l}.
\end{align}
Therefore, we could choose
\begin{align}
\tilde{V}_{2,3}^{(1)} = \sum_{i=1}^N \sum_{j_1,j_2,l=1}^N
\frac{\bra{i}\tilde{F}_2\ket{j_1,j_2}}{\lambda_{j_1}+\lambda_{j_2}-\lambda_i} \ketbra{i,l}{j_1,j_2,l}
= (N_2 \odot \tilde{F}_2) \otimes I.
\end{align}
By swapping the role of the first and second subsystems, we obtain the expression for $\tilde{V}_{2,3}^{(2)}$:
\begin{align}
\tilde{V}_{2,3}^{(2)} = I \otimes (N_2 \odot \tilde{F}_2).
\end{align}
This leads to the final expression for $\tilde{V}_{2,3}$:
\begin{align}
\tilde{V}_{2,3}=  (N_2 \odot \tilde{F}_2) \otimes I +  I \otimes (N_2 \odot \tilde{F}_2). 
\end{align}

Finally, we will derive the expression for $\tilde{V}_{1,3}$ using Eq.~\eqref{eq:diag_ex_2} as well as the expression for $\tilde{V}_{2,3}$. Note that
\begin{align}
-\Lambda_1 \tilde{V}_{1,3} + \tilde{V}_{1,3} \Lambda_3 = M_3 \odot \tilde{V}_{1,3},
\end{align}
where 
\begin{align}
M_3=\sum_{i=1}^N \sum_{i_1,i_2,i_3=1}^N \lrb{\lambda_{i_1}+\lambda_{i_2}+\lambda_{i_3}-\lambda_{i}} \ketbra{i}{i_1,i_2,i_3}
\end{align}
has non-zero entries only by the no-resonance condition. Therefore, we need to set
\begin{align}
\tilde{V}_{1,3}=N_3 \odot (\tilde{A}_{1,2}\tilde{V}_{2,3})
= N_3 \odot \lrcb{\tilde{F}_2 \lrsb{\lrb{N_2 \odot \tilde{F}_2} \otimes I}}
+ N_3 \odot \lrcb{\tilde{F}_2 \lrsb{I \otimes \lrb{N_2 \odot \tilde{F}_2}}},
\end{align}
where
\begin{align}
N_3=\sum_{i=1}^N \sum_{i_1,i_2,i_3=1}^N \frac{1}{\lambda_{i_1}+\lambda_{i_2}+\lambda_{i_3}-\lambda_{i}} \ketbra{i}{i_1,i_2,i_3}
\end{align}
is well-defined.

This finishes the construction of $\tilde{V}_{1,2}$, $\tilde{V}_{2,3}$ and $\tilde{V}_{1,3}$ for the case $k=3$. We remark that even when the truncation order $k$ exceeds $3$, the expressions for $\tilde{V}_{1,2}$, $\tilde{V}_{2,3}$ and $\tilde{V}_{1,3}$ with respect to the order-$k$ Carleman matrix $\tilde{A}$ will remain identical to those given above.

\subsubsection{Blockwise expressions for the diagonalizing matrix}
\label{subsec:expression_for_vij_diagonal_f1}

The essence of our strategy for constructing the matrices $\tilde{V}_{i,j}$ can be gleaned from the example above. We derive their expressions iteratively, in a layer-by-layer manner, where $\tilde{V}_{i, j}$ belongs to the $(j-i)$-th layer, for any $1\le i < j \le k$. The expression of $\tilde{V}_{i,j}$ is derived from that of $\tilde{V}_{i+1,j}$ using the identity
\begin{align}
    \tilde{V}_{i,j}\Lambda_j-\Lambda_i \tilde{V}_{i,j} = M_{i,j}\odot \tilde{V}_{i,j} =  \tilde{A}_{i,i+1} \tilde{V}_{i+1,j},    
\label{eq:condition_vij}
\end{align}
where $\tilde{A}_{i,i+1}=\sum_{l=0}^{i-1} I^{\otimes l} \otimes \tilde{F}_2 \otimes I^{\otimes i-l-1}$ (i.e., $\tilde{A}_{i,i+1}$ applies $\tilde{F}_2$ on two adjacent subsystems, with $i$ possible choices for its location), and
\begin{align}
    M_{i,j} \defeq \sum_{a_1,a_2,\dots,a_i=1}^N \sum_{b_1,b_2,\dots,b_j=1}^N \lrb{\sum_{l=1}^j \lambda_{b_l}-\sum_{l=1}^i \lambda_{a_l}} \ketbra{a_1,a_2,\dots,a_i}{b_1,b_2,\dots,b_j}.
\label{eq:def_mij}
\end{align}
Under the no-resonance condition, Eq.~\eqref{eq:condition_vij} is guaranteed to have a solution (as will be proven later in Proposition~\ref{prop:valid_vij}), 
and the resulting $\tilde{V}_{i,j}$ can be written as the sum of multiple terms, where each term is  generated by first applying $\tilde{F}_2$ to two adjacent subsystems of a term in the expression for $\tilde{V}_{i+1,j}$ and then performing a Hadamard product with 
\begin{align}
N_l \defeq \sum_{i=1}^N\sum_{i_1,i_2.\dots,i_l=1}^N \frac{1}{\lambda_{i_1}+\lambda_{i_2}+\dots+\lambda_{i_l}-\lambda_i} \ketbra{i}{i_1,i_2,\dots,i_l}    
\label{eq:def_nl}
\end{align}
for an appropriate $l \in [k]$. We begin the procedure with $\tilde{V}_{i-1,i}$ for $2 \le i \le k$, then proceed to 
 $\tilde{V}_{i-2,i}$ for $3 \le i \le k$, and continue in this manner until finally obtaining the expression for $\tilde{V}_{1,k}$. 

This iterative algebraic construction has a natural graphical interpretation in terms of binary trees, which can be read from the leaves upward. A trivial tree corresponds to a tensor factor that is not acted on, and is therefore assigned the identity. A nontrivial binary tree is evaluated recursively: first evaluate the operators associated with its two subtrees, then apply $\tilde F_2$ to fuse their two outputs into one, and finally multiply entrywise by $N_l$, where $l$ is the number of leaves under the tree. This last step inserts the nonresonance denominator associated with fusing those $l$ leaves into a single output. A forest represents the tensor product of several such independent fusion histories.

Formally, we define the matrices $\tilde{V}_{i,j}$ through the following procedure:
\begin{definition}
For every binary tree $\mbt$, we assign a linear operator $\tilde{f}(\mbt)$ to it. This operator $\tilde{f}(\mbt)$ is defined recursively: 
\begin{itemize}
    \item If $\mbt$ is trivial (i.e. $\mbt$ consists of a single node), then $\tilde{f}(\mbt)=I$, the $N$-dimensional identity operator;
    \item Otherwise, let $r$ be the root of $\mbt$, and let $c_{\mathrm{left}}$ and $c_{\mathrm{right}}$ be the left and right children of $r$, respectively. Define $\mbt_{\mathrm{left}}$ and $\mbt_{\mathrm{right}}$ as the subtrees of $\mbt$ rooted at $c_{\mathrm{left}}$ and $c_{\mathrm{right}}$, respectively. Additionally, let $v(\mbt)$ be the number of leaves in $\mbt$. The operator assigned to $\mbt$ is then defined as
    \begin{align}
        \tilde{f}(\mbt)=N_{v(\mbt)} \odot \lrsb{\tilde{F}_2\lrb{\tilde{f}(\mbt_{\mathrm{left}}) \otimes \tilde{f}(\mbt_{\mathrm{right}})}}.        
    \end{align}
\end{itemize}
Then, for a binary forest $\mbf=(\mbt_1,\mbt_2,\dots,\mbt_m)$, where $\mbt_l$ is the $l$-th tree in it, the operator assigned to the forest is the tensor product of the operators assigned to the constituent trees:
\begin{align}
    \tilde{f}(\mbf)=\tilde{f}(\mbt_1) \otimes \tilde{f}(\mbt_2) \otimes \dots \otimes \tilde{f}(\mbt_m).
\end{align}
Finally, for any $i \le j$,  let $\mathcal{T}^i_j$ denote the set of all binary forests consisting of $i$ trees and a total of $j$ leaves. Then $\tilde{V}_{i,j}$ is defined as the sum of the operators assigned to all forests in $\mathcal{T}^i_j$:
\begin{align}
\tilde{V}_{i,j}=\sum_{\mbf \in \mathcal{T}^i_j} \tilde{f}(\mbf).
\end{align}
\label{def:vij_expression}
\end{definition}

Note that in each term for $\tilde{V}_{i,j}$, the operator $\tilde{F}_2$ (along with appropriate $N_l$ matrices) appears exactly $j-i$ times. Consequently, $\tilde{V}_{i,j}$ maps an $N^j$-dimensional vector into an $N^i$-dimensional vector, as intended.

\begin{prop}
The matrices $\tilde{V}_{i,j}$ specified in Definition~\ref{def:vij_expression} satisfy the following identity:
\begin{align}
\tilde{V}_{i,j}\Lambda_j-\Lambda_i \tilde{V}_{i,j} = \tilde{A}_{i,i+1} \tilde{V}_{i+1,j}.    ~~~\forall i<j.
\end{align}
\label{prop:valid_vij}
Therefore, we have $\tilde{A}=\tilde{V} D \tilde{V}^{-1}$, where $D$ and $\tilde{V}$ are defined in Eq.~\eqref{eq:def_D_tildeV}.
\end{prop}
\begin{proof}
See Appendix~\ref{app:proof_lemma_valid_vij}.
\end{proof}

To illustrate this, below we give the explicit expressions for $\tilde{V}_{i,j}$ with $1\le i<j\le 4$, together with their binary forests representations:

\begingroup
\setlength{\jot}{12pt}
\tikzset{
    level distance=16pt,
    every tree node/.style = {align=center, anchor=north,nodecolor},
    edge from parent/.style = {draw,thick,edgecolor},
    edge from parent path = {
        ([yshift=4.5pt]\tikzparentnode.south) -- 
        ([yshift=-3.5pt]\tikzchildnode.north)
    }
}
\begin{align}
    \tilde{V}_{1,2} & =   N_2 \odot \tilde{F}_2 =\tilde{f}\left(
    \raisebox{-0.35cm}{
        \begin{tikzpicture}[scale=1]
            \Tree [.$\bt$ [.$\bt$ ] [.$\bt$ ] ]
        \end{tikzpicture}
    }
    \right) &
    \\
    \tilde{V}_{2,3} & =
    (N_2 \odot \tilde{F}_2) \otimes I +  I \otimes (N_2 \odot \tilde{F}_2) = \tilde{f}\left(
    \raisebox{-0.35cm}{
        \begin{tikzpicture}[scale=1]
            \Tree [.$\bt$ [.$\bt$ ] [.$\bt$ ] ]
        \end{tikzpicture}
        \hspace{\treeTree}
        \begin{tikzpicture}[scale=1]
            \Tree [.$\bt$ ]
        \end{tikzpicture}
    }
    \right)
    +\tilde{f}\left(
    \raisebox{-0.35cm}{
        \begin{tikzpicture}[scale=1]
            \Tree [.$\bt$ ]
        \end{tikzpicture}
        \hspace{\treeTree}
        \begin{tikzpicture}[scale=1]
            \Tree [.$\bt$ [.$\bt$ ] [.$\bt$ ] ]
        \end{tikzpicture}
    }
    \right) &
    \\
    \tilde{V}_{3,4} & =
    (N_2 \odot \tilde{F}_2) \otimes I \otimes I +  I \otimes (N_2 \odot \tilde{F}_2) \otimes I +  I \otimes I \otimes (N_2 \odot \tilde{F}_2) &\\
    & =
    \tilde{f}\left(
    \raisebox{-0.35cm}{
        \begin{tikzpicture}[scale=1]
            \Tree [.$\bt$ [.$\bt$ ] [.$\bt$ ] ]
        \end{tikzpicture}
        \hspace{\treeTree}
        \begin{tikzpicture}[scale=1]
            \Tree [.$\bt$ ]
        \end{tikzpicture}
        \hspace{\treeTree}
        \begin{tikzpicture}[scale=1]
            \Tree [.$\bt$ ]
        \end{tikzpicture}
    }
    \right)
    +\tilde{f}\left(
    \raisebox{-0.35cm}{
        \begin{tikzpicture}[scale=1]
            \Tree [.$\bt$ ]
        \end{tikzpicture}
        \hspace{\treeTree}
        \begin{tikzpicture}[scale=1]
            \Tree [.$\bt$ [.$\bt$ ] [.$\bt$ ] ]
        \end{tikzpicture}
        \hspace{\treeTree}
        \begin{tikzpicture}[scale=1]
            \tikzset{}
            \Tree [.$\bt$ ]
        \end{tikzpicture}
    }
    \right)
    +\tilde{f}\left(
    \raisebox{-0.35cm}{
        \begin{tikzpicture}[scale=1]
            \Tree [.$\bt$ ]
        \end{tikzpicture}
        \hspace{\treeTree}
        \begin{tikzpicture}[scale=1]
            \Tree [.$\bt$ ]
        \end{tikzpicture}
        \hspace{\treeTree}
        \begin{tikzpicture}[scale=1]
            \Tree [.$\bt$ [.$\bt$ ] [.$\bt$ ] ]
        \end{tikzpicture}
    }
    \right) &
    \\
    \tilde{V}_{1,3} & = 
    N_3 \odot \lrcb{\tilde{F}_2\lrsb{(N_2 \odot \tilde{F}_2) \otimes I}}+ N_3 \odot \lrcb{\tilde{F}_2\lrsb{ I \otimes (N_2 \odot \tilde{F}_2)}} &\\
    & = \tilde{f}\left(
    \raisebox{-0.6cm}{
        \begin{tikzpicture}[scale=1]
            \Tree [.$\bt$ 
            [.$\bt$ [.$\bt$ ] [.$\bt$ ] ] 
            [.$\bt$ ] 
            ]
        \end{tikzpicture}
    }
    \right)
    +\tilde{f}\left(
    \raisebox{-0.6cm}{
        \begin{tikzpicture}[scale=1]
            \Tree [.$\bt$ 
            [.$\bt$ ] 
            [.$\bt$ [.$\bt$ ] [.$\bt$ ] ] 
            ]
        \end{tikzpicture}
    }
    \right)&
    \\
    \tilde{V}_{2,4} & = 
    (N_2 \odot \tilde{F}_2) \otimes (N_2 \odot \tilde{F}_2) + N_3 \odot \lrcb{\tilde{F}_2 \lrsb{(N_2 \odot \tilde{F}_2) \otimes I}} \otimes I  & \nonumber \\
    & \phantom{=} + I \otimes N_3 \odot \lrcb{\tilde{F}_2 \lrsb{(N_2 \odot \tilde{F}_2) \otimes I}} + N_3 \odot \lrcb{\tilde{F}_2 \lrsb{I \otimes (N_2 \odot \tilde{F}_2) }} \otimes I  & \nonumber \\
    & \phantom{=} + I \otimes N_3 \odot \lrcb{\tilde{F}_2 \lrsb{I \otimes (N_2 \odot \tilde{F}_2)}} & \\
    & = 
    \tilde{f}\left(
    \raisebox{-0.35cm}{
        \begin{tikzpicture}[scale=1]
            \Tree [.$\bt$ [.$\bt$ ] [.$\bt$ ] ]
        \end{tikzpicture}
        \hspace{\treeTree}
        \begin{tikzpicture}[scale=1]
            \Tree [.$\bullet$ [.$\bt$ ] [.$\bt$ ] ]
        \end{tikzpicture}
    }
    \right)
    +
    \tilde{f}\left(
    \raisebox{-0.6cm}{
        \begin{tikzpicture}[scale=1]
            \Tree [.$\bt$ 
            [.$\bt$ [.$\bt$ ] [.$\bt$ ] ] 
            [.$\bt$ ] 
            ]
        \end{tikzpicture}
        \hspace{\treeTree}
        \begin{tikzpicture}[scale=1]
            \Tree [.$\bt$  ]
        \end{tikzpicture}
    }
    \right)
    +
    \tilde{f}\left(
    \raisebox{-0.6cm}{
        \begin{tikzpicture}[scale=1]
            \Tree [.$\bt$  ]
        \end{tikzpicture}
        \hspace{\treeTree}
        \begin{tikzpicture}[scale=1]
            \Tree [.$\bt$ 
            [.$\bt$ [.$\bt$ ] [.$\bt$ ] ] 
            [.$\bt$ ] 
            ]
        \end{tikzpicture}
    }
    \right) &  \nonumber \\
    & \phantom{=} 
    +\tilde{f}\left(
    \raisebox{-0.6cm}{
        \begin{tikzpicture}[scale=1]
            \Tree [.$\bt$ 
            [.$\bt$ ] 
            [.$\bt$ [.$\bt$ ] [.$\bt$ ] ] 
            ]
        \end{tikzpicture}
        \hspace{\treeTree}
        \begin{tikzpicture}[scale=1]
            \Tree [.$\bt$  ]
        \end{tikzpicture}
    }
    \right)
    +
    \tilde{f}\left(
    \raisebox{-0.6cm}{
        \begin{tikzpicture}[scale=1]
            \Tree [.$\bt$  ]
        \end{tikzpicture}
        \hspace{\treeTree}
        \begin{tikzpicture}[scale=1]
            \Tree [.$\bt$ 
            [.$\bt$ ] 
            [.$\bt$ [.$\bt$ ] [.$\bt$ ] ] 
            ]
        \end{tikzpicture}
    }
    \right) &
    \\
    \tilde{V}_{1,4} & = 
    N_4 \odot \lrcb{\tilde{F}_2 \lrsb{(N_2 \odot \tilde{F}_2) \otimes (N_2 \odot \tilde{F}_2)}} & \nonumber \\
    &\phantom{=}+ N_4 \odot \lrcb{\tilde{F}_2 \lrsb{N_3 \odot \lrcb{\tilde{F}_2 \lrsb{(N_2 \odot \tilde{F}_2) \otimes I}} \otimes I}} & \nonumber \\
    & \phantom{=} + N_4 \odot \lrcb{\tilde{F}_2 \lrsb{N_3 \odot \lrcb{\tilde{F}_2 \lrsb{I \otimes (N_2 \odot \tilde{F}_2) }} \otimes I}} & \nonumber \\
    &\phantom{=} + N_4 \odot \lrcb{\tilde{F}_2 \lrsb{I \otimes N_3 \odot \lrcb{\tilde{F}_2 \lrsb{(N_2 \odot \tilde{F}_2) \otimes I}}}}& 
    \nonumber \\
    & \phantom{=} + N_4 \odot \lrcb{\tilde{F}_2 \lrsb{I \otimes N_3 \odot \lrcb{\tilde{F}_2 \lrsb{I \otimes (N_2 \odot \tilde{F}_2)}}}} &\\
    & = 
    \tilde{f}\left(
    \raisebox{-0.6cm}{
        \begin{tikzpicture}[scale=1]
            \Tree [.$\bt$ [.$\bt$ [.$\bt$ ] [.$\bt$ ] ] [.$\bt$ [.$\bt$ ] [.$\bt$ ] ] ]
        \end{tikzpicture}				
    }
    \right)
    +
    \tilde{f}\left(
    \raisebox{-0.9cm}{
        \begin{tikzpicture}[scale=1]
            \Tree [.$\bt$ 
            [.$\bt$ [.$\bt$ [.$\bt$ ] [.$\bt$ ] ] [.$\bt$ ] ] 
            [.$\bt$ ] 
            ]
        \end{tikzpicture}
    }
    \right)
    +
    \tilde{f}\left(
    \raisebox{-0.9cm}{
        \begin{tikzpicture}[scale=1]
            \Tree [.$\bt$ 
            [.$\bt$ [.$\bt$ ] [.$\bt$ [.$\bt$ ] [.$\bt$ ] ] ] 
            [.$\bt$ ] 
            ]
        \end{tikzpicture}
    }
    \right) &    \nonumber \\
    & \phantom{=} +
    \tilde{f}\left(
    \raisebox{-0.9cm}{
        \begin{tikzpicture}[scale=1]
            \Tree [.$\bt$ 
            [.$\bt$ ] 
            [.$\bt$ [.$\bt$ [.$\bt$ ] [.$\bt$ ] ] [.$\bt$ ] ] 
            ]
        \end{tikzpicture}
    }
    \right)
    +
    \tilde{f}\left(
    \raisebox{-0.9cm}{
        \begin{tikzpicture}[scale=1]
            \Tree [.$\bt$ 
            [.$\bt$ ] 
            [.$\bt$ [.$\bt$ ] [.$\bt$ [.$\bt$ ] [.$\bt$ ] 	] ] 
            ]
        \end{tikzpicture}
    }
    \right) & 
\end{align}
\endgroup

The expressions for $V_{i,j}$ can be obtained from  those of $\tilde{V}_{i,j}$ by left-multiplying with $Q^{\otimes i}$.
Specifically, for any $\mbf \in \mathcal{T}^i_j$,
we define 
\begin{align}
f(\mbf)=Q^{\otimes i}\tilde{f}(\mbf).    
\end{align}
Then 
\begin{align}
    V_{i,j}= \sum_{\mbf \in \mathcal{T}^i_j} f(\mbf).
\end{align}

Importantly, both $V_{i,j}$ and $\tilde{V}_{i,j}$ are independent of the truncation order $k$ in the Carleman embedding, provided that $k \ge j$. Moreover, while the expressions for $V_{i,j}$ and $\tilde{V}_{i,j}$ grow increasingly intricate with larger $i$ and $j$, we emphasize that they are intended primarily for the error analysis in Section~\ref{subsec:error_analysis_nonresonant_systems}, rather than for explicit computation.

\subsubsection{Blockwise expressions for the inverse of the diagonalizing matrix}
\label{subsec:expression_for_v_inv_ij_diagonal_f1}

Our error analysis in Section~\ref{subsec:error_analysis_nonresonant_systems}
relies not only on explicit expressions for the blocks of
$\tilde{V}$ but also those of $\tilde{V}^{-1}$. Deriving these expressions is more involved than in the case of  $\tilde{V}$.

Specifically, recall that we diagonalize the order-$k$ Carleman matrix $\tilde{A}$ as $\tilde{A}=\tilde{V} D \tilde{V}^{-1}$, where 
\begin{align}
D=\begin{pmatrix}
        \Lambda_1  &           &        & \\
                 & \Lambda_2 &        & \\
                 &           & \ddots & \\
                 &           &      &\Lambda_k \end{pmatrix},~~~
        \tilde{V}=\begin{pmatrix}
        \tilde{V}_{1,1}  & \tilde{V}_{1,2}   &  \tilde{V}_{1,3} & \dots  & \tilde{V}_{1,k} \\
                 & \tilde{V}_{2,2}  &  \tilde{V}_{2,3} & \dots & \tilde{V}_{2,k} \\
                 &           & \ddots &   \ddots       & \vdots \\
                 &           &          & \tilde{V}_{k-1,k-1} & \tilde{V}_{k-1, k} \\
                &           &          &  & \tilde{V}_{k,k} \\\end{pmatrix},     
\end{align}
with 
\begin{align}
\Lambda_j=\diag{\sum_{l=1}^j \lambda_{i_l}:~i_1, i_2,\dots,i_j\in [N]}, ~~~\forall j \in [k],    
\end{align}
and $\tilde{V}_{i,j} \in \myC^{N^i \times N^j}$ for $1 \le i \le j \le k$.
We have provided explicit expressions for the $\tilde{V}_{i,j}$ operators in Section~\ref{subsec:expression_for_vij_diagonal_f1}, i.e. 
\begin{align}
    \tilde{V}_{i,j} &= \sum_{\mbf \in \mathcal{T}^i_j} \tilde{f}(\mbf) \\
    &= \sum_{(\mbt_1,\mbt_2,\dots,\mbt_i) \in \mathcal{T}^i_j}
    \tilde{f}(\mbt_1) \otimes \tilde{f}(\mbt_2) \otimes \dots \otimes \tilde{f}(\mbt_i).
\end{align}
In particular, we have $\tilde{V}_{i,i}=I^{\otimes i}$ for $i \in [k]$.

Next, we examine the block structure of the matrix $\tilde{V}^{-1}$. Using the block upper-triangular structure of $V$, one can show that
\begin{align}
    \tilde{V}^{-1} = \begin{pmatrix}
        (\tilde{V}^{-1})_{1,1}  & (\tilde{V}^{-1})_{1,2}   &  (\tilde{V}^{-1})_{1,3} & \dots  & (\tilde{V}^{-1})_{1,k} \\
                 & (\tilde{V}^{-1})_{2,2}   &  (\tilde{V}^{-1})_{2,3} & \dots & (\tilde{V}^{-1})_{2,k} \\
                 &           & \ddots &   \ddots       & \vdots \\
                 &           &          & (\tilde{V}^{-1})_{k-1,k-1} & (\tilde{V}^{-1})_{k-1, k} \\
                &           &          &  & (\tilde{V}^{-1})_{k,k} \\\end{pmatrix},
\end{align}
where $(\tilde{V}^{-1})_{i,i}=I^{\otimes i}$ for $i \in [k]$, and 
\begin{align}
(\tilde{V}^{-1})_{i,j} = \sum_{p=1}^{j-i} \sum_{i=i_1<i_2<\dots<i_p<i_{p+1}=j} (-1)^p \tilde{V}_{i_1,i_2}\tilde{V}_{i_2,i_3}\dots \tilde{V}_{i_p,i_{p+1}}, \qquad 1 \le i < j \le k,
\label{eq:wjk_def}
\end{align}
where the inner summation is over all increasing chains from $i$ to $j$ with $p$ steps. When $p=1$, the chain is simply $i=i_1<i_2=j$, so the corresponding term is $-\tilde{V}_{i,j}$.

The inverse matrix $\tilde{V}^{-1}$ admits a similar forest expansion, but the associated tree function is different from $\tilde{f}$. The reason is that the off-diagonal blocks of the inverse of a block upper-triangular matrix are alternating sums of products of the blocks of $\tilde{V}$, as in Eq.~\eqref{eq:wjk_def}. When the expressions for the blocks of $\tilde{V}$ are substituted into these products, a fixed binary tree can arise in several different ways, corresponding to different admissible orders in which its internal vertices are generated in the inverse expansion. These admissible orders are precisely the topological orders of the internal vertices: an internal vertex is listed before any of its descendants, reflecting the order in which larger blocks are decomposed into smaller fusion subtrees. The function $\tilde{g}$ therefore has the same $\tilde{F}_2$-entry factors as $\tilde{f}$, but its scalar weight is different: instead of a single product of inverse gap factors attached to the internal vertices, it is obtained by summing over admissible topological orders, with each order contributing its own product of inverse gap factors.

To make this intuition precise, we introduce the following notation.

Given an undirected graph \( G = (V, E) \), let \( U \subseteq V \) be any subset of its vertices. The \emph{induced subgraph} of \( G \) on \( U \), denoted by $G[U]$, consists of the vertices in \( U \) and all edges from \( E \) that connect pairs of vertices in \( U \). 

For any $i \le j$, we use $\mathcal{T}^i_j$ to denote the set of binary forests consisting of $i$ binary trees with a total of $j$ leaves. We have shown that there exist at most $4^{j-i}\binom{j-1}{i-1}$ such forests in Lemma~\ref{lem:count_binary_forests}, i.e. 
\begin{align}
|\mathcal{T}^i_j|\le 4^{j-i}\binom{j-1}{i-1}.    
\end{align}

Given any binary tree $\mbt=(V, E)$, we use $V_I$ and $V_L$ to denote the set of internal and leaf nodes of $\mbt$, respectively. For any $v \in V_I$, let $c_1(v)$ and $c_2(v)$ denote its left and right children, respectively, and let $C(v)=\lrcb{c_1(v), c_2(v)}$ denote the set of its children. Moreover, let $D(v)$ denote the set of $v$'s descendants in $\mbt$, and let $B(v) \defeq D(v)\cap V_L$ denote the set of $v$'s descendants which are also leaves of $\mbt$. Finally, for any $S \subseteq V_I$, let 
\begin{align}
C(S) \defeq \lrb{\bigcup_{v \in S} C(v)} \setminus S.    
\end{align}
Note that if $S \neq \emptyset$ and $\mbt[S]$ is connected, then we have 
\begin{align}
|C(S)|=|S|+1.    
\end{align}

A \emph{topological order} of a binary tree is a linear ordering of its internal nodes such that every node comes before its descendants in this order\footnote{Normally, a topological order is defined for a directed acyclic graph (DAG) and includes all its nodes. Here, we treat a binary tree as a DAG with all edges directed downward (i.e. from an internal node to its children). Moreover, we focus on ordering the internal nodes while disregarding the ordering of the leaf nodes.}. That is, $(v_1,v_2,\dots,v_n)$ is a topological order of $\mbt=(V, E)$ if and only if $V_I=\lrcb{v_1,v_2,\dots,v_n}$ and for every $1 \le i < j \le n$, $v_i$ is not a descendant of $v_j$ in $\mbt$. Note that there is no requirement regarding the ordering of nodes that do not share an ancestor-offspring relationship (e.g., the two children of the same node). This leads to potentially multiple valid topological orders of the same tree. Let $O(\mbt)$ denote the set of topological orders of~$\mbt$. Then, we have 
\begin{align}
|O(\mbt)| \le (|V_I|-1)!,    
\end{align}
since every topological order begins with the root of $\mbt$.

A labeling of a binary tree \( \mbt = (V, E) \) is a mapping \( h: V \to [N] \), meaning that each node in \( \mbt \) is assigned an integer label from \( 1 \) to \( N \). Let \( L_N(\mbt) \) denote the set of all such labelings of \( \mbt \). For better readability, we will omit the subscript \( N \) and simply write \( L(\mbt) \) in what follows. Moreover,  given any labeling $h$ of a binary tree $\mbt$, we define 
\begin{align}
    \alpha(\mbt, h) \defeq \prod_{v \in V_I} \bra{h(v)}\tilde{F}_2 \ket{h(c_1(v)), h(c_2(v))},
    \label{eq:alpha_t_h_def}
\end{align}
if $\mbt$ is nontrivial, and $\alpha(\mbt, h)=1$ otherwise. 
That is, $\alpha(\mbt, h)$ is the product of the entries of $\tilde{F}_2$ corresponding to the labels assigned to each internal node and its two children. It is nonzero if and only if all of these entries are nonzero. In this case, we say that $h$ is a \emph{feasible} labeling of $\mbt$ with respect to $\tilde{F}_2$. The number of feasible labelings of $\mbt$ could be much smaller than $|L(\mbt)|=N^{|V(\mbt)|}$ when $\tilde{F}_2$ is sparse.

Now, suppose $\mbt$ contains $m$ leaves (if $\mbt$ is trivial, then $m=1$). Let $r$ denote the root of $\mbt$, and let $l_1,l_2,\dots,l_m$ represent the leaves of $\mbt$ (ordered from left to right). Define
\begin{align}
    \beta(\mbt, h) \defeq \prod_{ v \in V_I} \frac{1}{\sum_{w \in B(v)} \lambda_{h(w)} -\lambda_{h(v)}},
    \label{eq:beta_t_h_def}
\end{align}
if $\mbt$ is nontrivial, and $\beta(\mbt, h)=1$ otherwise. Then, the function $\tilde{f}(\mbt)$ introduced in Section~\ref{subsec:expression_for_vij_diagonal_f1} can be re-written as
\begin{align}
    \tilde{f}(\mbt) = \sum_{h \in L(\mbt)} \alpha(\mbt, h) \beta(\mbt, h) \ketbra{h(r)}{h(l_1), h(l_2),\dots, h(l_m)}.
    \label{eq:ft_def2}
\end{align}
Note that if $\mbt$ is trivial, then $\tilde{f}(\mbt)=I$. This definition of $\tilde{f}$ is consistent with the recursive definition in Section~\ref{subsec:expression_for_vij_diagonal_f1}.
Moreover, define 
\begin{align}
    \gamma(\mbt, h) \defeq \sum_{(v_1,v_2,\dots,v_{m-1}) \in O(\mbt)} \prod_{i=1}^{m-1}\frac{1}{\sum_{w \in C(\lrcb{v_1,v_2,\dots,v_i})} \lambda_{h(w)} -\lambda_{h(r)}},
    \label{eq:gamma_t_h_def}
\end{align}
if $\mbt$ is nontrivial, and $\gamma(\mbt, h)=1$ otherwise.
Then, define
\begin{align}
    \tilde{g}(\mbt) \defeq \sum_{h \in L(\mbt)} \alpha(\mbt, h) \gamma(\mbt, h) \ketbra{h(r)}{h(l_1), h(l_2),\dots, h(l_m)}.     
\end{align}
Note that if $\mbt$ is trivial, then $\tilde{g}(\mbt)=I$. 

Finally, we extend the definitions of the functions $f$ and $g$ from binary trees to binary forests. For any binary forest $\mbf = (\mbt_1, \mbt_2, \dots, \mbt_j)$, we define:
\begin{align}
    \tilde{f}(\mbf)&=\tilde{f}(\mbt_1)\otimes \tilde{f}(\mbt_2) \otimes \dots \otimes \tilde{f}(\mbt_j), \\
    \tilde{g}(\mbf)&=\tilde{g}(\mbt_1)\otimes \tilde{g}(\mbt_2) \otimes \dots \otimes \tilde{g}(\mbt_j).     
\end{align}

The functions $\tilde{f}$ and $\tilde{g}$, defined over binary trees and extended to binary forests, play a central role in characterizing the blocks of $\tilde{V}$ and $\tilde{V}^{-1}$, respectively.

With this notation, the $(\tilde{V}^{-1})_{i,j}$ operators can be expressed compactly in terms of the $\tilde{g}(\mbt)$ operators:
\begin{prop}
For any $i \le j$, we have
    \begin{align}
    (\tilde{V}^{-1})_{i,j} &= (-1)^{j-i} \sum_{\mbf \in \mathcal{T}^i_j} \tilde{g}(\mbf)\\
    &=(-1)^{j-i} \sum_{(\mbt_1,\mbt_2,\dots,\mbt_i) \in \mathcal{T}^i_j} \tilde{g}(\mbt_1) \otimes \tilde{g}(\mbt_2) \otimes \dots \otimes \tilde{g}(\mbt_i).    
    \end{align}
    \label{prop:v_inv_ij_expression}
\end{prop}
\begin{proof}
See Appendix~\ref{app:proof_lemma_v_inv_ij_expression}.    
\end{proof}

Below are the explicit expressions for $\lrb{\tilde{V}^{-1}}_{i,j}$ where $1 \le i<j\le 4$, together with their binary forests representations:

\begingroup
\setlength{\jot}{12pt} 
\tikzset{
	level distance=16pt,
	every tree node/.style = {align=center, anchor=north,nodecolor},
	edge from parent/.style = {draw,thick,edgecolor},
	edge from parent path = {
		([yshift=4.5pt]\tikzparentnode.south) -- 
		([yshift=-3.5pt]\tikzchildnode.north)
	}
}
\begin{align}
    \lrb{\tilde{V}^{-1}}_{1,2}=&-\sum_{a,b,c=1}^N \bra{a}\tilde{F}_2\ket{b,c} \cdot \frac{1}{\lambda_b+\lambda_c-\lambda_a} \ket{a}\bra{b,c}
    = -\tilde{g}\left(
	\raisebox{-0.35cm}{
		\begin{tikzpicture}[scale=1]
			\Tree [.$\bt$ [.$\bt$ ] [.$\bt$ ] ];
			\node at (0,0) {$a$};
			\node at (-0.5,-0.75) {$b$};
			\node at (0.5,-0.75) {$c$};
		\end{tikzpicture}
	}
	\right),
    \\
    \lrb{\tilde{V}^{-1}}_{2,3}=& -\sum_{a,b,c,d=1}^N \bra{a}\tilde{F}_2\ket{b,c} \cdot \frac{1}{\lambda_b+\lambda_c-\lambda_a} \ket{a,d}\bra{b,c,d} \nonumber \\
    & -\sum_{a,b,c,d=1}^N \bra{a}\tilde{F}_2\ket{c,d} \cdot \frac{1}{\lambda_c+\lambda_d-\lambda_a} \ket{b,a}\bra{b,c,d}\\
    =&
	-\tilde{g}\left(
	\raisebox{-0.35cm}{
		\begin{tikzpicture}[scale=1]
			\Tree [.$\bt$ [.$\bt$ ] [.$\bt$ ] ];
			\node at (0,0) {$a$};
			\node at (-0.5,-0.75) {$b$};
			\node at (0.5,-0.75) {$c$};
		\end{tikzpicture}
		\hspace{\treeTree}
		\begin{tikzpicture}[scale=1]
			\Tree [.$\bt$ ];
			\node at (0,0.05) {$d$};
		\end{tikzpicture}
	}
	\right)
	-\tilde{g}\left(
	\raisebox{-0.35cm}{
		\begin{tikzpicture}[scale=1]
			\Tree [.$\bt$ ];
			\node at (0,0.05) {$b$};
		\end{tikzpicture}
		\hspace{\treeTree}
		\begin{tikzpicture}[scale=1]
			\Tree [.$\bt$ [.$\bt$ ] [.$\bt$ ] ];
			\node at (0,0) {$a$};
			\node at (-0.5,-0.75) {$c$};
			\node at (0.5,-0.75) {$d$};
		\end{tikzpicture}
	}
	\right),
    \\
    \lrb{\tilde{V}^{-1}}_{3,4}=& -\sum_{a,b,c,d,e=1}^N \bra{a}\tilde{F}_2\ket{b,c} \cdot \frac{1}{\lambda_b+\lambda_c-\lambda_a} \ket{a,d,e}\bra{b,c,d,e} \nonumber \\
    & -\sum_{a,b,c,d,e=1}^N \bra{a}\tilde{F}_2\ket{c,d} \cdot \frac{1}{\lambda_c+\lambda_d-\lambda_a} \ket{b,a,e}\bra{b,c,d,e} \nonumber \\
    & -\sum_{a,b,c,d,e=1}^N \bra{a}\tilde{F}_2\ket{d,e} \cdot \frac{1}{\lambda_d+\lambda_e-\lambda_a} \ket{b,c,a}\bra{b,c,d,e}
    \\
    =&
	-\tilde{g}\left(
	\raisebox{-0.35cm}{
		\begin{tikzpicture}[scale=1]
			\Tree [.$\bt$ [.$\bt$ ] [.$\bt$ ] ];
			\node at (0,0) {$a$};
			\node at (-0.5,-0.75) {$b$};
			\node at (0.5,-0.75) {$c$};
		\end{tikzpicture}
		\hspace{\treeTree}
		\begin{tikzpicture}[scale=1]
			\Tree [.$\bt$ ];
			\node at (0,0.05) {$d$};
		\end{tikzpicture}
		\hspace{\treeTree}
		\begin{tikzpicture}[scale=1]
			\Tree [.$\bt$ ];
			\node at (0,0.05) {$e$};
		\end{tikzpicture}
	}
	\right)
	-\tilde{g}\left(
	\raisebox{-0.35cm}{
		\begin{tikzpicture}[scale=1]
			\Tree [.$\bt$ ];
			\node at (0,0.05) {$b$};
		\end{tikzpicture}
		\hspace{\treeTree}
		\begin{tikzpicture}[scale=1]
			\Tree [.$\bt$ [.$\bt$ ] [.$\bt$ ] ];
			\node at (0,0) {$a$};
			\node at (-0.5,-0.75) {$c$};
			\node at (0.5,-0.75) {$d$};
		\end{tikzpicture}
		\hspace{\treeTree}
		\begin{tikzpicture}[scale=1]
			\tikzset{}
			\Tree [.$\bt$ ];
			\node at (0,0.05) {$e$};
		\end{tikzpicture}
	}
	\right)
	-\tilde{g}\left(
	\raisebox{-0.35cm}{
		\begin{tikzpicture}[scale=1]
			\Tree [.$\bt$ ];
			\node at (0,0.05) {$b$};
		\end{tikzpicture}
		\hspace{\treeTree}
		\begin{tikzpicture}[scale=1]
			\Tree [.$\bt$ ];
			\node at (0,0.05) {$c$};
		\end{tikzpicture}
		\hspace{\treeTree}
		\begin{tikzpicture}[scale=1]
			\Tree [.$\bt$ [.$\bt$ ] [.$\bt$ ] ];
			\node at (0,0) {$a$};
			\node at (-0.5,-0.75) {$d$};
			\node at (0.5,-0.75) {$e$};
		\end{tikzpicture}
	}
	\right),
    \\
\lrb{\tilde{V}^{-1}}_{1,3}=&\sum_{a,b,c,d,e=1}^N \bra{a}\tilde{F}_2\ket{b,c} \cdot \bra{b}\tilde{F}_2\ket{d,e} \cdot \frac{1}{\lambda_b+\lambda_c-\lambda_a}
    \cdot \frac{1}{\lambda_c+\lambda_d+\lambda_e-\lambda_a}
    \ket{a}\bra{d,e,c} \nonumber \\
    & +\sum_{a,b,c,d,e=1}^N \bra{a}\tilde{F}_2\ket{b,c} \cdot \bra{c}\tilde{F}_2\ket{d,e} \cdot \frac{1}{\lambda_b+\lambda_c-\lambda_a}
    \cdot \frac{1}{\lambda_b+\lambda_d+\lambda_e-\lambda_a}
    \ket{a}\bra{b,d,e}
    \\
    =&
	\tilde{g}\left(
	\raisebox{-0.6cm}{
		\begin{tikzpicture}[scale=1]
			\Tree [.$\bt$ 
			[.$\bt$ [.$\bt$ ] [.$\bt$ ] ] 
			[.$\bt$ ] 
			];
			\node at (0,0) {$a$};
			\node at (-0.6,-0.75) {$b$};
			\node at (0.6,-0.75) {$c$};
			\node at (-0.9,-1.35) {$d$};
			\node at (0.15,-1.35) {$e$};
		\end{tikzpicture}
	}
	\right)
	+
	\tilde{g}\left(
	\raisebox{-0.6cm}{
		\begin{tikzpicture}[scale=1]
			\Tree [.$\bt$ 
			[.$\bt$ ] 
			[.$\bt$ [.$\bt$ ] [.$\bt$ ] ] 
			];
			\node at (0,0) {$a$};
			\node at (-0.6,-0.75) {$b$};
			\node at (0.6,-0.75) {$c$};
			\node at (-0.15,-1.35) {$d$};
			\node at (0.9,-1.35) {$e$};
		\end{tikzpicture}
	}
	\right),
    \\
    \lrb{\tilde{V}^{-1}}_{2,4}
    =& \sum_{a,b,c,d,e,f=1}^N \bra{a}\tilde{F}_2\ket{b,c} \cdot \bra{b}\tilde{F}_2\ket{d,e} \cdot \frac{1}{\lambda_b+\lambda_c-\lambda_a}
    \cdot \frac{1}{\lambda_c+\lambda_d+\lambda_e-\lambda_a}
    \ket{a,f}\bra{d,e,c,f} \nonumber \\
    & +\sum_{a,b,c,d,e,f=1}^N \bra{a}\tilde{F}_2\ket{b,c} \cdot \bra{c}\tilde{F}_2\ket{d,e} \cdot \frac{1}{\lambda_b+\lambda_c-\lambda_a}
    \cdot \frac{1}{\lambda_b+\lambda_d+\lambda_e-\lambda_a}
    \ket{a,f}\bra{b,d,e,f} \nonumber \\
    & + \sum_{a,b,c,d,e,f=1}^N \bra{a}\tilde{F}_2\ket{b,c} \cdot \bra{b}\tilde{F}_2\ket{e,f} \cdot \frac{1}{\lambda_b+\lambda_c-\lambda_a}
    \cdot \frac{1}{\lambda_c+\lambda_e+\lambda_f-\lambda_a}
    \ket{d,a}\bra{d,e,f,c} \nonumber \\
    & +\sum_{a,b,c,d,e,f=1}^N \bra{a}\tilde{F}_2\ket{b,c} \cdot \bra{c}\tilde{F}_2\ket{e,f} \cdot \frac{1}{\lambda_b+\lambda_c-\lambda_a}
    \cdot \frac{1}{\lambda_b+\lambda_e+\lambda_f-\lambda_a}
    \ket{d,a}\bra{d,b,e,f} \nonumber \\
    & +\sum_{a,b,c,d,e,f=1}^N \bra{a}\tilde{F}_2\ket{c,d} \cdot \bra{b}\tilde{F}_2\ket{e,f} \cdot \frac{1}{\lambda_c+\lambda_d-\lambda_a}
    \cdot \frac{1}{\lambda_e+\lambda_f-\lambda_b}
    \ket{a,b}\bra{c,d,e,f}\\
    =& 
	\tilde{g}\left(
	\raisebox{-0.6cm}{
		\begin{tikzpicture}[scale=1]
			\Tree [.$\bt$ 
			[.$\bt$ [.$\bt$ ] [.$\bt$ ] ] 
			[.$\bt$ ] 
			];		
			\node at (0,0) {$a$};
			\node at (-0.6,-0.75) {$b$};
			\node at (0.6,-0.75) {$c$};
			\node at (-0.9,-1.35) {$d$};
			\node at (0.15,-1.35) {$e$};
		\end{tikzpicture}
		\hspace{\treeTree}
		\begin{tikzpicture}[scale=1]
			\Tree [.$\bt$  ];
			\node at (0,0.05) {$f$};
			\node at (0,-0.2) {\phantom{$f$}};
		\end{tikzpicture}
	}
	\right)
	+
	\tilde{g}\left(
	\raisebox{-0.6cm}{
		\begin{tikzpicture}[scale=1]
			\Tree [.$\bt$ 
			[.$\bt$ ] 
			[.$\bt$ [.$\bt$ ] [.$\bt$ ] ] 
			];
			\node at (0,0) {$a$};
			\node at (-0.6,-0.75) {$b$};
			\node at (0.6,-0.75) {$c$};
			\node at (-0.15,-1.35) {$d$};
			\node at (0.9,-1.35) {$e$};
		\end{tikzpicture}
		\hspace{\treeTree}
		\begin{tikzpicture}[scale=1]
			\Tree [.$\bt$  ];
			\node at (0,0.05) {$f$};
			\node at (0,-0.2) {\phantom{$f$}};
		\end{tikzpicture}
	}
	\right)
	+
	\tilde{g}\left(
	\raisebox{-0.6cm}{
		\begin{tikzpicture}[scale=1]
			\Tree [.$\bt$  ];
			\node at (0,0.05) {$d$};
			\node at (0,-0.2) {\phantom{$f$}};
		\end{tikzpicture}
		\hspace{\treeTree}
		\begin{tikzpicture}[scale=1]
			\Tree [.$\bt$ 
			[.$\bt$ [.$\bt$ ] [.$\bt$ ] ] 
			[.$\bt$ ] ];
			\node at (0,0) {$a$};
			\node at (-0.6,-0.75) {$b$};
			\node at (0.6,-0.75) {$c$};
			\node at (-0.9,-1.35) {$e$};
			\node at (0.15,-1.35) {$f$};
		\end{tikzpicture}
	}
	\right)
	\nonumber\\
    &
	+
	\tilde{g}\left(
	\raisebox{-0.6cm}{
		\begin{tikzpicture}[scale=1]
			\Tree [.$\bt$  ];
			\node at (0,0.05) {$d$};
			\node at (0,-0.2) {\phantom{$f$}};
		\end{tikzpicture}
		\hspace{\treeTree}
		\begin{tikzpicture}[scale=1]
			\Tree [.$\bt$ 
			[.$\bt$ ] 
			[.$\bt$ [.$\bt$ ] [.$\bt$ ] ] 
			];
			\node at (0,0) {$a$};
			\node at (-0.6,-0.75) {$b$};
			\node at (0.6,-0.75) {$c$};
			\node at (-0.15,-1.35) {$e$};
			\node at (0.9,-1.35) {$f$};
		\end{tikzpicture}
	}
	\right)
	+
	\tilde{g}\left(
	\raisebox{-0.35cm}{
		\begin{tikzpicture}[scale=1]
			\Tree [.$\bt$ [.$\bt$ ] [.$\bt$ ] ];
			\node at (0,0.05) {$a$};
			\node at (-0.5,-0.75) {$c$};
			\node at (0.5,-0.75) {$d$};
		\end{tikzpicture}
		\hspace{\treeTree}
		\begin{tikzpicture}[scale=1]
			\Tree [.$\bullet$ [.$\bt$ ] [.$\bt$ ] ];
			\node at (0,0.05) {$b$};
			\node at (-0.5,-0.75) {$e$};
			\node at (0.5,-0.75) {$f$};
		\end{tikzpicture}
	}
	\right),\\     
    \lrb{\tilde{V}^{-1}}_{1,4}=&-\sum_{a,b,c,d,e,f,g=1}^N \bra{a}\tilde{F}_2\ket{b,c} \cdot \bra{b}\tilde{F}_2\ket{d,e} \cdot 
    \bra{d}\tilde{F}_2\ket{f,g} \cdot \frac{1}{\lambda_b+\lambda_c-\lambda_a}  \cdot \frac{1}{\lambda_c+\lambda_d+\lambda_e-\lambda_a} \nonumber \\
    &
    \cdot \frac{1}{\lambda_c+\lambda_e+\lambda_f+\lambda_g-\lambda_a}
    \ket{a}\bra{f,g,e,c} \nonumber \\
    & -\sum_{a,b,c,d,e,f,g=1}^N \bra{a}\tilde{F}_2\ket{b,c} \cdot \bra{b}\tilde{F}_2\ket{d,e} \cdot 
    \bra{e}\tilde{F}_2\ket{f,g} \cdot \frac{1}{\lambda_b+\lambda_c-\lambda_a}  \cdot \frac{1}{\lambda_c+\lambda_d+\lambda_e-\lambda_a} \nonumber \\
    &
    \cdot \frac{1}{\lambda_c+\lambda_d+\lambda_f+\lambda_g-\lambda_a}
    \ket{a}\bra{d,f,g,c} \nonumber \\
    & -\sum_{a,b,c,d,e,f,g=1}^N \bra{a}\tilde{F}_2\ket{b,c} \cdot \bra{c}\tilde{F}_2\ket{d,e} \cdot 
    \bra{d}\tilde{F}_2\ket{f,g} \cdot \frac{1}{\lambda_b+\lambda_c-\lambda_a}  \cdot \frac{1}{\lambda_b+\lambda_d+\lambda_e-\lambda_a} \nonumber \\
    &
    \cdot \frac{1}{\lambda_b+\lambda_e+\lambda_f+\lambda_g-\lambda_a}
    \ket{a}\bra{b,f,g,e} \nonumber \\
    & -\sum_{a,b,c,d,e,f,g=1}^N \bra{a}\tilde{F}_2\ket{b,c} \cdot \bra{c}\tilde{F}_2\ket{d,e} \cdot 
    \bra{e}\tilde{F}_2\ket{f,g} \cdot \frac{1}{\lambda_b+\lambda_c-\lambda_a}  \cdot \frac{1}{\lambda_b+\lambda_d+\lambda_e-\lambda_a} \nonumber \\
    &
    \cdot \frac{1}{\lambda_b+\lambda_d+\lambda_f+\lambda_g-\lambda_a}
    \ket{a}\bra{b,d,f,g} \nonumber \\
    & -\sum_{a,b,c,d,e,f,g=1}^N \bra{a}\tilde{F}_2\ket{b,c} \cdot \bra{b}\tilde{F}_2\ket{d,e} \cdot 
    \bra{c}\tilde{F}_2\ket{f,g} \cdot \frac{1}{\lambda_b+\lambda_c-\lambda_a}  \cdot \left ( \frac{1}{\lambda_c+\lambda_d+\lambda_e-\lambda_a} \right. \nonumber \\
    & \left. +\frac{1}{\lambda_b+\lambda_f+\lambda_g-\lambda_a} \right )
    \cdot \frac{1}{\lambda_d+\lambda_e+\lambda_f+\lambda_g-\lambda_a}
    \ket{a}\bra{d,e,f,g}
    \\
    =&
    -\tilde{g}\left(
	\raisebox{-0.9cm}{
		\begin{tikzpicture}[scale=1]
			\Tree [.$\bt$ 
			[.$\bt$ [.$\bt$ [.$\bt$ ] [.$\bt$ ] ] [.$\bt$ ] ] 
			[.$\bt$ ] 
			];
			\node at (0,0) {$a$};
			\node at (-0.7,-0.75) {$b$};
			\node at (0.7,-0.75) {$c$};
			\node at (-1.05,-1.35) {$d$};
			\node at (0.15,-1.35) {$e$};
			\node at (-1.25,-2) {$f$};
			\node at (-0.25,-2) {$g$};
		\end{tikzpicture}
	}
	\right)
	-
	\tilde{g}\left(
	\raisebox{-0.9cm}{
		\begin{tikzpicture}[scale=1]
			\Tree [.$\bt$ 
			[.$\bt$ [.$\bt$ ] [.$\bt$ [.$\bt$ ] [.$\bt$ ] ] ] 
			[.$\bt$ ] 
			];
			\node at (0,0) {$a$};
			\node at (-0.8,-0.75) {$b$};
			\node at (0.8,-0.75) {$c$};
			\node at (-1.15,-1.35) {$d$};
			\node at (0.05,-1.35) {$e$};
			\node at (-0.7,-2) {$f$};
			\node at (0.3,-2) {$g$};
		\end{tikzpicture}
	}
	\right)
	-
	\tilde{g}\left(
	\raisebox{-0.9cm}{
		\begin{tikzpicture}[scale=1]
			\Tree [.$\bt$ 
			[.$\bt$ ] 
			[.$\bt$ [.$\bt$ [.$\bt$ ] [.$\bt$ ] ] [.$\bt$ ] ] 
			];
			\node at (0,0) {$a$};
			\node at (-0.8,-0.75) {$b$};
			\node at (0.8,-0.75) {$c$};
			\node at (-0.1,-1.35) {$d$};
			\node at (1.1,-1.35) {$e$};
			\node at (-0.35,-2) {$f$};
			\node at (0.65,-2) {$g$};
		\end{tikzpicture}
	}
	\right)
	\nonumber
	\\
	&
	-
	\tilde{g}\left(
	\raisebox{-0.9cm}{
		\begin{tikzpicture}[scale=1]
			\Tree [.$\bt$ 
			[.$\bt$ ] 
			[.$\bt$ [.$\bt$ ] [.$\bt$ [.$\bt$ ] [.$\bt$ ] 	] ] 
			];
			\node at (0,0) {$a$};
			\node at (-0.7,-0.75) {$b$};
			\node at (0.7,-0.75) {$c$};
			\node at (-0.2,-1.35) {$d$};
			\node at (1,-1.35) {$e$};
			\node at (0.25,-2) {$f$};
			\node at (1.25,-2) {$g$};
		\end{tikzpicture}
	}
	\right)
	-
	\tilde{g}\left(
	\raisebox{-0.6cm}{
		\begin{tikzpicture}[scale=1]
			\Tree [.$\bt$ [.$\bt$ [.$\bt$ ] [.$\bt$ ] ] [.$\bt$ [.$\bt$ ] [.$\bt$ ] ] ];
			\node at (0,0) {$a$};
			\node at (-0.7,-0.75) {$b$};
			\node at (0.7,-0.75) {$c$};
			\node at (-0.8,-1.6) {$d$};
			\node at (-0.25,-1.6) {$e$};
			\node at (0.25,-1.6) {$f$};
			\node at (0.8,-1.6) {$g$};
		\end{tikzpicture}				
	}
	\right).
\end{align}
\endgroup

Given the expressions for the blocks of $\tilde{V}^{-1}$, the corresponding blocks of $V^{-1}$ can be obtained by right-multiplying with $(Q^{-1})^{\otimes j}$, where $j$ is the column index. Specifically, one can verify that
\begin{align}
    {V}^{-1} = \begin{pmatrix}
        ({V}^{-1})_{1,1}  & ({V}^{-1})_{1,2}   &  ({V}^{-1})_{1,3} & \dots  & ({V}^{-1})_{1,k} \\
                 & ({V}^{-1})_{2,2}   &  ({V}^{-1})_{2,3} & \dots & ({V}^{-1})_{2,k} \\
                 &           & \ddots &   \ddots       & \vdots \\
                 &           &          & ({V}^{-1})_{k-1,k-1} & ({V}^{-1})_{k-1, k} \\
                &           &          &  & ({V}^{-1})_{k,k} \\\end{pmatrix},
\end{align}
where 
\begin{align}
    ({V}^{-1})_{i,j} = (\tilde{V}^{-1})_{i,j}(Q^{-1})^{\otimes j},\quad ~\forall 1 \le i \le j \le k.
\end{align}
For arbitrary $\mbf \in \mathcal{T}^i_j$, we define
\begin{align}
    g(\mbf) = \tilde{g}(\mbf) \lrb{Q^{-1}}^{\otimes j}.
\end{align}
Then 
\begin{align}
    ({V}^{-1})_{i,j} = (-1)^{j-i}\sum_{\mbf \in \mathcal{T}^i_j} g(\mbf).
\end{align}

As in the case of $V_{i,j}$ and $\tilde{V}_{i,j}$, both $(V^{-1})_{i,j}$ and $(\tilde{V}^{-1})_{i,j}$ are independent of the truncation order $k$ in the Carleman embedding, as long as $k \ge j$. Moreover, the expressions for $(V^{-1})_{i,j}$ and $(\tilde{V}^{-1})_{i,j}$ become quite involved for large $i$ and $j$, but they are primarily used for the error analysis in Section~\ref{subsec:error_analysis_nonresonant_systems} rather than for explicit computation.

\subsubsection{Extension to broader classes of systems}
\label{subsec:carleman_matrix_diagonalization_extension}
Thus far, our framework has focused on the explicit diagonalization of Carleman matrices for quadratic differential equations without external driving:
\begin{align}
\dot{x}(t)=F_1 x(t) + F_2 x(t)^{\otimes 2}.
\end{align}
In what follows, we outline how this framework may be extended to encompass a wider variety of differential equations.

\paragraph{Higher-order nonlinearities.} First, our framework can be extended to handle general polynomial nonlinearities. Consider the ODE system:
\begin{align}
\dot{x}(t) = F_1 x(t) + F_2 x(t)^{\otimes 2} + F_3 x(t)^{\otimes 3} + \dots + F_m x(t)^{\otimes m}.
\end{align}
One can follow the same strategy to derive blockwise expressions for the  linear transformation and its inverse that diagonalize the Carleman matrices associated with this system. Each term in these expressions corresponds to a rooted  tree or forest; however, the internal nodes of these trees and forests may now have degrees $2, 3, \dots, m$. In other words, the representation generalizes from binary to general rooted trees and forests. Nonetheless, the core structure of the representation remains the same as in the quadratic case.

\paragraph{Driven nonlinear systems.} Second, if a system includes an external driving term but is defined in a neighborhood of a known equilibrium, our framework can be extended to analyze its behavior near that equilibrium. Specifically, consider the ODE system:
\begin{align}
\dot{x}(t)  = F_0 + F_1 x(t) + F_2 x(t)^{\otimes 2},
\end{align}
with initial condition $x(0)$ lying in the basin of attraction of a stable fixed point  $x^*$ which satisfies
\begin{align}
F_0 + F_1 x^* + F_2 (x^*)^{\otimes 2} = 0.
\end{align}
Introducing the change of variables $y(t) = x(t) - x^*$, the system can be rewritten as
\begin{align}
\dot{y}(t) = \hat{F}_1 y(t) + F_2 y(t)^{\otimes 2},
\end{align}
where
\begin{align}
\hat{F}_1 = F_1 + F_2 (I \otimes x^* + x^* \otimes I).
\end{align}
Note that the driving term is eliminated in this transformed system, allowing our framework to be directly applied to analyze the dynamics of $y(t)$. Moreover, since $\hat{F}_1$ depends on the choice of fixed point $x^*$, different choices yield distinct transformed systems. The local behavior around each equilibrium is thus determined by the properties of the corresponding $\hat{F}_1$ matrix. This approach extends naturally to higher-order nonlinear systems as well.

\paragraph{Resonant, diagonalizable systems.} Finally, certain resonant systems still admit diagonalizable Carleman matrices and can therefore be treated within our framework. Specifically, consider the ODE system
\begin{align}
\dot{x}(t) = F_1 x(t) + F_2 x(t)^{\otimes 2},
\end{align}
where $F_1 = Q \cdot \mathrm{diag}(\lambda_1, \lambda_2, \dots, \lambda_N) \cdot Q^{-1}$ with possibly resonant eigenvalues $\lambda_1, \lambda_2, \dots, \lambda_N$. The order-$k$ Carleman matrix $A$ associated with this system is diagonalizable via a similarity transformation if we can construct matrices $\lrcb{\tilde{V}_{i,j}:~1\le i\le j\le k}$ satisfying
\begin{align}
\tilde{V}_{i,j} \Lambda_j - \Lambda_i \tilde{V}_{i,j} = \tilde{A}_{i,i+1} \tilde{V}_{i+1,j}, \quad \forall 1 \le i < j \le k,
\label{eq:V_recurrence2}
\end{align}
and $\tilde{V}_{j,j}=I^{\otimes j}$ for all $j \in [k]$. Proposition~\ref{prop:valid_vij} shows that when the tuple $(\lambda_1, \lambda_2, \dots, \lambda_N)$ is nonresonant, these equations always admit a solution. In the resonant case, however, a solution may still exist, provided that any zero entry on the right-hand side of Eq.~\eqref{eq:V_recurrence2} is matched by a corresponding zero entry on the left-hand side. This condition leads to a system of homogeneous polynomial equations in the entries of $\tilde{F}_2$, with coefficients determined by the eigenvalues of $F_1$. As long as these constraints are satisfied, the Carleman matrix $A$ remains diagonalizable, and our framework can be applied essentially as before, with minor adjustments to account for the problematic zero entries. We leave the details of this extension to the reader.

\subsection{Error analysis of the Carleman embedding}
\label{subsec:error_analysis_nonresonant_systems}

Next, we utilize the framework developed in Section~\ref{subsec:carleman_matrix_diagonalization} to derive upper bounds on the approximation errors of truncated Carleman linearizations for nonresonant differential equations.

Let us briefly recall the problem setting. Consider the ODE system:
\begin{align}
\dot{x}(t) = F_1 x(t) + F_2 x(t)^{\otimes 2},    
\label{eq:general_ode_system2}
\end{align}
where $x(t) \in \mathbb{C}^N$, $F_2 \in \mathbb{C}^{N \times N^2}$, and $F_1 = Q \Lambda Q^{-1}$ with $\Lambda = \operatorname{diag}(\lambda_1, \lambda_2, \dots, \lambda_N)$. We assume that $\operatorname{Re}(\lambda_i) \le 0$ for all $i \in [N]$. Unless otherwise stated, we further assume that the spectrum $(\lambda_1, \lambda_2, \dots, \lambda_N)$ is nonresonant throughout this subsection.

Let 
\begin{align}
\dot{y}(t) =Ay(t)    
\end{align}
be the order-$k$ Carleman linearization for this system, where $y(t)=[y^{[1]}(t), y^{[2]}(t),\dots,y^{[k]}(t)]^T$ with $y^{[i]}(t) \in \myC^{N^i}$ for $i \in [k]$. We aim to prove an upper bound on the truncation error, i.e. $\norm{y^{[i]}(t)-x(t)^{\otimes i}}$, for each $i \in [k]$, under mild assumptions on the system.  

To this end, we first simplify the problem through a change of variables $\tilde{x}(t)=Q^{-1}x(t)$. Then $\tilde{x}(t)$ is governed by the ODE system:
\begin{align}
\dot{\tilde{x}}(t) = \Lambda \tilde{x}(t) + \tilde{F}_2 \tilde{x}(t)^{\otimes 2},   
\label{eq:simple_ode_system2}
\end{align}
where $\tilde{F}_2 = Q^{-1}F_2 Q^{\otimes 2}$. Let 
\begin{align}
\dot{\tilde{y}}(t) =\tilde{A}\tilde{y}(t)    
\end{align}
be the order-$k$ Carleman linearization for this system, where $\tilde{y}(t)=[\tilde{y}^{[1]}(t), \tilde{y}^{[2]}(t),\dots,\tilde{y}^{[k]}(t)]^T$ with $\tilde{y}^{[i]}(t) \in \myC^{N^i}$ for $i \in [k]$. We will first establish an upper bound on $\norm{\tilde{y}^{[i]}(t)-\tilde{x}(t)^{\otimes i}}$, and then convert it into an upper bound on $\norm{y^{[i]}(t)-x(t)^{\otimes i}}$, for each $i \in [k]$.

We say that $(\lambda_1, \lambda_2, \dots, \lambda_N)$ lies in the \emph{Poincar\'{e} domain} if $0$ is outside the convex hull of $\lambda_1, \lambda_2, \dots, \lambda_N$; otherwise, it is said to lie in the \emph{Siegel domain}. As we shall see, the error analysis depends crucially on which of these two domains the system falls into. Accordingly, we address the Poincar\'{e} and Siegel domain cases separately in what follows.

\subsubsection{Error analysis for systems in the Poincar\'{e} domain}
\label{subsubsec:error_analysis_poincare_domain}

We begin by considering the case where $(\lambda_1, \lambda_2, \dots, \lambda_N)$ is nonresonant and lies in the Poincar\'{e} domain. More precisely, we assume that the following condition holds:
\begin{assumption}
There exists a positive normalized ``no-resonance'' gap $\Delta$ defined as:
\begin{align}
\Delta \defeq \min_{i \in [N]} \inf_{\substack{\alpha_j \in \myN_0 \\ \sum_{j=1}^N \alpha_j \ge 2}} \frac{\abs{\lambda_i-\sum_{j=1}^N \alpha_j\lambda_j}}{\sum_{j=1}^N \alpha_j-1} > 0.   
\label{eq:linear_no_resonance_gap}
\end{align}
Under this condition, we say that $(\lambda_1, \lambda_2, \dots, \lambda_N)$ is $\Delta'$-nonresonant for every $\Delta' \in [0, \Delta]$.
\label{assump:linear_no_resonance_gap}
\end{assumption}

Assumption~\ref{assump:linear_no_resonance_gap} might seem unusual at first glance, but it is in fact equivalent to the condition that $(\lambda_1, \lambda_2, \dots, \lambda_N)$ is nonresonant and lies in the Poincar\'{e} domain.

\begin{lemma}
Assumption~\ref{assump:linear_no_resonance_gap} holds if and only if $(\lambda_1, \lambda_2, \dots, \lambda_N)$ is nonresonant and lies in the Poincar\'{e} domain.
\label{lem:equivalence_linear_no_resonance_gap_poincare_domain}
\end{lemma}

The proof of this lemma is given in Appendix \ref{app:proof_norm_bound_vij_v_inv_ij_poincare_domain}. We remark that, under the condition $\realpart{\lambda_j} \le 0$ for all $j \in [N]$, the origin lies outside the convex hull of $\lambda_1, \lambda_2, \dots, \lambda_N$ if and only if these eigenvalues do not lie on both the positive and negative parts of the imaginary axis. That is, either of the following holds:
\begin{itemize}
    \item $\forall j \in [N]$, $\realpart{\lambda_j}<0$ or $\imaginarypart{\lambda_j} > 0$;
    \item $\forall j \in [N]$, $\realpart{\lambda_j}<0$ or $\imaginarypart{\lambda_j} < 0$.
\end{itemize}
In particular, if the system is dissipative, i.e. $\realpart{\lambda_j} < 0$ for all $j \in [N]$, then $(\lambda_1, \lambda_2, \dots, \lambda_N)$ lies in the Poincar\'{e} domain.

In Section~\ref{subsec:carleman_matrix_diagonalization}, we have studied the diagonalization of the order-$k$ Carleman matrix \( \tilde{A} \) and provided explicit expressions for the blocks of \( \tilde{V} \) and \( \tilde{V}^{-1} \), where \( \tilde{V} \) is the similarity transformation that diagonalizes \( \tilde{A} \), i.e., \( \tilde{A} = \tilde{V} D \tilde{V}^{-1} \). In what follows, we will continue to use the notation introduced in that section and establish upper bounds on the spectral norms of the individual blocks of \( \tilde{V} \) and \( \tilde{V}^{-1} \), which will play a crucial role in the subsequent error analysis.

\begin{lemma}
Under Assumption~\ref{assump:linear_no_resonance_gap}, supposing $\tilde{F}_2 = Q^{-1}F_2 Q^{\otimes 2}$ is $s$-column-sparse, we have 
\begin{align}
\norm{\tilde{V}_{i,j}} \le \binom{j-1}{i-1} \lrb{\frac{4s\norm{\tilde{F}_2}}{\Delta}}^{j-i},~&~\forall~i \le j.
\end{align}
\label{lem:vij_norm_bound_poincare_domain}
\end{lemma}

\begin{lemma}
Under Assumption~\ref{assump:linear_no_resonance_gap}, supposing $\tilde{F}_2 = Q^{-1}F_2 Q^{\otimes 2}$ is $s$-column-sparse, we have 
\begin{align}
\norm{\lrb{\tilde{V}^{-1}}_{i,j}} \le \binom{j-1}{i-1}\lrb{\frac{4s\norm{\tilde{F}_2}}{\Delta}}^{j-i},~&~\forall i \le j.
\end{align}
\label{lem:wij_norm_bound_poincare_domain}    
\end{lemma}

The proofs of these key Lemmas~\ref{lem:vij_norm_bound_poincare_domain} and \ref{lem:wij_norm_bound_poincare_domain} are presented in 
Appendix~\ref{app:proof_norm_bound_vij_v_inv_ij_poincare_domain}. In particular,  Lemma~\ref{lem:wij_norm_bound_poincare_domain} relies crucially on the binary forest representation of $\tilde{V}^{-1}_{i,j}$.

Equipped with these bounds on the norms of the blocks of $\tilde{V}$ and $\tilde{V}^{-1}$, we are now ready to derive an upper bound on the error of truncated Carleman linearization for the system of interest.

\begin{thm}[Carleman error bound for nonresonant systems in the Poincar\'{e} domain]
\label{thm:non-resonant_Poincare}
Let $x(t)$ be the solution to the $N$-dimensional quadratic ODE system
    \begin{align}
        \dot{x}(t) = F_1 x(t) + F_2 x(t)^{\ot 2},
    \end{align}
and
    \begin{align}
        \dot{y}(t) = A y(t),\qquad y(0)=[x(0),x(0)^{\ot 2},\dots,x(0)^{\ot k}]^T,
        \label{eq:carleman_linearized_system}
    \end{align}
be the Carleman linearization of the system truncated at level $k$, see Eq.~\eqref{eq:CarlemanODE}. Assume $F_1=Q\Lambda Q^{-1}$, where $\Lambda=\diag{\lambda_1, \lambda_2, \dots, \lambda_N}$ with $\realpart{\lambda_i} \le 0$ for all $i \in [N]$. Suppose $(\lambda_1, \lambda_2, \dots, \lambda_N)$ is $\Delta$-nonresonant for some $\Delta>0$, and $\tilde{F_2}=Q^{-1}F_2Q^{\otimes 2}$ is $s$-column-sparse. Let 
$\norm{\tilde{x}_{\rm max}}=\max_{t \in [0,T]} \norm{Q^{-1}x(t)}$, and define
\begin{align}
    R_{\Delta} \defeq \frac{8 s \norm{\tilde{F}_2} \norm{\tilde{x}_{\rm max}}}{\Delta}.
\end{align} 
Then the Carleman error vectors $\eta_i(t) = x(t)^{\ot i}-y^{[i]}(t)$ satisfy 
\begin{align}
        \norm{\eta_i(t)} \le k t  \norm{Q}^i \norm{\tilde{F}_2} \norm{\tilde{x}_{\rm max}}^{i+1} R_{\Delta}^{k-i}\binom{k-1}{i-1} ,
        \label{eq:error_bound_ith_block_poincare_domain}
\end{align}
for all $i\in [k]$ and $t \in [0, T]$. In particular, for $i=1$, we have
\begin{align}
\norm{\eta_1(t)} 
\le k t   \norm{Q}\norm{\tilde{F}_2} \norm{\tilde{x}_{\rm max}}^2 R_{\Delta}^{k-1},    
\label{eq:error_bound_1st_block_poincare_domain}
\end{align}
for all $t \in [0, T]$.
\label{thm:error_bound_nonresonant_poincare_domain}
\end{thm}

\begin{proof}
Recall that $\tilde{x}(t)=Q^{-1}x(t)$ is governed by the ODE system:
\begin{align}
\dot{\tilde{x}}(t)  = \Lambda \tilde{x}(t) + \tilde{F}_2 \tilde{x}(t)^{\otimes 2},   
\end{align}
where $\tilde{F}_2 = Q^{-1}F_2 Q^{\otimes 2}$. The order-$k$ Carleman linearization for this system is  
\begin{align}
\dot{\tilde{y}}(t) =\tilde{A}\tilde{y}(t)    
\end{align}
 where $\tilde{y}(t)=[\tilde{y}^{[1]}(t), \tilde{y}^{[2]}(t),\dots,\tilde{y}^{[k]}(t)]^T$ with $\tilde{y}^{[i]}(t) \in \myC^{N^i}$ for $i \in [k]$. We have found $\tilde{V}$ such that $\tilde{A}=\tilde{V}D\tilde{V}^{-1}$, where $D=\diag{\Lambda_1, \Lambda_2, \dots, \Lambda_k}$, 
with
\begin{align}
\Lambda_i=\diag{\lambda_{j_1}+\lambda_{j_2}+\dots+\lambda_{j_i}:~j_1,j_2,\dots,j_i \in [N]}, ~~~\forall i \in [k].
\end{align}
We will first establish an upper bound on $\norm{\tilde{x}(t)^{\otimes i}-\tilde{y}^{[i]}(t)}$, and then convert it into an upper bound on $\norm{x(t)^{\otimes i}-y^{[i]}(t)}$, for each $i \in [k]$.

Let $\tilde{b}(t)=[\tilde{b}_1(t),\tilde{b}_2(t),\dots,\tilde{b}_k(t)]^T$, where $\tilde{b}_i(t)=\vec 0 \in \myC^{N^i}$ for $i \in [k-1]$, and 
\begin{align}
\tilde{b}_k(t)=\tilde{A}_{k,k+1}\tilde{x}(t)^{\otimes (k+1)}.    
\end{align}
Moreover, let $\tilde{\eta}_i(t)=\tilde{x}(t)^{\otimes i}-\tilde{y}^{[i]}(t)$ for $i \in [k]$, and let $\tilde{\eta}(t)=[\tilde{\eta}_1(t), \tilde{\eta}_2(t), \dots, \tilde{\eta}_k(t)]^T$. Then $\tilde{\eta}(t)$ is governed by the ODE system:
\begin{align}
    \dot{\tilde{\eta}}(t) = \tilde{A} \tilde{\eta}(t)+ \tilde{b}(t).
\end{align}
As a consequence, we have
\begin{align}
\tilde{\eta}(t) =\int_0^t \tilde{V}e^{D(t-\tau)}\tilde{V}^{-1}\tilde{b}(\tau)d\tau.
\end{align}
In other words, for each $i \in [k]$, we have
\begin{align}
\tilde{\eta}_i(t) = \sum_{j=i}^k \int_0^t \tilde{V}_{i,j} e^{\Lambda_j(t-\tau)} \lrb{\tilde{V}^{-1}}_{j,k} \tilde{b}_k(\tau) 
d\tau.
\end{align}
It follows that
\begin{align}
\norm{\tilde{\eta}_i(t)} \le \sum_{j=i}^k \int_0^t \norm{\tilde{V}_{i,j}} \norm{e^{\Lambda_j(t-\tau)}} \norm{\lrb{\tilde{V}^{-1}}_{j,k}} \norm{\tilde{b}_k(\tau)} 
d\tau.
\end{align}

Note that $\Lambda_j$ is a diagonal matrix whose diagonal entries have non-positive real parts. Thus, we have $\norm{e^{\Lambda_j r}} \le 1$ for all $r \in \R$. Then using the definition of $\tilde{A}_{k, k+1}$, $\tilde{x}(t)$ and $\norm{\tilde{x}_{\rm max}}$, we obtain 
\begin{align}
    \norm{\tilde{b}_k(\tau)} \le \norm{\tilde{A}_{k, k+1}} \norm{\tilde{x}(\tau)}^{k+1}
    \le k \norm{\tilde{F}_2} \cdot \norm{\tilde{x}_{\rm max}}^{k+1}, ~~~\forall \tau \in [0, T]. 
\end{align}
Finally, since the given system satisfies Assumption~\ref{assump:linear_no_resonance_gap} and $\tilde{F}_2$ is $s$-column-sparse, Lemma~\ref{lem:vij_norm_bound_poincare_domain} provides an upper bound on $\norm{\tilde{V}_{i,j}}$:
\begin{align}
   \norm{\tilde{V}_{i,j}} \le \binom{j-1}{i-1} \lrb{\frac{4s\norm{\tilde{F}_2}}{\Delta}}^{j-i},
\end{align}
while Lemma~\ref{lem:wij_norm_bound_poincare_domain} provides an upper bound on $\norm{\lrb{\tilde{V}^{-1}}_{j,k}}$:
\begin{align}
    \norm{\lrb{\tilde{V}^{-1}}_{j,k}} \le \binom{k-1}{j-1}\lrb{\frac{4s\norm{\tilde{F}_2}}{\Delta}}^{k-j}.
\end{align}
Combining the above facts yields
\begin{align}
\norm{\tilde{\eta}_i(t)} &\le \sum_{j=i}^k \int_0^t \norm{\tilde{V}_{i,j}} \norm{\lrb{\tilde{V}^{-1}}_{j,k}} \norm{\tilde{b}_k(\tau)} 
d\tau \\
& \le \sum_{j=i}^k k t \norm{\tilde{F}_2} \norm{\tilde{x}_{\rm max}}^{k+1} 
    \binom{k-1}{j-1}\binom{j-1}{i-1}\lrb{\frac{4s\norm{\tilde{F}_2}}{\Delta}}^{k-i} \\
    & = k t \norm{\tilde{F}_2} \norm{\tilde{x}_{\rm max}}^{k+1} \binom{k-1}{i-1} \lrb{\frac{8 s\norm{\tilde{F}_2}}{\Delta}}^{k-i},
    \label{eq:tilde_eta_i_bound}
\end{align}
for all $i \in [k]$ and $t \in [0, T]$.

Now since $x(t)=Q\tilde{x}(t)$ and $y^{[i]}(t)=Q^{\otimes i} \tilde{y}^{[i]}(t)$, we have ${\eta}_i(t)={x}(t)^{\otimes i}-{y}^{[i]}(t)
=Q^{\otimes i} \tilde{\eta}_i(t)$, for all $i \in [k]$ and $t \ge 0$. Then, by Eq.~\eqref{eq:tilde_eta_i_bound}, we obtain
\begin{align}
\norm{\eta_i(t)}
&\le \norm{Q}^{i} \norm{\tilde{\eta}_i(t)} \\
&\le 
k t \norm{Q}^{i}\norm{\tilde{F}_2} \norm{\tilde{x}_{\rm max}}^{k+1} \binom{k-1}{i-1} \lrb{\frac{8 s\norm{\tilde{F}_2}}{\Delta}}^{k-i}\\
&= k t  \norm{Q}^i\norm{\tilde{F}_2} \norm{\tilde{x}_{\rm max}}^{i+1}\binom{k-1}{i-1} \lrb{\frac{8 s \norm{\tilde{F}_2}\norm{\tilde{x}_{\rm max}}}{\Delta}}^{k-i},
\end{align}
for all $i \in [k]$ and $t \in [0, T]$.

This completes the proof of the theorem.
\end{proof}

Theorem~\ref{thm:error_bound_nonresonant_poincare_domain} implies that, for nonresonant systems in the Poincaré domain, the solution of the Carleman linearized system converges to that of the original system as the truncation order $k$ increases, provided the system satisfies a suitable condition. 
Specifically, suppose the following condition holds:
\begin{align}
R_\Delta = \frac{8 s  \norm{\tilde{F}_2} \norm{\tilde{x}_{\rm max}}}{\Delta} < 1.
\end{align}
Then, from Eq.~\eqref{eq:error_bound_1st_block_poincare_domain}, it follows that 
\begin{align}
    \norm{x(t) - y^{[1]}(t)} \to 0, ~\textrm{as}~k \to +\infty.
\end{align}
Furthermore, we can derive an upper bound on $k$ to ensure that the approximation error remains below a prescribed threshold. Specifically, let 
\begin{align}
\beta = \frac{8s}{t\Delta \norm{Q} \norm{\tilde{x}_{\rm max}}}.
\end{align}
Then for any $\epsilon > 0$, by choosing  
\begin{align}
k \ge \left \lceil \frac{W_{-1}(\zeta)}{\myln{R_\Delta}} \right \rceil,
\end{align}
where $W_{-1}$ is the lower branch of the Lambert W function and 
\begin{align}
\zeta=\max(\beta\epsilon \myln{R_\Delta}, -1/e),    
\end{align}
we can guarantee that $\norm{x(t) - y^{[1]}(t)} \le \epsilon$. Similar bounds on $k$ can be derived to ensure that $ \norm{x(t)^{\otimes j} - y^{[j]}(t)} \le \epsilon$ for any given $j \in \myN$.

\subsubsection{Error analysis for systems in the Siegel domain}
\label{subsubsec:error_analysis_siegel_domain}

We now turn to the case where $(\lambda_1, \lambda_2, \dots, \lambda_N)$ is nonresonant and lies in the Siegel domain. In contrast to the Poincar\'{e} domain case, we do not establish a general Carleman convergence result in this setting. Instead, in Appendix~\ref{app:error_analysis_siegel_domain}, we derive an upper bound on the Carleman truncation error,  which may not vanish as the truncation order $k$ increases, and we discuss the technical challenges that hinder a convergence proof in the Siegel domain. Whether such a result can be obtained through alternative methods remains an open question.

The obstruction in the Siegel domain can be understood as a small-denominator problem. In the Poincar\'{e} domain, the relevant no-resonance denominators grow with the order: 
\begin{align}
\left|\lambda_i-\lambda_{j_1}-\cdots-\lambda_{j_l}\right|
\geq \Delta(l-1).
\end{align}
Consequently, the corresponding inverse gap factors become smaller as $l$ grows, which is crucial for controlling the binary-forest expansion of the diagonalizing transformation and its inverse. By contrast, in the Siegel domain, the relevant gap denominators can become arbitrarily small as the order increases. Although nonresonance ensures that no denominator vanishes at any finite order, the available worst-case bounds allow the inverse gap factors to grow polynomially with the order, and the resulting estimates for the norms of the blocks of $\tilde{V}$ and $\tilde{V}^{-1}$ acquire factorial-type overheads. This is why the general error bound in Appendix~\ref{app:error_analysis_siegel_domain} does not by itself imply convergence as $k\to\infty$, even for arbitrarily small $\|\tilde{F}_2\|$. This does not show that Carleman convergence necessarily fails in the Siegel domain; rather, it shows that the Poincar\'{e}-domain argument does not directly extend without additional structure or a sharper analysis of the small-denominator contributions.

However, under additional assumptions on the nonlinearity, one can establish the convergence of Carleman linearization for systems in the Siegel domain. The key idea is to decompose the original system into two subsystems, each residing within the Poincar\'{e} domain, provided that the quadratic term $F_2$ does not induce coupling between the subspaces associated with these subsystems. Under this structural constraint, convergence for the full system follows directly from the convergence of Carleman linearization applied independently to each subsystem. We consider this special case because the technical result derived below forms the foundation for efficient quantum algorithms targeting a class of nonlinear Schr\"{o}dinger equations, as discussed in Section~\ref{subsec:nonlinear_schrodinger_equations}. These equations, which are closely related to coupled oscillator systems, are shown to be BQP-hard. Consequently, we obtain a BQP-complete problem formulated in terms of nonlinear Schr\"{o}dinger dynamics.

Formally, we assume that the system satisfies the following structural condition:
\begin{assumption}
There exist subsets $S_+$ and $S_-$ of $[N]$ such that $[N]=S_+\sqcup S_-$ and:
\begin{itemize}
    \item For each $j \in S_+$, we have $\realpart{\lambda_j} < 0$ or $\imaginarypart{\lambda_j} > 0$;
    \item For each $j \in S_-$, we have $\realpart{\lambda_j} < 0$ or $\imaginarypart{\lambda_j} < 0$.    
    \item $\tilde{F}_2=Q^{-1} F_2 Q^{\otimes 2}$ satisfies the conditions:
    \begin{align}
        \tilde{F}_2 \ket{a}\ket{b} &= \tilde{F}_2 \ket{b}\ket{a} = 0, ~&\forall a \in S_+, ~\forall b \in S_-; \\
        \tilde{F}_2 \ket{a}\ket{b} &\in \myspan{\ket{c}:~c \in S_+}, ~&\forall a,b \in S_+;  \\    
        \tilde{F}_2 \ket{a}\ket{b} &\in 
        \myspan{\ket{c}:~c \in S_-}, ~&\forall a,b \in S_-.         
    \end{align} 
    Equivalently, $\tilde{F}_2=P_+\tilde{F}_2(P_+\otimes P_+) + P_-\tilde{F}_2(P_-\otimes P_-)$, where $P_{\pm}=\sum_{j \in S_{\pm}} \ket{j}\bra{j}$.
\end{itemize}
\label{assump:block_diagonal_f2}
\end{assumption}
Then, we can easily show the following.
\begin{lemma}
If $(\lambda_1,\lambda_2,\dots,\lambda_N)$ is nonresonant and Assumption~\ref{assump:block_diagonal_f2} holds, then there exists $\Delta > 0$ such that
both $(\lambda_j:~j \in S_+)$ and $(\lambda_j:j \in S_-)$ are $\Delta$-nonresonant.
\label{lem:delta_constant_diagonal_f2}
\end{lemma}
\begin{proof}
First, note that $\vec\lambda_+\defeq (\lambda_j:~j \in S_+)$ is nonresonant and lies in the Poincar\'{e} domain. Then by Lemma~\ref{lem:equivalence_linear_no_resonance_gap_poincare_domain}, there exists $\Delta_1>0$ such that
for any $i$, $j_1$, $j_2$, $\dots$, $j_m \in S_+$, where $m \ge 2$, we have
\begin{align}
\abs{\lambda_i-\lambda_{j_1}-\lambda_{j_2}-\dots-\lambda_{j_m}} \ge (m-1)\Delta_1.    
\end{align}
Similarly, since $\vec\lambda_-\defeq (\lambda_j:~j \in S_-)$ is also nonresonant and belongs to the Poincar\'{e} domain, there exists $\Delta_2>0$ such that
for any $i$, $j_1$, $j_2$, $\dots$, $j_m \in S_-$ with $m \ge 2$, we have 
\begin{align}
\abs{\lambda_i-\lambda_{j_1}-\lambda_{j_2}-\dots-\lambda_{j_m}} \ge (m-1)\Delta_2.    
\end{align}
Setting $\Delta=\mymin{\Delta_1, \Delta_2}$ yields the desired result.     
\end{proof}

Now, under Assumption~\ref{assump:block_diagonal_f2}, we can obtain a bound on the error of truncated Carleman linearization, which closely resembles the one established in Theorem~\ref{thm:error_bound_nonresonant_poincare_domain}. The proof of the following theorem can be found in Appendix~\ref{app:error_bound_nonresonant_siegel_domain_decompose_f2}.

\begin{thm}
Let $x(t)$ be the solution to the $N$-dimensional quadratic ODE system
\begin{align}
    \dot{x}(t) = F_1 x(t) + F_2 x(t)^{\ot 2},
\end{align}
and
\begin{align}
    \dot{y}(t) = A y(t),\qquad y(0)=[x(0),x(0)^{\ot 2},\dots,x(0)^{\ot k}]^T,
\end{align}
be the Carleman linearization of the system truncated at level $k$, see Eq.~\eqref{eq:CarlemanODE}. Assume $F_1=Q\Lambda Q^{-1}$, where $\Lambda=\diag{\lambda_1, \lambda_2, \dots, \lambda_N}$ with  $\realpart{\lambda_i} \le 0$ for all $i \in [N]$. Suppose the system satisfies Assumption 
\ref{assump:block_diagonal_f2}, and 
let $\Delta > 0$ be as defined in Lemma~\ref{lem:delta_constant_diagonal_f2}.
Moreover, suppose
$\tilde{F_2}=Q^{-1}F_2Q^{\otimes 2}$ is $s$-column-sparse. Let 
$\norm{\tilde{x}_{\rm max}}=\max_{t \in [0,T]} \norm{Q^{-1}x(t)}$, and define
\begin{align}
    R_{\Delta} \defeq \frac{8 s \norm{\tilde{F}_2} \norm{\tilde{x}_{\rm max}}}{\Delta}.
\end{align}
Then the Carleman error vectors $\eta_i(t) = x(t)^{\ot i}-y^{[i]}(t)$ satisfy 
\begin{align}
\norm{\eta_i(t)} \le \sqrt{2} k t  \norm{Q}^i \norm{\tilde{F}_2} \norm{\tilde{x}_{\rm max}}^{i+1} R_{\Delta}^{k-i} \binom{k-1}{i-1} ,
\end{align}
for all $i\in [k]$ and $t \in [0, T]$.
In particular, for $i=1$, we have
\begin{align}
\norm{\eta_1(t)} 
\le \sqrt{2} k t \norm{Q}\norm{\tilde{F}_2} \norm{\tilde{x}_{\rm max}}^2 R_{\Delta}^{k-1},    
\end{align}
for all $t \in [0, T]$.
\label{thm:error_bound_nonresonant_siegel_domain_decompose_f2}
\end{thm}

\subsection{Addressing a class of time-dependent nonlinear systems}
\label{subsec:certain_time_dependent_systems}

We next introduce a method for solving a class of differential equations with certain \emph{time-dependent} nonlinearities. The linear part of the system has spectrum which may lie in the Siegel domain and may also exhibit resonances. Our basic idea is to shift the spectrum vertically into the Poincar\'{e} domain via a simple change of variables, eliminating resonances simultaneously, and then apply the convergence results for Carleman linearization to the transformed system. In other words, we build a reduction from the more challenging case -- where the spectrum may lie in the Siegel domain and may be resonant -- to the more tractable case, where the spectrum lies in the Poincar\'{e} domain and is nonresonant. 

Specifically, consider the time-dependent ODE system: 
\begin{align}
\dot{x}(t) = F_1 x(t) + e^{\i \omega t} F_2 x(t)^{\otimes 2},  
\label{eq:original_ode_system_time_dependent}
\end{align}
where $\omega \in \mathbb{R}^+$ \footnote{A similar result holds for $\omega \in \mathbb{R}^-$ as well, in which case the condition on the input becomes $|\operatorname{Im}(\lambda_i)| \le -\omega/4$. This extension is straightforward, and we therefore leave its details to the reader.}, $x(t) \in \mathbb{C}^N$, $F_2 \in \mathbb{C}^{N \times N^2}$, and $F_1 = Q \Lambda Q^{-1}$ with $\Lambda = \operatorname{diag}(\lambda_1, \lambda_2, \dots, \lambda_N)$. We assume that $\operatorname{Re}(\lambda_i) \le 0$ and $|\operatorname{Im}(\lambda_i)| \le \omega/4$ for all $i \in [N]$. Note that the spectrum of $F_1$ may lie within the Siegel domain and may also be resonant.

We introduce the change of variables \( z(t) = e^{\i \omega t} x(t) \),  transforming the system into a time-independent form:
\begin{align}
    \dot{z}(t) = (F_1 + \i \omega I) z(t) + F_2  z(t)^{\otimes 2},
\label{eq:transformed_ode_system_time_independent}
\end{align}
where the eigenvalues of \( F_1 + \i \omega I \), denoted by \( \lambda_j' = \lambda_j + \i \omega \), satisfy \( \realpart{\lambda_j'} \le 0 \) and
\begin{align}
0 < \frac{3 \omega}{4} \le \imaginarypart{\lambda_j'} \le \frac{5 \omega}{4}, \quad \forall j \in [N].
\end{align}
Clearly, this implies that \( (\lambda_1', \lambda_2', \dots, \lambda_N') \) lies in the Poincar\'{e} domain. Furthermore, it follows that the shifted eigenvalues are nonresonant: for any \( i, j_1, j_2, \dots, j_l \in [N] \) with \( l \ge 2 \), we have
\begin{align}
\imaginarypart{\lambda_{j_1}' + \lambda_{j_2}' + \dots + \lambda_{j_l}'-\lambda_i'}
\ge l \cdot \frac{3 \omega}{4} - \frac{5 \omega}{4}
\ge (l - 1) \cdot \frac{\omega}{4}.
\end{align}
This inequality ensures that $( \lambda_1', \lambda_2', \dots, \lambda_N')$ is $\frac{\omega}{4}$-nonresonant.

Suppose \( \tilde{F}_2 = Q^{-1} F_2 Q^{\otimes 2} \) is \( s \)-column-sparse, and let $\norm{\tilde{x}_{\rm max}}=\max_{t \in [0, T]}\norm{Q^{-1}x(t)}$. Then, by Theorem~\ref{thm:error_bound_nonresonant_poincare_domain}, the solution of the Carleman linearization for the transformed system~\eqref{eq:transformed_ode_system_time_independent} converges to \( z(t) \) as the truncation order \( k \) increases, provided that 
\begin{align}
    8 s \norm{\tilde{F}_2} \norm{\tilde{x}_{\rm max}} < \frac{\omega}{4}.
\end{align}
Multiplying the resulting approximation of \( z(t) \) by \( e^{-\i \omega t} \), we obtain an accurate approximation of the solution \( x(t) \) to the original system~\eqref{eq:original_ode_system_time_dependent}. Note that we are not claiming convergence of Carleman linearization for the original time-dependent system~\eqref{eq:original_ode_system_time_dependent}, which remains an open question.

To summarize, we have the following result on solving a class of time-dependent differential equations:
\begin{thm}
\label{thm:resonant_oscillating_F2}
Let $x(t)$ be the solution to the $N$-dimensional quadratic ODE system
\begin{align}
    \dot{x}(t) = F_1 x(t) + e^{\i \omega t} F_2 x(t)^{\ot 2},
\end{align}
where $\omega \in \R^+$.
Define the shifted ODE system
\begin{align}
\dot{z}(t) = (F_1+\i \omega I) z(t)  + F_2 z(t)^{\otimes 2},  \qquad 
z(0)=x(0),
\end{align}
and let
\begin{align}
    \dot{y}(t) = A y(t),\qquad y(0)=[z(0),z(0)^{\ot 2},\dots,z(0)^{\ot k}]^T,
\end{align}
be the Carleman linearization of the system truncated at level $k$, where $y(t)=[y^{[1]}(t)$, $y^{[2]}(t)$, $\dots$, $y^{[k]}(t)]^T$ with $y^{[j]}(t) \in \myC^{N^j}$ for $j \in [k]$. Assume $F_1=Q\Lambda Q^{-1}$, where $\Lambda=\diag{\lambda_1, \lambda_2, \dots, \lambda_N}$ with $\realpart{\lambda_i} \le 0$ and $\abs{\imaginarypart{\lambda_i}} \le \omega/4$ for all $i \in [N]$. Suppose 
$\tilde{F_2}=Q^{-1}F_2Q^{\otimes 2}$ is $s$-column-sparse. Let $\norm{\tilde{x}_{\rm max}}=\max_{t \in [0,T]}\norm{Q^{-1}x(t)}$, and define
\begin{align}
    R_\omega \defeq  \frac{32 s \norm{\tilde{F}_2} \norm{\tilde{x}_{\rm max}}}{\omega}.
\end{align}
Then for all $j\in [k]$ and $t \in [0, T]$, we have the error bound:
\begin{align}
    \norm{x(t)^{\otimes j}-e^{-\i j \omega t}y^{[j]}(t)} \le k t \norm{Q}^j \norm{\tilde{F}_2} \norm{\tilde{x}_{\rm max}}^{j+1} R_\omega^{k-j} \binom{k-1}{j-1} .
\end{align}
In particular, for $j=1$, we have
\begin{align}
    \norm{x(t)-e^{-\i \omega t}y^{[1]}(t)} 
    \le k t \norm{Q} \norm{\tilde{F}_2} \norm{\tilde{x}_{\rm max}}^2 R_\omega^{k-1},    
\end{align}
for all $t \in [0, T]$.
\label{thm:error_bound_nonresonant_siegel_domain_time_dependent_f2}
\end{thm}

In Section~\ref{subsec:quantum_algorithm_nonresonant_systems}, we employ  Theorem~\ref{thm:error_bound_nonresonant_siegel_domain_time_dependent_f2} to develop quantum algorithms for a class of time-dependent nonlinear differential equations, including certain nonlinear Schr\"{o}dinger equations. In Section~\ref{subsubsec:coupled_nonlinear_oscillators}, we apply these algorithms to efficiently simulate the dynamics of a class of coupled nonlinear oscillators, thereby establishing the BQP-completeness of the associated computational problem.

\newpage 

\section{Quantum algorithms for nonlinear dynamical systems}
\label{sec:quantum_algorithms_nonlinear_systems}

In this section, we propose quantum algorithms for solving several classes of nonlinear differential equations, including stable, conservative, and nonresonant systems -- as well as certain resonant cases. These algorithms build upon the error analysis of the Carleman scheme developed in Sections~\ref{sec:stable}, \ref{sec:conservative} and \ref{sec:nonresonant_systems}, and they serve as the foundation for the BQP-completeness results presented in Section~\ref{sec:exponential_quantum_advantage}.

To simplify the exposition, we introduce notation that will be used throughout this section and in Section~\ref{sec:exponential_quantum_advantage}. For any matrix $A$, we define $\operatorname{cs}(A)$ as the column sparsity of a matrix $A$ when $A \neq 0$, and set $\operatorname{cs}(A)=1$ when $A=0$.  Additionally, for any vector $v=(v_1,v_2,\dots,v_N) \in \mathbb{C}^N$, we define $\ket{v}$ as the normalized quantum state proportional to $v$, i.e.,
\begin{align}
\ket{v} \defeq \frac{\sum_{j=1}^N v_j \ket{j}}{\norm{v}}.    
\end{align}

\subsection{Nonresonant and certain resonant nonlinear problems}
\label{subsec:quantum_algorithm_nonresonant_systems}

We begin by designing a quantum algorithm for preparing a quantum state approximately proportional to the solution of a nonresonant differential equation in the Poincar\'{e} domain. The algorithm builds on the Carleman linearization of nonlinear systems and leverages existing quantum algorithms for solving linear differential equations, but the analysis requires nontrivial additions that we shall discuss below. Our main result is stated in the following theorem, whose proof is given in Appendix~\ref{app:poincare_domain_algorithm}.
\begin{thm}
\label{thm:Poincarealgorithm}
Consider the $N$-dimensional quadratic ODE system
\begin{align}
\label{eq:Poincaresystem}
    \dot{x}(t) = F_1 x(t) + F_2 x(t)^{\ot 2},
\end{align}
where $F_1=Q\Lambda Q^{-1}$ with $\Lambda=\diag{\lambda_1, \lambda_2, \dots, \lambda_N}$ and $\realpart{\lambda_i} \le 0$ for all $i \in [N]$. Suppose $(\lambda_1, \lambda_2, \dots, \lambda_N)$ is $\Delta$-nonresonant for some $\Delta>0$. Let $\tilde{F}_2=Q^{-1}F_2Q^{\otimes 2}$, and assume that 
\begin{align}
8\rho s \norm{Q^{-1}} \norm{\tilde{F}_2} \le c\Delta,   \end{align}
where $\rho \ge \max_{t \in [0, T]} \norm{x(t)}$, $s=\operatorname{cs}(\tilde{F}_2)$, and $c \in (0,1/2)$ is a constant. 

Let $U_i$ be an $(\alpha_i, a_i, 0)$-block-encoding of $F_i$, for $i=1,2$, and let $V_0$ be a unitary that prepares the normalized initial state $\ket{x(0)}$. Then for any $\epsilon \in (0, 1)$, we can prepare a quantum state that is $\epsilon$-close (in Euclidean norm) to $\ket{x(T)}$ using 
$$\mytO{\alpha g T  2^k \kappa_Q^k \mylog{1/\epsilon}}$$
queries to (controlled-) $U_1$, $U_2$, $V_0$ and their inverses, where 
\begin{align}
g=\frac{\max_{t \in [0, T]} {\norm{x(t)}}}{\min_{t \in [0, T]} {\norm{x(t)}}}, \quad
\kappa_Q=\norm{Q}\cdot \norm{Q^{-1}}, \quad
\alpha=\alpha_1+{\rho\kappa_Q\alpha_2}, 
\end{align} 
and
\begin{align}
k=\myO{\max(\mylog{ g  \kappa_Q T\Delta/(s\epsilon)}, 1)}.    
\end{align}
In addition to these queries, the algorithm uses
$$\mytO{\alpha g T 2^k \kappa_Q^k \mylog{1/\epsilon}\cdot \mypoly{\mylog{N}}}$$
elementary quantum gates.
\label{thm:algorithm_complexity_poincare_domain_case}
\end{thm}

We emphasize that the column-sparsity parameter $s=\operatorname{cs}(\tilde{F}_2)$ in this theorem plays a different role from the block-encoding assumptions on $F_1$ and $F_2$. The column-sparsity parameter enters through the convergence analysis of the Carleman truncation in Section~\ref{sec:nonresonant_systems}, and affects the sufficient small-nonlinearity condition and the truncation order $k$. The block-encoding assumptions, by contrast, are the algorithmic input assumptions used to implement the truncated Carleman system. In particular, larger column-sparsity leads to more restrictive small-nonlinearity conditions; however, once these conditions are satisfied, the algorithmic cost is still governed by the displayed complexity bound. Thus, $s$ should not be interpreted as an additional cost parameter in the assumed block-encodings of $F_1$ and $F_2$. The same distinction applies to Theorems~\ref{thm:algorithm_complexity_decomposable_f2_case} and~\ref{thm:algorithm_complexity_time_dependent_f2_case} below.

Let us now highlight why Theorem~\ref{thm:algorithm_complexity_poincare_domain_case} does not follow trivially from prior analysis. The analysis of the Carleman scheme presented in Section~\ref{subsec:error_analysis_nonresonant_systems} enables one to approximate the nonlinear ODE system~\eqref{eq:Poincaresystem} by a linear ODE system with controlled error. This linear system can then be solved using the quantum algorithm in Ref.~\cite{jennings2023cost}, which essentially builds on the algorithm of Ref.~\cite{berry2017quantum} while improving its original analysis. However, applying this result requires an \emph{a priori} bound on the norm of the exponential of the Carleman matrix $A$, which is crucial for bounding the condition number of the linear system that embeds the Carleman ODE. This norm bound is closely tied to the structure of the matrix $V$ and its inverse $V^{-1}$, which jointly diagonalize the Carleman matrix $A$, and it is established by leveraging Lemmas \ref{lem:vij_norm_bound_poincare_domain} and \ref{lem:wij_norm_bound_poincare_domain}, both of which are based on the formalism developed in Section~\ref{subsec:carleman_matrix_diagonalization}.

Another aspect that requires careful attention concerns the specific data structure obtained when solving the Carleman ODE system. In particular, it is necessary to ensure that a state close to $\ket{x(T)}$ can be efficiently extracted from the output of the ODE solver. This involves appropriately rescaling the original system \eqref{eq:Poincaresystem}, making use of the flexibility discussed in Section~\ref{sec:rescaling}. The remainder of the analysis is relatively standard, but for the reader’s convenience, we present it in full in Appendix~\ref{app:poincare_domain_algorithm}.

Theorem~\ref{thm:algorithm_complexity_poincare_domain_case} applies to nonresonant systems in the Poincar\'{e} domain and builds on the error bound established in Theorem~\ref{thm:error_bound_nonresonant_poincare_domain}. Similarly, certain systems in the Siegel domain can be addressed using the error bounds provided by Theorems~\ref{thm:error_bound_nonresonant_siegel_domain_decompose_f2} and~\ref{thm:error_bound_nonresonant_siegel_domain_time_dependent_f2}, both of which are derived from Theorem~\ref{thm:error_bound_nonresonant_poincare_domain}. These results lead to the following algorithmic consequences:

\begin{thm}
Consider the $N$-dimensional quadratic ODE system
\begin{align}
    \dot{x}(t) = F_1 x(t) + F_2 x(t)^{\ot 2},
\end{align}
where $F_1=Q\Lambda Q^{-1}$ with $\Lambda=\diag{\lambda_1, \lambda_2, \dots, \lambda_N}$ and $\realpart{\lambda_i} \le 0$ for all $i \in [N]$. Suppose there exist disjoint subsets $S_+$ and $S_-$ of $[N]$ such that $[N]=S_+\sqcup S_-$, and both $\lrb{\lambda_j:~j\in S_+}$ and $\lrb{\lambda_j:~j \in S_-}$ are $\Delta$-nonresonant for some $\Delta>0$. Moreover, let  $\tilde{F}_2=Q^{-1}F_2Q^{\otimes 2}$, and assume 
\begin{align}
\tilde{F}_2=P_+ \tilde{F}_2 (P_+\otimes P_+) + P_- \tilde{F}_2 (P_-\otimes P_-),
\end{align}
where $P_{\pm}=\sum_{j \in S_\pm}\ket{j}\bra{j}$, and that
\begin{align}
8\rho s \norm{Q^{-1}} \norm{\tilde{F}_2} \le c\Delta,    \end{align}
where $\rho \ge \max_{t \in [0, T]} \norm{x(t)}$, $s=\operatorname{cs}(\tilde{F}_2)$, and $c \in (0,1/2)$ is a constant.

Let $U_i$ be a $(\alpha_i, a_i, 0)$-block-encoding of $F_i$, for $i=1,2$, and let $V_0$ be a unitary that prepares the normalized initial state $\ket{x(0)}$. Then for any $\epsilon \in (0, 1)$, we can prepare a quantum state that is $\epsilon$-close (in Euclidean norm) to $\ket{x(T)}$ using 
$$\mytO{\alpha g T  2^k \kappa_Q^k\mylog{1/\epsilon}}$$
queries to (controlled-) $U_1$, $U_2$, $V_0$ and their inverses, where 
\begin{align}
g=\frac{\max_{t \in [0, T]} {\norm{x(t)}}}{\min_{t \in [0, T]} {\norm{x(t)}}},~
\kappa_Q=\norm{Q}\cdot \norm{Q^{-1}},~
\alpha=\alpha_1+{\rho\kappa_Q\alpha_2},~ 
\end{align} 
and
\begin{align}
k=\myO{\max(\mylog{ g  \kappa_Q T\Delta/(s\epsilon)}, 1)}.    
\end{align} 
In addition to these queries, the algorithm uses
$$\mytO{\alpha g T 2^k \kappa_Q^k \mylog{1/\epsilon}\cdot \mypoly{\mylog{N}}}$$
elementary quantum gates.
\label{thm:algorithm_complexity_decomposable_f2_case}
\end{thm}

\begin{thm}
Consider the $N$-dimensional quadratic ODE system
\begin{align}
    \dot{x}(t) = F_1 x(t) + e^{\i \omega t}F_2 x(t)^{\ot 2},
\end{align}
where $\omega \in \R^+$ \footnote{A similar result holds for $\omega \in \R^-$ as well. This extension is straightforward and hence we leave the details to the reader.}, 
$F_1=Q\Lambda Q^{-1}$ with $\Lambda=\diag{\lambda_1, \lambda_2, \dots, \lambda_N}$, $\realpart{\lambda_i} \le 0$ and $\abs{\imaginarypart{\lambda_i}} \le \omega/4$, for all $i \in [N]$. Let $\tilde{F}_2=Q^{-1}F_2Q^{\otimes 2}$, and assume that
\begin{align}
32\rho s \norm{Q^{-1}} \norm{\tilde{F}_2} \le c\omega,    
\end{align}
where $\rho \ge \max_{t \in [0, T]} \norm{x(t)}$, $s=\operatorname{cs}(\tilde{F}_2)$, and 
$c \in (0,1/2)$ is a constant.

Let $U_i$ be a $(\alpha_i, a_i, 0)$-block-encoding of $F_i$, for $i=1,2$, and let $V_0$ be a unitary that prepares the normalized initial state $\ket{x(0)}$. Then for any $\epsilon \in (0, 1)$, we can prepare a quantum state that is $\epsilon$-close (in Euclidean norm) to $\ket{x(T)}$ using 
$$\mytO{\alpha g T  2^k \kappa_Q^k\mylog{1/\epsilon}}$$
queries to (controlled-) $U_1$, $U_2$, $V_0$ and their inverses, where 
\begin{align}
g=\frac{\max_{t \in [0, T]} {\norm{x(t)}}}{\min_{t \in [0, T]} {\norm{x(t)}}},~
\kappa_Q=\norm{Q}\cdot \norm{Q^{-1}},~
\alpha=\alpha_1+{\rho\kappa_Q\alpha_2},~ 
\end{align} 
and
\begin{align}
k=\myO{\max(\mylog{ g  \kappa_Q T\omega/(s\epsilon)}, 1)}.    
\end{align}
In addition to these queries, the algorithm uses
$$\mytO{\alpha g T 2^k \kappa_Q^k \mylog{1/\epsilon}\cdot \mypoly{\mylog{N}}}$$
elementary quantum gates.
\label{thm:algorithm_complexity_time_dependent_f2_case}
\end{thm}

The proofs of Theorems~\ref{thm:algorithm_complexity_decomposable_f2_case} and \ref{thm:algorithm_complexity_time_dependent_f2_case} closely follow that of Theorem~\ref{thm:algorithm_complexity_poincare_domain_case} and are thus omitted here.

Under the assumption $\kappa_Q = \myO{1}$, the query and gate complexities of the quantum algorithms in Theorems~\ref{thm:algorithm_complexity_poincare_domain_case}, \ref{thm:algorithm_complexity_decomposable_f2_case}, and \ref{thm:algorithm_complexity_time_dependent_f2_case} scale polynomially with $T$, $1/\epsilon$ and other relevant parameters. Notably, this includes the case where $F_1$ is normal, thereby encompassing certain nonlinear Schr\"{o}dinger equations as a special case.

However, the complexities of these algorithms contain the undesirable factor $2^k\kappa_Q^k$, which arises from our upper bound on the norm of the exponential of the associated Carleman matrix. Even when $\kappa_Q=\myO{1}$, this factor could lead to unfavorable polynomial dependence on $T/\epsilon$, though this is not always the case. It remains unclear whether such a factor is inherent to the problem or merely a consequence of our proof technique. We leave it as an open question whether a refined analysis could reduce or eliminate this factor. In the next section, we shall see that a reduction is possible at the price of introducing a new oracle.

\subsection{Nonlinear stable and conservative problems via quantum Lyapunov transform}
\label{subsec:quantum_algorithm_stable_conservative_systems}

We shall now turn to quantum algorithms for nonlinear stable systems, under the settings of the convergence results in Theorem~\ref{thm:Carlemanstable}. Here, the issue is, similarly to the previous case, to get bounds on the condition number of the linear system of equation resulting from Carleman linearization. It is useful to recall that this condition number can be upper bounded by a quantity involving $\max_{t \in [0,T]} \| e^{At}\|$. Since $P$ is a Lyapunov matrix for $F_1$, one can take $\mathcal{P} = \oplus_{l=1}^k P^{\otimes l}$ and prove by continuity that it is a Lyapunov matrix for the Carleman system. Under these conditions, the bound $\|e^{A t}\| \leq \kappa_{\mathcal{P}} e^{\mu_{\mathcal{P}}t}$ holds, with  $\kappa_{\mathcal{P}}$ the condition number of $\mathcal{P}$ and $\mu_{\mathcal{P}}$ a $\mathcal{P}$-log norm which turns negative for small enough nonlinearity. Even leaving aside the issue of how small the nonlinearity ought to be in order to have $\mu_{\mathcal{P}}<0$, we also have that $\kappa_{\mathcal{P}} = \kappa_P^k$, which means that even in the optimistic setting where $\kappa_{P} = \myO{1}$ we have an undesired polynomial scaling in $T/\epsilon$, in analogy with the algorithm in Section~\ref{subsec:quantum_algorithm_nonresonant_systems}.

Given these considerations, here we shall introduce a different algorithm. The algorithm requires an extra oracle -- access to a block-encoding of $Q$, the square root of $P$:
\begin{align}
    P = Q^\dag Q.
\end{align}
Note that this implicitly replaces $Q$ with $Q^{-1}$ relative to earlier sections, but we do so for notational simplicity in this subsection.

At high level, the idea is that the extra oracle is used to encode the problem in a coordinate system where the problem is stable with respect to the identity matrix. Given the data structure resulting from solving the Carleman ODE system with a quantum ODE solver, we shall rotate only the relevant components (corresponding to the first Carleman block) back to the original frame, and amplitude amplify them. 

Note that the most natural setting is perhaps one where one has access to a block-encoding of $P$, rather than to its square root. However, if a block-encoding of $P$ is available, the square root of $P$ can be constructed using known techniques, such as phase estimation~\cite{babbush2023exponential} or quantum signal processing~\cite{gilyen2019quantum}, which only gives an extra complexity factor of $\sqrt{\kappa_P}$ (and potentially an extra factor $1/\epsilon$ if we used phase estimation). This does not change the fact that, under these stronger assumptions, we avoid the undesired scaling with $\kappa_P^k = \myO{(T/\epsilon)^{\log \sqrt{\kappa_P}}}$, rather getting $\myO{\kappa_P^2}$ overheads which are $\myO{1}$ under the previously mentioned optimistic setting where $\kappa_P = \myO{1}$.


\begin{thm}
\label{thm:algorithmstable}
    Consider the setting and assumptions of Theorem~\ref{thm:Carlemanstable}. Assume access to an $\alpha_{F_i}$ block-encoding of $F_i$ with $i=0,1,2$, and an $\alpha_Q = \myO{\|Q\|}$ block-encoding of a matrix $Q$ such that $Q^\dag Q = P$. Finally, assume access to a state preparation for $\ket{\tilde{x}(0)}$, where $\tilde{x}(0) = Q x(0)$.
    Setting $h = k \left(\alpha_{F_0} \alpha_Q + \alpha_{F_1} \kappa_Q + \alpha_{F_2} \kappa_Q^2/\alpha_Q\right)$ with $k = \myO{\log(T/\epsilon)}$, there is a quantum algorithm that outputs a quantum state $\epsilon$-close to the \emph{ history state}
\begin{align}
    \ket{\psi_1} \propto \sum_{m=1}^{\lceil T/h \rceil} \| x(mh)\| \ket{x(m h)},
\end{align}
    with cost 
\footnotesize
\begin{align}
\myO{\|Q^{-1}\|(\alpha_{F_0} \alpha_Q + \alpha_{F_1} \kappa_Q + \frac{\alpha_{F_2} \kappa_Q^2}{\alpha_Q})\log(1/\epsilon) \log^2(T/\epsilon)} & \quad \textrm{calls to each of the block-encodings of the } F_i, \nonumber \\
\myO{\|Q^{-1}\|(\alpha_{F_0} \alpha_Q + \alpha_{F_1} \kappa_Q + \frac{\alpha_{F_2} \kappa_Q^2}{\alpha_Q})\log^3(1/\epsilon) \log^2(T/\epsilon) \kappa^3_Q} & \quad  \textrm{calls to the block-encoding of } Q, \nonumber \\
 \myO{\kappa_Q (\alpha_{F_0} \alpha_Q + \alpha_{F_1} \kappa_Q + \frac{\alpha_{F_2} \kappa_Q^2}{\alpha_Q})\log(1/\epsilon)\log^2(T/\epsilon)} & \quad \textrm{calls to the state preparation of $\ket{x(0)}$}. 
\end{align}
\normalsize
If $Q$ is unitary then $\kappa_Q = \alpha_Q = \|Q\|^{-1} =1$ and the $\log(1/\epsilon)$ factors drop. If $\kappa_P = \myO{1}$ these reduce to 
\footnotesize
\begin{align}
\myO{(\alpha_{F_0}  + \alpha_{F_1}  + \alpha_{F_2})\log(1/\epsilon) \log^2(T/\epsilon)} & \quad \textrm{calls to each of the block-encodings of the } F_i, \nonumber \\
\myO{(\alpha_{F_0}  + \alpha_{F_1}  + \alpha_{F_2} ) \log^3(1/\epsilon) \log^2(T/\epsilon)} & \quad  \textrm{calls to the block-encoding of } Q, \nonumber \\
 \myO{(\alpha_{F_0}  + \alpha_{F_1}  + \alpha_{F_2})\log(1/\epsilon)\log^2(T/\epsilon)} & \quad \textrm{calls to the state preparation of $\ket{x(0)}$}. 
\end{align}
\normalsize
\end{thm}
The proof of this theorem is provided in Appendix \ref{app:proof_thm_algorithmstable}.

\begin{rmk}[Partial Differential Equations]
    Finally note that some care needs to be taken when evaluating complexity for settings where the number of variables scales with some parameter $N$, as is the case when the nonlinear ODE problem at hand results from the discretization of a partial differential equation. Then, the time discretization $h$ normally gets a dependence from $N$ (since we set $\|A\| h \leq 1$), so that $N$ and the number of time steps $M$ are linked. We then end up with a family of ODEs parameterized by $N$, which leads to some obstacles when evaluating the asymptotic complexity. First, convergence of the algorithm is typically ensured by the $R$ number conditions only up to finite $N$. Second, we can have situations where for fixed $N$ the scaling of the complexity with $M$ is $\log(M)$, but as $N \rightarrow \infty$ the log-norm of the system goes to zero, and we transition from stable to marginally stable behavior, hence from $O(\log T)$ to $O(T)$ scaling. In some settings it may be clearer to assess the performance with non-asymptotic query complexities such as those developed in Ref.~\cite{jennings2023cost} (updated with the tightest available constant factor upper bounds).
\end{rmk}

\begin{rmk}[Conservative systems]
    Note that the above result also applies to conservative systems if we take $Q$ to be the matrix diagonalizing $F_1$. The condition number will introduce an extra scaling with $T$ compared to the scalings presented above. We shall see this directly in action in Section~\ref{subsec:bqphard}.
\end{rmk}

\newpage 

\section{Exponential quantum advantage for nonlinear problems}
\label{sec:exponential_quantum_advantage}

In Section~\ref{subsec:coupled_oscillators}, we apply the algorithms from the previous section to a class of coupled nonlinear oscillator problems and establish their BQP-completeness. In Section~\ref{subsec:nonlinear_schrodinger_equations}, we apply the same framework to a class of nonlinear Schr\"{o}dinger equations and show that these problems are also BQP-complete. In both cases, BQP-completeness leverages the BQP-hardness recently established for systems consisting of exponentially many coupled linear oscillators~\cite{babbush2023exponential}. Finally, in Section~\ref{subsec:bqphard}, we go slightly further by presenting a BQP-complete family of problems that includes inherently nonlinear BQP-hard instances. These instances involve only limited nonlinearity, with the nonlinearity strength-measured in operator norm-scaling inverse-polynomially with the number of qubits.

\subsection{Coupled oscillator systems}
\label{subsec:coupled_oscillators}

In this subsection, we consider the problem of simulating the dynamics of coupled linear and nonlinear oscillators. We begin by reviewing the result of Ref.~\cite{babbush2023exponential}, which establishes the BQP-completeness of simulating coupled linear oscillator dynamics. We then extend this result to a nonlinear setting, which encompasses the linear case as a special instance, and show that -- under reasonable assumptions -- the problem remains efficiently solvable and is therefore also BQP-complete. 

\subsubsection{Linear oscillators}
\label{subsubsec:coupled_linear_oscillators}

We begin with a brief overview of the results in Ref.~\cite{babbush2023exponential} on simulating the dynamics of coupled linear oscillators.

Consider a classical system of coupled harmonic oscillators, consisting of \( N \) point masses $m_1$, $m_2$, $\dots$, $m_N$ connected by springs. For \( 1 \le i < j \le N \), let \( \kappa_{i,j} = \kappa_{j,i} \ge 0 \) denote the spring constant coupling the \( i \)-th and \( j \)-th masses. If there is no spring connecting masses \( i \) and \( j \), we set \( \kappa_{i,j} = \kappa_{j,i} = 0 \). In addition, each mass may be connected to a fixed wall by a spring, and we let \( \kappa_{i,i} \ge 0 \) denote the spring constant of the spring attached to the \( i \)-th mass for all \( i \in [N] \).

The dynamics of this system is governed by the differential  equations
\begin{align}
m_i \ddot{x}_i(t) = \sum_{j \ne i} \kappa_{i,j} \left[ x_j(t) - x_i(t) \right] - \kappa_{i,i} x_i(t),    
\end{align}
for each \( i \in [N] \), where \( x_i(t) \) denotes the position of the \( i \)-th mass at time \( t \). 

The coupling structure of the system can be concisely described using graph theory. We represent the network of springs as a weighted graph $G$, where the vertex set is $[N]$ and the edge weights are given by $\lrcb{\kappa_{i,j}:~1\le i\le j\le N}$. Note that $G$ may contain self-loops, corresponding to springs that connect individual masses to fixed walls. The \emph{Laplacian} matrix of $G$ is defined as
\begin{align}
L_G = \sum_{i=1}^N \lrb{\sum_{j=1}^N \kappa_{i,j}} \ket{i}\bra{i} - \sum_{1\le i<j\le N} \kappa_{i,j} \lrb{\ket{i}\bra{j}+\ket{j}\bra{i}}.
\label{eq:laplacian_matrix_def}
\end{align}

Now let \( x(t)=(x_1(t), x_2(t), \dots, x_N(t))^T \in \mathbb{R}^N \) denote the vector of all positions at time \( t \), and define the mass matrix \( M = \sum_{i=1}^N m_i \ket{i}\bra{i} \). With the above  definitions, the equations of motion can be written compactly as
\begin{align}
M \ddot{x}(t) = -L_G x(t).    
\label{eq:coupled_linear_oscillators_equation}
\end{align}

The kinetic energy of the masses at time $t$ is given by
\begin{align}
K(t) \coloneqq \frac{1}{2} \sum_{i=1}^N m_i \dot{x}_i(t)^2,
\end{align}
and the potential energy stored in the springs at time $t$ is
\begin{align}
U(t) \coloneqq \frac{1}{2} \sum_{1\le i<j\le N} \kappa_{i,j} \left( x_i(t) - x_j(t) \right)^2+\frac{1}{2}\sum_{i=1}^N\kappa_{i,i}x_i(t)^2.
\end{align}
The total mechanical energy, defined as \( E = K(t) + U(t) \), is conserved over time.

Suppose the masses $\{ m_i : 1 \le i \le N \}$, coupling constants $\{ \kappa_{i,j} : 1 \le i \le j \le N \}$, and the initial positions $x(0) \in \mathbb{R}^N$ and velocities $\dot{x}(0) \in \mathbb{R}^N$ are provided in a suitable format. We are interested in estimating the kinetic energy of a specified subset of masses at a given time $T$. Specifically, for a subset $V \subseteq [N]$, let
\begin{align}
K_V(t) = \frac{1}{2} \sum_{i \in V} m_i \dot{x}_i(t)^2
\end{align}
denote the kinetic energy of the masses in $V$ at time $t$. The goal is to compute an estimate $\hat{k}_V$ such that
\begin{align}
\abs{\hat{k}_V - \frac{K_V(T)}{E}} \le \epsilon,
\end{align}
for given parameters $T > 0$ and $\epsilon > 0$, where $E$ is the conserved total energy of the system. (Alternatively, we may wish to estimate the ratio of the potential energy stored in a subset $V \subseteq [N] \times [N]$ of the springs to the total energy $E$
at time $T$, up to an additive error $\epsilon$. Similar results hold for this alternative problem as well.)

Ref.~\cite{babbush2023exponential} proves that a decision version of this problem is BQP-complete under suitable input encodings and parameter regimes. Specifically, the decision problem is defined as follows:
\begin{definition}[Coupled Linear Oscillator Problem (CLOP)]
Let $K=(\kappa_{i,j})_{1\le i,j\le N}$ be the $N\times N$ matrix of spring constants $\kappa_{i,j} \ge 0$ and assume it is $d$-sparse. Let $M=\diag{m_1,m_2,\dots,m_N}$ be the $N \times N$ diagonal matrix of masses $m_i>0$. Define the state 
\begin{align}
\ket{\psi(t)}=\frac{1}{\sqrt{2E}}\begin{pmatrix}
    \sqrt{M}\dot{x}(t) \\
    \i \mu(t)
\end{pmatrix},
\end{align}
where $\mu(t) \in \mathbb{R}^{\frac{N(N+1)}{2}}$ is a vector consisting of $N$ entries of the form $\sqrt{\kappa_{i,i}}x_i(t)$ and $N(N-1)/2$ entries  of the form $\kappa_{j,k}(x_j(t)-x_k(t))$. Suppose we are given access to oracles that compute the entries of $K$ and $M$, along with a unitary $\mathcal{W}$ that prepares the initial state $\ket{\psi(0)}$. 
Let $V \subseteq [N]$ be a subset for which there exists an oracle that decides its membership. Given a time $T>0$ and real numbers $0\le a < b \le 1$, the goal is to decide whether $K_V(T)/E < a$ or $K_V(T)/E > b$, promised that one of these holds.
\end{definition}

\begin{definition}[Poly-CLOP]
The setup follows that of CLOP, with additional constraints on the input encodings and parameter settings. Specifically, we require that $N=2^n$ and 
\begin{align}
d, T, \kappa_{\max}, m_{\max}, \frac{1}{m_{\min}} = \myO{\operatorname{poly}(n)},
\end{align}
where $m_{\max} = \max_{i \in [N]} m_i$, $m_{\min} = \min_{i \in [N]} m_i$ and $\kappa_{\max} = \max_{i,j \in [N]} \kappa_{i,j}$. In addition, we assume that the oracles for $K$, $M$, $V$, and the unitary $\mathcal{W}$ that prepares $\ket{\psi(0)}$ can be implemented by quantum circuits of  $\myO{\operatorname{poly}(n)}$ sizes. Finally, we require that the promise gap satisfies 
\begin{align}
b - a = \myOmega{\frac{1}{\operatorname{poly}(n)}}.    
\end{align}
\end{definition}

\begin{thm}[Ref.~\cite{babbush2023exponential}]
Poly-CLOP is BQP-complete. In fact, the problem remains BQP-complete even if we further require that $\kappa_{\rm max}\le 4$,
$m_{\rm max}=m_{\rm min}=1$, $d \le 4$, $x(0)=(0,0,\dots,0)^T$ and $\dot{x}(0)=(1, -1, 0, \dots, 0)^T$.
\end{thm}

Next, we briefly describe the quantum algorithm in Ref.~\cite{babbush2023exponential} for solving Poly-CLOP in $\myO{\mypoly{n}}$ time. This serves as the foundation for our generalization to the nonlinear setting. Ref.~\cite{babbush2023exponential} also proves that Poly-CLOP is BQP-hard by constructing a reduction from simulating polynomial-size quantum circuits over the gate set consisting of Hadamard, Pauli-$X$, and Toffoli gates to this problem. However, this aspect is not the main focus of our work.

The key idea of Ref.~\cite{babbush2023exponential} is to embed the dynamics of coupled linear oscillators into a Hamiltonian simulation problem. Before presenting this construction, it is useful to first make some observations about the Laplacian matrix
$L_G$. Clearly, it is a real symmetric matrix. We claim that the eigenvalues of $L_G$ lie in the interval  $[\tilde{\kappa}_{\rm min}, 2\tdeg{G}]$,
where $\tilde{\kappa}_{\rm min} \defeq \min_{i \in [N]} \kappa_{i,i}$, and
\begin{align}
\tdeg{G} \defeq \max_{i \in [N]} \lrb{\sum_{j\neq i} \kappa_{i,j}+\frac{\kappa_{i,i}}{2}}  \le d\kappa_{\rm max}.  
\end{align}
To see this, note that for any vector $x \in \R^N$, we have
\begin{align}
x^T L_G x &= \sum_{1\le i<j \le N} \kappa_{i,j}(x_i-x_j)^2 + \sum_{i=1}^N \kappa_{i,i}x_i^2 
\ge \sum_{i=1}^N \kappa_{i,i}x_i^2 
\ge \tilde{\kappa}_{\mathrm{min}} x^Tx,
\end{align}
and
\begin{align}
x^T L_G x &= \sum_{1\le i<j \le N} \kappa_{i,j}(x_i-x_j)^2 + \sum_{i=1}^N \kappa_{i,i}x_i^2 \\
& \le 
2\sum_{1\le i<j \le N} \kappa_{i,j}(x_i^2+x_j^2)
+\sum_{i=1}^N \kappa_{i,i}x_i^2 \\
&=\sum_{i=1}^N \lrb{2\sum_{j\neq i}\kappa_{i,j}+\kappa_{i,i}} x_i^2 \\
&\le 2 \tdeg{G} \cdot x^Tx.
\end{align}

The Laplacian matrix $L_G$ admits a convenient factorization into two rectangular matrices that are transposes of each other\footnote{This decomposition of $L_G$ has been leveraged to develop several algorithms, including the ones presented in Refs.~\cite{wang2017efficient} and \cite{chakraborty2019thepower}.}. Specifically, the \emph{signed incidence} matrix of $G$ is an $N \times \frac{N(N+1)}{2}$ matrix defined by
\begin{align}
B_G = \sum_{1\le i<j\le N} \sqrt{\kappa_{i,j}} \left( \ket{i} - \ket{j} \right) \bra{i,j} + \sum_{i=1}^N \sqrt{\kappa_{i,i}}\ket{i}\bra{i,i}.
\label{eq:incidence_matrix_def}
\end{align}
It is straightforward to verify that 
\begin{align}
B_G B_G^\dagger = L_G.    
\end{align}
This implies that the singular values of $B_G$ lie in the interval $[\sqrt{\tilde{\kappa}_{\mathrm{min}}}, \sqrt{2\tdeg{G}}]$. 

Next, we define $\bar{L}_G=\sqrt{M}^{-1}L_G \sqrt{M}^{-1}$ and 
$\bar{B}_G=\sqrt{M}^{-1}B_G$, and obtain
\begin{align}
\bar{B}_G \bar{B}_G^{\dagger}=\bar{L}_G.  
\end{align}
Note that the singular values of $\bar{B}_G$ lie in the interval $[\sqrt{\tilde{\kappa}_{\mathrm{min}}/m_{\mathrm{max}}}, \sqrt{2\tdeg{G}/m_{\mathrm{min}}}]$.

We now introduce a change of variables by setting $y(t)=\sqrt{M} x(t)$. In terms of $y(t)$, the system obeys the differential equation:
\begin{align}
  \frac{d}{dt}  \begin{pmatrix}
\dot{y}(t) \\
\i \bar{B}_G^{\dagger} y(t)
    \end{pmatrix}
    = \i \begin{pmatrix}
        0 & \bar{B}_G \\
        \bar{B}_G^{\dagger} & 0
    \end{pmatrix}
    \begin{pmatrix}
\dot{y}(t) \\
\i \bar{B}_G^{\dagger} y(t)
    \end{pmatrix}.
    \label{eq:embedding_linear_oscillator_schrodinger_equation}
\end{align}
Note that \( \dot{y}(t) \) encodes the kinetic energy of the masses, while \( \i \bar{B}_G^{\dagger} y(t) \) encodes the potential energy stored in the springs. Specifically, the kinetic energy at time $t$ is given by
\begin{align}
K(t) = \frac{1}{2} \| \dot{y}(t) \|^2 = \frac{1}{2} \sum_{i=1}^N m_i \dot{x}_i(t)^2,
\end{align}
and the potential energy at time $t$ is
\begin{align}
U(t) = \frac{1}{2} \| \i \bar{B}_G^{\dagger} y(t) \|^2 = \frac{1}{2} \sum_{1\le i<j\le N} \kappa_{i,j} \left( x_i(t) - x_j(t) \right)^2+\frac{1}{2}\sum_{i=1}^N\kappa_{i,i}x_i(t)^2.
\end{align}
Since the total energy \( E = K(t) + U(t) \) is conserved, it is valid to introduce the normalized state vector
\begin{align}
\ket{\psi(t)} = \frac{1}{\sqrt{2E}} \begin{pmatrix}
\dot{y}(t) \\
\i \bar{B}_G^{\dagger} y(t)
\end{pmatrix}.
\end{align}
Moreover, define the Hermitian operator
\begin{align}
H = -\begin{pmatrix}
0 & \bar{B}_G \\
\bar{B}_G^{\dagger} & 0
\end{pmatrix},
\end{align}
which can be interpreted as a Hamiltonian.  
The dynamics then satisfies
\begin{align}
\frac{d}{dt} \ket{\psi(t)} = -\i H \ket{\psi(t)},
\label{eq:schrodinger_equation_for_linear_oscillators}
\end{align}
indicating that the evolution of the coupled linear oscillators is embedded into a Schr\"{o}dinger equation. Note that we have the following upper bound on the spectral norm of $H$:
\begin{align}
\norm{H} \le \sqrt{2\tdeg{G}/m_{\mathrm{min}}} \le \sqrt{2 d\kappa_{\rm max}/m_{\mathrm{min}}}.    
\end{align}

To solve Poly-CLOP, we apply an existing Hamiltonian simulation algorithm to Eq.~\eqref{eq:schrodinger_equation_for_linear_oscillators} and obtain a sufficiently accurate approximation of $\ket{\psi(T)}$. Then, letting $\Pi_V= \ket{0}\bra{0}\otimes \sum_{j \in V} \ket{j}\bra{j}$, we can estimate
\begin{align}
   \bra{\psi(T)} \Pi_V \ket{\psi(T)} =  \frac{K_V(T)}{E}
\end{align}
by repeatedly performing the projective measurement
$\lrcb{\Pi_V, I-\Pi_V}$ on $\ket{\psi(T)}$ and recording the frequency of observing the outcome corresponding to $\Pi_V$. The projective measurement $\lrcb{\Pi_V, I - \Pi_V}$ can be implemented using one query to the $V$-membership oracle. This number of repetitions can be reduced quadratically using amplitude estimation.

It remains to analyze the cost of the above algorithm. We need to prepare the state $\ket{\psi(T)}$ within $\myO{\epsilon}$ precision to ensure that the final estimate is $\epsilon$-close to $K_V(T)/E$. Using Ref.~\cite{low2017optimal}, the Hamiltonian evolution can be simulated with 
$$\myO{T\Lambda + \log(1/\epsilon)}$$ calls to a $(\Lambda, m, \epsilon')$-block-encoding of $H$, where $\epsilon' = \mytO{\epsilon / (T\Lambda)}$. Ref.~\cite{babbush2023exponential} demonstrates how to construct such a block-encoding of $H$ with a rescaling factor 
\begin{align}
\Lambda = \sqrt{2d\kappa_{\rm max}/m_{\rm min}}    
\end{align}
using $\myO{1}$ queries to the oracles for $K$ and $M$, along with $\myO{\log^2(N T\Lambda m_{\rm max} / (\epsilon m_{\rm min}))}$ elementary gates. In addition, the Hamiltonian simulation also requires one query to the unitary $\mathcal{W}$ that prepares the initial state $\ket{\psi(0)}$. Finally, amplitude estimation requires $\myO{1/\sqrt{\epsilon}}$ repetitions of the Hamiltonian simulation stage to achieve $\myO{\epsilon}$ accuracy in estimating the target quantity, and each repetition entails one query to the $V$-membership oracle. 

In Poly-CLOP, since the promise gap is $\myOmega{1/\mypoly{n}}$, it suffices to set $\epsilon = \myTheta{1/\mypoly{n}}$. Moreover, under the given parameter regime, we have $\Lambda = \myO{\mypoly{n}}$ and $m_{\rm max}/m_{\rm min}=\myO{\mypoly{n}}$. Therefore, it is straightforward to verify that the above algorithm requires a total of $\myO{\mypoly{n}}$ calls to the oracles for $K$, $M$, $V$ and  $\ket{\psi(0)}$, in addition to $\myO{\mypoly{n}}$ elementary gates. Since each of these oracles can be implemented in $\myO{\mypoly{n}}$ time, we conclude that Poly-CLOP is in BQP.

\subsubsection{Nonlinear oscillators}
\label{subsubsec:coupled_nonlinear_oscillators}

Next, we consider the simulation of coupled nonlinear oscillators governed by the equation:
\begin{align}
M\ddot{x}(t) = -L_G x(t)+ e^{\i \omega t}\lrsb{Q_1 (\dot{x}(t) \otimes \dot{x}(t)) + Q_2 (\dot{x}(t) \otimes x(t)) + Q_3 (x(t)\otimes x(t))
},
\label{eq:coupled_nonlinear_oscillators_equation}
\end{align}
where $\omega \in \R^+$ \footnote{A similar result holds for $\omega \in \R^-$ as well. Since this extension is straightforward, we leave the details to the reader.} denotes a driving frequency that modulates the nonlinear interaction terms. Compared to Eq.~\eqref{eq:coupled_linear_oscillators_equation} for coupled linear oscillators, this equation incorporates additional nonlinearities through the terms $\dot{x}(t)\otimes \dot{x}(t)$, $\dot{x}(t)\otimes x(t)$ and $x(t)\otimes x(t)$. The matrices $Q_1, Q_2, Q_3 \in \R^{N \times N^2}$ specify the coefficients associated with these respective nonlinear interactions. Here,  $M=\diag{m_1,m_2,\dots,m_N}$ is the same mass matrix as in Section~\ref{subsubsec:coupled_linear_oscillators}, and $L_G$, introduced in Eq.~\eqref{eq:laplacian_matrix_def}, is the Laplacian matrix of the network $G$ that encodes the linear interactions.

Due to the presence of the factor $e^{\i \omega t}$ in the quadratic terms, the solution $x(t)$ to Eq.~\eqref{eq:coupled_nonlinear_oscillators_equation} is generally complex-valued, even when the initial condition $x(0)$ is a real vector. Consequently, we interpret $x(t)$ as the position vector of $N$ coupled \emph{complex} nonlinear oscillators. Equations of this type are often used as models in contexts where oscillatory driving interacts with nonlinear effects, including nonlinear circuits, dispersive wave propagation, and nonlinear optics.

To facilitate the implementation of quantum algorithm for this problem, we impose the following structure on the matrices $Q_1$, $Q_2$, $Q_3$:
\begin{align}
Q_1 &= \sqrt{M} G_1 \lrb{\sqrt{M}\otimes \sqrt{M}},
\label{eq:def_q1}\\
Q_2 &= -\sqrt{M} G_2 \lrb{\sqrt{M} \otimes B_G^\dagger}, 
\label{eq:def_q2} \\
Q_3 &= -\sqrt{M} G_3 \lrb{B_G^\dagger \otimes B_G^\dagger}, 
\label{eq:def_q3}
\end{align}
where $B_G$ is the signed incidence matrix of the network $G$, as defined in Eq.~\eqref{eq:incidence_matrix_def}, and
$G_1 \in \R^{N \times N^2}$, $G_2 \in \R^{N \times \frac{N^2(N+1)}{2}}$ and
$G_3 \in \R^{N \times \frac{N^2(N+1)^2}{4}}$ admit efficient block-encodings.

We now introduce the variables $\alpha(t)=\sqrt{M}\dot{x}(t)$
and $\beta(t)=\i B_G^\dagger x(t)$.
Then, using Eqs.~\eqref{eq:coupled_nonlinear_oscillators_equation}, \eqref{eq:def_q1}, \eqref{eq:def_q2} and \eqref{eq:def_q3}, one can  verify that the pair $(\alpha(t), \beta(t))$  evolves according to the equation:
\begin{align}
    \begin{pmatrix}
        \dot{\alpha}(t) \\
        \dot{\beta}(t)        
    \end{pmatrix}
    =\i \begin{pmatrix}
0 & \bar{B}_G \\
\bar{B}_G^{\dagger} & 0
\end{pmatrix} \begin{pmatrix}
        \alpha(t) \\
        \beta(t)        
    \end{pmatrix}
    +e^{\i \omega t}\begin{pmatrix}
    G_1 & \i G_2 & G_3 \\
    0 & 0 & 0
\end{pmatrix} \begin{pmatrix}
    \alpha(t) \otimes \alpha(t) \\
    \alpha(t) \otimes \beta(t) \\
    \beta(t) \otimes \beta(t)
\end{pmatrix},
\end{align}
where $\bar{B}_G=\sqrt{M}^{-1}B_G$, as before.

Letting $z(t)=\begin{pmatrix}
    \alpha(t) \\
    \beta(t)
\end{pmatrix}$, the system can be written compactly as:
\begin{align}
\dot{z}(t) = -\i H z(t) + e^{\i \omega t} F_2 \lrb{z(t)\otimes z(t)},    
\end{align}
where
\begin{align}
H = -\begin{pmatrix}
0 & \bar{B}_G \\
\bar{B}_G^{\dagger} & 0
\end{pmatrix},~&~
F_2=\begin{pmatrix}
    G_1 & \i G_2 & 0 & G_3 \\
    0 & 0 & 0 & 0
\end{pmatrix}. 
\end{align}
Therefore, the dynamics of the coupled nonlinear oscillators has been reformulated as a nonlinear Schr\"{o}dinger-type equation. In particular, when the nonlinear interaction terms $G_1$, $G_2$ and $G_3$ vanish, the equation reduces to a linear Schr\"{o}dinger equation.

The vectors $\alpha(t)$ and $\beta(t)$ encode the kinetic and potential energy information of the oscillators at time $t$, respectively. Specifically, the kinetic energy of the system at time $t$ is given by
\begin{align}
    K(t) \defeq \frac{1}{2}\norm{\alpha(t)}^2= \frac{1}{2} \dot{x}(t)^\dagger M \dot{x}(t) = \frac{1}{2}\sum_{i=1}^N m_i \abs{\dot{x}_i(t)}^2,
\end{align}
while the potential energy is given by
\begin{align}
    U(t) \defeq \frac{1}{2}\norm{\beta(t)}^2= \frac{1}{2} x(t)^\dagger L_G x(t)=\frac{1}{2}\sum_{1\le i<j\le N}\kappa_{i,j}\abs{x_i(t)-x_j(t)}^2 + \frac{1}{2}\sum_{i=1}^N \kappa_{i,i}\abs{x_i(t)}^2.
\end{align}
Thus, the total energy of the system at time $t$ is 
\begin{align}
    E(t) \defeq  \frac{1}{2}\lrb{\norm{\alpha(t)}^2 + \norm{\beta(t)}^2}
    =\frac{1}{2}\dot{x}(t)^\dagger M \dot{x}(t) + \frac{1}{2}x(t)^\dagger L_G x(t).
    \label{eq:total_energy_nonlinear_setting}
\end{align}
In contrast to the linear setting, the total energy $E(t)$ is not necessarily conserved due to the presence of nonlinear driving terms.

Given the initial condition $(x(0), \dot{x}(0))$, a time $T$ and a subset $V \subset [N]$, we aim to compute the ratio of the kinetic energy localized in the subset $V$ to the total energy at time $T$. Precisely, let
\begin{align}
    K_V(t) \defeq \frac{1}{2}\sum_{i \in V} m_i \abs{\dot{x}_i(t)}^2
\end{align}
denote the kinetic energy associated with the subset $V$ at time $t$. Our goal is to estimate the ratio $K_V(T)/E(T)$ within a specified additive error $\epsilon>0$. (Alternatively, one may wish to estimate the ratio of the potential energy stored in a subset $V \subseteq [N] \times [N]$ of the springs to the total energy at time $T$, up to an additive error $\epsilon$. Similar results hold for this alternative problem as well.)
Next, we demonstrate that, under appropriate input encodings and parameter settings, a decision version of the above problem is BQP-complete. Specifically, the decision problem is defined as follows:

\begin{definition}[Coupled Nonlinear Oscillator Problem (CNOP)]
Let $M=\diag{m_1,m_2,\dots,m_N}$ be an $N \times N$ diagonal matrix of masses $m_i>0$, and let $K=(\kappa_{i,j})_{1\le i,j\le N}$ be the $N\times N$ matrix of spring constants $\kappa_{i,j} \ge 0$, assumed to be $d$-sparse. Consider a network of coupled nonlinear oscillators governed by the equation:
\begin{align}
M\ddot{x}(t) = -L_G x(t)+ e^{\i \omega t}\lrsb{Q_1 (\dot{x}(t) \otimes \dot{x}(t)) + Q_2 (\dot{x}(t) \otimes x(t)) + Q_3 (x(t)\otimes x(t))
},
\end{align}
where $\omega \in \R^+$ and $L_G$ is defined as in Eq.~\eqref{eq:laplacian_matrix_def}. Suppose
\begin{align}
Q_1 = \sqrt{M} G_1 \lrb{\sqrt{M}\otimes \sqrt{M}},
~Q_2 = -\sqrt{M} G_2 \lrb{\sqrt{M} \otimes B_G^\dagger}, 
~Q_3 = -\sqrt{M} G_3 \lrb{B_G^\dagger \otimes B_G^\dagger},   
\end{align}
where $B_G$ is defined as in Eq.~\eqref{eq:incidence_matrix_def}, and
$G_1 \in \R^{N \times N^2}$, $G_2 \in \R^{N \times \frac{N^2(N+1)}{2}}$,
$G_3 \in \R^{N \times \frac{N^2(N+1)^2}{4}}$ are arbitrary.
Define the normalized state 
\begin{align}
\ket{\psi(t)}=\frac{1}{\sqrt{2E(t)}}\begin{pmatrix}
    \sqrt{M}\dot{x}(t) \\
    \i B_G^\dagger x(t)
\end{pmatrix},
\end{align} 
where $E(t)=\frac{1}{2}\dot{x}(t)^\dagger M \dot{x}(t)+\frac{1}{2}x(t)^\dagger L_G x(t)$. Suppose we are given access to oracles that compute the entries of $K$ and $M$, block-encodings of $G_1$, $G_2$ and $G_3$, along with a unitary $\mathcal{W}$ that prepares $\ket{\psi(0)}$. 
Let $V \subseteq [N]$ be a subset for which there exists an oracle that decides its membership. Given a time $T>0$ and real numbers $0\le a < b \le 1$, the goal is to decide whether $K_V(T)/E(T) <a$ or $K_V(T)/E(T) > b$, promised that one of these holds.    
\end{definition}

\begin{definition}[Poly-CNOP]
The setup follows that of CNOP, with additional constraints on the input encodings and parameter settings. Specifically, we require that $N=2^n$,  
\begin{align}
d, T, \kappa_{\max}, m_{\max}, \frac{1}{m_{\min}}, E_{\rm max}, \frac{1}{E_{\rm min}} = \myO{\operatorname{poly}(n)}, 
\end{align}
where $m_{\max} = \max_{i \in [N]} m_i$, $m_{\min} = \min_{i \in [N]} m_i$, $\kappa_{\max} = \max_{i,j \in [N]} \kappa_{i,j}$, $E_{\rm max}=\max_{t \in [0, T]} E(t)$
and
$E_{\rm min}=\min_{t \in [0, T]} E(t)$, and
\begin{align}
b - a = \myOmega{\frac{1}{\operatorname{poly}(n)}}.    
\end{align}
 
Moreover, let $\rho=\myO{\mypoly{n}}$ be a known upper bound on $E_{\rm max}$. Suppose $\bar{B}_G$ has the singular value decomposition $\bar{B}_G=U \Gamma W^\dagger$, where $\Gamma=\sum_{j=1}^N \gamma_j \ket{j}\bra{j}$ is diagonal, and $U=\sum_{j=1}^N \ket{u_j}\bra{j}$ and $W=\sum_{j=1}^{N} \ket{v_j}\bra{j}$ are unitary \footnote{Technically, $\bar{B}_G$ is a rectangular matrix and $W$ is an isometry, as each $\ket{v_j}$ resides in a $\frac{N(N+1)}{2}$-dimensional space. However, the system's state always lies within the $N$-dimensional subspace spanned by $\ket{v_1}, \ket{v_2}, \dots, \ket{v_N}$. Therefore, we restrict our attention to this subspace, effectively treating $\bar{B}_G$ as a square matrix and $W$ as a unitary.}. Define
\begin{align}
\tilde{G}_1=U^\dagger G_1 (U \otimes U),~
\tilde{G}_2=U^\dagger G_2 (U \otimes W),~
\tilde{G}_3=U^\dagger G_3 (W \otimes W).
\end{align}
Let $s_i=\operatorname{cs}(\tilde{G}_i)$, for $i=1,2,3$ \footnote{As in Section~\ref{subsec:quantum_algorithm_nonresonant_systems}, these column-sparsity parameters enter through the convergence condition for the Carleman truncation. They should not be interpreted as additional cost parameters in the assumed block-encodings of $G_1$, $G_2$ and $G_3$.}. We assume that
\begin{align}
    64 \rho (s_1+s_2+s_3)\sqrt{\norm{\tilde{G}_1}^2+\norm{\tilde{G}_2}^2+\norm{\tilde{G}_3}^2}\le c \omega
\end{align}
where $c \in (0,1/2)$ is a constant. In addition, we assume that 
\begin{align}
\omega \ge 4\sqrt{2d\kappa_{\rm max}/m_{\rm min}}, \quad \omega = \myO{(s_1+s_2+s_3) \cdot\mypoly{n}}.
\end{align}

Furthermore, we assume that the oracles for $K$, $M$, $V$, and the unitary $\mathcal{W}$ that prepares $\ket{\psi(0)}$ can be implemented by quantum circuits of $\myO{\operatorname{poly}(n)}$ sizes. In addition, for each $i=1,2,3$, 
a block-encoding of $G_i$ with a rescaling factor $\myO{\mypoly{n}}$ can be implemented using $\myO{\mypoly{n}}$ elementary gates. 

\end{definition}

\begin{thm}
Poly-CNOP is BQP-complete.
\end{thm}

\begin{proof}
Since Poly-CNOP generalizes Poly-CLOP -- corresponding to the special case where $G_1$, $G_2$ and $G_3$ are all zero -- and the latter has been shown to be BQP-hard in Ref.~\cite{babbush2023exponential}, it follows that Poly-CNOP is also BQP-hard.

It remains to prove that Poly-CNOP can be solved in $\myO{\mypoly{n}}$ quantum time. This is accomplished by embedding the dynamics of the coupled nonlinear oscillators into a nonlinear Schr\"{o}dinger equation and applying the quantum algorithm from Theorem~\ref{thm:algorithm_complexity_time_dependent_f2_case} to solve it.

Specifically, as mentioned earlier, the vector
\begin{align}
z(t)=\begin{pmatrix}
    \sqrt{M}\dot{x}(t)\\
    \i B_G^\dagger x(t)
\end{pmatrix} 
\end{align}
evolves according to the differential equation
\begin{align}
\dot{z}(t) = -\i H z(t) + e^{\i \omega t} F_2 z(t)^{\otimes 2},    
\label{eq:nonlinear_schrodinger_equation_for_nonlinear_oscillators}
\end{align}
where
\begin{align}
H = -\begin{pmatrix}
0 & \bar{B}_G \\
\bar{B}_G^{\dagger} & 0
\end{pmatrix},~&~
F_2=\begin{pmatrix}
    G_1 & \i G_2 & 0 & G_3 \\
    0 & 0 & 0 & 0
\end{pmatrix}. 
\end{align}

Next, we show that this nonlinear Schr\"{o}dinger equation satisfies all the conditions required by Theorem~\ref{thm:algorithm_complexity_time_dependent_f2_case}:
\begin{itemize}
\item First, in this problem, we have $F_1=-\i H$, where $H$ is Hermitian and has
eigenvalues $\pm \gamma_1, \pm \gamma_2$, $\dots$, $\pm \gamma_N$. As shown in Section~\ref{subsubsec:coupled_linear_oscillators}, each $\gamma_j$ satisfies $\gamma_j\le \sqrt{2 d \kappa_{\rm max} / m_{\rm min}}$. Thus, $F_1$ is anti-Hermitian, and all eigenvalues of $F_1$ have zero real parts, and imaginary parts lying in the interval $[-\omega/4, \omega/4]$.
\item Second, $F_1=-\i H$ admits the spectral decomposition
\begin{align}
    F_1 = \i\sum_{j=1}^N \gamma_j \lrb{\ket{+_j}\bra{+_j} - \ket{-_j}\bra{-_j}},    
\end{align}
where the eigenvectors are given by
\begin{align}
    \ket{\pm_j}=\frac{1}{\sqrt{2}}\lrb{\ket{0}\ket{u_j}\pm \ket{1}\ket{v_j}}.
\end{align}
Equivalently, $F_1$ can be written as 
\begin{align}
    F_1 = Q \Lambda Q^\dagger,
\end{align}
where
\begin{align}
Q &= 
\sum_{j=1}^N \ket{+_j}\bra{0}\bra{j}
+
\sum_{j=1}^N \ket{-_j}\bra{1}\bra{j} \\
&=\frac{1}{\sqrt{2}}\lrb{\ket{0}\bra{0}\otimes U+\ket{1}\bra{0}\otimes W+\ket{0}\bra{1}\otimes U-\ket{1}\bra{1}\otimes W}
\end{align}
is unitary and
\begin{align}
\Lambda &= \lrb{\ket{0}\bra{0} - \ket{1}\bra{1}} \otimes \i 
 \sum_{j=1}^N \gamma_j \ket{j}\bra{j}        
\end{align}
is diagonal. Meanwhile, observe that
\begin{align}
F_2 = \ket{0}\bra{0}\bra{0}\otimes G_1
    +\i\ket{0}\bra{0}\bra{1}\otimes G_2
    +\ket{0}\bra{1}\bra{1}\otimes G_3.    
\end{align}
Thus, we have
\begin{align}
\tilde{F}_2 & = Q^\dagger F_2 Q^{\otimes 2} \\
&=\ket{+}\bra{+}\bra{+} \otimes U^\dagger G_1 (U \otimes U)  +\ket{+}\bra{+}\bra{-} \otimes \i U^\dagger G_2 (U \otimes W) \nonumber\\
&\quad +\ket{+}\bra{-}\bra{-} \otimes U^\dagger G_3 (W \otimes W) \\
& 
=\ket{+}\bra{+}\bra{+} \otimes \tilde{G}_1 +\ket{+}\bra{+}\bra{-} \otimes \i \tilde{G}_2 +\ket{+}\bra{-}\bra{-} \otimes \tilde{G}_3,
\end{align}
where $\ket{\pm}=\frac{1}{\sqrt{2}}(\ket{0} \pm \ket{1})$. Then,
since $\tilde{G}_i$ is $s_i$-column-sparse, for $i=1,2,3$, it follows that $\tilde{F}_2$ is at most $2(s_1+s_2+s_3)$-column-sparse. Furthermore, it is straightforward to verify that
\begin{align}
    \norm{\tilde{F}_2} \le \sqrt{\norm{\tilde{G}_1}^2+\norm{\tilde{G}_2}^2+\norm{\tilde{G}_3}^2}. 
\end{align}
Then we have
\begin{align}
    32 \rho \cdot \operatorname{cs}(\tilde{F}_2) \cdot \norm{\tilde{F}_2} \le c \omega,
\end{align}
as required by Theorem~\ref{thm:algorithm_complexity_time_dependent_f2_case}.
\end{itemize}

To solve Poly-CNOP, we first apply the quantum algorithm from Theorem~\ref{thm:algorithm_complexity_time_dependent_f2_case} to the nonlinear Schr\"{o}dinger equation \eqref{eq:nonlinear_schrodinger_equation_for_nonlinear_oscillators} and obtain a state that is $\frac{b-a}{4}$-close to $\ket{\psi(T)}$. Next, define the projector $\Pi_V= \ket{0}\bra{0}\otimes \sum_{j \in V} \ket{j}\bra{j}$. We then estimate
\begin{align}
   \bra{\psi(T)} \Pi_V \ket{\psi(T)} = \frac{K_V(T)}{E(T)}
\end{align}
within additive error $\frac{b-a}{2}$ by repeatedly performing the projective measurement
$\lrcb{\Pi_V, I-\Pi_V}$ on $\ket{\psi(T)}$ and recording the frequency of observing the outcome corresponding to $\Pi_V$. This projective measurement can be implemented using a single query to the $V$-membership oracle. Moreover, the number of repetitions required for estimating the expectation value can be quadratically reduced using amplitude estimation.

It remains to analyze the cost of the above algorithm. As shown in Section~\ref{subsubsec:coupled_linear_oscillators}, a block-encoding of $F_1=-\i H$ with a rescaling factor $\myO{\mypoly{n}}$ can be implemented using $\myO{1}$ queries to the oracles for $K$ and $M$, along with $\myO{\mypoly{n}}$ elementary gates. 

Meanwhile, a block-encoding of $F_2$ can be constructed from block-encodings of $G_1$, $G_2$ and $G_3$ using standard linear combination of unitaries (LCU) technique. Since each of these block-encodings has a rescaling factor $\myO{\mypoly{n}}$ and can be implemented using $\myO{\mypoly{n}}$ elementary gates, it follows that 
a block-encoding of $F_2$ with a rescaling factor $\myO{\mypoly{n}}$ can also be implemented using $\myO{\mypoly{n}}$ elementary gates. 

Moreover, since $b-a=\myOmega{\frac{1}{\mypoly{n}}}$, it suffices to prepare $\ket{\psi(T)}$ within $\myTheta{\frac{1}{\mypoly{n}}}$ precision and to repeat this procedure $\myO{\mypoly{n}}$ times. 

Putting everything together, and invoking Theorem~\ref{thm:algorithm_complexity_time_dependent_f2_case}, we conclude that the quantum algorithm for Poly-CNOP requires $\myO{\mypoly{n}}$ calls to the oracles for $K$, $M$ and $V$, to the block-encodings of $G_1$, $G_2$ and $G_3$ with $\myO{\mypoly{n}}$ rescaling factors, and to the unitary $\mathcal{W}$ that prepares $\ket{\psi(0)}$, in addition to $\myO{\mypoly{n}}$ elementary gates. Since each of the aforementioned unitaries can be implemented in $\myO{\mypoly{n}}$ time, we conclude that the above algorithm runs in $\myO{\mypoly{n}}$ quantum time. Hence, Poly-CNOP is in BQP. This completes the proof of the theorem.
\end{proof}

\subsection{Nonlinear Schr\"{o}dinger equations}
\label{subsec:nonlinear_schrodinger_equations}

Finally, we consider a class of nonlinear Schr\"{o}dinger equations motivated by the coupled oscillator problem and demonstrate that solving these equations is BQP-complete.

Suppose the variables $\alpha(t) \in \mathbb{C}^N$ and $\beta(t) \in \mathbb{C}^K$, where $N \le K$,  evolve according to the following equations: 
\begin{align}
    \dot{\alpha}(t)&=\i A \beta(t) + C\alpha(t)\otimes \alpha(t) 
    +
    D\beta(t)\otimes \beta(t), \\ 
    \dot{\beta}(t)&=\i A^\dagger \alpha(t) + E\alpha(t) \otimes \beta(t) +
    F\beta(t) \otimes \alpha(t),  
\end{align}
where $A\in \mathbb{C}^{N \times K}$ has rank $N$, and  $C \in \mathbb{C}^{N \times N^2}$, $D \in \mathbb{C}^{N \times K^2}$, and $E, F \in \mathbb{C}^{K \times NK}$ are matrices to be specified later.

Let $z(t)=\begin{pmatrix}
    \alpha(t) \\
    \beta(t)
\end{pmatrix}$. Then the system can be written in the compact form:
\begin{align}
\dot{z}(t) = -\i H z(t) + F_2 z(t)^{\otimes 2},    
\label{eq:quadratic_nonlinear_schrodinger_equation}
\end{align}
where
\begin{align}
H = -\begin{pmatrix}
0 & A \\
A^{\dagger} & 0
\end{pmatrix},~&~
F_2=\begin{pmatrix}
    C & 0 & 0 & D \\
    0 & E & F & 0
\end{pmatrix}.    
\end{align}

In particular, the Schr\"{o}dinger equation \eqref{eq:embedding_linear_oscillator_schrodinger_equation}, which encodes the dynamics of coupled linear oscillators, corresponds to a special case of the above system with $A=\bar{B}_G$, $C=0$, $D=0$, and $E=F=0$. In this case,  $\alpha(t)=\sqrt{M}\dot{x}(t)$, $\beta(t)=\i B_G^\dagger x(t)$, where $x(t)$ denotes the vector of mass positions at time $t$.
We also assume that $\kappa_{i,i}>0$ for each $i \in [N]$, ensuring that $\bar{B}_G$ has full rank $N$. Note that the instances of Poly-CLOP that are proven to be BQP-hard in Ref.~\cite{babbush2023exponential} satisfy this condition.

Compared to Eq.~\eqref{eq:embedding_linear_oscillator_schrodinger_equation}, Eq.~\eqref{eq:quadratic_nonlinear_schrodinger_equation}
allows the evolution of $\alpha(t)$ and $\beta(t)$ to depend not only on linear terms but also on quadratic terms in these variables. One might interpret Eq.~\eqref{eq:quadratic_nonlinear_schrodinger_equation} as describing a system of coupled nonlinear oscillators, although $\alpha(t)$ and $\beta(t)$ no longer retain their original interpretations as the vectors encoding kinetic and potential energies, respectively.

Given the coefficient matrices $A, C, D, E, F$, an initial condition $z(0)$, a projector $\Pi$,
a time $T>0$, and a precision parameter $\epsilon>0$,
the goal is to estimate the quantity
$$\frac{z(T)^\dagger \Pi z(T)}{z(T)^\dagger z(T)}$$
within additive error $\epsilon$.

We aim to solve this problem by applying the quantum algorithm presented in Theorem~\ref{thm:algorithm_complexity_decomposable_f2_case}. To this end, the matrices $C$, $D$, $E$ and $F$ must satisfy certain constraints. Specifically, suppose $A$ has singular value decomposition $A=U \Gamma W^\dagger$, where $\Gamma=\sum_{j=1}^N \gamma_j \ket{j}\bra{j}$ is diagonal, and
$U=\sum_{j=1}^N \ket{u_j}\bra{j}$ and $W=\sum_{j=1}^{N} \ket{v_j}\bra{j}$ are unitary. Then $F_1=-\i H$ has spectral decomposition 
\begin{align}
    F_1 = Q \Lambda Q^\dagger = \i\sum_{j=1}^N \gamma_j \lrb{\ket{+_j}\bra{+_j} - \ket{-_j}\bra{-_j}},
\end{align}
where     
\begin{align}
\ket{\pm_j}&=\frac{1}{\sqrt{2}}\lrb{\ket{0}\ket{u_j}\pm \ket{1}\ket{v_j}},    \\
Q &= 
\sum_{j=1}^N \ket{+_j}\bra{0}\bra{j}
+
\sum_{j=1}^N \ket{-_j}\bra{1}\bra{j} \\
&=\frac{1}{\sqrt{2}}\lrb{\ket{0}\bra{0}\otimes U+\ket{1}\bra{0}\otimes W+\ket{0}\bra{1}\otimes U-\ket{1}\bra{1}\otimes W},\\
\Lambda &= \lrb{\ket{0}\bra{0} - \ket{1}\bra{1}} \otimes \i 
 \sum_{j=1}^N \gamma_j \ket{j}\bra{j}.        
\end{align}
Now we partition the eigenvalues of $F_1$ into two groups $\lrcb{\i\gamma_j:~j \in [N]}$ and $\lrcb{-\i\gamma_j:~j\in [N]}$. Correspondingly, we define the projectors 
$P_+=\ket{0}\bra{0}\otimes I$
and
$P_-=\ket{1}\bra{1}\otimes I$.

Meanwhile, note that
\begin{align}
    F_2&=\ket{0}\bra{0}\bra{0}\otimes C + \ket{0}\bra{1}\bra{1}\otimes D + \ket{1}\bra{0}\bra{1}\otimes E +
    \ket{1}\bra{1}\bra{0}\otimes F.
\end{align}
It follows that
\begin{align}
    \tilde{F}_2&=Q^\dagger F_2 Q^{\otimes 2} \\
    &=\ket{+}\bra{+}\bra{+}\otimes U^\dagger C (U \otimes U)
    + \ket{+}\bra{-}\bra{-}\otimes U^\dagger D (W \otimes W) \\
    &\quad 
    +\ket{-}\bra{+}\bra{-}\otimes W^\dagger E (U \otimes W) 
    + \ket{-}\bra{-}\bra{+}\otimes W^\dagger F (W \otimes U)
    \\
    &=\ket{+}\bra{+}\bra{+}\otimes \tilde{C}
    + \ket{+}\bra{-}\bra{-}\otimes \tilde{D} \\
    &\quad +\ket{-}\bra{+}\bra{-}\otimes \tilde{E} 
    + \ket{-}\bra{-}\bra{+}\otimes \tilde{F}, 
\end{align}
where 
\begin{align}
\tilde{C}&=U^\dagger C (U \otimes U), \quad    
\tilde{D}=U^\dagger D (W \otimes W), \\
\tilde{E}&=W^\dagger E (U \otimes W), \quad
\tilde{F}=W^\dagger F (W \otimes U).
\end{align}

We would like $\tilde{F}_2$ to satisfy 
\begin{align}
    \tilde{F}_2 = P_+ \tilde{F}_2 (P_+\otimes P_+)
    + P_- \tilde{F}_2 (P_-\otimes P_-),
\end{align}
and one can verify by direct calculation that this condition holds if and only if
\begin{align}
  \tilde{C}=\tilde{D}=\tilde{E}=\tilde{F}.  
\end{align}
Assume this condition holds from now on. Then we have
\begin{align}
\norm{\tilde{F}_2}=\sqrt{2}\norm{\tilde{C}}=\sqrt{2}\norm{C}.
\end{align}
Moreover, the column sparsity of $\tilde{F}_2$ equals that of $\tilde{C}$.\footnote{As in Section~\ref{subsec:quantum_algorithm_nonresonant_systems}, this column-sparsity parameter enters through the convergence condition for the Carleman truncation. It should not be interpreted as an additional cost parameter in the assumed block-encodings of the coefficient matrices. In the formal Poly-NLSE definition below, the same role is played by $\operatorname{cs}(B)$, since $\tilde{C}=B$.}

Now suppose $\lrb{\gamma_j:~j\in [N]}$ is $\Delta$-nonresonant, and let $\rho$ be an upper bound on $\max_{t \in [0,t]}\norm{z(t)}$. Furthermore, suppose that $C$ satisfies 
\begin{align}
    8 \sqrt{2} \rho  \cdot \operatorname{cs}(\tilde{C}) \cdot  \norm{C} \le c \Delta,
\end{align}
for some constant $c \in (0, 1/2)$. Then it follows that
\begin{align}
    8 \rho \cdot \operatorname{cs}(\tilde{F}_2)\cdot \norm{\tilde{F}_2} \le c \Delta.
\end{align}
This ensures that $\tilde{F}_2$ satisifies the conditions required by Theorem~\ref{thm:algorithm_complexity_decomposable_f2_case}.

Now suppose $N=2^n$ and we can construct block-encodings of $A$, $C$, $D$, $E$ and $F$ with rescaling factors $\myO{\mypoly{n}}$ in $\myO{\mypoly{n}}$ time.
Then we can use them to build block-encodings of $F_1$ and $F_2$ with rescaling factors $\myO{\mypoly{n}}$ in $\myO{\mypoly{n}}$ time. Moreover, suppose we have $T=\myO{\mypoly{n}}$, and
\begin{align}
\max_{t\in [0, T]} \norm{z(t)}=\myO{\mypoly{n}}, ~    
\min_{t\in [0, T]} \norm{z(t)}=\frac{1}{\myOmega{\mypoly{n}}}.
\end{align} 
Furthermore, assume the initial state $\ket{z(0)}$ can be prepared in $\myO{\mypoly{n}}$ time. 

Then by Theorem~\ref{thm:algorithm_complexity_decomposable_f2_case}, we can prepare a state that is $\myO{\frac{1}{\mypoly{n}}}$-close to $\ket{z(T)}$ in $\myO{\mypoly{n}}$ time. We can then estimate the target quantity
\[
\bra{z(T)}\Pi \ket{z(T)}=\frac{z(T)^\dagger \Pi z(T)}{ z(T)^\dagger z(T)}    
\]
within $\myO{\frac{1}{\mypoly{n}}}$ accuracy by performing the measurement $\lrcb{\Pi, I-\Pi}$ on the prepared state and repeating this procedure $\myO{\mypoly{n}}$ times. Alternatively, amplitude estimate can be used to reduce the number of repetitions quadratically. Assuming the measurement $\lrcb{\Pi, I-\Pi}$ can be implemented in $\myO{\mypoly{n}}$ time, this entire algorithm runs in $\myO{\mypoly{n}}$ time.

We summarize the above argument in the following result. To begin, we formally define the associated decision problem:

\begin{definition}[NLSE] 
Let $N \le K$ be arbitrary. Let $A \in \mathbb{C}^{N \times K}$ has singular value decomposition $A=U \Gamma W^\dagger$, where $\Gamma=\sum_{j=1}^N \gamma_j \ket{j}\bra{j}$ with $\gamma_j>0$ for all $j \in [N]$. Let $B \in \mathbb{C}^{N \times N^2}$ be arbitrary. Consider the nonlinear Schr\"{o}dinger equation:
\begin{align}
    \dot{z}(t)=-\i H z(t) + F_2 z(t)^{\otimes 2},
\end{align}
where $z(t) \in \mathbb{C}^{N+K}$, 
\begin{align}
H = -\begin{pmatrix}
0 & A \\
A^{\dagger} & 0
\end{pmatrix}, ~
F_2 =\begin{pmatrix}
    C & 0 & 0 & D \\
    0 & E & F & 0
\end{pmatrix},    
\end{align}
with 
\begin{align}
    C&=UB(U^\dagger \otimes U^\dagger),\quad 
    D=UB(W^\dagger \otimes W^\dagger), \\
    E&=WB(U^\dagger \otimes W^\dagger),\quad
    F=WB(W^\dagger \otimes U^\dagger).
\end{align}
Suppose we are given access to block-encodings of $A$, $C$, $D$, $E$, $F$ and to a unitary $\mathcal{W}$ that prepares the normalized initial state $\ket{z(0)}$. Given a projector $\Pi \in \mathbb{C}^{(N+K) \times (N+K)}$, a time $T>0$ and real numbers $0<a<b\le 1$, the goal is to decide whether 
$$\frac{z(T)^\dagger\Pi z(T)}{z(T)^\dagger z(T)}<a ~~~\mathrm{or}~~~\frac{z(T)^\dagger\Pi z(T)}{z(T)^\dagger z(T)}>b,$$
promised that one of these holds.
    
\end{definition}

\begin{definition}[Poly-NLSE]
The setup follows that of NLSE, with additional constraints on the input encodings and parameter settings. Specifically, we require that $N=2^n$, and
\begin{align}
T, \max_{t\in [0, T]} \norm{z(t)}
&=\myO{\mypoly{n}}, \\    
b-a, \min_{t\in [0, T]} \norm{z(t)}&=\frac{1}{\myOmega{\mypoly{n}}}.
\end{align}

Suppose $\lrb{\gamma_j:~j\in [N]}$ is $\Delta$-nonresonant, where $\Delta=\myO{\mypoly{n} \cdot \operatorname{cs}(B)}$. Let $\rho=\myO{\mypoly{n}}$ be a known upper bound on $\max_{t \in [0, T]} \norm{z(t)}$. Assume that
\begin{align}
    8 \sqrt{2} \rho \cdot \operatorname{cs}(B)\cdot \norm{B} \le c\Delta
\end{align}
for some constant $c \in (0, 1/2)$.

Finally, suppose block-encodings of $A$, $C$, $D$, $E$ and $F$ with rescaling factors $\myO{\mypoly{n}}$ can be implemented by quantum circuits of $\myO{\mypoly{n}}$ sizes. Assume the measurement $\lrcb{\Pi, I-\Pi}$ and the unitary $\mathcal{W}$ that prepares $\ket{z(0)}$ can be  also implemented using $\myO{\mypoly{n}}$ elementary gates.
\end{definition}

Then the above argument establishes the following result:
\begin{thm}
Poly-NLSE is BQP-complete.    
\end{thm}


\subsection{Robustness of BQP-completeness result}
\label{subsec:bqphard}

In the previous subsections, we presented two families of nonlinear BQP-complete problems. To prove their BQP-hardness, we relied on the fact that these problems encompass the linear systems shown to be BQP-hard in Ref.~\cite{babbush2023exponential}, which involve simulating exponentially many coupled harmonic oscillators. It is also known that small dissipative perturbations preserve BQP-completeness~\cite{krovi2024quantum}.

A natural question is whether BQP-completeness still holds when nonlinearity is introduced directly into the original setting of Ref.~\cite{babbush2023exponential}. In our constructions, the nonlinear terms are tunable, and the original linear system is recovered when these terms are set to zero. BQP-hardness then follows from the fact that the linear instance is explicitly contained within the problem family. However, one may ask whether BQP-hardness persists in settings where the problems are inherently nonlinear -- that is, when the linear system is not explicitly embedded.

The theorem below addresses this question through a different approach. Rather than relying on the presence of a BQP-hard linear subsystem, it considers nonlinear systems obtained by slightly perturbing the BQP-hard linear system, and shows that BQP-hardness is preserved by demonstrating that the solution to the perturbed system remains sufficiently close to that of the original one. Moreover, these nonlinear problems can still be efficiently solved using our Carleman analysis for conservative systems, in combination with existing quantum algorithms for linear differential equations.

\begin{thm}
\label{thm:bqp}
    Consider the following dynamical system
\begin{align}
\label{eq:bqpproblem}
    \ddot{q} + R \dot{q} + D q + W \dot{q} \otimes \dot{q} =0, 
\end{align}
with:
\begin{itemize}
    \item Linear couplings: $D$ is an $N \times N$ matrix which is (a) symmetric, (b) {at most $\Theta(1)$-sparse, (c) its eigenvalues $\lambda_{i}(D)$ satisfy $ 0< \lambda_{\mathrm{min}} \leq \lambda_{i}(D) = O(1)$.}
    \item Initial condition: $q(0) = 0$ and $\dot{q}(0) = (1,-1,0,\dots, 0)$.
\end{itemize}  

Then, there exists a choice of $R$ and $W$ such that:
\begin{enumerate}
    \item $R$ is an $N \times N$ matrix with support on  $N' \leq N$ oscillators and $\|R\| = \Theta(1/\mathrm{polylog}(N))$;
    \item $W$ is an $N \times N^2$ matrix with $\|W\| = \Theta(1/\mypolylog{N})$, supported on the same oscillators on which $R$ is supported;
\end{enumerate}
such that the problem of determining whether 
\begin{align}
 \frac{1}{2} \dot{q}_1(t_\star)^2, 
\end{align}
is at least $1/\mathrm{polylog}(N)$ or at most $1/e^{\sqrt{\log N}}$ after time $t_\star= \mypolylog{N}$, promised that one of these holds, is BQP-complete.
\end{thm} 

The proof can be found in Appendix~\ref{sec:proofbqpcompleteness}.


\section{Outlook}
\label{sec:outlook}

This work advances the theoretical foundations for quantum simulation of nonlinear dynamical systems by rigorously establishing the convergence of Carleman linearization under broad and physically meaningful conditions, and then analyzing how to construct quantum algorithms that output encodings of these dynamics. 
In particular, it substantially expands the class of nonlinear dynamics amenable to efficient quantum algorithms, extending beyond purely dissipative systems with negative log-norm to encompass general stable, conservative, oscillating, and nonresonant dynamics under well-defined constraints. Along the way, we identified and addressed technical gaps present in previous works. We further establish the BQP-completeness of a class of nonlinear oscillator systems and demonstrate that the exponential quantum speedup available for certain linear oscillators can be robustly extended to the nonlinear regime.

While we have made several contributions in our analysis, important directions remain open. Our current convergence theory does not fully extend to general non-autonomous systems~\cite{wiggins2003introduction}, and developing a robust Carleman framework for such systems remains a valuable challenge, both for theoretical understanding and practical implementation. Addressing this may require exploiting variants of the Carleman method tailored to the problem at hand~\cite{kowalski1991nonlinear}, and much work remains. Notably, our results do not yet cover a range of natural nonlinear problems, such as in the context of nonlinear Hamiltonian dynamics~\cite{arnol2013mathematical}, and so it is important to extend our techniques to cover such systems. 

Several technical aspects of our results may also be refined. For instance, can the convergence condition for Carleman linearization of nonresonant systems in the Poincar\'{e} domain be strengthened? In particular, is the dependence on the column sparsity of $\tilde{F}_2=Q^{-1}F_2Q^{\otimes 2}$ truly necessary? Furthermore, the quantum algorithms for nonresonant (and certain resonant) differential equations are efficient only when the matrix $Q$ that diagonalizes $F_1$ has condition number $\myO{1}$. Extending these algorithms to cases where $F_1$ is far from normal, and understanding the impact of non-normality on quantum performance, remains an open challenge. In addition, although we establish convergence for nonresonant systems in the Siegel domain in certain special cases, the general case remains unresolved.

Finally, while our results broaden the theoretical regimes in which Carleman-based quantum algorithms can be rigorously analyzed, identifying concrete high-impact applications remains an important open direction. The sufficient conditions developed here should be viewed as criteria for analyzing candidate nonlinear models, rather than as a complete demonstration of practical quantum advantage for any particular real-world system. Our results provide tools to construct a broader bridge between theory and applications linking the Carleman approach to real-world models where nonlinearities are both ubiquitous and computationally challenging. These may include fluid dynamics~\cite{batchelor2000introduction, pozrikidis2009fluid}, chemical reaction networks~\cite{feinberg2019foundations}, plasma physics~\cite{tajima2018computational}, or complex ecological and neural systems~\cite{meron2015nonlinear}. Establishing such connections will be essential for demonstrating practical quantum advantage and guiding the development of quantum algorithms that address the structure and constraints of real-world nonlinear dynamics. In this context, it will be important to explore whether quantum algorithms for nonlinear differential equations can be tailored to the constraints of early generations of fault-tolerant quantum computers~\cite{lin2022heisenberg, dong2022groundstate, zhang2022computingground, wang2022statepreparation, wang2023quantumalgorithm, katabarwa2024early, wang2025efficient}. Such devices will likely access limited computational resources, necessitating highly resource-efficient algorithms. Identifying low-depth methods that leverage partial Carleman truncations or problem-specific structure could provide a practical path toward early demonstrations of quantum advantage in nonlinear dynamics.


\section{Author contributions} 
Authors are listed alphabetically. DJ helped develop the research direction connecting with Normal Forms, and developed the polynomial extension to conserved linear observables in Section~\ref{sec:conservative}. KK developed the theory behind conservative systems, derived most of the results in Section~\ref{sec:conservative}, and helped proving the robustness of BQP-completeness result and designing examples for stable systems. ML developed the theory for stable systems in Sections~\ref{sec:stable}, quantum algorithms for stable and conservative systems in Section~\ref{subsec:quantum_algorithm_stable_conservative_systems}, and the BQP-completeness results in Section~\ref{subsec:bqphard}. GW developed the framework for diagonalizing nonresonant Carleman matrices and derived the error bounds based on this framework in Section~\ref{sec:nonresonant_systems}. He also designed quantum algorithms for nonresonant and certain resonant systems in Section \ref{subsec:quantum_algorithm_nonresonant_systems}, proved the BQP-completeness results in Sections \ref{subsec:coupled_oscillators} and \ref{subsec:nonlinear_schrodinger_equations}, and contributed to the proofs of Lemmas \ref{lem:fusion_sums} and \ref{lem:fusion_sums_general} in Appendix~\ref{app:thm_cons}. DJ, KK, ML, and GW wrote the paper and YS contributed to structuring and reviewing the article. ATS helped structure the research direction to NNSA mission needs and reviewed the article.

\section{Acknowledgements} 
ATS and YS are supported by the NNSA’s Advanced
Simulation and Computing Beyond Moore’s Law Program at Los
Alamos National Laboratory. Los Alamos
National Laboratory is operated by Triad National Security, LLC, for
the National Nuclear Security Administration of the U.S.
Department of Energy (Contract No. 89233218CNA000001). Many thanks to Bjorn Berntson for useful comments on a draft of this work, and to Paul Mannix for his help with running the large-truncation-order Carleman simulations. DJ, KK, ML, and GW are thankful for collaboration with our colleagues at PsiQuantum on other related topics.

\newpage
\appendix


\section{Issues with some of the proofs in previous works}
\label{app:issue_with_previous_proofs}
Here we point out errors we identified in previous proofs, which we believe help shed light on some of the underlying complexities of the topic \footnote{After this manuscript was posted on arXiv, Hari Krovi informed us that he had independently identified the error in his proof presented in Ref.~\cite{krovi2023improved} and later updated the arXiv version to fix it.}. Despite our best efforts, some  proofs or statements in the present work may also contain mistakes. We warmly welcome corrections!  :)


\subsection{Issue with the proof in \texorpdfstring{Ref.~\cite{liu2021efficient}}{arXiv:2011.03185}}
\label{app:issueliu}

Lemma 2 of Ref.~\cite{liu2021efficient}, is the first result providing analytical guarantees on the convergence of the Carleman method used in a quantum algorithm. We showed in Section~\ref{sec:solutionnormbounds} that the proof require solution norm bounds that implicitly assume the system has negative log-norm, $\mu(F_1)<0$. Here we hence assume this stronger assumption has been made.

We note that [Eq.~(4.17) \cite{liu2021efficient}] claims that for all $j=1,\dots,k$ the error vectors $\eta_j$ defined in Eq.~\eqref{eq:errorvectorcomponents} satisfy
\begin{align}
   \eta^\dag_j (A_{j,j} + (A_{j,j})^\dag) \eta_j \leq 2j \alpha(F_1) \|\eta\|^2.
\end{align}
We have already discussed in Remark~\ref{issue1liu} that this inequality does not hold for the case $j=1$, so clearly it does not hold for all $j=1,\dots,k$. This inequality can be replaced with the weaker inequality 
\begin{align}
   \eta^\dag_j (A_{j,j} + (A_{j,j})^\dag) \eta_j \leq 2j \mu(F_1) \|\eta\|^2,
\end{align}
which then forces us to invoke $\mu(F_1)<0$ in the rest of the proof, which is what we shall do. 

Now we describe the core proof idea in Ref.~\cite{liu2021efficient}. A rescaling $\gamma$ is fixed such that $\mu(F_1) + \| F_0\| + \| F_2\| < 0$ and  $\|x(0)\|<1$. With this, it is claimed that
\begin{align}
\label{eq:coreinequality}
    \alpha(A) <0.
\end{align}
In Ref.~\cite{liu2021efficient}, \eqref{eq:coreinequality} is proved in two steps. Firstly, the bound is derived [Eq.~(4.18) \cite{liu2021efficient}]
\begin{align}
   \eta^\dag (A + A^\dag)\eta \leq  \eta_G^\dag (G + G^\dag) \eta_G.
\end{align}
where $\eta_G = ( \|\eta_1\|,  \|\eta_2\|, \dots, )$, and the matrix $G$ is defined by $G_{j-1,j} = j\| F_0\|$ for $j=2,\dots, k$, $G_{j,j} = j \mu(F_1)$ for $j=1,\dots, k$, $G_{j+1,j} = j\|F_2\|$ for $j=1,\dots, k-1$, and all other elements equal to zero. 
Secondly, it is argued that [Eq.~(4.19) \cite{liu2021efficient}]
\begin{align}
\label{eq:incorrectbound}
   \eta_G^\dag (G + G^\dag) \eta_G \leq 2 \alpha(G) \| \eta_G\|^2. 
\end{align}
The core inequality \eqref{eq:coreinequality}  then follows by proving that $\alpha(G)<0$.

Now, the chain of implications is broken by the fact that \eqref{eq:incorrectbound} does not hold. This is not immediately a consequence of Remark~\ref{issue1liu}, since the matrix $G$ has nontrivial structure to it. Nonetheless, simple counterexamples can be constructed. Take $k=2$, $\| F_0\| = 1/2$, $\|F_2\| = 1/4$, $\mu(F_1) =-1$, which satisfies $\mu(F_1) + \| F_0\| + \| F_2\| <0$. With these choices, $G = \big[\begin{smallmatrix}
  -1 & 1\\
  1/4 & -2
\end{smallmatrix}\big]$. Take \mbox{$\eta_G = [(4+\sqrt{41})/5,1]^T $}. Then we have
\begin{align}
\nonumber
\eta_G^\dag (G+G^\dag) \eta_G & = \frac{-328-7\sqrt{41}}{50}  > 2 \alpha(G) \| \eta_G\|^2  = \frac{2}{25} \left(\sqrt{2}-3\right) \left(4 \sqrt{41}+41\right).
\nonumber
\end{align}
The inequality~\eqref{eq:incorrectbound} can be corrected as
\begin{align}
   \eta_G^\dag (G + G^\dag) \eta_G \leq 2 \mu(G) \| \eta_G\|^2,
\end{align}
which holds by definition of log-norm. For the Carleman convergence proof to go through with these corrections, we then need to prove
\begin{align}
    \mu(A)<0.
\end{align}
instead of~\eqref{eq:coreinequality}, from the assumption $\mu(F_1) + \| F_0\| + \| F_2\| <0$. However this also does not hold, as shown by means of a counterexample in Section~\ref{sec:comparisons}.


\subsection{Issue with the proof in \texorpdfstring{Ref.~\cite{krovi2023improved}}{arXiv:2202.01054}}
\label{app:issuekrovi}

A second proof that $\mu(F_1) + \| F_0\| + \| F_2\| <0$ implies $\mu(A)<0$ was given in Ref.~\cite{krovi2023improved}, see Lemma~16 from Eq.~(7.17). The critical problem with the proof is Eq.~(7.19), which claims
\begin{align}
    \mu(H_1) = k \mu(F_1),
\end{align}
with  $H_1$ defined as in Eq.~\eqref{eq:CarlemannMatrix2}, namely $H_1 = \sum_{j=1}^{k} A_{j,j}$, where recall that each $A_{j,j}$ acts on a distinct space of dimension $d^j$. An important comment about this equation is that, while a tensor product notation has been used in Ref.~\cite{krovi2023improved}, 
\begin{align}
    H_1 = \sum_j \ketbra{j}{j} \otimes A_{j,j},
\end{align}
formally $H_1$ does not have this form, since each $A_{j,j}$ acts on a vector space of a different dimension.  In fact, as one can see from Eq.~(7.3) of Ref.~\cite{krovi2023improved}, the above notation is used as a shortcut for matrices of the form (example for $k =3)$
\begin{align}
    H_1 = \begin{bmatrix} 
    A_{1,1} & 0 & 0  \\ 
    0 & A_{2,2} & 0  \\
    0 & 0 & A_{3,3}
    \end{bmatrix} = \bigoplus_{j=1}^{k} A_{j,j},
\end{align}
where each $A_{j,j}$ has a different dimension. 

With these comments in place, we shall now see that
\begin{align}
    \mu(H_1) = \mu(F_1).
\end{align}

In fact, denoting by $\lambda_1(X)$ the largest eigenvalue of a matrix $X$
    \begin{align}
    \mu(H_1):=\mu \left ( \oplus_{j=1}^{k} A_{j,j}\right) &:= \lambda_1 \left ( \oplus_{j=1}^{k} \frac{1}{2}(A_{j,j}+(A_{j,j})^\dagger)\right) \nonumber 
    = \max_{1\le j \le k} \lambda_1 \left (  \frac{1}{2}A_{j,j}+(A_{j,j})^\dagger \right) \nonumber \\
& = \max_{1\le j \le k} \mu(A_{j,j})
= \max_{1\le j \le k} \mu(F_1 \otimes I^{\otimes (j-1)} + \mathrm{shifts})
 \nonumber  \\ & = \max_{1\le j \le k} \sum_{k=1}^j \mu(I^{\otimes (k-1)} \otimes F_1 \otimes I^{\otimes (j-k)})
     =  \max_{1\le j \le k} j\mu \left ( F_1\right)
          = \mu \left ( F_1\right).
\end{align}
For an explicit example, take $F_1 =  \left[\begin{smallmatrix}
 -1 & 1/2 \\ 1/4 & -1
\end{smallmatrix}\right]$ and $k =2$. Then,
\begin{align}
    H_1 =  \left[\begin{smallmatrix}
 -1 & 1/2 & 0 & 0 &0 &0  \\ 
 1/4 & -1  & 0 & 0 &0 &0  \\
 0 & 0   & -2 & 1/2 & 1/2 & 0 \\
  0 & 0  & 1/4 & -5/2 & 0 & 1/2 \\
    0 & 0 & 1/4 & 0 & -5/2 & 1/2 \\
      0 & 0 & 0 & 1/4 & 1/4 & -3.
\end{smallmatrix}\right].
\end{align}
We have $\mu(A^1_1) = \mu(F_1) = -5/8$, $\mu(A^2_2) = -5/4 = 2 \times  \mu(A^1_1)$, but $\mu(H_1) = -5/8$. 

With this, the bound on $\mu(A)$ in Ref.~\cite{krovi2023improved} is weakened to
\begin{align}
    \mu(A) \leq \mu(F_1) + k \|F_0\| + k \| F_2\|,
\end{align}
which does not allow us to prove $\mu(A) <0$ independently of $k$. In fact, the condition $\mu(F_1) + \| F_2\| + \|F_0\| <0 $ only allows us to prove $ \mu(A) < 0$ for $k = 1$, and is hence insufficient to prove the convergence of Carleman errors.


\subsection{Issue with the proof in \texorpdfstring{Ref.~\cite{wu2024quantum}}{arXiv:2405.12714}}
\label{app:issuenonresonant}

There has been one prior paper on Carleman criteria for non-dissipative dynamics. The results are similar to those we obtain in Section~\ref{sec:nonresonant_systems}, however our analysis was more involved and required a different handling of the $V^{-1}$ blocks. One problem arises with the approximation of the binomial term $\binom{2j-k}{j-k}$ which in Ref.~\cite{wu2024quantum} is upper bounded with $(2e)^{j-k}$. This can be seen to be incorrect if, e.g., one sets $j-k=1$ so that $\binom{2j-k}{j-k}= k+2$, which grows with $k$, and so the bound used is not correct for $j-1=k\geq 4$. Similar violations occur for $j-k=c$ a constant, and $k$ increasing.

A second, more pressing reason for us using the binary forest formalism is that the norm-bound analysis of Ref.~\cite{wu2024quantum} for $V^{-1}$ leads to a $k!$ term that prohibits obtaining a Carleman convergence criterion. We can illustrate this with a one-dimensional system.

Consider the following differential equation:
\begin{align}
    \dot{x}(t) = a x(t) + b x(t)^2, 
\end{align}
where $x(t) \in \myC$, $a \in \myC\setminus \lrcb{0}$ with $\realpart{a} \le 0$, and $b \in \myC$. This corresponds to the case $N=1$, $F_1 = a$, and
$F_2 = b$. In this setting, we have $Q=1$, $s=1$ and $\Delta = \abs{a}$.

The order-$k$ Carleman linearization for this equation is
\begin{align}
\frac{d}{dt}\begin{pmatrix}
    y^{[1]} \\
    y^{[2]} \\
    y^{[3]} \\
    \vdots \\
    y^{[k-1]} \\
    y^{[k]}
    \end{pmatrix}=\begin{pmatrix}
    a & b & & & &\\
     & 2a & 2b & & &\\
            &  & 3a & 3b &       & \\
            &         & \ddots   & \ddots   & \ddots & \\
            &         &         &  & (k-1)a  & (k-1)b\\      
            &         &         &         &   & ka
    \end{pmatrix}\begin{pmatrix}
    y^{[1]} \\
    y^{[2]} \\
    y^{[3]} \\
    \vdots \\
    y^{[k-1]} \\
    y^{[k]}
    \end{pmatrix},      
\end{align}
where $y^{[i]}(t) \in \myC$ approximates $x(t)^i$ for $i \in [k]$.  In this case, the order-$k$ Carleman matrix $A$ satisfies $A_{j,j}=ja$ for $j \in [k]$ and $A_{j,j+1}=jb$ for $j \in [k-1]$. The initial conditions are given by $y^{[i]}(0)=x(0)^i$ for all $i \in [k]$. 

By Lemmas~\ref{lem:vij_norm_bound_poincare_domain}
and \ref{lem:wij_norm_bound_poincare_domain}, the matrix $V$ that diagonalizes $A$ satisfies
\begin{align}
\abs{V_{i,j}}, \abs{(V^{-1})_{i,j}} \le \binom{j-1}{i-1} \lrb{\frac{4\abs{b}}{\abs{a}}}^{j-i}, \quad \forall 1\le i\le j \le k. 
\label{eq:vij_v_inv_ij_abs_bounds_1d}
\end{align}

Furthermore, let $\abs{x_{\rm max}}=\max_{t \in [0,T]}\abs{x(t)}$. Then, by Theorem~\ref{thm:error_bound_nonresonant_poincare_domain}, we obtain the following error bound for any $i \in [k]$ and $t \in [0, T]$:
\begin{align}
    \abs{x(t)^{i}-y^{[i]}(t)} \le k t  \abs{x_{\rm max}}^{i+1} \abs{b} \binom{k-1}{i-1} \lrb{\frac{8 \abs{b}\abs{x_{\rm max}}}{\abs{a}}}^{k-i}.
\end{align}
In particular, for $i=1$, we have
\begin{align}
    \abs{x(t) -y^{[1]}(t)} \le k t  \abs{x_{\rm max}}^{2} \abs{b}  \lrb{\frac{8 \abs{b} \abs{x_{\rm max}}}{\abs{a}}}^{k-1}.
\end{align}
Thus, if $8\abs{b}\abs{x_{\rm max}}<\abs{a}$, then we have $|y^{[1]}(t) - x(t)| \to 0$ as $k \to +\infty$. This holds for all $t \ge 0$.

Note that in this setting, one may obtain an alternative upper bound on $\abs{(V^{-1})_{i,j}}$  using the identity: 
\begin{align}
({V}^{-1})_{i,j} = \sum_{p=1}^{j-i} \sum_{{i=i_1}<i_2<\dots<i_p<{i_{p+1}=j}} (-1)^p {V}_{{i_1},i_2}{V}_{i_2,i_3}\dots {V}_{i_p,{i_{p+1}}}.
\label{eq:wij_in_terms_of_vij}
\end{align}
However, the resulting bound is significantly looser than the one provided by Lemma~\ref{lem:wij_norm_bound_poincare_domain} in the regime where $i \ll j$. Specifically, applying triangle inequality to Eq.~\eqref{eq:wij_in_terms_of_vij} yields
\begin{align}
\abs{({V}^{-1})_{i,j}} \le \sum_{p=1}^{j-i} \sum_{{i=i_1}<i_2<\dots<i_p<{i_{p+1}=j}} \abs{{V}_{{i_1},i_2}}\abs{{V}_{i_2,i_3}}\dots \abs{{V}_{i_p,{i_{p+1}}}}.
\end{align}
This is the approach followed in Ref.~\cite{wu2024quantum}. The right-hand side of this equation is at least 
$$\abs{{V}_{i,i+1}}\abs{{V}_{i+1,i+2}}\dots \abs{{V}_{j-1,j}}.$$ 
Meanwhile, one can verify by direct calculation that
\begin{align}
V_{m,m+1}=\frac{mb}{a}, \qquad \forall m \in [k-1].
\end{align}
Thus, we obtain a lower bound of
$$\frac{(j-1)!}{(i-1)!} \cdot \lrb{\frac{\abs{b}}{\abs{a}}}^{j-i},$$
which is considerably larger than the bound from Lemma~\ref{lem:wij_norm_bound_poincare_domain},
$$\binom{j-1}{i-1}\lrb{\frac{4\abs{b}}{\abs{a}}}^{j-i},$$
when $j-i \gg 4e$. In particular, the best possible bound on $\abs{(V^{-1})_{1,k}}$ obtained via Eq.~\eqref{eq:wij_in_terms_of_vij}
is at least 
$$(k-1)! \lrb{\frac{\abs{b}}{\abs{a}}}^{k-1},$$
while Lemma~\ref{lem:wij_norm_bound_poincare_domain} gives a much tighter bound
$$ \lrb{\frac{4\abs{b}}{\abs{a}}}^{k-1}.$$
If we were to use the looser bound to estimate $\abs{x(t)-y^{[1]}(t)}$, we would not obtain anything better than
\begin{align}
    \abs{x(t)-y^{[1]}(t)} \le k! \cdot t \abs{x_{\rm max}}^{2} \abs{b} \lrb{\frac{\abs{b}\abs{x_{\rm max}}}{\abs{a}}}^{k-1},
\end{align}
which fails to yield convergence as $k \to +\infty$
regardless of the value of $b$ (unless $b=0$). 

This example illustrates the limitation of proof strategies that solely rely on the blockwise structure of $V$, and highlights the necessity of deriving blockwise expressions for $V^{-1}$, as we do in Section~\ref{subsec:expression_for_v_inv_ij_diagonal_f1}.


\section{Proofs of Lemma~\ref{lem:normdecreasegeneral} and Theorem~\ref{thm:Carlemanstable}}
\label{sec:proofCarlemanstable}

First we prove Lemma~\ref{lem:normdecreasegeneral}, and then proceed to the main result in Theorem~\ref{thm:Carlemanstable} for the special case where $P=I$, which means that $F_1$ has negative log-norm. The general proof will involve a reduction to this case. The basic strategy is to show that if $F_1$ has negative log-norm, then under a $R_\mu<1$ condition there is a rescaling for which (a) the Carleman matrix has negative log-norm (b) $\|x(0)\|<1$. These two conditions then imply that the error in the Carleman vector is exponentially suppressed with the truncation order~$k$. The result can then be extended to $P \neq I$ by an appropriate change of coordinates.


\subsection{Proof of Lemma~\ref{lem:normdecreasegeneral}}

\begin{proof}
 
Define the Lyapunov function $V(x):=\tfrac12\|x\|_P^2$. Along any solution
\begin{align*}
\dot V(t)& =\mathrm{Re} \langle P^{1/2} x(t), P^{1/2} \dot x(t)\rangle
            \\ & = \mathrm{Re} \langle P^{1/2} x(t), P^{1/2} F_2(x(t)\otimes x(t))\rangle
           + \mathrm{Re} \langle P^{1/2} x(t), P^{1/2} F_1 x(t)\rangle
           + \mathrm{Re} \langle P^{1/2} x(t), P^{1/2} F_0\rangle \\
           & \leq \|F_2\|_P \| x(t)\|_P^3 + \mu_P(F_1) \| x(t)\|_P^2 + \|F_0\|_P \| x(t)\|_P,
\end{align*}
where we used Lemma 3.31 of Ref.~\cite{plischke2005transient} to introduce $\mu_P(F_1)$. Define the  sphere in the norm induced by $P$: $\mathcal{S} = \{ x: \|x\|_P = \|x(0)\|_P\}$. For any $t\geq 0$ such that $x(t) \in \mathcal{S}$ we have
\begin{align}
    \dot{V}(t) \leq (-\mu_P(F_1)) \|x(0)\|^2_P  (R_P -1) \leq 0. 
\end{align}
This implies for any $t$ such that the solution is at $\mathcal{S}$, $\frac{d}{dt} \|x(t)\|_P^2 \leq 0$, which means the solution cannot exit $\mathcal{S}$. Equivalently, $\|x(t)\|_P \leq \|x(0)\|_P$ for all $t \geq 0$. 
     
\end{proof}

\subsection{Error bound in terms of matrix of norms \texorpdfstring{$\hh{G}$ }{}for the case \texorpdfstring{$P=I$}{P=I}}

In the first part of the proof we shall consider the special case $P=I$, which means $F_1$ has negative log-norm. For this case, we shall construct a $k \times k$-dimensional symmetric matrix $\hh{G}$, whose entries only involve $\|F_0\|$, $\|F_2\|$ and $\mu(F_1)$, and whose largest eigenvalue of $\lambda_1(\hh{G})$ upper bounds the log-norm of the Carleman matrix. 
\begin{align}
\label{eq:Carlemanerrorboundproof}
    \| \eta(t)\| \leq \frac{1- e^{\lambda_1(\hh{G}) t/2}}{-\lambda_1(\hh{G})/2 } k \| F_2\| \| x(0)\|^{k+1}.
\end{align}

For convenience, let us start by recalling Eq.~\eqref{eq:error_evolution}:
\begin{equation}
    \dot{\eta}(t) = A \eta(t) + \zeta(t).
\end{equation}
We can then consider the equation for the norm $\| \eta(t)\|$:
\begin{align}
    \frac{d\| \eta(t)\|^2}{dt} & = \frac{d\eta^\dag(t)}{dt} \eta(t)+ \eta^\dag(t)\frac{d\eta(t)}{dt} \\
    &  =\eta^\dag(t) (A^\dag + A) \eta(t) +  \eta^\dag(t)\zeta(t)+ \zeta^\dag(t) \eta(t). 
\end{align}
First, analyze
\begin{align}
 \eta^\dag(t) (A^\dag + A) \eta(t)
   & = \sum_{j=1}^{k}  \eta_j^\dag(t)(A^{\dag}_{j,j} + A_{j,j}) \eta_j(t) \\ & + \sum_{j=2}^{k}  \eta_{j-1}^\dag(t) A^{\dag}_{j,j-1} \eta_{j}(t) + \eta_{j}^\dag(t) A_{j,j-1} \eta_{j-1}(t) 
   \\ & + \sum_{j=1}^{k-1}  \eta_{j+1}^\dag(t) A^{\dag}_{j,j+1} \eta_{j}(t) + \eta_{j}^\dag(t) A_{j, j+1} \eta_{j+1}(t).
\end{align}
For the first line,
\begin{align}
     \sum_{j=1}^{k}  \eta_j^\dag(t)(A^{\dag}_{j,j} + A_{j,j}) \eta_j(t)  & \leq \sum_{j=1}^{k}  \eta_j^\dag(t)\sum_{p=1}^j ( (F^\dag_1+F_1) \otimes I^{\otimes (j-1)} + \textrm{shifts } ) \eta_j(t) \\ 
     \leq \sum_{j=1}^{k} 2j \mu(F_1)  \| \eta_j(t)\|^2. 
   \end{align}
For the second and third lines,
\small
\begin{align}
\nonumber
    \sum_{j=2}^{k}  \eta_{j-1}^\dag(t) A^{\dag}_{j,j-1} \eta_{j}(t) + \eta_{j}^\dag(t) A_{j,j-1} \eta_{j-1}(t) \leq \sum_{j=2}^{k} 2 \| A_{j-1}^j\| \| \eta_j(t)\| \| \eta_{j-1}(t)\| \le \sum_{j=2}^{k} 2 j \| F_0\| \| \eta_j(t)\| \| \eta_{j-1}(t)\|.
\end{align}
\begin{align}
\nonumber
    \sum_{j=1}^{k-1}  \eta_{j+1}^\dag(t) A^{\dag}_{j,j+1} \eta_{j}(t) + \eta_{j}^\dag(t) A_{j, j+1} \eta_{j+1}(t) & \leq \sum_{j=1}^{k-1} 2 \| A_{j+1,j}\| \| \eta_j(t)\| \| \eta_{j+1}(t)\| \le \sum_{j=1}^{k-1} 2 j \| F_2\| \| \eta_j(t)\| \| \eta_{j+1}(t)\|.
\end{align}
\normalsize
So, setting for convenience $\eta_{G,j} = \| \eta_j\|$, we obtain the bound
   \begin{align}
    \eta^\dag(t) (A^\dag + A) \eta(t) \leq \sum_{j=1}^{k} 2j \mu(F_1)  \eta_{G,j}^2 + \sum_{j=2}^{k}   2 j \| F_0\|\eta_{G,j} \eta_{G,j-1} + \sum_{j=1}^{k-1}  2 j \| F_2\| \eta_{G,j} \eta_{G,j+1}.
   \end{align}
If we define the $k \times k$ matrix $G$ with $G_{j,j} = 2j\mu(F_1) $ for $1 \leq j \leq k$, $G_{j,j-1} = 2j \| F_0\|$ for $2 \leq j \leq k$, $G_{j,j+1} = 2j \|F_2\|$ for $1 \leq j \leq k-1$, and zero otherwise, we can rewrite the above bounds as
\begin{align}
    \eta^\dag(t) (A^\dag + A) \eta(t) \leq \eta_G^\dag G \eta_G. 
\end{align}
Indeed,
\begin{align}
    \eta_G^\dag G \eta_G & = \sum_{j=1}^{k} G_{jj} \eta^2_{G,j} + \sum_{j=2}^{k} G_{j,j-1} \eta_{G,j} \eta_{G,j-1} + \sum_{j=1}^{k-1} G_{j,j+1}  \eta_{G,j} \eta_{G,j+1}\\ & 
    = \sum_{j=1}^{k} 2\mu(F_1) j \eta_{G,j}^2 + \sum_{j=2}^{k} 2 j \| F_0\| \eta_{G,j} \eta_{G,j-1} + \sum_{j=1}^{k-1} 2j \| F_2\| \eta_{G,j} \eta_{G,j+1}. 
\end{align}
Note that
\begin{align}
     \eta_G^\dag G^\dag \eta_G & = \sum_{j=1}^{k} G_{jj} \eta^2_{G,j} + \sum_{j=2}^{k} G_{j-1,j} \eta_{G,j-1} \eta_{G,j} + \sum_{j=1}^{k-1} G_{j+1,j} \eta_{G,j+1} \eta_{G,j} \\
     & = \sum_{j=1}^{k} G_{jj} \eta^2_{G,j} + \sum_{j=1}^{k-1} G_{j,j+1} \eta_{G,j} \eta_{G,j+1} + \sum_{j=2}^{k} G_{j,j-1} \eta_{G,j} \eta_{G,j-1} =  \eta_G^\dag G \eta_G,
\end{align}
so
\begin{align}
    \label{eq:ReGinequality}
    \eta^\dag(t) (A^\dag + A) \eta(t) \leq  \eta_G^\dag \hh{G} \eta_G \leq \lambda_1 (\hh{G}) \| \eta_G\|^2 = \lambda_1 (\hh{G}) \| \eta(t)\|^2,
\end{align}
where $\hh{G}:= (G+G^\dag)/2$ is a hermitization of $G$ and $\lambda_1 (\hh{G})$ denotes its largest eigenvalue. 

Now we proceed with the analysis of the term $\eta^\dag(t)\zeta(t)+ \zeta^\dag(t) \eta(t)$. 
We have
\begin{align}
    \eta^\dag(t)\zeta(t)+ \zeta^\dag(t) \eta(t) & \leq 2 \| A^{k,k+1}\| \| w_{k+1}\| \| \eta(t)\| \\ & \leq 2 k \| F_2\| \|x(t)^{\otimes k+1}\| \| \eta(t)\| \\ & \leq 2 k \| F_2\| \|x(0)^{\otimes k+1}\|  \| \eta(t)\|,
\end{align}
where in the last step we assumed $\|x(t) \| \le \|x(0)\|$ for all $t\ge 0$, which follows from
 Lemma~\ref{lem:normdecrease} under the assumption that $R<1$. 

Putting everything together
\begin{align}
    2 \| \eta(t)\| \frac{d \| \eta(t)\|}{dt} = \frac{d\| \eta(t)\|^2}{dt} \leq  \lambda_1 ( \hh{G}) \| \eta(t)\|^2 +  2 k \| F_2\| \|x_0^{\otimes k+1}\|  \| \eta(t)\|,
\end{align}
and so, assuming $\| \eta(t)\| >0$,
\begin{align}
     \frac{d \| \eta(t)\|}{dt}  \leq \frac{\lambda_1 ( \hh{G})}{2} \| \eta(t)\| +   k \| F_2\| \|x(0)\|^{k+1}.
\end{align}
This inequality also holds for $\|\eta(t)\| = 0$, since the derivative vanishes in such cases.
Solving the above with equality and using Gr\"onwall's lemma we obtain inequality~\eqref{eq:Carlemanerrorboundproof}.


\subsection{Negativity of \texorpdfstring{$\lambda_1(\hh{G})$}{the largest eigenvalue} under appropriate rescaling for the case \texorpdfstring{$P=I$}{P=I}} 

In the second part of the proof we analyze $\hh{G}$. We will use Gershgorin's circle theorem to obtain a bound on $\lambda_1(\hh{G})$. We shall then prove that, under the condition $R<1$, under an appropriate choice of rescaling we can always make sure that $\lambda_1(\hh{G})<0$, while also ensuring \mbox{$\|x(0)\|\leq 1$}. Meeting these two conditions simultaneously ensures that the Carleman error converges exponentially to zero as we increase the truncation scale $k$.

\bigskip

\textbf{Gershgorin's circle bound on $\lambda_1(\hh{G})$.}    If we look at a non-boundary row of $\hh{G}$ ($2 \leq j \leq k-1$)
 \begin{align}
     \hh{G}_{j,j-1} &= \frac{1}{2}(G_{j,j-1}+G_{j-1,j}) = j \| F_0\| + (j-1)\|F_2\|, \\ \hh{G}_{j,j} &= 2\mu(F_1) j, \\ \hh{G}_{j,j+1} & = \frac{1}{2}(G_{j,j+1}+G_{j+1,j}) = (j+1) \| F_0\| + j \|F_2\|,
 \end{align}
all other elements in the row being zero.

For $j=1$, the row has nonzero elements 
\begin{align}
    \hh{G}_{1,1} = 2 \mu(F_1), \quad \hh{G}_{1,2} = 2 \| F_0\| + \| F_2\|.
\end{align} 

For $j=k$, we similarly have 
\begin{align}
    \hh{G}_{k,k} = 2 k \mu(F_1), \quad \hh{G}_{k,k-1} =  k \| F_0\| + (k-1) \| F_2\|.
\end{align} 

By applying the Gershgorin's circle theorem to the $\hh{G}$ matrix, and using the fact that it is symmetric, it then follows that all its eigenvalues are contained in 
\begin{align}
    \bigcup_{j=1}^{k} [c_j - r_j, c_j + r_j],
\end{align}
where
\begin{align}
    c_j &= 2 j \mu(F_1), \quad \textrm{for } j=1,\dots, k \\ 
    r_j & =  (2j+1) \| F_0\| + (2j-1) \| F_2\|, \quad \textrm{for } j=2,\dots, k -1, \\ r_1 & = 2 \|F_0 \| + \|F_2\|, \quad r_{k} = k \|F_0\| + (k-1) \|F_2\|. 
\end{align}
For the largest eigenvalue we then have
\begin{align}
     \lambda_1(\hh{G}) \leq \max_{j=1, \dots, k} (c_j + r_j). 
\end{align}

Under the assumption $R<1$ it follows by continuity there exists a rescaling such that \mbox{$\| x(0)\| <1$} and
\begin{align}
\label{eq:littler}
   \bar{r} := \mu(F_1) + \| F_0\| + \| F_2\| < 0.
\end{align} 

We shall assume we fixed such rescaling. Then, 

\begin{align}
 \max_{j=2, \dots, k-1}  (c_j + r_j) & =  \max_{j=2, \dots, k}  2j( \mu(F_1) + \| F_0\| + \| F_2 \|) +  \| F_0\| - \| F_2 \| \\
 & =  \max_{j=2, \dots, k}  2j \bar{r} +  \| F_0\| - \| F_2 \|
 \\
 & \stackrel{j=2}{=} 4 \bar{r} + \| F_0\| - \| F_2\|. 
\end{align}
Also note that for $k\geq 3$
\begin{equation}
    c_{k} + r_{k} = 2 k \bar{r} - k\| F_0\|
    -    
    (k+1)\| F_2\| \leq 4 \bar{r} + \| F_0\| - \| F_2\| = c_2 + r_2
\end{equation}
whereas for $j=1$,
\begin{align}
    c_1 + r_1 = 2\mu(F_1) + 2\| F_0\| + \| F_2 \| = 2 \bar{r} - \| F_2\|. 
\end{align}
In conclusion, under the assumption $R<1$,
\begin{align}
\label{eq:ReGmaxeigenvaluemaxbound}
     \lambda_1(\hh{G}) \leq \max \{ 4 \bar{r} + \| F_0\| - \| F_2\|, 2 \bar{r} - \| F_2\|\}.
\end{align}

\bigskip

\textbf{There exists rescaling choices that  $\lambda_1(\hh{G})<0$ and $\|x(0)\|<1$.} Replacing back the expression for $\bar{r}$ from above, we find a sufficient set of conditions for \mbox{$\lambda_1(\hh{G})<0$}:
\begin{align}
    4 \mu + 5 \| F_0\| + 3 \| F_2\| <0 
\end{align}
and
\begin{align}
       2 \mu + 2 \| F_0\| +  \| F_2\| <0, 
\end{align}
where we note that the second one is implied by the first one. To this, we need to add that the rescaling parameter $\gamma$ should obey
\begin{align}
   0< \gamma < 1/ \| x(0)\|,
\end{align}
and 
\begin{align}
    \mu + \| F_0\| + \|F_2\| < 0,
\end{align}
the latter one having being used in the derivation.

First note that, if $\|F_0\| = 0$, it suffices to satisfy
\begin{align}
   \mu + \| F_2\|  = \bar{r} < 0
\end{align}
notwithstanding the condition $0< \gamma < 1/\| x(0)\|$. We discussed above that a scaling ensuring these condition can be found, by continuity. 

So we can now consider the case $\|F_0\| \neq 0$. Solving the above set of inequalities we find that the following are sufficient:
\begin{align}
\label{eq:gammainequality}
    0 < \gamma < \min \left\{\frac{|\mu| }{2\| F_0\|}, \frac{1}{\| x(0)\|}\right\},
\end{align}
and 
\begin{align}
    \| F_2\| < | \mu| \gamma -  \| F_0\| \gamma^2. 
\end{align}

Let us analyze the situation depending on the value of $\| x(0)\|$.
\begin{itemize}
    \item[Case (a)] $0 < \| x(0)\| < \frac{\| F_0\|}{|\mu|}$
\end{itemize}
One can immediately check this condition implies
\begin{align}
    R \geq  \frac{\| F_0\|}{|\mu| \| x(0)\|}>1,
\end{align}
and so it is incompatible with $R<1$. Hence we do not need to consider this possibility.
\begin{itemize}
    \item[Case (b)] $\frac{\| F_0\|}{|\mu|} < \| x(0)\| < \frac{2 \| F_0\|}{|\mu|}$
\end{itemize}
Under this condition, inequality~\eqref{eq:gammainequality} simplifies to
\begin{align}
\label{eq:gammainequalitysimpler}
    0 < \gamma < \frac{|\mu| }{2\| F_0\|}.
\end{align}
Furthermore, note that $R<1$ implies
\begin{align}
    \mu^2 > \left( \frac{\|F_0\|}{\| x(0)\|} + \| F_2\| \| x(0)\|\right)^2
\end{align}
Note that
\begin{align}
    \left( \frac{\|F_0\|}{\| x(0)\|} + \| F_2\| \| x(0)\|\right)^2 - 4 \| F_0\| \| F_2\| = \frac{(\| F_0\| - \| F_2\| \| x(0)\|^2)^2}{\| x(0)\|^2}\geq 0 
\end{align}
and so
\begin{align}
\label{eq:modifiedR}
  \mu^2 > 4 \| F_0\| \|F_2\| 
\end{align}
Let us choose the rescaling 
\begin{align}
    \gamma = \frac{|\mu|}{2\| F_0\|} - \epsilon,
\end{align} 
with $\epsilon>0$, which is a valid rescaling for $\epsilon$ small enough. The $\|F_2\|$ condition becomes
\begin{align}
\label{eq:modifiedRepsilon}
       \| F_2\| < | \mu| \left(\frac{|\mu|}{2\| F_0\|} - \epsilon\right) -  \| F_0\| \left(\frac{|\mu|}{2\| F_0\|} - \epsilon\right)^2. 
\end{align}
which for $\epsilon = 0$ reads
\begin{align}
    \| F_2\| < \frac{\mu^2}{4 \| F_0\|},
\end{align}
that is just \eqref{eq:modifiedR}. Since the latter is satisfied, by continuity there is $\epsilon>0$ small enough which satisfies \eqref{eq:modifiedRepsilon}, with a valid rescaling.

\begin{itemize}
    \item[Case (c)] $\| x(0)\| \geq \frac{2 \| F_0\|}{|\mu|}$
\end{itemize}
The condition on $\gamma$ becomes
\begin{align}
    0 < \gamma < \frac{1 }{\| x(0)\|}.
\end{align}

We shall then choose the rescaling
\begin{align}
    \gamma = \frac{1}{ \| x(0)\|}-\epsilon.
\end{align}
Again, this condition ensures $\gamma \| x(0)\| < 1$, and $\gamma>0$ for $\epsilon$ small enough, so it is a valid rescaling. 
The $\|F_2\|$ inequality becomes
\begin{align}
\label{eq:Rmodified2}
    \| F_2\| < \frac{- \| F_0\| + 2 \| F_0\| \| x(0)\| \epsilon - \| F_0\| \| x(0)\|^2 \epsilon^2 - \| x(0)\| \mu + \| x(0)\|^2 \mu \epsilon}{\|x_0\|^2}
\end{align}
For $\epsilon = 0$
\begin{align}
   \|F_2\| < \frac{|\mu| \| x(0)\| - \| F_0\|}{\| x(0)\|^2}
\end{align}
which can be rearranged as $R<1$. It follows by continuity that Eq.~\eqref{eq:Rmodified2} follows from $R<1$ for $\epsilon>0$ small enough under a valid choice of rescaling.

This completes the proof for $P=I$.

\subsection{Extending to \texorpdfstring{$P \neq I$}{P/=I}}
 Let us now consider the general case and show it can be reduced to the case $P=I$ that we solved in the previous section. We start with the general quadratic problem
\begin{align}
\label{eq:xsystem}
    \dot{x} = F_0 +  F_1 x  + F_2 x \otimes x,
\end{align}
and perform the change of coordinates $y = Q x$, where $Q^\dag Q = P$. Setting
\begin{align}
\label{eq:Ftildetransformation}
    \tilde{F}_2 = Q F_2 Q^{-1} \otimes Q^{-1}, \quad \tilde{F}_1 = Q F_1 Q^{-1}, \quad \tilde{F}_0 = Q F_0,
\end{align}
we get
\begin{align}
\label{eq:ysystem}
    \dot{y} =  \tilde{F}_0 + \tilde{F}_1 y + \tilde{F}_2 y \otimes y.
\end{align}
We assumed $\mu_P(F_1) <0$. Then, one can show [Lemma 3.31 \cite{plischke2005transient}] that
\begin{align}
\label{eq:connectingtoPlognorm}
\mu(\tilde{F_1}) = \mu(Q F_1 Q^{-1}) = \mu_P(F_1) 
\end{align} 
and so we get $\mu(\tilde{F_1})<0$. 
Furthermore, 
\begin{align*}
   \tilde{R}_{I} = \frac{1}{-\mu(\tilde{F}_{1})} \left (  \|\tilde{F}_{2}\| \|y(0)\| + \frac{\|\tilde{F}_0\|}{ \|y(0)\|} \right )  = \frac{1}{-\mu_P(F_{1})} \left (  \|F_{2}\|_P \|x(0)\|_P + \frac{\|F_0\|_P}{ \|x(0)\|_P} \right ) = R_P <1,
\end{align*}
since
\begin{align}
    \| y(0) \| = \sqrt{ x(0)^\dag Q^\dag Q x(0)} = \sqrt{x(0)^\dag P  x(0)} =   \| x(0)\|_P,
\end{align}
and similarly $\| \tilde{F}_0\| = \| F_0\|_P$; furthermore, using the fact that $Q = U P^{1/2}$ for some unitary $U$ and that the operator norm is invariant under unitaries 
\begin{align}
\label{eq:connectingtoPoperatornorm}
    \| \tilde{F}_2\| = \| F_2\|_P.
\end{align}
Finally, note that from Lemma~\ref{lem:normdecreasegeneral} and $\| x\|_P = \| y\|$ we have $\| y(t)\| \leq \| y(0)\|$. Hence the system in Eq.~\eqref{eq:ysystem} respects all conditions of the special case proved above.
It then follows that the Carleman error vector
\begin{align}
    \tilde{\eta}_j = y^{\otimes j} - \tilde{z}_j,
\end{align}
where $\tilde{z}$ is the solution to the Carleman ODE system constructed from Eq.~\eqref{eq:ysystem}, satisfies
 \begin{align}
    \| \tilde{\eta}(t)\| \leq \frac{1- e^{ \xi t}}{-\xi } k \| \tilde{F}_2\| \| y(0)\|^{k+1},
\end{align}
with $\| y(0) \|<1$ and $
\xi \leq  4 \mu(\tilde{F}_1) + 5 \| \tilde{F}_0\| + 3 \| \tilde{F}_2\| <0$. 
Equivalently, if $\eta_j = x^{\otimes j} - z_j$, with $z = P_{k}^{-1/2} \tilde{z}$, where $P_{k} = \oplus_{j=1}^{k} P^{\otimes j}$, then we have
 \begin{align}
 \label{eq:dissipativeerrorbound}
    \| \eta(t)\|_{P_{k}} \leq \frac{1- e^{ \xi t}}{-\xi } k \|F_2\|_P \| x(0)\|_P^{k+1},
\end{align}
with $\| x(0) \|_P<1$ 
and
\begin{align}
\xi_P \leq  4 \mu(F_1) + 5 \| F_0\|_P + 3 \| F_2\|_P <0.
\end{align}
To complete the proof, we only need to verify that knowing that $\tilde{z}$ is the solution of the truncated Carleman ODE system for \eqref{eq:ysystem}, $z = P_{k}^{-1/2} \tilde{z}$ is the solution of the truncated Carleman ODE system for \eqref{eq:xsystem}:
\begin{align}
\dot{z} = A z + \zeta,
\end{align}
with $A$ the Carleman matrix~\eqref{eq:CarlemannMatrix2} constructed from $F_0$, $F_1$ and $F_2$, and $\zeta = [0, \dots, 0, A_{k,k+1} x^{\otimes k+1}]$.  

To see this, note that by construction $z$ satisfies
\begin{align}
    \dot{z} = P_{k}^{-1/2} \tilde{A} P_{k} z + P_{k} \tilde{\zeta},
\end{align}
where by construction $\tilde{A}$ satisfies the tilde version of Eq.~\eqref{eq:CarlemannMatrix2}. With the slight abuse of tensor product notation flagged before,
\begin{align*}
    P_{k}^{-1/2} \tilde{A} P_{k}^{1/2} & = \sum_{j=2}^{k} \ketbra{j}{j-1} \otimes (P^{-1/2})^{\otimes j} \tilde{A}_{j,j-1} (P^{1/2})^{\otimes (j-1)} 
 + \sum_{j=1}^{k} \ketbra{j}{j} \otimes (P^{-1/2})^{\otimes j}  \tilde{A}_{j,j} (P^{1/2})^{\otimes j}  \\ & + \sum_{j=1}^{k-1} \ketbra{j}{j+1} \otimes (P^{-1/2})^{\otimes j}  \tilde{A}_{j,j+1}  (P^{1/2})^{\otimes (j+1)} \\
 & = \sum_{j=2}^{k} \ketbra{j}{j-1} \otimes A_{j,j-1} 
 + \sum_{j=1}^{k} \ketbra{j}{j} \otimes  A_{j,j}  + \sum_{j=1}^{k-1} \ketbra{j}{j+1} \otimes  A_{j,j+1}, 
\end{align*}
where in the third line we used Eq.~\eqref{eq:Ftildetransformation}. Similarly, since $\tilde{\zeta} = [0, \dots, 0, \tilde{A}_{k,k+1} y^{\otimes k+1}] $ we have $P_{k}^{-1/2} \tilde{\zeta} = \zeta$. Hence, we proved that $z$ satisfies $\dot{z} = A z + \zeta$.

Finally, we have
\begin{align}
\sqrt{\lambda_{\mathrm{min}} (P^{\otimes j}) }\sqrt{\eta_j^\dag \eta_j} \leq  \sqrt{\eta_j^\dag P^{\otimes j} \eta_j} = \| \eta_j(t)\|_{P^{\otimes j}} \leq \| \eta(t)\|_{P_{k}},
\end{align}
and so, using that $\lambda_{\mathrm{min}} (P^{\otimes j}) = \lambda_{\mathrm{min}}(P)^j = \|P^{-1}\|^{-j}$ as well as Eq.~\eqref{eq:dissipativeerrorbound}
\begin{align}
  \|\eta_j(t)\|  \leq \frac{1- e^{ \xi t}}{-\xi } k \|F_2\|_P \| P^{-1}\|^{j/2} \| x(0)\|_P^{k+1}.
\end{align}


\section{Proof of Theorem~\ref{thm:Carleman_cons}}
\label{app:thm_cons}

In what follows, we first derive a bound for the Carleman error $\|\eta_j\|$ of the $j$-th block, expressed with the help of a combinatorial quantity $S_k^{j}$ that we call a weighted fusion sum. Next, we explain how to relate $S_k^{1}$ to the number of binary trees with $k+1$ leaves, and show that it is equal to the Catalan numbers~$C_k$. As a result, we give a closed-form bound on the error $\|\eta_1\|$ of the first block, showing that when $R_\delta<1$, it converges to 0 as $k$ grows. We then generalize our reasoning, relating $S_k^{j}$ to the number of binary forests with $j$ trees and a total of $k+1$ leaves, and show that it is equal to the $j$-th fold convolution $C^{j}_k$ of the Catalan numbers. This, in turn, allows us to derive a closed-form bound on the errors $\|\eta_j\|$ of the general $j$-th block, showing that when $R_\delta<1$, they all converge to 0 as~$k$ grows, and so does the total Carleman error $\|\eta\|$.


\subsection{Carleman error bounds using weighted fusion sums}

In the proof of Proposition~5.3 of Ref.~\cite{forets2017explicit}, the authors show that the $j$-th component of the Carleman error vector for a quadratic system without driving and at time $s_{j-1}=t$ can be expressed exactly as the following nested integral:
\begin{align}
\begin{split}
     \label{eq:nested_int}
     \eta_j(s_{j-1}) &=  \int_{0}^{s_{j-1}}  e^{A_{j,j}(s_{j-1}-s_{j})} A_{j,j+1}  \ ds_j \cdots \int_{0}^{s_{k-1}} e^{A_{k,k}(s_{k-1}-s_k)} A_{k,k+1} \ x^{\otimes k+1}(s_k) \  ds_k,
\end{split}
\end{align}
where $A_{j,j}$ and $A_{j,j+1}$ denote the blocks of the Carleman matrix defined in Eqs.~\eqref{eq:A_F1}-\eqref{eq:A_F2}. We will now transform the above to a more amenable form by diagonalizing the exponential terms using $Q$:
\begin{align}
\begin{split}
    \label{eq:nested_int_diag}
    \eta_j(s_{j-1}) &=  Q^{\ot j} \int_{0}^{s_{j-1}}  e^{\tilde{A}_{j,j}(s_{j-1}-s_{j})} \tilde{A}_{j,j+1} \ ds_j \cdots   \int_{0}^{s_{k-1}} e^{\tilde{A}_{k,k}(s_{k-1}-s_k)} \tilde{A}_{k,k+1} \tilde{x}^{\otimes k+1}(s_k) ds_k,
\end{split}
\end{align}
where
\begin{equation}
    \tilde{A}_{j,j}= (Q^{-1})^{\otimes j} A_{j,j} Q^{\otimes j},\qquad \tilde{A}_{j,j+1} = (Q^{-1})^{\otimes j} A_{j,j+1} Q^{\otimes (j+1)}.
\end{equation}
Clearly, $\tilde{A}_{j,j}$ are diagonal matrices given by
\begin{align}
    \tilde{A}_{j,j}  &:= \Lambda \otimes I^{\otimes (j-1)}  + \mbox{shifts},
\end{align}
whereas
\begin{equation}
    \tilde{A}_{j,j+1} = \tilde{F}_2\otimes I^{\otimes (j-1)}+ \mbox{shifts} =: \sum_{i=1}^j \tilde{F}_2(i),
\end{equation}
where we have introduced a convenient notation with $\tilde{F}_2(i)$ acting non-trivially as $\tilde{F}_2$ on subsystems $i$ and $i+1$, and trivially as identity on all other subsystems. Using this notation, $\eta_j(s_{j-1})$ becomes:
\begin{align}
\begin{split}
    \label{eq:nested_int_F2}
    \eta_j(s_{j-1}) = &  Q^{\ot j} \sum_{l_j=1}^{j} \cdots \sum_{l_k=1}^k \int_{0}^{s_{j-1}} ds_{j} \cdots \int_{0}^{s_{k-1}} ds_k \\
    &\qquad  \times e^{\tilde{A}_{j,j}(s_{j-1}-s_{j})}  \tilde{F}_2(l_j) \dots e^{\tilde{A}_{k,k}(s_{k-1}-s_k)} \tilde{F}_2(l_k) \tilde{x}^{\otimes k+1}(s_k).
\end{split}
\end{align}

Now, we assume that the norm of the transformed solution is bounded for all times,
\begin{equation}
    \forall t:~ \|\tilde{x}(t)\| \leq \|\tilde{x}_{\mathrm{max}}\|,
\end{equation}
and introduce a subnormalized $\tilde{x}$ state:
\begin{equation}
    \tilde{\xi}(t) = \frac{\tilde{x}(t)}{\|\tilde{x}_{\mathrm{max}}\|},
\end{equation}
as well as normalized $\tilde{F}_2(l)$ matrices:
\begin{equation}
    \tilde{f}_2(l) = \frac{\tilde{F}_2(l)}{\|\tilde{F}_2\|}.
\end{equation}
Then, taking the norm of both sides of Eq.~\eqref{eq:nested_int_F2} and using the triangle inequality, we get
\begin{align}
\begin{split}
    \label{eq:error_cons_1}
    \|\eta_j(s_{j-1})\| \leq & \|Q\|^j \|\tilde{x}_{\mathrm{max}}\|^{k+1} \|\tilde{F}_2\|^{k-j+1}   \sum_{l_j=1}^{j} \cdots \sum_{l_k=1}^k \int_{0}^{s_{j-1}} ds_{j} \cdots \int_{0}^{s_{k-1}} ds_k \\ 
    &  \qquad\times \left\| e^{\tilde{A}_{j,j}(s_{j-1}-s_{j})} \tilde{f}_2(l_j)  \dots e^{\tilde{A}_{k,k}(s_{k-1}-s_k)} \tilde{f}_2(l_k) \tilde{\xi}^{\otimes k+1}(s_k)\right\|.
\end{split}
\end{align}

To bound the integrand of the above equation, let us introduce a convenient fusion formalism. First, note that as we sequentially apply matrices $\tilde{f}_2(l_i)$ and $e^{\tilde{A}_{i,i}(s_{i-1}-s_{i})}$ to the state $\tilde{\xi}^{\otimes k+1}$, at each step $i\in\{k,\dots,j\}$ the resulting vector is still a product state. The application of $\tilde{f}_2(l_i)$ \emph{fuses} two subsystems $l_i$ and $l_i+1$ into a single subsystem state, thus reducing the number of subsystems by 1. Such a fusion does not increase the norm as $\|\tilde{f}_2(l_i)\|=1$, and the resulting state of the $l_i$-th subsystem has no support on the first $M$ computational basis states due to Eq.~\eqref{eq:f2_cons}. We will call such subsystems \emph{excited}, as opposed to \emph{unexcited} subsystems that may have the support on the first $M$ basis states. Such terminology is justified by the action of $e^{\tilde{A}_{i,i}(s_{i-1}-s_{i})}$ on the system state, which is in fact the application of $e^{\tilde{F}_1(s_{i-1}-s_{i})}$ to each of the $i$ subsystems. Namely, its action cannot increase the norm when acting on an unexcited system (since $\tilde{F}_1$ is a diagonal matrix with entries having non-positive real parts), but when acting on an excited subsystem, it reduces the norm by at least a factor of $e^{-\delta(F_1)(s_{i-1}-s_{i})}$. We can then bound the integrand by analyzing subsequent fusions and keeping track of how many subsystems are excited at every step. More precisely, let us introduce a function $\Xi_i(l_i,\dots,l_k)$ that counts the number of excited subsystems obtained from the initial $k+1$ unexcited subsystems through a sequence of fusions corresponding to applications of $\tilde{f}_2(l_k),\dots \tilde{f}_2(l_i)$. Then, we have the following bound
\begin{equation}
    \label{eq:Xi_bound}
    \left\| e^{\tilde{A}_{j,j}(s_{j-1}-s_{j})} \tilde{f}_2(l_j) \dots e^{\tilde{A}_{k,k}(s_{k-1}-s_k)} \tilde{f}_2(l_k) \tilde{\xi}^{k+1}(s_k)\right\| \leq \prod_{i=j}^k e^{-\Xi_i(l_i,\dots,l_k)(s_{i-1}-s_{i})\delta(F_1)}.
\end{equation}

Next, for a fixed $(l_i,\dots,l_k)$, we will bound the nested integral term appearing in the expression for $\|\eta_j(s_{j-1})\|$ in Eq.~\eqref{eq:error_cons_1}. Using Eq.~\eqref{eq:Xi_bound}, we get
\begin{align}
    &\int_{0}^{s_{j-1}} ds_{j} \cdots \int_{0}^{s_{k-1}} ds_k \left\| e^{\tilde{A}_{j,j}(s_{j-1}-s_{j})} \tilde{f}_2(l_j) \dots e^{\tilde{A}_{k,k}(s_{k-1}-s_k)} \tilde{f}_2(l_k) \tilde{\xi}^{k+1}(s_k)\right\| \\
     &\qquad \leq \int_{0}^{s_{j-1}} e^{-\Xi_j(l_j,\dots,l_k)(s_{j-1}-s_j)\delta(F_1)} ds_{j} \cdots \int_{0}^{s_{k-1}} e^{-\Xi_k(l_k)(s_{k-1}-s_k)\delta(F_1)} ds_k \\
     &\qquad \leq \frac{1}{\delta(F_1)^{k-j+1}} \left(\prod_{i=j}^k \Xi_i(l_i,\dots l_k)\right)^{-1},
\end{align}
where in the last line we used the following inequality that holds for $a>0$:
\begin{equation}
    \int_0^{s_{i-1}} e^{-a(s_{i-1}-s_i)} \ ds_{i} = \frac{\left(1-e^{-a s_{i-1}}\right)}{a} \leq a^{-1}.
\end{equation}
Then, introducing the weighted fusion sum,
\begin{equation}   
    \label{eq:fusion_sum}
    S^{j}_k := \sum_{l_j=1}^j\cdots \sum_{l_k=1}^k \left(\prod_{i=j}^k \Xi_i(l_i,\dots l_k)\right)^{-1},
\end{equation}
Eq.~\eqref{eq:error_cons_1} can be rewritten as
\begin{equation}
    \label{eq:error_cons_2}
    \|\eta_j(s_{j-1})\| \leq  \|Q\|^j \|\tilde{x}_{\mathrm{max}}\|^{k+1} \left(\frac{\|\tilde{F}_2\|}{\delta(F_1)}\right)^{k-j+1} S^{j}_k.
\end{equation}


\subsection{Weighted fusion sums\texorpdfstring{ $S_k^{1}$}{} and bounds for \texorpdfstring{$\|\eta_1\|$}{the first component of the Carleman error}}
\label{app:weighted_fusion_for_tree}

Recall that Catalan numbers $C_k$ are given by
\begin{equation}
    C_k := \frac{1}{k+1} \binom{2k}{k}.
\end{equation}
We then have the following crucial lemma, the proof of which we postpone for a moment.
\begin{lemma}
    \label{lem:fusion_sums}
    The weighted fusion sum $S_k^{1}$, defined in Eq.~\eqref{eq:fusion_sum}, is given by:
    \begin{equation}
        S_k^{1} = C_k.
    \end{equation}
\end{lemma}
Noting that
\begin{equation}
    \frac{C_{k+1}}{C_k}=\frac{4(k+\frac{1}{2})}{k+2} \leq 4
\end{equation}
and using $C_1=1$, one gets a simple bound for the $k$-th Catalan number for $k\geq 1$:
\begin{equation}
    C_k\leq 4^{k-1}.
\end{equation}
Using Lemma~\ref{lem:fusion_sums} with the above in Eq.~\eqref{eq:error_cons_2} for $j=1$, we obtain
\begin{equation}
    \label{eq:error_cons_3}
    \|\eta_1(s_0)\| \leq  \left(\frac{\|\tilde{x}_{\mathrm{max}}\| \|Q\|}{4}\right) \left(\frac{4\|\tilde{x}_{\mathrm{max}}\| \|\tilde{F}_2\|}{\delta(F_1)} \right)^k.
\end{equation}
Now, we note that employing the rescaling degree of freedom (recall Eq.~\eqref{eq:rescalingrelations}), the first factor can always be made smaller than the second one (which is invariant under rescaling). More precisely, choosing 
\begin{equation}
    \gamma = p \gamma_0,\qquad \gamma_0=\frac{16\|\tilde{F}_2\|}{\delta(F_1)\|Q\|},\qquad p\in(0,1],
\end{equation}
we arrive at
\begin{equation}
    \|\eta_1(s_0)\| \leq  p \left(\frac{4\|\tilde{x}_{\mathrm{max}}\| \|\tilde{F}_2\|}{\delta(F_1)} \right)^{k+1},
\end{equation}
which is the claimed result.

We now proceed to the proof of Lemma~\ref{lem:fusion_sums}. The main idea behind the proof is as follows. First, we will explain how every fusion path from $k+1$ subsystems to one subsystem can be mapped to a binary tree with $k+1$ leaves. Then, we will argue that the fusion paths corresponding to a given binary tree can be seen as ``growing from excitations''. As a result, the sum over fusion paths corresponding to a given tree, weighted by the inverse product of excitations, is always equal to 1. Thus, $S^{1}_k$ is given by the number of binary trees with $k+1$ leaves, which is known to be the $k$-th Catalan number (see, e.g., 1.5.1~Theorem of Ref.~\cite{stanley2015catalan}, where what we call binary trees is referred to as planar trees).

We start by introducing a graphical formalism for representing fusion paths, choosing $k=3$ as an example. We start with $k+1$ unexcited subsystems represented by '$\bt$', and at each time step $i\in\{k,\dots,1\}$ ($i$ tells us how many systems will be there after that time step) we can choose two neighboring subsystems and fuse them into one excited subsystem denoted by '$\btx$', while leaving other subsystems unchanged. We will denote the set of such $k!$ fusion paths from $k+1$ to 1 subsystem by $\F^{1}_{k+1}$, with $\F^1_4$ given by:
\begingroup
\tikzset{
    level distance=32pt,
    grow'=up,
    every tree node/.style = {align=center, anchor=north,nodecolor},
    edge from parent/.style = {draw,edgecolor,thick}
}
\begin{equation}
    \F^{1}_4=
    \left\{\raisebox{-4em}{
        \begin{tikzpicture}[scale=0.8]
            \Tree [.$\btx$ [.$\btx$ [.$\btx$ [.$\bt$ ] [.$\bt$ ] ]
            [.$\bt$ $\bt$ ] ]      
            [.$\bt$ [.$\bt$ $\bt$ ] ] ]
        \end{tikzpicture},~~
        \begin{tikzpicture}[scale=0.8]
            \Tree [.$\btx$ [.$\btx$ [.$\bt$ [.$\bt$ ] ]
            [.$\btx$ [.$\bt$ ] [.$\bt$ ] ] ]       
            [.$\bt$ [.$\bt$ $\bt$ ] ] ]         ] 
        \end{tikzpicture},~~
        \begin{tikzpicture}[scale=0.8]
            \Tree [.$\btx$ [.$\btx$ [.$\bt$ $\bt$ ] [.$\bt$  $\bt$ ] ]
            [.$\btx$ [.$\btx$ [.$\bt$ ] [.$\bt$ ] ] ] ]          
        \end{tikzpicture},~~
        \begin{tikzpicture}[scale=0.8]
            \Tree [.$\btx$ [.$\btx$ [.$\btx$ [.$\bt$ ] [.$\bt$ ] ] ]
            [.$\btx$ [.$\bt$ $\bt$ ] [.$\bt$ $\bt$ ] ] ]          
        \end{tikzpicture},~~
        \begin{tikzpicture}[scale=0.8]
            \Tree [.$\btx$ [.$\bt$ [.$\bt$ [.$\bt$ ] ] ]
            [.$\btx$ [.$\btx$ [.$\bt$ ] [.$\bt$ ] ]
            [.$\bt$ [.$\bt$ ] ] ] ]         
        \end{tikzpicture},~~
        \begin{tikzpicture}[scale=0.8]
            \Tree [.$\btx$ [.$\bt$ [.$\bt$ [.$\bt$ ] ] ]
            [.$\btx$ [.$\bt$ [.$\bt$ ] ]
            [.$\btx$ [.$\bt$ ] [.$\bt$ ] ] ] ]              
        \end{tikzpicture}} 
    \right\}.
\end{equation}
\endgroup
Using this notation, we can rewrite Eq.~\eqref{eq:fusion_sum} for $j=1$ as
\begin{equation}
    \label{eq:fusion_sum2}
    S^{1}_k = \sum_{\v{f}\in\F^{1}_{k+1}} \left(\prod_{i=1}^k \Xi_i(\v{f})\right)^{-1},
\end{equation}
where, with a slight abuse of notation, $\Xi_i(\v{f})$ denotes the number of excitations at time step $i$ of a fusion path $\v{f}$. For example, the first two and last two fusion paths in $\F^{1}_4$ have $\Xi_1=\Xi_2=\Xi_3=1$, whereas the middle two fusion paths have $\Xi_1=\Xi_3=1$ and $\Xi_2=2$.

Now, in order to map each fusion path to a binary tree, we contract all the trivial links, i.e., the ones that just push the subsystem to the next time step, without fusing it. For the considered example with $k=3$, we illustrate the links to be contracted in gray:
\begingroup
\tikzset{
    level distance=32pt,
    grow'=up,
    every tree node/.style = {align=center, anchor=north,nodecolor},
    edge from parent/.style = {draw,edgecolor,thick}
}
\begin{equation}
    \left\{\raisebox{-4em}{
        \begin{tikzpicture}[scale=0.8]
        \Tree [.$\btx$ [.$\btx$ [.$\btx$ [.$\bt$ ] [.$\bt$ ] ]
                      [.$\bt$ \edge[draw=mygray]; \textcolor{mygray}{$\bt$} ] ]      
                  [.$\bt$ \edge[draw=mygray]; [.\textcolor{mygray}{$\bt$} \edge[draw=mygray]; \textcolor{mygray}{$\bt$} ] ] ]
        \end{tikzpicture},~~
        \begin{tikzpicture}[scale=0.8]
        \Tree [.$\btx$ [.$\btx$ [.$\bt$ \edge[draw=mygray]; [.\textcolor{mygray}{$\bt$} ] ]
                      [.$\btx$ [.$\bt$ ] [.$\bt$ ] ] ]       
                  [.$\bt$ \edge[draw=mygray]; [.\textcolor{mygray}{$\bt$} \edge[draw=mygray]; \textcolor{mygray}{$\bt$} ] ] ]         ] 
        \end{tikzpicture},~~
        \begin{tikzpicture}[scale=0.8]
        \Tree [.$\btx$ [.$\btx$ [.$\bt$ \edge[draw=mygray]; \textcolor{mygray}{$\bt$} ] [.$\bt$ \edge[draw=mygray];  \textcolor{mygray}{$\bt$} ] ]
                  [.$\btx$ \edge[draw=mygray]; [.\textcolor{mygray}{$\btx$} [.$\bt$ ] [.$\bt$ ] ] ] ]          
        \end{tikzpicture},~~
        \begin{tikzpicture}[scale=0.8]
        \Tree [.$\btx$ [.$\btx$ \edge[draw=mygray];[.\textcolor{mygray}{$\btx$} [.$\bt$ ] [.$\bt$ ] ] ]
                  [.$\btx$ [.$\bt$ \edge[draw=mygray]; \textcolor{mygray}{$\bt$} ] [.$\bt$ \edge[draw=mygray]; \textcolor{mygray}{$\bt$} ] ] ]        \end{tikzpicture},~~
        \begin{tikzpicture}[scale=0.8]
        \Tree [.$\btx$ [.$\bt$ \edge[draw=mygray]; [.\textcolor{mygray}{$\bt$} \edge[draw=mygray]; [.\textcolor{mygray}{$\bt$} ] ] ]
                  [.$\btx$ [.$\btx$ [.$\bt$ ] [.$\bt$ ] ]
                      [.$\bt$ \edge[draw=mygray]; [.\textcolor{mygray}{$\bt$} ] ] ] ]         
        \end{tikzpicture},~~
        \begin{tikzpicture}[scale=0.8]
        \Tree [.$\btx$ [.$\bt$ \edge[draw=mygray]; [.\textcolor{mygray}{$\bt$} \edge[draw=mygray]; [.\textcolor{mygray}{$\bt$} ] ] ]
                  [.$\btx$ [.$\bt$ \edge[draw=mygray]; [.\textcolor{mygray}{$\bt$} ] ]
                      [.$\btx$ [.$\bt$ ] [.$\bt$ ] ] ] ]              
        \end{tikzpicture}} 
    \right\}
\end{equation}
\endgroup
This way, one obtains all binary trees with $k+1$ leaves, the set of which we will denote by $\T_{k+1}$, potentially with many different fusion paths resulting in the same binary tree. For example, when $k=3$, the third and fourth fusion path get mapped to the same binary tree, so $\T_4$ is given by
\begingroup
\tikzset{
    level distance=32pt,
    grow'=up,
    every tree node/.style = {align=center, anchor=north,nodecolor},
    edge from parent/.style = {draw,edgecolor,thick}
}
\begin{equation}
    \T_4=
    \left\{\raisebox{-4em}{
        \begin{tikzpicture}[scale=0.8]
        \Tree [.$\btx$ [.$\btx$ [.$\btx$ [.$\bt$ ] [.$\bt$ ] ]
                      [.$\bt$ ] ]      
                  [.$\bt$ ] ]
        \end{tikzpicture},~~
        \begin{tikzpicture}[scale=0.8]
        \Tree [.$\btx$ [.$\btx$ [.$\bt$ ]
                      [.$\btx$ [.$\bt$ ] [.$\bt$ ] ] ]       
                  [.$\bt$ ] ] ] 
        \end{tikzpicture},~~
        \begin{tikzpicture}[scale=0.8]
        \Tree [.$\btx$ [.$\btx$ [.$\bt$ ] [.$\bt$ ] ]
                  [.$\btx$  [.$\bt$ ] [.$\bt$ ] ] ]          
        \end{tikzpicture},~~
        \begin{tikzpicture}[scale=0.8]
        \Tree [.$\btx$ [.$\bt$ ]
                  [.$\btx$ [.$\btx$ [.$\bt$ ] [.$\bt$ ] ]
                      [.$\bt$ ] ] ]         
        \end{tikzpicture},~~
        \begin{tikzpicture}[scale=0.8]
        \Tree [.$\btx$ [.$\bt$ ]
                  [.$\btx$ [.$\bt$ ]
                      [.$\btx$ [.$\bt$ ] [.$\bt$ ] ] ] ]              
        \end{tikzpicture}}
    \right\}.
\end{equation}
\endgroup
Note that the excited subsystems from the fusion path are mapped to internal nodes of the tree, whereas unexcited subsystems become leaves. For a given binary tree $\v{t}$, let us then denote by $\Phi(\v{t})$ the set of all fusion paths corresponding to $\v{t}$. For the considered example, we have
\begingroup
\tikzset{
    level distance=32pt,
    grow'=up,
    every tree node/.style = {align=center, anchor=north,nodecolor},
    edge from parent/.style = {draw,edgecolor,thick}
}
\begin{equation}
    \Phi\left(
        \raisebox{-2.8em}{
            \begin{tikzpicture}[scale=0.8]
            \tikzset{grow'=up,level distance=32pt}
              \Tree [.$\btx$ [.$\btx$ [.$\bt$ ] [.$\bt$ ] ]
                      [.$\btx$ [.$\bt$ ] [.$\bt$ ]  ] ]  
            \end{tikzpicture}
        }
    \right) 
    = 
    \left\{
        \raisebox{-4em}{
            \begin{tikzpicture}[scale=0.8]
            \tikzset{grow'=up,level distance=32pt}
              \Tree [.$\btx$ [.$\btx$ [.$\bt$ $\bt$ ] [.$\bt$  $\bt$ ] ]
                      [.$\btx$ [.$\btx$ [.$\bt$ ] [.$\bt$ ] ] ] ]               
            \end{tikzpicture},~~
            \begin{tikzpicture}[scale=0.8]
            \tikzset{grow'=up,level distance=32pt}
              \Tree [.$\btx$ [.$\btx$ [.$\btx$ [.$\bt$ ] [.$\bt$ ] ] ]
                      [.$\btx$ [.$\bt$ $\bt$ ] [.$\bt$ $\bt$ ] ] ] 
            \end{tikzpicture}
        }
    \right\}
\end{equation}
\endgroup
and the remaining four sets contain each a single fusion path. Using $\Phi$, we can rewrite Eq.~\eqref{eq:fusion_sum2} as
\begin{equation}
    \label{eq:fusion_sum3}
    S^{1}_k = \sum_{\v{t}\in \T_{k+1}} \sum_{\v{f}\in \Phi(\v{t})} \left(\prod_{i=1}^k \Xi_i(\v{f})\right)^{-1}.
\end{equation}

So far we have discussed how, given $\v{f}\in\F^{1}_{k+1}$, to obtain $\v{t}\in\T_{k+1}$ such that $\v{f}\in\Phi(\v{t})$. Now, we will explain the reverse process, i.e., how, given a tree $\v{t}$, find all fusion paths belonging to $\Phi(\v{t})$. One starts at the excited root, and at each time step $i$ grows one of the internal nodes (i.e., one of '$\btx$') into two of its children, while trivially evolving all other nodes (i.e., just pushing them to the next time step, without changing). Crucially, at every time step $i$, the number of different ways to grow a tree (i.e., a number of ways to create distinct fusion paths) is equal to the number of internal nodes, so the number of excited states $\Xi_i(\v{f})$ at this time step. Therefore, when summing over all fusion paths corresponding to a given tree, the number of new ways to construct fusions at each step $i$ exactly cancels out the weight $\Xi_i(\v{f})^{-1}$, and so:
\begin{equation}
    \sum_{\v{f}\in \Phi(\v{t})} \left(\prod_{i=1}^k \Xi_i(\v{f})\right)^{-1} = 1.
    \label{eq:identity_fusion_paths_for_trees}
\end{equation}
As a result, Eq.~\eqref{eq:fusion_sum3} simplifies to
\begin{equation}
    \label{eq:fusion_sum4}
    S_k^{1} = \sum_{\v{t}\in \T_{k+1}} 1 = C_k,
\end{equation}
where we used the known fact concerning the number of binary trees with $k+1$ leaves~\cite{stanley2015catalan}.

Note that one can prove the crucial identity from Eq.~\eqref{eq:identity_fusion_paths_for_trees} in a more formal and rigorous way using the concepts appearing in Section~\ref{subsec:expression_for_v_inv_ij_diagonal_f1}. More precisely, given a binary tree $\mbt$, we use $V_I=V_I(\mbt)$ and $V_L=V_L(\mbt)$ to denote the set of its internal nodes and leaves, respectively, and use $O(\mbt)$ to denote the set of topological orders of the nodes in $V_I$. For any $v \in V_I$, let $C(v)$ be the set of its children. Then for any $S \subseteq V_I$, define 
\begin{align}
C(S)=\bigcup\limits_{v \in S} C(v) \setminus S, \quad    
E(S)=C(S)\cap V_I.
\end{align}

A topological order of a binary tree is a linear ordering of its internal nodes such that every node comes before its descendants in this order. Observe that for any $\mbt \in \mathcal{T}_{k+1}$, there is a one-to-one correspondence between the set of associated fusion paths $\Phi(\mbt)$
and the set of topological orders $O(\mbt)$. Furthermore, suppose a fusion path $\v{f}$ corresponds to a topological order $\vec v=(v_1, v_2, \dots, v_k)$. Then for all $i \in [k-1]$,  we have
\begin{align}
    \Xi_{i+1}(\v{f}) = |E(\lrcb{v_1, v_2,\dots, v_i})|.
\end{align}
Meanwhile, we always have $\Xi_{1}(\v{f})=1$. Therefore, proving Eq.~\eqref{eq:identity_fusion_paths_for_trees} amounts to showing that
\begin{align}
    \sum_{\vec v \in O(\mbt)} h(\vec v) = 1,
    \label{eq:identity_about_topological_order}
\end{align}
where
\begin{align}
h(v_1,v_2,\dots,v_{k}) = \lrb{\prod_{i=1}^{k-1} |E(\lrcb{v_1, v_2, \dots, v_i})|}^{-1}.
\end{align}

We prove Eq.~\eqref{eq:identity_about_topological_order} by induction on the number of leaves in $\mbt$. In the base case, $\mbt$ consists of a root and two leaves; in this case, both sides of Eq.~\eqref{eq:identity_about_topological_order} evaluate to $1$, and the identity holds.

Now suppose that Eq.~\eqref{eq:identity_about_topological_order} holds for all binary trees with at most $n+1$ leaves, where $n \ge 1$. Let $\mbt \in \mathcal{T}_{n+2}$ be arbitrary. Since $n+2 \ge 3$, the tree has depth at least $2$. Therefore, there exists an internal node $v$ with a parent $u$, and two children $w_1$ and $w_2$, both of which are leaves. 

Let $\mbt'$ be the binary tree obtained from $\mbt$ by deleting the nodes $w_1$, $w_2$, along with their adjacent edges $(v, w_1)$, $(v, w_2)$. Then, $\mbt'$ is a binary tree with $n+1$ leaves. So, by the inductive assumption, we have
\begin{align}
    \sum_{\vec v \in O(\mbt')} h'(\vec v)=1,
    \label{eq:inductive_assumption_topological_order}
\end{align}
where 
\begin{align}
h'(v_1,v_2,\dots,v_n) = \lrb{\prod_{i=1}^{n-1} |E'(\lrcb{v_1, v_2, \dots, v_i})|}^{-1},
\end{align}
with $E'$ being defined analogously to $E$ but with respect to the reduced tree $\mbt'$. 

Given any $\vec u =(u_1, u_2, \dots, u_n) \in O(\mbt')$, we map it to one or more elements of $O(\mbt)$. This mapping induces a partition of $O(\mbt)$. Specifically, since $\lrcb{u_1, u_2, \dots, u_n}$ contains all internal nodes of $\mbt'$, which includes $u$, there must exist $m \in [n]$ such that $u_m=u$. Then for $j\in\{m, \dots, n\}$, we define
\begin{align}
    \vec u^{(j)}=(u_1, u_2, \dots, u_j, v, u_{j+1}, \dots, u_n).
\end{align}
That is, $\vec u^{(j)}$ is obtained from $\vec u$ by inserting $v$ right after $u_j$. Clearly, $\vec u^{(j)}$ is a valid topological order of $\mbt$. Furthermore, let
\begin{align}
    \mathcal{Z}(\vec u) = \lrcb{\vec u^{(j)}:~j\in \lrcb{m, m+1, \dots, n}}.
\end{align}
Then one can verify that
\begin{align}
    O(\mbt) = \bigsqcup_{\vec u \in O(\mbt')} \mathcal{Z}(\vec u).
    \label{eq:partition_topological_order_set}
\end{align}
In other words, every topological order of $\mbt$ can be generated by extending a corresponding topological order of $\mbt_i'$ in the above manner.

Now for $i \in [n-1]$, let $x_i=|E'(\lrcb{u_1, u_2, \dots, u_i})|$.    
Moreover, for $j\in\{m, \dots, n\}$ and $l \in [n]$, let  
$\vec u^{(j)}_{1:l}$ be the set consisting of the first $l$ elements of $\vec u^{(j)}$.
Then, one can verify that for $j\in\{m,\dots, n-1\}$, 
\begin{align}
    |E(\vec u^{(j)}_{1:l})| = \begin{cases}
    x_l, \quad &\textrm{if}~1 \le l < m; \\
    x_l+1, \quad &\textrm{if}~m \le l \le j; \\
    x_{l-1}, \quad &\textrm{otherwise},        
    \end{cases}
\end{align}
while for $j=n$,
\begin{align}
    |E(\vec u^{(n)}_{1:l})| = \begin{cases}
    x_l, \quad &\textrm{if}~1 \le l < m; \\
    x_l+1, \quad &\textrm{if}~m \le l \le n-1; \\
    1, \quad &\textrm{otherwise}.       
    \end{cases}
\end{align}
Here $E$ is defined with respect to the original tree $\mbt$. The difference of $1$ in the middle case arises from the fact that $v$ is an internal node of $\mbt$, but not of $\mbt'$.

It follows that if $m<n$, then
\begin{align}
\sum_{\vec v \in \mathcal{Z}(\vec u)} h(\vec v)
&=\sum_{j=m}^{n} \lrb{\prod_{l=1}^n |E(\vec u^{(j)}_{1:l})|}^{-1}\\
&=    
\lrb{\prod_{l=1}^{m-1} \frac{1}{x_l}} \cdot 
\left [\sum_{j=m}^{n-1} 
\lrb{\prod_{l=m}^j \frac{1}{x_l+1} \cdot \prod_{l=j}^{n-1} \frac{1}{x_l}}  +\prod_{l=m}^{n-1}\frac{1}{x_l+1} \right ].
\end{align}
Using the identity, valid for arbitrary $y_1, y_2, \dots, y_q \in \R \setminus \lrcb{0, -1}$,
\begin{align}
\frac{1}{y_1} \frac{1}{y_2} \cdots \frac{1}{y_q}
=
\frac{1}{y_1+1} \frac{1}{y_2+1} \cdots \frac{1}{y_q+1}
+\sum_{j=1}^{q} \frac{1}{y_1} \frac{1}{y_2} \cdots \frac{1}{y_j}
\frac{1}{y_j+1}\frac{1}{y_{j+1}+1} \cdots \frac{1}{y_q+1},
\end{align}
which can be proven by induction on $q$, we conclude that
\begin{align}
\sum_{\vec v \in \mathcal{Z}(\vec u)} h(\vec v)
=\prod_{l=1}^{n-1} \frac{1}{x_l} 
=h'(\vec u).
\end{align}
On the other hand, if $m=n$, then 
\begin{align}
\mathcal{Z}(\vec u) = \lrcb{(u_1, u_2, \dots, u_n=u, v)}
\end{align}
and therefore we also have
\begin{align}
\sum_{\vec v \in \mathcal{Z}(\vec u)} h(\vec v)
=\prod_{l=1}^{n-1} \frac{1}{x_l} 
=h'(\vec u).
\end{align}
Then it follows from Eqs.~\eqref{eq:inductive_assumption_topological_order} and \eqref{eq:partition_topological_order_set} that
\begin{align}
    \sum_{\vec v \in O(\mbt)} h(\vec v)
    =   \sum_{\vec u \in O(\mbt')} \sum_{\vec v \in \mathcal{Z}(\vec u)} h(\vec v)
    =   \sum_{\vec u \in O(\mbt')}  h'(\vec u)
    = 1.
\end{align}
Thus, Eq.~\eqref{eq:identity_about_topological_order} also holds for the binary tree $\mbt$ with $n+2$ leaves.


\subsection{Weighted fusion sums\texorpdfstring{ $S_k^{j}$}{} and bounds for \texorpdfstring{$\|\eta_j\|$}{higher components of the Carleman error}}

We will now show how to generalize the reasoning from the previous section to calculate $S_k^{j}$ and use it to bound $\|\eta_j\|$ for general $j$. First, let us introduce the Catalan $j$-fold convolution $C_k^{j}$:
\begin{equation}
    C_k^{j} := \sum_{\substack{k_1+\dots+k_j = k \\ k_1,\dots, k_j \geq 0 }} C_{k_1} \dots C_{k_j}. 
\end{equation}
It is known (see, e.g., Eq.~(14) of Ref.~\cite{larcombe2003catalan}) that the above can be expressed via the following simple formula:
\begin{equation}
    C_k^{j} = \frac{j}{k+j} \binom{2k+j-1}{k}.
\end{equation}
We then have the following lemma that generalizes Lemma~\ref{lem:fusion_sums}.
\begin{lemma}
    \label{lem:fusion_sums_general}
    The weighted fusion sums $S_k^{j}$, defined in Eq.~\eqref{eq:fusion_sum}, are given by:
    \begin{equation}
        S_k^{j} = C^{j}_{k-j+1} = \frac{j}{k+1} \binom{2k-j+1}{k-j+1}.
    \end{equation}
\end{lemma}
Postponing for a moment the proof, let us upper bound $S_k^{j}$ using the following known bound for the binomial coefficient:
\begin{equation}
   \binom{n}{k} \leq \left(\frac{e n}{k}\right)^k.
\end{equation}
Employing it, we have
\begin{equation}
    S_k^{j} \leq \frac{ e^{k+1-j} j}{k+1} \left(\frac{2k+1-j}{k+1-j}\right)^{k+1-j} = \frac{ e^{k+1-j} j}{k+1} \left[\left(1 + \frac{1}{\frac{k+1-j}{k}}\right)^{\frac{k+1-j}{k}}\right]^k \leq \frac{ 2^k e^{k+1-j} j}{k+1}.
\end{equation}
Using Lemma~\ref{lem:fusion_sums_general} with the above in Eq.~\eqref{eq:error_cons_2}, we arrive at
\begin{align}
    \|\eta_j(s_{j-1})\| &\leq  2^k \|Q\|^j  \|\tilde{x}_{\mathrm{max}}\|^{k+1} \left(\frac{e\|\tilde{F}_2\|}{\delta(F_1)}\right)^{k+1-j}\frac{j}{k+1}\\
    & = \left( 2 \|\tilde{x}_{\mathrm{max}}\| \|Q\|  \right)^j \left(\frac{2e\|\tilde{x}_{\mathrm{max}}\|\|\tilde{F}_2\|}{\delta(F_1)}\right)^{k+1-j}\frac{j}{2(k+1)}.
\end{align}
Now, we note that employing the rescaling degree of freedom (recall Eq.~\eqref{eq:rescalingrelations}), the first term can always be made smaller than the second one (which is invariant under rescaling). More precisely, choosing 
\begin{equation}
    \gamma = p \gamma_0,\qquad \gamma_0=\frac{e\|\tilde{F}_2\|}{\delta(F_1)\|Q\|},\qquad p\in(0,1],
\end{equation}
we obtain
\begin{align}
    \label{eq:cons_err_final}
    \|\eta_j(s_{j-1})\| &\leq \frac{j}{2(k+1)}p^j R_\delta^{k+1},
\end{align}
where
\begin{equation}
    R_\delta :=\frac{2e\|\tilde{x}_{\mathrm{max}}\|\|\tilde{F}_2\|}{\delta(F_1)}.
\end{equation}
Clearly, if $R_\delta<1$, the Carleman error in every block $j$ can be made arbitrarily small, and also the total Carleman error $\|\eta\|$ goes to zero as $k\to\infty$.

We now proceed to the proof of Lemma~\ref{lem:fusion_sums_general}. The proof is analogous to the proof of Lemma~\ref{lem:fusion_sums}. First, we will explain how every fusion path from $k+1$ subsystems to $j$ subsystems can be mapped to a binary forest consisting of $j$ binary trees with a total number of $k+1$ leaves. Then, we will argue that the fusion paths corresponding to a given binary forest can be seen as growing from excitations. As a result, the sum over fusion paths corresponding to a given forest, weighted by the product of excitations, is always equal to 1. Thus, $S^{j}_k$ is given by the number of binary forests with $j$ trees and a total of $k+1$ leaves, which is known to be the $j$-th fold convolution of the Catalan sequence evaluated at $k-j+1$ (see, e.g., Ref.~\cite{dimitrov2024simple}).

We will employ the graphical formalism introduced in the previous section, choosing $k=3$ and $j=2$ as an example. We start with $k+1$ unexcited subsystems and at each time step $i\in\{k,\dots,j\}$ we can fuse two neighboring subsystems into one excited subsystem. We will denote the set of such $k!/(j-1)!$ fusion paths from $k+1$ to $j$ subsystems by $\F^{j}_{k+1}$, with $\F^{2}_4$ given by:
\begingroup
\tikzset{
    level distance=32pt,
    grow'=up,
    every tree node/.style = {align=center, anchor=north,nodecolor},
    edge from parent/.style = {draw,edgecolor,thick}
}
\begin{equation}
    \F^{2}_4=
    \left\{\raisebox{-2.8em}{
        \begin{tikzpicture}[scale=0.8]
            \Tree [.$\btx$ [.$\btx$ [.$\bt$ ] [.$\bt$ ] ]
                          [.$\bt$ $\bt$ ] ] 
        \end{tikzpicture}
        \begin{tikzpicture}[scale=0.8]
            \Tree   [.$\bt$ [.$\bt$ $\bt$ ] ] ]
        \end{tikzpicture},~~
        \begin{tikzpicture}[scale=0.8]
            \Tree [.$\btx$ [.$\bt$ [.$\bt$ ] ]
                          [.$\btx$ [.$\bt$ ] [.$\bt$ ] ] ]  
        \end{tikzpicture}
        \begin{tikzpicture}[scale=0.8]
            \Tree  [.$\bt$ [.$\bt$ $\bt$ ] ] ]      
        \end{tikzpicture},~~
        \begin{tikzpicture}[scale=0.8]
            \Tree  [.$\btx$ [.$\bt$ $\bt$ ] [.$\bt$  $\bt$ ] ]      
        \end{tikzpicture}
        \begin{tikzpicture}[scale=0.8]
            \Tree  [.$\btx$ [.$\btx$ [.$\bt$ ] [.$\bt$ ] ] ]        
        \end{tikzpicture},~~
        \begin{tikzpicture}[scale=0.8]
            \Tree [.$\btx$ [.$\btx$ [.$\bt$ ] [.$\bt$ ] ] ]       
        \end{tikzpicture}
        \begin{tikzpicture}[scale=0.8]
            \Tree [.$\btx$ [.$\bt$ $\bt$ ] [.$\bt$ $\bt$ ] ]           
        \end{tikzpicture},~~
        \begin{tikzpicture}[scale=0.8]
            \Tree [.$\bt$ [.$\bt$ [.$\bt$ ] ] ]    
        \end{tikzpicture}
        \begin{tikzpicture}[scale=0.8]
            \Tree [.$\btx$ [.$\btx$ [.$\bt$ ] [.$\bt$ ] ]
                          [.$\bt$ [.$\bt$ ] ] ]        
        \end{tikzpicture},~~
        \begin{tikzpicture}[scale=0.8]
            \Tree [.$\bt$ [.$\bt$ [.$\bt$ ] ] ]
        \end{tikzpicture}
        \begin{tikzpicture}[scale=0.8]
            \Tree [.$\btx$ [.$\bt$ [.$\bt$ ] ]
                          [.$\btx$ [.$\bt$ ] [.$\bt$ ] ] ] 
        \end{tikzpicture}
        } 
    \right\}.
\end{equation}
\endgroup
Using this notation, we can rewrite Eq.~\eqref{eq:fusion_sum} for general $j$ as
\begin{equation}
    \label{eq:fusion_sum2_j}
    S^{j}_k = \sum_{\v{f}\in\F^{j}_{k+1}} \left(\prod_{i=j}^k \Xi_i(\v{f})\right)^{-1},
\end{equation}
where again $\Xi_i(\v{f})$ denotes the number of excitations at time step $i$ of a fusion path $\v{f}$. For example, the first two and last two fusion paths in $\F^{2}_4$ have $\Xi_2=\Xi_3=1$, whereas the middle two fusion paths have $\Xi_3=1$ and $\Xi_2=2$.

Now, in order to map each fusion path to a binary forest, we contract all the trivial links, i.e., the ones that just push the subsystem to the next time step, without fusing it. For the considered example with $k=3$ and $j=2$, we illustrate the links to be contracted in gray:
\begingroup
\tikzset{
    level distance=32pt,
    grow'=up,
    every tree node/.style = {align=center, anchor=north,nodecolor},
    edge from parent/.style = {draw,edgecolor,thick}
}
\begin{equation}
    \left\{\raisebox{-2.8em}{
        \begin{tikzpicture}[scale=0.8]
            \Tree [.$\btx$ [.$\btx$ [.$\bt$ ] [.$\bt$ ] ]
                          [.$\bt$ \edge[draw=mygray]; \textcolor{mygray}{$\bt$} ] ]      
        \end{tikzpicture}
        \begin{tikzpicture}[scale=0.8]
            \Tree [.$\bt$ \edge[draw=mygray]; [.\textcolor{mygray}{$\bt$} \edge[draw=mygray]; \textcolor{mygray}{$\bt$} ] ] 
        \end{tikzpicture},~~
        \begin{tikzpicture}[scale=0.8]
            \Tree  [.$\btx$ [.$\bt$ \edge[draw=mygray]; [.\textcolor{mygray}{$\bt$} ] ]
                          [.$\btx$ [.$\bt$ ] [.$\bt$ ] ] ]       
        \end{tikzpicture}
        \begin{tikzpicture}[scale=0.8]
            \Tree [.$\bt$ \edge[draw=mygray]; [.\textcolor{mygray}{$\bt$} \edge[draw=mygray]; \textcolor{mygray}{$\bt$} ] ] ]         
        \end{tikzpicture},~~
        \begin{tikzpicture}[scale=0.8]
            \Tree [.$\btx$ [.$\bt$ \edge[draw=mygray]; \textcolor{mygray}{$\bt$} ] [.$\bt$ \edge[draw=mygray];  \textcolor{mygray}{o} ] ]        
        \end{tikzpicture}
        \begin{tikzpicture}[scale=0.8]
            \Tree [.$\btx$ \edge[draw=mygray]; [.\textcolor{mygray}{x} [.$\bt$ ] [.$\bt$ ] ] ]     
        \end{tikzpicture},~~
        \begin{tikzpicture}[scale=0.8]
            \Tree [.$\btx$ \edge[draw=mygray];[.\textcolor{mygray}{x} [.$\bt$ ] [.$\bt$ ] ] ]
        \end{tikzpicture}
        \begin{tikzpicture}[scale=0.8]
            \Tree [.$\btx$ [.$\bt$ \edge[draw=mygray]; \textcolor{mygray}{o} ] [.$\bt$ \edge[draw=mygray]; \textcolor{mygray}{$\bt$} ] ] ]      
        \end{tikzpicture},~~
        \begin{tikzpicture}[scale=0.8]
            \Tree [.$\bt$ \edge[draw=mygray]; [.\textcolor{mygray}{$\bt$} \edge[draw=mygray]; \textcolor{mygray}{$\bt$} ] ] 
        \end{tikzpicture}
        \begin{tikzpicture}[scale=0.8]
            \Tree [.$\btx$ [.$\btx$ [.$\bt$ ] [.$\bt$ ] ]
                      [.$\bt$ \edge[draw=mygray]; [.\textcolor{mygray}{$\bt$} ] ] ] 
        \end{tikzpicture},~~
        \begin{tikzpicture}[scale=0.8]
            \Tree [.$\bt$ \edge[draw=mygray]; [.\textcolor{mygray}{$\bt$} \edge[draw=mygray]; [.\textcolor{mygray}{$\bt$} ] ] ]       
        \end{tikzpicture}
        \begin{tikzpicture}[scale=0.8]
            \Tree [.$\btx$ [.$\bt$ \edge[draw=mygray]; [.\textcolor{mygray}{$\bt$} ] ]
                      [.$\btx$ [.$\bt$ ] [.$\bt$ ] ] ] 
        \end{tikzpicture}
    } \right\}
\end{equation}
\endgroup
This way, one obtains all binary forests with $j$ tress and a total of $k+1$ leaves, the set of which we will denote by $\T^{j}_{k+1}$, potentially with many different fusion paths resulting in the same binary forest. For example, when $k=3$ and $j=2$, the third and fourth fusion path get mapped to the same binary forest, so $\T^{2}_4$ is given by
\begingroup
\tikzset{
    level distance=32pt,
    grow'=up,
    every tree node/.style = {align=center, anchor=north,nodecolor},
    edge from parent/.style = {draw,edgecolor,thick}
}
\begin{equation}
    \T_4^{2}=
    \left\{\raisebox{-2.8em}{
        \begin{tikzpicture}[scale=0.8]
            \Tree [.$\btx$ [.$\btx$ [.$\bt$ ] [.$\bt$ ] ]
                      [.$\bt$ ] ]      
                    ]
            \end{tikzpicture}
        \begin{tikzpicture}[scale=0.8]
            \Tree [.$\bt$ ]
        \end{tikzpicture},~~
        \begin{tikzpicture}[scale=0.8]
            \Tree [.$\btx$ [.$\bt$ ]
                      [.$\btx$ [.$\bt$ ] [.$\bt$ ] ] ]       
        \end{tikzpicture}
        \begin{tikzpicture}[scale=0.8]
            \Tree [.$\bt$ ]
        \end{tikzpicture},~~
        \begin{tikzpicture}[scale=0.8]
            \Tree [.$\btx$ [.$\bt$ ] [.$\bt$ ] ]
        \end{tikzpicture}
        \begin{tikzpicture}[scale=0.8]
            \Tree [.$\btx$ [.$\bt$ ] [.$\bt$ ] ]
        \end{tikzpicture},~~
        \begin{tikzpicture}[scale=0.8]
            \Tree [.$\bt$ ]
        \end{tikzpicture}
        \begin{tikzpicture}[scale=0.8]
            \Tree [.$\btx$ [.$\btx$ [.$\bt$ ] [.$\bt$ ] ]
                      [.$\bt$ ] ]      
                     ]
        \end{tikzpicture},~~
        \begin{tikzpicture}[scale=0.8]
            \Tree [.$\bt$ ]
        \end{tikzpicture}
        \begin{tikzpicture}[scale=0.8]
            \Tree [.$\btx$ [.$\bt$ ]
                      [.$\btx$ [.$\bt$ ] [.$\bt$ ] ] ]       
        \end{tikzpicture}     
    }\right\}.
\end{equation}
\endgroup
Note that the excited subsystems from the fusion path are mapped to internal nodes of the trees, whereas unexcited subsystems become leaves. Also note that trivial trees with zero internal nodes and a single leaf are also included in a forest. For a given binary forest $\v{T}$, let us then denote by $\Phi(\v{T})$ the set of all fusion paths corresponding to $\v{T}$. For the considered example, we have
\begingroup
\tikzset{
    level distance=32pt,
    grow'=up,
    every tree node/.style = {align=center, anchor=north,nodecolor},
    edge from parent/.style = {draw,edgecolor,thick}
}
\begin{equation}
    \Phi\left(
        \raisebox{-1.4em}{
            \begin{tikzpicture}[scale=0.8]
                  \Tree [.$\btx$ [.$\bt$ ] [.$\bt$ ] ] 
            \end{tikzpicture}
            \begin{tikzpicture}[scale=0.8]
                  \Tree [.$\btx$ [.$\bt$ ] [.$\bt$ ] ] 
            \end{tikzpicture}
        }
        \right) 
        = 
        \left\{
            \raisebox{-2.8em}{
            \begin{tikzpicture}[scale=0.8]
                \Tree  [.$\btx$ [.$\bt$ $\bt$ ] [.$\bt$  $\bt$ ] ]      
            \end{tikzpicture}
            \begin{tikzpicture}[scale=0.8]
                \Tree  [.$\btx$ [.$\btx$ [.$\bt$ ] [.$\bt$ ] ] ]        
            \end{tikzpicture},~~
            \begin{tikzpicture}[scale=0.8]
                \Tree [.$\btx$ [.$\btx$ [.$\bt$ ] [.$\bt$ ] ] ]       
            \end{tikzpicture}
            \begin{tikzpicture}[scale=0.8]
                \Tree [.$\btx$ [.$\bt$ $\bt$ ] [.$\bt$ $\bt$ ] ]           
            \end{tikzpicture}
            }
        \right\}
\end{equation}
\endgroup
and the remaining four sets contain each a single fusion path. Using~$\Phi$, we can rewrite Eq.~\eqref{eq:fusion_sum2_j} as
\begin{equation}
    \label{eq:fusion_sum3_j}
    S^{j}_k = \sum_{\v{T}\in \T^{j}_{k+1}} \sum_{\v{f}\in \Phi(\v{T})} \left(\prod_{i=j}^k \Xi_i(\v{f})\right)^{-1}.
\end{equation}

So far we have discussed how, given $\v{f}\in\F^{j}_{k+1}$, to obtain $\v{T}\in\T^{j}_{k+1}$ such that $\v{f}\in\Phi(\v{T})$. Now, we will explain the reverse process, i.e., how, given a forest $\v{T}$, find all fusion paths belonging to $\Phi(\v{T})$. One starts at the excited roots, and at each time step $i$ grows one of the internal nodes of one of the trees into two of its children, while trivially evolving all other nodes (i.e., just pushing them to the next time step, without changing). Crucially, at every time step $i$, the number of different ways to grow a forest (i.e., a number of ways to create distinct fusion paths) is equal to the number of internal nodes, so the number of excited states $\Xi_i(\v{f})$ at this time step. Therefore, when summing over all fusion paths corresponding to a given forest, the number of new ways to construct fusions at each step $i$ exactly cancels out the weight $\Xi_i(\v{f})^{-1}$, and so:
\begin{equation}
    \sum_{\v{f}\in \Phi(\v{T})} \left(\prod_{i=j}^k \Xi_i(\v{f})\right)^{-1} = 1.
    \label{eq:identity_fusion_paths_for_forests}
\end{equation}
As a result, Eq.~\eqref{eq:fusion_sum3_j} simplifies to
\begin{equation}
    \label{eq:fusion_sum4_j}
    S_k^{j} = \sum_{\v{T}\in \T^{j}_{k+1}} 1 = C^{j}_{k-j+1},
\end{equation}
where we used the known fact concerning the number of binary forests with $j$ trees and a total of $k+1$ leaves~\cite{dimitrov2024simple}.

Note that Eq.~\eqref{eq:identity_fusion_paths_for_forests} can be rigorously proved using an approach analogous to that of Eq.~\eqref{eq:identity_fusion_paths_for_trees}, as presented in Appendix~\ref{app:weighted_fusion_for_tree}. Specifically, the identity can be reformulated in terms of topological orders on the internal nodes of the forest, i.e., any linear ordering in which each node precedes all of its descendants. This reformulation enables a proof by induction on the total number of leaves in the forest, mirroring the strategy used in Appendix~\ref{app:weighted_fusion_for_tree}. As the argument closely parallels that of Eq.~\eqref{eq:identity_fusion_paths_for_trees}, we omit the details and leave the proof to the reader.


\section{Proof of Corollary~\ref{cor:Carleman_cons_extension}}
\label{app:thm_cons_ext}

First, denote by $\eta_j(t) = z(t)^{\ot j} - y^{[j]}(t)$ the Carleman truncation error vector and by $\eta_j'(t) = x(t)^{\ot j} - \D_j y^{[j]}(t) $ the error of approximating the solution of the original problem with the Carleman solution after discarding the last entry of each of $j$ subsystems. Then, from the triangle inequality, we clearly have $\|\eta_j'(t)\|\leq\|\eta_j(t)\|$. Hence, we will focus on bounding $\|\eta_j(t)\|$ to prove the claimed result.

Now, note that $G_1$ is just $F_1$ embedded in a space of dimension $N+1$, whereas $G_2$ is $F_2$ embedded in that space plus a term that transforms only the 1-dimensional extra subspace of the two copies of the extended system into the original $N$-dimensional subspace. Thus, $G_1$ is diagonalized by $Q\oplus 1$ and has the same spectrum as $F_1$ with an extra zero corresponding to the state $\ket{N+1}$. Moreover, this zero eigenvalue corresponds to a conserved quantity as $\bra{N+1}G_2=0$, and all oscillating quantities of the original system are also the oscillating quantities of the new system. What is more,
\begin{equation}
    \|(Q\oplus 1)^{-1} z(t)\| \leq \sqrt{\gamma^2\|Q^{-1} x(t)\|^2+\gamma^2 \upsilon^2} \leq \gamma \sqrt{\|\tilde{x}_{\mathrm{max}}\|^2+\upsilon^2},
\end{equation}
where here and in what follows we add explicit dependence on the rescaling parameter $\gamma$ to improve the clarity of the derivation. Therefore, we are exactly in the scenario described by Theorem~\ref{thm:Carleman_cons} with $F_0=0$, but with $G_1$, $G_2$ and $V\oplus 1$ instead of $F_1$, $F_2/\gamma$ and $Q$, and with $\gamma\|\tilde{x}_{\mathrm{max}}\|$ replaced by $\gamma\sqrt{\|\tilde{x}_{\mathrm{max}}\|^2+ a^2}$. 

Making the above substitutions in Eq.~\eqref{eq:cons_err_final} and choosing $\gamma=p\gamma_0$ as in the statement of the theorem, we get
\begin{equation}
    \label{eq:error_cons_4}
    \|\eta_j'(t)\|\leq \|\eta_j(t)\| \leq   \frac{j}{2(k+1)}p^j \left(\frac{2e\gamma \sqrt{\|\tilde{x}_{\mathrm{max}}\|^2+ a^2} \|\tilde{G}_2\|}{\delta(G_1)} \right)^{k+1},
\end{equation}
where $\tilde{G}_2:=(Q\oplus 1)^{-1} G_2(Q\oplus 1)^{\otimes 2}$. We can now use the fact that $F_2$, and so $G_2$, can always be chosen to be symmetric under the permutation of the two systems it acts upon, which results in
\begin{equation}
    \|\tilde{G}_2\|\ = \max_{\|\xi\|=1}  \|\tilde{G}_2 \xi^{\otimes 2}\|. 
\end{equation}
Then, writing $\xi=[\zeta,b]$, where $\zeta$ is an $N$-dimensional vector and $b$ is a complex number, we have
\begin{align}
    \|\tilde{G}_2\|\ &= \max_{\|\zeta\|=\sqrt{1-|b|^2}}  \left\| \frac{\tilde{F}_2}{\gamma} \zeta^{\otimes 2} + \frac{b^2}{\gamma^2 \upsilon^2} \gamma \tilde{F}_0\right\| \leq (1-|b|^2) \frac{\|\tilde{F}_2\|}{\gamma} + \frac{|b|^2}{\gamma \upsilon^2}\|\tilde{F}_0\| \leq \frac{1}{\gamma}\max\left\{\|\tilde{F}_2\|,\frac{\|\tilde{F}_0\|}{\upsilon^2}\right\}. 
\end{align}
Combining the above with the fact that $\delta(G_1)=\delta(F_1)$, we get 
\begin{equation}
     \|\eta_j'(t)\|\leq  \frac{j}{2(k+1)}p^j \left(\frac{2e\sqrt{\|\tilde{x}_{\mathrm{max}}\|^2+\upsilon^2} \max\left\{\|\tilde{F}_2\|,\frac{\|\tilde{F}_0\|}{\upsilon^2}\right\}}{\delta(F_1)} \right)^{k+1},
\end{equation}
which is the claimed result. Moreover, since we have a freedom to choose $\upsilon$, we can make an optimal choice of $\upsilon=\sqrt{\|\tilde{F}_0\|/\|\tilde{F}_2\|}$ to minimize the upper bound, which results in
\begin{equation}
    \|\eta_j'(t)\| \leq   \frac{j}{2(k+1)}p^j \left(\frac{2e\sqrt{\|\tilde{x}_{\mathrm{max}}\|^2+\frac{\|\tilde{F}_0\|}{\|\tilde{F}_2\|}} \|\tilde{F}_2\|}{\delta(F_1)} \right)^{k+1}.
\end{equation}

\section{Proof of Theorem~\ref{thm:Carleman_poly_cons}}
\label{apps:Carleman_poly_cons}
We first show that every set of conserved quantities that are either linear or quadratic in $x$ can be realized as linearly conserved quantities in a larger dimensional system with variable $\check{x}$. The situation where we have higher-order polynomial conserved quantities follows from this either by a simple generalization of the scheme, or (less efficiently) by iteratively applying the quadratic embedding to reduce the order of the terms. For this, we recall the terms $A_{j,j+1}$ and $A_{j,j}$ as defined in Section~\ref{sec:Core-Carleman}. 

Given that $\dot{x} = F_1 x + F_2 x^{\otimes 2}$, it follows that
\begin{equation}\label{eq:2nd-level-Carleman}
    \frac{d}{dt} \begin{bmatrix} x(t) \\ x(t)^{\otimes 2}\end{bmatrix} = \begin{bmatrix}
        A_{1,1} & A_{1,2} \\ 0 & A_{2,2}
    \end{bmatrix} \begin{bmatrix} x(t) \\ x(t)^{\otimes 2}\end{bmatrix} + \begin{bmatrix} 0 \\  A_{2,3} x^{\otimes 3} \end{bmatrix}.
\end{equation}
We consider the direct sum space $\mathcal{V}=\mathbb{C}^N \oplus \mathbb{C}^{N^2}$, and for any $x \in \mathbb{C}^N$ we can embed the vector into this space as $x \rightarrow \check{x} \in \mathcal{V}$, where
\begin{equation}
    \check{x} = x \oplus x^{\otimes 2} = \ket{1,x } + \ket{2,x^{\otimes 2}} = \sum_{i=1}^N x_i \ket{1,i} + \sum_{i,j=1}^N x_i x_j \ket{2,ij}.
\end{equation} 
Here, we label the two sectors by $1$ and $2$ inside the kets, since the full space is not a tensor product space. Next, observe that $\check{x}^{\otimes 2}$ can be written as
\begin{equation}
    \check{x}^{\otimes 2} = \ket{1,x }\otimes \ket{1,x } + \ket{1,x }\otimes \ket{2,x^{\otimes 2}} + \ket{2,x^{\otimes 2}}\otimes \ket{1,x} + \ket{2,x^{\otimes 2}}\otimes \ket{2,x^{\otimes 2}}.
\end{equation}
Note that the subspace $\mbox{span} \{\ket{1,x }\otimes \ket{2,x^{\otimes 2}} : x \in \mathbb{C}^N \} $ is isomorphic to $\mbox{span} \{x^{\otimes 3} : x \in \mathbb{C}^N\}$. Thus, letting $\mathcal{V}_3:=\mbox{span} \{\ket{1,i_1 }\otimes \ket{2,i_2i_3} : i_1,i_2,i_3 \in [N] \}$, we may define a linear operator $\check{F}_2$ on the full space $\mathcal{V}^{\otimes 2}$, so that it is zero on $\mathcal{V}_3^\perp$ and its action on $\mathcal{V}_3$ gives
\begin{equation}
    \check{F}_2  \check{x}^{\otimes 2} = \ket{2, A_{2,3} x^{\otimes 3}} \mbox{ for all } x \in \mathbb{C}^N.
\end{equation}
We also consider the initial conditions $\check{x}(0) = \ket{1,x(0)} + \ket{2,(x(0))^{\otimes 2}}$.
Given the above, we see that Eq.~\eqref{eq:2nd-level-Carleman} can be rewritten as
\begin{equation}\label{eq:lifted-dynamics}
    \frac{d}{dt} \check{x}(t) = \check{F}_1 \check{x} + \check{F}_2 \check{x}^{\otimes 2}, \mbox{ for all } t \ge 0,
\end{equation}
with initial conditions $\check{x}(0)$, $\check{F}_2$ defined as above, and $\check{F}_1$ defined as the linear action on $\mathcal{V}$ arising from the first matrix term in Eq.~\eqref{eq:2nd-level-Carleman}, which is given explicitly as
\begin{equation}
    \check{F}_1 \check{x} = \ket{1,A_{1,1}x} + \ket{1,A_{1,2}x^{\otimes 2}} + \ket{2,A_{2,2}x^{\otimes 2}}.
\end{equation}

Within this setting, we consider both linear quantities
\begin{equation}
    \mathcal{Q}^{(m)}_1 := \langle 1,q_1^{(m)}\ket{1,x}, \quad \ket{1, q_1^{(m)}} = \sum_{j=1}^Nq_{1,j}^{(m)} \ket{1,j},
\end{equation}
with $m=1,2,\dots, M_1$, and also \emph{quadratic-order} homogeneous quantities,
\begin{equation}
    \mathcal{Q}^{(m)}_2 := \langle 2,q_2^{(m)}| 2,x^{\otimes 2} \rangle, \quad \ket{2,q_2^{(m)}} = \sum_{i,j=1}^{N}q_{2,ij}^{(m)} \ket{2,ij},
\end{equation}
with $m=1,2,\dots, M_2$. Next, define
\begin{equation}
    \check{\mathcal{Q}}^{(m)}_1 := \left\{ 
    \begin{array}{ll}
        \langle 1,q_1^{(m)}| 1,x^{\otimes 2} \rangle,\quad& \mbox{ for } m=1,2,\dots M_1\\
        \langle 2,q_2^{(m-M_1)}| 2,x^{\otimes 2} \rangle,\quad&\mbox{ for } m= M_1+1, \dots M_1+M_2.
    \end{array}
    \right.
\end{equation}
Therefore, in the embedding space $\mathcal{V}$, the quantities $\check{\mathcal{Q}}_1^{(m)}$ correspond to $\check{M}_1:=M_1+M_2$ quantities that are \emph{linear} in the variable $\check{x}$. Given the prior theory on conserved linear observables in the main text, these quantities are conserved under the dynamics given by Eq.~\eqref{eq:lifted-dynamics} if 
\begin{aligns}
    \bra{1,q_1^{(m)}} \check{F}_1 \check{x} &=0 \mbox{ for } m=1,\dots, M_1  \\
        \bra{2, q_2^{(m)}} \check{F}_1 \check{x} &=0 \mbox{ for } m=1,\dots, M_2  \\
        \bra{1,q_1^{(m)}} \check{F}_2 \check{x}^{\otimes 2} &=0 \mbox{ for } m=1,\dots, M_1 \\
   \bra{2,q_2^{(m)}} \check{F}_2\check{x}^{\otimes 2} &=0 \mbox{ for } m=1,\dots, M_2,
\end{aligns}
for all $x \in \mathbb{C}^N$.

We now recast these conditions in the unembedded setting. The third condition holds by definition of the operator $\check{F}_2$, while the first condition holds for all $x$ if and only if $\bra{q_1^{(m)}} F_1 = 0$ and $\bra{q_1^{(m)}} F_2 = 0$. These recover the prior conditions for the linear conserved quantities. The additional constraints for the quadratic order terms are that $\bra{2,q_2^{(m)}} \check{F}_1 = 0 $ and $\bra{2,q_2^{(m)}} \check{F}_2 = 0 $. These, in turn, hold if and only if
\begin{aligns}
    \bra{q_2^{(m)}} A_{2,2}=\sum_{i,j=1}^N q_{2,ij}^{(m)*} \bra{ij} (F_1\otimes I + I\otimes F_1) &= 0, \\
      \bra{q_2^{(m)}} A_{2,3}=  \sum_{i,j=1}^N q_{2,ij}^{(m)*} \bra{ij} (F_2 \otimes I + I\otimes F_2) &= 0.
\end{aligns}
Given these conditions hold, we have that linear and quadratic homogeneous conserved quantities in the original ODE can be recast as purely linear conserved quantities in the embedded system. Therefore, the conditions for Carleman convergence can be obtained in this embedding system and convergence occurs if $\check{R}_\delta <1$.

To see that these conditions indeed hold, let us consider more generally and compactly $\ket{r, q_r^{(m)}}$ as an order-$r$ monomial conserved quantity for $m=1,\dots, M_r$ via
\begin{equation}
    \mathcal{Q}^{(m)}_r := \bra{r, q_r^{(m)}} r,x^{\otimes r} \rangle,
\end{equation}
and we have a range of such conserved quantities with orders $r=1,2,\dots ,r_{\max}$. Then, conservation under the autonomous, quadratic ODE requires that
\begin{equation}
    \bra{q_r^{(m)}} A_{r,r} = 0 \quad\mbox{ and }\quad\bra{q_r^{(m)}} A_{r,r+1} = 0 ,
\end{equation}
which in turn implies that the conserved quantities are in the left-nullspace of the Carleman matrix truncated to $k \ge r_{\max}$. Finally, if we have a polynomial conserved quantities $\mathcal{Q}^{(m)}_r$ of order $r$ then similarly $\mathcal{Q}^{(m)}_r = \bra{q_r^{(m)}} \check{x}$, where now the vector $\ket{q_r^{(m)}}$ has terms on each of the sectors of the direct sum from $j=1$ copies to $j=r$ copies, and a similar construction follows. More precisely, let the embedding space be the direct sum $\mathcal{V} = \oplus_{j=1}^r \mbox{span}\{\ket{j, i_1,\dots i_j}: i_1,i_2,\dots ,i_j \in [N]\}$, with notation as above for the quadratic case. We now have
\begin{equation}
    \mathcal{Q}^{(m)}_r = \sum_{j=1}^r\bra{j,q_j^{(m)}}j,x^{\otimes j}\rangle, 
\end{equation}
where $\ket{j,q_j^{(m)}}$ define the coefficients for the order $j$ terms in the polynomial. If we have $\bra{q_j^{(m)}} A_{j,j} =\bra{q_j^{(m)}} A_{j,j+1} =0 $ for each $j$ then it is readily seen that $\dot{\Q}_r^{(m)}=0$. The embedded dynamics is then given by the vector
\begin{equation}
    \check{x} := \sum_{j=1}^r\ket{j,x^{\otimes j}}, \mbox{ and initial condition } \check{x}(0) := \sum_{j=1}^r\ket{j,x(0)^{\otimes j}},
\end{equation}
with the associated linear observable $\check{\Q}_1^{(m)} = \bra{\check{q}^{(m)}} \check{x}$, where $\ket{\check{q}^{(m)}} := \sum_{j=1}^r\ket{j,q_j^{(m)}}$. The dynamics for $\check{x}(t)$ involves terms up to $\ket{r,x^{\otimes r}}$, defining an operator $\hat{F}_1$ on $\mathcal{V}$, plus a term involving $x^{\otimes r+1}$. By considering $\check{x}^{\otimes 2}$ we can define an operator $\check{F}_2$ acting on $\mathcal{V}^{\otimes 2}$ with support on the $\ket{1,x}\otimes \ket{r,x^{\otimes r}}$ sector, which realizes the order $r+1$ in the dynamics, exactly as happens for the quadratic case. Therefore, we realize the conserved polynomial observable as a linear-order term in $\check{x}$, with $\check{x}$ evolving under a quadratic ODE, as required. Hence the convergence criterion for systems with linear conserved observables can be applied. The dimension of $\mathcal{V}$ is seen to be $O(N^r)$.

\section{Proof of Proposition~\ref{prop:valid_vij}}
\label{app:proof_lemma_valid_vij}

Fix any $i<j$. To prove the claim, we need to show that 
 for arbitrary $a_1,a_2,\dots, a_i$, $b_1,b_2,\dots,b_j \in [N]$, 
\begin{align}
\lrb{\sum_{l=1}^j \lambda_{b_l}-
\sum_{l=1}^i \lambda_{a_l}}\bra{a_1,a_2,\dots, a_i}  \tilde{V}_{i,j}\ket{b_1,b_2,\dots,b_j} = \bra{a_1,a_2,\dots, a_i} \tilde{A}_{i,i+1} \tilde{V}_{i+1,j} \ket{b_1,b_2,\dots,b_j}.     
\label{eq:vij_recursive_relation}
\end{align}

For convenience, we introduce the following notation. Given any binary forest $\mbf=(\mbt_1,\mbt_2,\dots,\mbt_i)$ $\in \mathcal{T}^i_j$, we label the leaves of $\mbf$ sequentially from $1$ to $j$ and use $\mathcal{I}_l$ to denote the index set for the leaves in $\mathbf{t}_l$, for each $l \in [i]$. Moreover, let $\mathcal{A}(\mbf)$ and $\mathcal{B}(\mbf)=[i] \setminus \mathcal{A}(\mbf)$ contain the indices of trivial and nontrivial trees in $\mbf$, respectively. That is, for any $l \in [i]$, $\mathbf{t}_l$ is trivial if $l \in \mathcal{A}(\mbf)$, and nontrivial otherwise. For each \( \beta \in \mathcal{B}(\mbf) \), let \( \tilde{\mbf}_\beta \) denote the binary forest obtained from $\mbf$ by removing the root of $\mathbf{t}_{\beta}$ and splitting it into two binary trees. Clearly,  $\tilde{\mbf}_\beta \in \mathcal{T}^{i+1}_{j}$.  

In addition, for any binary forest \( \mbf = (\mathbf{t}_1, \mathbf{t}_2, \dots, \mathbf{t}_i, \mathbf{t}_{i+1}) \in \mathcal{T}^{i+1}_{j} \), let $\hat{\mbf}_l$ denote the binary forest obtained from $\mbf$ by combining $\mathbf{t}_l$ and $\mathbf{t}_{l+1}$ into a single binary tree, for each $l \in [i]$. Clearly,  $\hat{\mbf}_l \in \mathcal{T}^i_j$, for each $l \in [i]$.  

For any $\mbf \in \mathcal{T}^i_j$, define
\begin{align}
    \Omega(\mbf) \defeq \lrcb{(l, \mbf') \in [i]\times \mathcal{T}^{i+1}_j: \hat{\mbf}'_l = \mbf}.
\end{align}
One can easily verify that
\begin{align}
    \Omega(\mbf) = \lrcb{(\beta, \tilde{\mbf}_\beta): \beta \in \mathcal{B}(\mbf)}.
    \label{eq:omegaf_def}
\end{align}
Moreover, $\lrcb{\Omega(\mbf):~\mbf \in \mathcal{T}^i_j}$ forms a partition of $[i]\times \mathcal{T}^{i+1}_j$, i.e.
\begin{align}
[i]\times \mathcal{T}^{i+1}_{j} =  \bigsqcup_{\mbf \in \mathcal{T}^i_j} \Omega(\mbf).
\label{eq:omegaf_partition}
\end{align}

Finally, let $\tilde{F}_2^{(i,l)}=I^{\otimes l} \otimes \tilde{F}_2 \otimes I^{\otimes i-l-1}$ for any $l<i$. Then we have 
\begin{align}
\tilde{A}_{i,i+1}=\sum_{l=0}^{i-1} \tilde{F}_2^{(i,l)}.    
\end{align}

Now we are ready to prove Eq.~\eqref{eq:vij_recursive_relation}. Observe that \begin{align}
    \tilde{V}_{i,j}=\sum_{\mbf \in \mathcal{T}^i_j} \tilde{f}(\mbf),
    \label{eq:vij_ff}
\end{align}
and
\begin{align}
    \tilde{A}_{i,i+1}\tilde{V}_{i+1,j}=\sum_{l=1}^{i}\sum_{\mbf' \in \mathcal{T}^{i+1}_j} \tilde{F}_2^{(i,l-1)}\tilde{f}(\mbf').
    \label{eq:aiivij_ff}
\end{align}
Then, by Eq.~\eqref{eq:omegaf_partition}, it suffices to prove that 
for every $\mbf = (\mathbf{t}_1, \mathbf{t}_2, \dots, \mathbf{t}_i)
\in \mathcal{T}^i_j$, the following holds:
\begin{align}
\hspace{-15pt} \lrb{\sum_{l=1}^j \lambda_{b_l} - \sum_{l=1}^i \lambda_{a_l}}\bra{a_1,a_2,\dots, a_i}  \tilde{f}(\mbf)\ket{b_1,b_2,\dots,b_j} = \bra{a_1,a_2,\dots, a_i} \sum_{(l, \mbf') \in \Omega(\mbf)} \tilde{F}_2^{(i,l-1)} \tilde{f}(\mbf')\ket{b_1,b_2,\dots,b_j}.   
    \label{eq:ffff}
\end{align}

Let $\mathcal{A}=\mathcal{A}(\mbf)$ and $\mathcal{B}=\mathcal{B}(\mbf)$. For $l \in \mathcal{A}$, let $s_l$ denote the only element of $\mathcal{I}_l$. Since 
\begin{align}
\tilde{f}(\mbf) = \tilde{f}(\mbt_1) \otimes \tilde{f}(\mbt_2) \otimes \dots \otimes \tilde{f}(\mbt_i),   
\end{align}
where $\tilde{f}(\mbt_l)=I$ for $l \in \mathcal{A}$, it follows that
$\bra{a_1,a_2,\dots, a_i}  V_{i,j}\ket{b_1,b_2,\dots,b_j} \neq 0$ only if $a_l=b_{s_l}$ for all $l \in \mathcal{A}$. Meanwhile, by the definition of $\Omega(\mbf)$, we know that for every $(l, \mbf') \in \Omega(\mbf)$, $\tilde{F}_2^{(i,l-1)}$ acts trivially on the subsystem corresponding to $\mathbf{t}_l$ for each $l \in \mathcal{A}$. Thus,  we have that
$\bra{a_1,a_2,\dots, a_i}  \tilde{F}_2^{(i,l-1)} \tilde{f}(\mbf') \ket{b_1,b_2,\dots,b_j} \neq 0$ only if $a_l=b_{s_l}$ for all $l \in \mathcal{A}$.

Therefore, we only need to consider the case where $a_l=b_{s_l}$ for all $l \in \mathcal{A}$ from now on. Note that this condition implies 
\begin{align}
\sum_{l=1}^j \lambda_{b_l} - \sum_{l=1}^i \lambda_{a_l} 
=\sum_{\beta \in \mathcal{B}} \lrb{\sum_{l \in \mathcal{I}_\beta} \lambda_{b_l}-\lambda_{a_\beta}}.
\label{eq:eigenval_sum_diff}
\end{align}
Meanwhile, by the construction of $\tilde{f}(\mbf)$, we have
\begin{align}
\bra{a_1,a_2,\dots, a_i} \tilde{f}(\mbf) \ket{b_1,b_2,\dots,b_j}
= \frac{\bra{a_1,a_2,\dots, a_i} \tilde{F}_2^{(i, \beta-1)} \tilde{f}(\tilde{\mbf}_{\beta}) \ket{b_1,b_2,\dots,b_j}}{\sum_{l \in \mathcal{I}_\beta} \lambda_{b_l} - \lambda_{a_\beta}}.
\end{align}
It follows that
\begin{align}
\hspace{-15pt} 
\sum_{\beta \in \mathcal{B}} \lrb{\sum_{l \in \mathcal{I}_\beta} \lambda_{b_l} - \lambda_{a_\beta} }
\bra{a_1,a_2,\dots, a_i} \tilde{f}(\mbf) \ket{b_1,b_2,\dots,b_j}=\sum_{\beta \in \mathcal{B}}
    \bra{a_1,a_2,\dots, a_i} \tilde{F}_2^{(i, \beta-1)} f(\tilde{\mbf}_{\beta}) \ket{b_1,b_2,\dots,b_j}.
\end{align}
This is equivalent to Eq.~\eqref{eq:ffff}, as implied by Eqs.~\eqref{eq:omegaf_def} and \eqref{eq:eigenval_sum_diff}. The proposition is thus proved.

\section{Proof of Proposition~\ref{prop:v_inv_ij_expression}}
\label{app:proof_lemma_v_inv_ij_expression}

We restate Proposition~\ref{prop:v_inv_ij_expression} as follows. For any $a \le b$, 
\begin{align}
(\tilde{V}^{-1})_{a,b} &= (-1)^{b-a} \sum_{\mbf \in \mathcal{T}^a_b} \tilde{g}(\mbf)\\
&=(-1)^{b-a} \sum_{(\mbt_1,\mbt_2,\dots,\mbt_a) \in \mathcal{T}^a_b} \tilde{g}(\mbt_1) \otimes \tilde{g}(\mbt_2) \otimes \dots \otimes \tilde{g}(\mbt_a).    
\label{eq:v_inv_ij_expression}
\end{align}
Here we have replaced $i$, $j$ with $a$, $b$, respectively, as we will frequently use $i$, $j$ as summation indices (or loop indices) in the proof.

If $a=b$, then both sides of Eq.~\eqref{eq:v_inv_ij_expression} reduce to $I^{\otimes a}$, verifying the equation. Therefore, we assume $b > a \ge 1$ from now on.

Before proceeding, it is necessary to introduce the following notation about binary forests (with binary trees as a special case). 

Given two binary forests $\mbf_1$ and $\mbf_2$, we say that $\mbf_2$ is \emph{concatenable} with $\mbf_1$ if the number of trees in $\mbf_2$ matches the number of leaves in $\mbf_1$. Then suppose $\mbf_1$ is a binary forest with $p$ trees and $q$ leaves, and $\mbf_2$ is a binary forest with $q$ trees and $r$ leaves, where $p \le q \le r$ are arbitrary. Then we define their \emph{concatenation}, denoted $\mbf_1 \circ \mbf_2$, as the binary forest obtained by ``gluing'' the root of the $i$-th tree in $\mbf_2$ to the $i$-th leaf of $\mbf_1$, for $i=1,2,\dots,q$. Furthermore, this operation can be extended to a sequence of binary forests, $\mbf_1, \mbf_2, \dots, \mbf_m$, provided that the number of leaves in $\mbf_i$ equals the number of roots in $\mbf_{i+1}$ for $i=1,2,\dots,m-1$. The resulting binary forest is denoted as $$\mbf_1 \circ \mbf_2 \circ \dots \circ \mbf_m.$$ Note that the concatenation operation $\circ$ is associative, i.e. 
\begin{align}
(\mbf_1 \circ \mbf_2) \circ \mbf_3 = \mbf_1 \circ (\mbf_2 \circ \mbf_3).    
\end{align}
So the expression $\mbf_1 \circ \mbf_2 \circ \dots \circ \mbf_m$ is well-defined.

Given a binary forest $\mbf=(V, E)$ and an integer $p \ge 1$, a decomposition of $\mbf$ into $p$ layers is a tuple $(V_1, V_2, \dots, V_p) \in P(V)^p$ such that
\begin{itemize}
    \item $\cup_{i=1}^p V_i = V$;
    \item Let $V_I$ be the set of internal nodes of $\mbf$, and for each $v \in V_I$, let $C(v)=\lrcb{c_1(v), c_2(v)}$ be the set of its children. Then for each $i \in [p]$ and each $v \in V_i \cap V_I$, either $C(v)\subset V_i$ or
    $C(v) \cap V_i = \emptyset$. Furthermore, there exists at least one $v \in V_i \cap V_I$ such that $C(v) \subset V_i$. This condition ensures that $\mbf[V_i]$ is a valid binary forest that contains at least one nontrivial binary tree.
    \item For each $i \in [p]$, let $V_{R, i}$ denote the set of roots of the trees in $\mbf[V_i]$, and let $V_{L, i}$ denote the set of their leaves. 
    Then we have
    \begin{align}
        V_i \cap V_{i+1}=V_{L,i}=V_{R, i+1},\quad\forall~i \in [p-1].        
    \end{align}
    This condition ensures that $\mbf[V_{i+1}]$ is concatenable with $\mbf[V_i]$ for $i \in [p-1]$.

\end{itemize}
In other words, $\mbf$ is decomposed into $p$ layers
$\mbf[V_1], \mbf[V_2], \dots, \mbf[V_p]$, each of which is a binary forest containing at least one nontrivial binary tree, such that
\begin{align}
    \mbf = \mbf[V_1] \circ \mbf[V_2] \circ \dots \circ \mbf[V_p].
\end{align}
We will use $\mathcal{L}_p(\mbf)$ to denote the set of $p$-layer decompositions of $\mbf$.

Conversely, if $\mbf=\mbf_1\circ \mbf_2 \circ \dots \circ \mbf_p$, where each $\mbf_i$ contains at least one nontrivial tree, then $\mbf$ can be decomposed into $p$ layers such that the $i$-th layer corresponds to $\mbf_i$ for all $i \in [p]$.

A labeling of a binary forest $\mbf=(V, E)$ is a map $h: V \to [N]$, where $N$ is the dimension of the system of interest. Given a binary forest $\mbf$ and a labeling $h$ of its nodes, we refer to the pair $(\mbf, h)$ a \emph{labeled binary forest}. Suppose $(\mbf_1, h_1)$ and $(\mbf_2, h_2)$ are two labeled forests such that $\mbf_2$ is concatenable with $\mbf_1$. Let $l_1,l_2,\dots,l_m$ be the leaves of the trees in $\mbf_1$, and let 
$r_1,r_2,\dots,r_m$ be the roots of the trees in $\mbf_2$.  
We say that $(\mbf_2, h_2)$ is \emph{compatible} with $(\mbf_1, h_1)$ if $h_1(l_j)=h_2(r_j)$ for all $j \in [m]$. In this case, we can concatenate the two labeled forests and denote the resulting labled forest as 
\begin{align}
(\mbf_1, h_1) \circ (\mbf_2, h_2) = (\mbf_1 \circ \mbf_2, h_1 \circ h_2).    
\end{align}
Furthermore, given any sequence of labeled forests $(\mbf_1, h_1), (\mbf_2, h_2)$, $\dots$, $(\mbf_m, h_m)$, where $(\mbf_{i+1}, h_{i+1})$ is compatible with $(\mbf_i, h_i)$ for all $i \in [m-1]$, we can concatenate them and obtain the labeled forest
$$(\mbf_1, h_1) \circ (\mbf_2, h_2) \circ \dots \circ (\mbf_m, h_m)
=(\mbf_1 \circ \mbf_2  \circ \dots \circ \mbf_m, 
h_1 \circ h_2 \circ \dots \circ h_m).$$

Given a binary forest $\mbf=(V, E)$, a labeling $h$ of its nodes, and a $p$-layer decomposition $(V_1, V_2, \dots, V_p)$ of $\mbf$, let us consider the $p$ labeled binary forests $(\mbf[V_1], h[V_1])$, $(\mbf[V_2], h[V_2])$, $\dots$, $(\mbf[V_p], h[V_p])$, where $h[V_i]$ denotes the restriction of $h$ to $V_i$. They satisfy
\begin{align}
    (\mbf, h) = (\mbf[V_1], h[V_1]) \circ (\mbf[V_2], h[V_2]) \circ \dots \circ (\mbf[V_p], h[V_p]).
\end{align}

Conversely, if $(\mbf, h) = (\mbf_1, h_1) \circ (\mbf_2, h_2) \circ \dots \circ (\mbf_m, h_m)$, where each $\mbf_i$ contains at least one nontrivial tree, then there exists a $p$-layer decomposition of $\mbf$ such that the $i$-th layer of $(\mbf, h)$ corresponds to $(\mbf_i, h_i)$ for all $i \in [p]$.

The functions $\alpha(\mbt, h)$, $\beta(\mbt, h)$ and $\gamma(\mbt, h)$ have been defined for any labeled binary tree $(\mbt, h)$ in  Eqs.~\eqref{eq:alpha_t_h_def}, \eqref{eq:beta_t_h_def} and \eqref{eq:gamma_t_h_def}, respectively. Now we extend their definition to labeled binary forests. Let $\mbf=(\mbt_1,\mbt_2,\dots,\mbt_m)$ be a binary forest and let $h$ be a labeling of its nodes. We define
\begin{align}
    \alpha(\mbf, h) \defeq \alpha(\mbt_1, h)\alpha(\mbt_2, h) \dots 
    \alpha(\mbt_m, h),
\end{align}
\begin{align}
    \beta(\mbf, h) \defeq \beta(\mbt_1, h)\beta(\mbt_2, h) \dots 
    \beta(\mbt_m, h),
\end{align}
and
\begin{align}
    \gamma(\mbf, h) \defeq \gamma(\mbt_1, h)\gamma(\mbt_2, h) \dots 
    \gamma(\mbt_m, h).
\end{align}
Let $L(\mbf)$ denote the set of all possible labelings of $\mbf$. A labeling $h$ of $\mbf$ is feasible if and only if $\alpha(\mbf, h) \neq 0$. The number of feasible labelings of $\mbf$ could be much smaller than $|L(\mbf)|=N^{|V(\mbf)|}$ when $\tilde{F}_2$ is sparse.

Finally, for a binary forest $\mbf=(\mbt_1,\mbt_2,\dots,\mbt_m)$, the operator $\tilde{f}(\mbf)$ assigned to it satisfies
\begin{align}
    \tilde{f}(\mbf)&=\tilde{f}(\mbt_1) \otimes \tilde{f}(\mbt_2) \otimes \dots \otimes \tilde{f}(\mbt_m) \\
    &=\sum_{h \in L(\mbf)} \alpha(\mbf, h) \beta(\mbf, h) \ket{h(r_1),h(r_2),\dots,h(r_m))}\bra{h(l_1), h(l_2),\dots, h(l_k)},
\end{align}
where $r_1,r_2,\dots,r_m$ are the roots of the trees in $\mbf$ and $l_1,l_2,\dots,l_k$ are the leaves of these trees.

Now recall that for any $i \le j$, we have 
\begin{align}
\tilde{V}_{i,j}=\sum_{\mbf \in \mathcal{T}^i_j} \tilde{f}(\mbf),
\end{align}
where $\mathcal{T}^i_j$ denotes the set of binary forests consisting of $i$ binary trees with a total of $j$ leaves. For any $j_2, j_3, \dots, j_p$ satisfying $j_1\defeq a<j_2<j_3<\dots<j_p<j_{p+1}\defeq b$, we have
\begin{align}
\tilde{V}_{a,j_2} \tilde{V}_{j_2, j_3} \dots \tilde{V}_{j_p, b}
&=\sum_{\mbf_1 \in \mathcal{T}^{a}_{j_2}}
\sum_{\mbf_2 \in \mathcal{T}^{j_2}_{j_3}}
\dots
\sum_{\mbf_p \in \mathcal{T}^{j_p}_b}
\tilde{f}(\mbf_1)\tilde{f}(\mbf_2)\dots \tilde{f}(\mbf_p) \\
&=\sum_{\mbf_1 \in \mathcal{T}^{a}_{j_2}}
\sum_{\mbf_2 \in \mathcal{T}^{j_2}_{j_3}}
    \dots
\sum_{\mbf_p \in \mathcal{T}^{j_p}_{b}}
\sum_{h_1 \in L(\mbf_1)}
\sum_{h_2 \in L(\mbf_2)}
\dots
\sum_{h_p \in L(\mbf_p)} \nonumber \\
&\quad \alpha(\mbf_1, h_1)
\beta(\mbf_1, h_1)
\alpha(\mbf_2, h_2)
\beta(\mbf_2, h_2)
\dots 
\alpha(\mbf_p, h_p)
\beta(\mbf_p, h_p) \nonumber \\
&\quad \ket{h_1(r_{1,1}), h_1(r_{1,2}), \dots, h_1(r_{1,j})}
\bra{h_1(l_{1,1}), h_1(l_{1,2}), \dots, h_1(l_{1, j_2}))}
\cdot \nonumber \\
&\quad \ket{h_2(r_{2,1}), h_2(r_{2,2}), \dots, h_2(r_{2, j_2})}
\bra{h_2(l_{2,1}), h_2(l_{2,2}), \dots, h_2(l_{2, j_3}))}
\cdot \nonumber \\
&\quad \dots  \nonumber \\
& \quad 
\cdot\ket{h_p(r_{p,1}), h_p(r_{p,2}), \dots, h_p(r_{p,j_p})}
\bra{h_p(l_{p,1}), h_p(l_{p,2}), \dots, h_p(l_{p,k})},
\label{eq:product_vijs_expansion}
\end{align}
where $r_{i,j}$ is the root of the $j$-th binary trees in $\mbf_i$,
for $i \in [p]$ and $j \in [j_i]$, 
and $l_{i,j}$ is the $j$-th leaf of $\mbf_i$, 
for $i \in [p]$ and $j \in [j_{i+1}]$. Note that the summand on the right-hand side of Eq.~\eqref{eq:product_vijs_expansion} is nonzero only if 
\begin{align}
h_i(l_{i,j})=h_{i+1}(r_{i+1, j})    
\end{align}
for all $i \in [p-1]$ and $j \in [j_{i+1}]$. This is equivalent to the condition that $(\mbf_{i+1}, h_{i+1})$ is compatible with $(\mbf_i, h_i)$ for all $i \in [p-1]$. Under this condition, we can concatenate these labeled forests and obtain 
\begin{align}
(\mbf, h) = (\mbf_1, h_1) \circ (\mbf_2, h_2) \circ \dots \circ (\mbf_p, h_p),    
\end{align}
where $\mbf$ is a binary forest consisting of $a$ binary trees with a total of $b$ leaves, and $h$ is a labeling of $\mbf$. Therefore, Eq.~\eqref{eq:product_vijs_expansion} implies that
\begin{align}
\sum_{a<j_2<j_3<\dots<j_p<b}  \tilde{V}_{a,j_2} \tilde{V}_{j_2, j_3} \dots \tilde{V}_{j_p, b}
    &= \sum_{\mbf \in \mathcal{S}^{a}_{b}}
    \sum_{h \in L(\mbf)}
    \sum_{\Omega \in \mathcal{L}_{p}(\mbf)}  \bar{\alpha}(\mbf, \Omega, h)
\bar{\beta}(\mbf, \Omega, h) \nonumber \\
&\quad  \ket{h(r_{1}), h_1(r_{2}), \dots, h_1(r_{a})}
\bra{h(l_{1}), h_1(l_{2}), \dots, h_1(l_{b}))},
\label{eq:product_vijs_expansion2}
\end{align}
where for any 
$\Omega=(U_1, U_2, \dots, U_p) \in \mathcal{L}_p(\mbf)$, we define
\begin{align}
\bar{\alpha}(\mbf, \Omega, h) &\defeq \alpha(\mbf[U_1], h[U_1])
\alpha(\mbf[U_2], h[U_2])
\dots 
\alpha(\mbf[U_p], h[U_p]), 
\end{align}
and
\begin{align}
\bar{\beta}(\mbf, \Omega, h) &\defeq \beta(\mbf[U_1], h[U_1])
\beta(\mbf[U_2], h[U_2])
\dots 
\beta(\mbf[U_p], h[U_p]).
\end{align}
Here $r_1, r_2, \dots, r_a$ are the roots of the binary trees in $\mbf$, and $l_1, l_2, \dots, l_b$ are the leaves of these trees.

It then follows from Eqs.~\eqref{eq:wjk_def} and \eqref{eq:product_vijs_expansion2} that
\begin{align}
\lrb{\tilde{V}^{-1}}_{a,b} &= 
\sum_{\mbf \in \mathcal{T}^{a}_{b}}
    \sum_{h \in L(\mbf)}
    \sum_{p=1}^{b-a} (-1)^p
    \sum_{\Omega \in \mathcal{L}_{p}(\mbf)}   \bar{\alpha}(\mbf, \Omega, h) \bar{\beta}(\mbf, \Omega, h)  \nonumber \\
&\quad \ket{h(r_{1}), h_1(r_{2}), \dots, h_1(r_a)}
\bra{h(l_{1}), h_1(l_{2}), \dots, h_1(l_b))}.
\label{eq:wjk_exp1}
\end{align}

On the other hand, note that
\begin{align}
    &(-1)^{b-a} \sum_{(\mbt_1,\mbt_2,\dots,\mbt_a) \in \mathcal{T}^a_b} \tilde{g}(\mbt_1) \otimes \tilde{g}(\mbt_2) \otimes \dots \otimes \tilde{g}(\mbt_a) \nonumber \\
    &=(-1)^{b-a} \sum_{\mbf \in \mathcal{T}^{a}_{b}}
    \sum_{h \in L(\mbf)} \alpha(\mbf, h) \gamma(\mbf, h) 
 \ket{h(r_{1}), h_1(r_{2}), \dots, h_1(r_{a})}
\bra{h(l_{1}), h_1(l_{2}), \dots, h_1(l_{b}))}.    
\label{eq:wjk_exp2}
\end{align}

Now, if we can show that for every $\mbf \in \mathcal{T}^a_b$ and $h \in L(\mbf)$, the following holds:
\begin{align}
\sum_{p=1}^{b-a} (-1)^p
    \sum_{\Omega \in \mathcal{L}_{p}(\mbf)} 
 \bar{\alpha}(\mbf, \Omega, h)
\bar{\beta}(\mbf, \Omega, h)
= (-1)^{b-a} \alpha(\mbf, h) \gamma(\mbf, h).
\label{eq:key_claim_labeled_forest}
\end{align}
then we can apply Eqs.~\eqref{eq:wjk_exp1} and \eqref{eq:wjk_exp2} to obtain 
\begin{align}
\lrb{\tilde{V}^{-1}}_{a,b} = (-1)^{b-a} \sum_{(\mbt_1,\mbt_2,\dots,\mbt_a) \in \mathcal{T}^a_b} \tilde{g}(\mbt_1) \otimes \tilde{g}(\mbt_2) \otimes \dots \otimes \tilde{g}(\mbt_a),    
\end{align}
as desired.

It remains to prove Eq.~\eqref{eq:key_claim_labeled_forest}. Fix any $\mbf \in \mathcal{T}^a_b$ and $h \in L(\mbf)$. Observe that for any $\Omega \in \mathcal{L}_{p}(\mbf)$, we have
\begin{align}
 \bar{\alpha}(\mbf, \Omega, h)
=\prod_{v \in V_I(\mbf)} \bra{h(v)}\tilde{F}_2\ket{h(c_1(v)), h(c_2(v))}
=\alpha(\mbf, h),
\end{align}
where $V_I(\mbf)$ denotes the set of internal nodes of $\mbf$. Thus, we only need to show that 
\begin{align}
 \sum_{p=1}^{b-a} (-1)^p
    \sum_{\Omega \in \mathcal{L}_{p}(\mbf)} 
\bar{\beta}(\mbf, \Omega, h)
 = (-1)^{b-a} \gamma(\mbf, h).
\label{eq:key_fact_beta_gamma}
\end{align}

This will be achieved by induction on $b-a$. We begin with the base case $b-a=1$. Here, $\mbf$ consists of a binary tree with $3$ nodes and $b-2$ isolated nodes. Suppose the root of the $3$-node tree is $r$, and the leaves of this tree are $l_1$ and $l_2$.
Then both sides of Eq.~\eqref{eq:key_fact_beta_gamma} simplify  to 
\begin{align}
\frac{1}{\lambda_{h(r)}-\lambda_{h(l_1)}-\lambda_{h(l_2)}},   
\end{align}
Thus, Eq.~\eqref{eq:key_fact_beta_gamma} holds in this case. 

Now we assume that Eq.~\eqref{eq:key_fact_beta_gamma} holds for all $\mbf' \in \mathcal{T}^{a}_{a+z}$ and $h' \in L(\mbf')$ for some $z \ge 1$.  We will prove that Eq.~\eqref{eq:key_fact_beta_gamma} also holds for all $\mbf \in \mathcal{T}^{a}_{a+z+1}$ and $h \in L(\mbf)$.

Suppose $\mbf=(\mbt_1,\mbt_2,\dots,\mbt_a)$ is a binary forest, where $\mbt_i$ is the $i$-th tree in $\mbf$ for $i \in [a]$, and the total number of leaves in these trees is $a+z+1$, and $h$ is a labeling of $\mbf$. We will consider two cases separately.
\begin{itemize}
\item Case 1: There exists a tree in $\mbf$ which contains exactly $3$ nodes. Let $\mbt_i$ be such a tree, and let $v$ be its root, and let $w_1$ and $w_2$ be its leaves. Then let $V'=V(\mbf) \setminus \lrcb{w_1, w_2}$, and let
$\mbf'=\mbf[V']$ and $h'=h[V']$. Clearly, $\mbf' \in \mathcal{T}^{a}_{a+z}$, and $h' \in L(\mbf')$. Thus, by assumption, we have
\begin{align}
\sum_{p=1}^{z} (-1)^p
\sum_{\Omega' \in \mathcal{L}_{p}(\mbf')} 
\bar{\beta}(\mbf', \Omega', h')
= (-1)^{z} \gamma(\mbf', h').
\label{eq:induction_assumption}
\end{align}

Given any $\Omega'=(U_1', U_2', \dots, U_p') \in \mathcal{L}_{p}(\mbf')$, we map it to $p$ elements of $\mathcal{L}_p(\mbf)$ and $p+1$ elements of $\mathcal{L}_{p+1}(\mbf)$. These mappings induce a partition of $\cup_{p=1}^{z+1} \mathcal{L}_{p}(\mbf)$. 

Before describing these mappings, we first note that $v$ is an isolated node in $\mbf'$, and hence $v \in U_j'$ for all $j \in [p]$.

Now, for each $c \in [p]$, we construct a $p$-layer decomposition $(U_{c,1}, U_{c,2}, \dots, U_{c,p})$ of $\mbf$ as follows:
\begin{itemize}
    \item For $j=1,2,\dots,c-1$, set $U_{c,j}=U_j'$;        
    \item Set $U_{c,c}=U_c' \cup \lrcb{w_1, w_2}$;    
    \item For $j=c+1,c+2,\dots,p$, set $U_{c,j}=(U_j' \setminus \lrcb{v}) \cup \lrcb{w_1, w_2}$.
\end{itemize}
In other words, we add the edges $(v, w_1)$ and $(v, w_2)$ in one of the $p$ layers, while replacing $v$ with $w_1$ and $w_2$ in the layers beneath it. Let $\Omega_c = (U_{c,1}, U_{c,2}, \dots, U_{c,p})$. Then one can verify that $\Omega_c  \in \mathcal{L}_p(\mbf)$, and
\begin{align}
\bar{\beta}(\mbf, \Omega_c, h)
=\frac{1}{\lambda_{h(w_1)}+\lambda_{h(w_2)}-\lambda_{h(v)}} \cdot
\bar{\beta}(\mbf', \Omega', h'),
\label{eq:induction_relation1}
\end{align}
for each $c \in [p]$. 

Meanwhile, for each $d \in [p+1]$, we construct a $(p+1)$-layer decomposition $(\tilde{U}_{d,1}, \tilde{U}_{d,2}, \dots, \tilde{U}_{d,p+1})$
of $\mbf$ as follows: 
\begin{itemize}
    \item For $j=1,2,\dots,d-1$, set $\tilde{U}_{d,j}=U'_j$;   
    \item If $d \le p$, then set $\tilde{U}_{d,d}=V_R(\mbf'[U'_d]) \cup \lrcb{w_1, w_2}$, where $V_R(\mbf'[U'_d])$ denotes the set of roots of the trees in $\mbf'[U'_d]$.
    \item If $d \le p$, then for $j=d+1,d+2,\dots,p+1$, set $\tilde{U}_{d,j}=(U_{j-1}' \setminus \lrcb{v}) \cup \lrcb{w_1, w_2}$.    
    \item If $d=p+1$, then set $\tilde{U}_{p+1, p+1}=V_L(\mbf') \cup \lrcb{w_1, w_2}$, where $V_L(\mbf')$ denotes the set of leaves of $\mbf'$.
\end{itemize}
In other words, we create a new layer containing the edges $(v, w_1)$ and $(v, w_2)$, place it between the original layers, or above the top layer, or below the bottom layer, while replacing $v$ with $w_1$ and $w_2$ in the layers beneath it. Let $\tilde{\Omega}_d = (\tilde{U}_{d,1}, \tilde{U}_{d,2}, \dots, \tilde{U}_{d,p+1})$. Then one can verify that
$\tilde{\Omega}_d\in \mathcal{L}_{p+1}(\mbf)$, and 
\begin{align}
\bar{\beta}(\mbf, \tilde{\Omega}_d, h)
=\frac{1}{\lambda_{h(w_1)}+\lambda_{h(w_2)}-\lambda_{h(v)}} \cdot
\bar{\beta}(\mbf', \Omega', h')
\label{eq:induction_relation2}
\end{align}
for each $d \in [p+1]$. 

Now let $\mathcal{X}(\Omega')\defeq\lrcb{\Omega_c:~c \in [p]}$, 
and $\mathcal{Y}(\Omega')\defeq\lrcb{\tilde{\Omega}_d:~d \in [p+1]}$. One can see that
\begin{align}
\bigsqcup_{p=1}^{z+1}\mathcal{L}_{p}(\mbf) = 
\bigsqcup_{p=1}^z
\bigsqcup_{\Omega' \in \mathcal{L}_p(\mbf')}
\lrb{\mathcal{X}(\Omega')
\bigsqcup \mathcal{Y}(\Omega')}.
\label{eq:induction_partition}
\end{align}
Thus, all decompositions of $\mbf$ can be obtained by modifying a corresponding decomposition of $\mbf'$ in the manner described above.

Now we combine Eqs.~\eqref{eq:induction_assumption}, \eqref{eq:induction_relation1},
\eqref{eq:induction_relation2} and \eqref{eq:induction_partition} and get
\begin{align}
\sum_{p=1}^{z+1} (-1)^p
\sum_{\Omega \in \mathcal{L}_{p}(\mbf)} 
\bar{\beta}(\mbf, \Omega, h)   
&=
-\frac{1}{\lambda_{h(w_1)}+\lambda_{h(w_2)}-\lambda_{h(v)}}
\sum_{p=1}^{z} (-1)^p
\sum_{\Omega' \in \mathcal{L}_{p}(\mbf')} 
\bar{\beta}(\mbf', \Omega', h')  \\
&=(-1)^{z+1}\frac{1}{\lambda_{h(w_1)}+\lambda_{h(w_2)}-\lambda_{h(v)}} \gamma(\mbf', h') \\
&=(-1)^{z+1}\gamma(\mbf, h),
\end{align}
where 
the last step follows from 
\begin{align}
\gamma(\mbf, h) = \gamma(\mbf', h') \cdot \frac{1}{\lambda_{h(w_1)}+\lambda_{h(w_2)}-\lambda_{h(v)}},
\end{align}
since the binary tree containing $v, w_1, w_2$ is isolated from the other trees in $\mbf$. Therefore, Eq.~\eqref{eq:key_fact_beta_gamma} holds for the labeled forest $(\mbf, h)$.

\item Case 2: Every nontrivial tree in $\mbf$ contains at  least $5$ nodes. Let $\mbt_i$ be such a tree. Then there exists an internal node $v$ in $\mbt_i$ such that its two children, denoted $w_1$ and $w_2$, are both leaves of $\mbt_i$, and $v$ has a parent, denoted $u$, within $\mbt_i$.

Let $V'=V(\mbf) \setminus \lrcb{w_1, w_2}$, and let
$\mbf'=\mbf[V']$ and $h'=h[V']$. Then we have $\mbf' \in \mathcal{T}^{a}_{a+z}$, and $h' \in L(\mbf')$. Thus, by assumption, we have
\begin{align}
\sum_{p=1}^{z} (-1)^p
\sum_{\Omega' \in \mathcal{L}_{p}(\mbf')} 
\bar{\beta}(\mbf', \Omega', h')
= (-1)^{z} \gamma(\mbf', h').
\label{eq:induction_assumption2}
\end{align}

In addition, we also consider another labeling $h''$ of $\mbf'$, defined as $h''(v)=N+1$, and $h''(w)=h'(w)$ for all $w \in V'\setminus\lrcb{v}$. Correspondingly, we extend the system's dimension by one and introduce a new \emph{dummy} eigenvalue $\lambda_{N+1}\defeq \lambda_{h(w_1)}+\lambda_{h(w_2)}$. Note that this eigenvalue is introduced solely for the purpose of proving Eq.~\eqref{eq:key_fact_beta_gamma} and does not affect the differential equations to be solved. Then since $\mbf' \in \mathcal{T}^{a}_{a+z}$ and $h'' \in L(\mbf')$,
by assumption, we have
\begin{align}
\sum_{p=1}^{z} (-1)^p
\sum_{\Omega' \in \mathcal{L}_{p}(\mbf')} 
\bar{\beta}(\mbf', \Omega', h'')
= (-1)^{z} \gamma(\mbf', h'').
\label{eq:induction_assumption3}
\end{align}

As in the previous case, given any $\Omega'=(U_1', U_2', \dots, U_p') \in \mathcal{L}_{p}(\mbf')$, we map it to multiple elements of $\mathcal{L}_p(\mbf)$ and multiple  elements of $\mathcal{L}_{p+1}(\mbf)$. These mappings induce a partition of $\cup_{p=1}^{z+1} \mathcal{L}_{p}(\mbf)$. 

Before describing these mappings, we first note that there exists a unique $q \in [p]$ such that $U_q'$ contains both $u$ and $v$. Then we have $v \in U_j'$ and $u \not \in U_j'$ for $j=q+1,q+2,\dots,p$.

Now, for each $c \in \lrcb{q, q+1,\dots,p}$, we construct a $p$-layer decomposition $(U_{c,1}, U_{c,2}, \dots, U_{c,p})$ of $\mbf$ as follows:
\begin{itemize}
    \item For $j=1,2,\dots,c-1$, set $U_{c,j}=U_j'$;        
    \item Set $U_{c,c}=U_c' \cup \lrcb{w_1, w_2}$;    
    \item For $j=c+1,c+2,\dots,p$, set $U_{c,j}=(U_j' \setminus \lrcb{v}) \cup \lrcb{w_1, w_2}$.
\end{itemize}
In other words, we add the edges $(v, w_1)$ and $(v, w_2)$ in one of the $p$ layers, while replacing $v$ with $w_1$ and $w_2$ in the layers beneath it. Let $\Omega_c = (U_{c,1}, U_{c,2}, \dots, U_{c,p})$. Then one can verify that $\Omega_c  \in \mathcal{L}_p(\mbf)$.

Note that there exists a key difference between $\Omega_q$ and $\Omega_{c}$ for $q+1 \le c \le p$. That is, $\mbf[U_{q,q}]$ contains a binary tree in which $v$ is an internal node with parent $u$ and children $w_1, w_2$. In contrast, $\mbf[U_{c,c}]$ contains a $3$-node tree where $v$ is the root and $w_1, w_2$ are the leaves, for $q+1 \le c \le p$. Due to this difference, we will analyze $\bar{\beta}(\mbf, \Omega_q, h)$ and $\bar{\beta}(\mbf, \Omega_c, h)$ for $q+1\le c\le p$ separately. Precisely, one can verify that
\begin{align}
\bar{\beta}(\mbf, \Omega_c, h)
=\frac{1}{\lambda_{h(w_1)}+\lambda_{h(w_2)}-\lambda_{h(v)}} \cdot
\bar{\beta}(\mbf', \Omega', h'),
\label{eq:induction_relation3}
\end{align}
for each $c \in \lrcb{q+1,q+2,\dots,p}$, and
\begin{align}
\bar{\beta}(\mbf, \Omega_q, h)
=\frac{1}{\lambda_{h(w_1)}+\lambda_{h(w_2)}-\lambda_{h(v)}} \cdot
\bar{\beta}(\mbf', \Omega', h'').
\label{eq:induction_relation4}
\end{align}

Moreover, for each $d \in \lrcb{q+1,q+2,\dots,p+1}$, we construct a $(p+1)$-layer decomposition $(\tilde{U}_{d,1}$, $\tilde{U}_{d,2}$, $\dots$, $\tilde{U}_{d,p+1})$
of $\mbf$ as follows: 
\begin{itemize}
    \item For $j=1,2,\dots,d-1$, set $\tilde{U}_{d,j}=U'_j$;   
    \item If $d \le p$, then set $\tilde{U}_{d,d}=V_R(\mbf'[U'_d]) \cup \lrcb{w_1, w_2}$, where $V_R(\mbf'[U'_d])$ denotes the set of roots of the trees in $\mbf'[U'_d]$.
    \item If $d \le p$, then for $j=d+1,d+2,\dots,p+1$, set $\tilde{U}_{d,j}=(U_{j-1}' \setminus \lrcb{v}) \cup \lrcb{w_1, w_2}$.    
    \item If $d=p+1$, then set $\tilde{U}_{p+1, p+1}=V_L(\mbf') \cup \lrcb{w_1, w_2}$, where $V_L(\mbf')$ denotes the set of leaves of $\mbf'$.
\end{itemize}
In other words, we create a new layer containing the edges $(v, w_1)$ and $(v, w_2)$, place it between the original layers, or above the top layer, or below the bottom layer, while replacing $v$ with $w_1$ and $w_2$ in the layers beneath it. Let $\tilde{\Omega}_d = (\tilde{U}_{d,1}, \tilde{U}_{d,2}, \dots, \tilde{U}_{d,p+1})$. Then one can verify that
$\tilde{\Omega}_d\in \mathcal{L}_{p+1}(\mbf)$, and 
\begin{align}
\bar{\beta}(\mbf, \tilde{\Omega}_d, h)
=\frac{1}{\lambda_{h(w_1)}+\lambda_{h(w_2)}-\lambda_{h(v)}} \cdot
\bar{\beta}(\mbf', \Omega', h')
\label{eq:induction_relation5}
\end{align}
for each $d \in \lrcb{q+1,q+2,\dots,p+1}$.

Now let $\mathcal{X}(\Omega')\defeq\lrcb{\Omega_c:~c \in \lrcb{q,q+1,\dots,p}}$, 
and $\mathcal{Y}(\Omega')\defeq\lrcb{\tilde{\Omega}_d:~d \in \lrcb{q+1,q+2, \dots, p+1}}$.
One can verify that
\begin{align}
\bigsqcup_{p=1}^{z+1}\mathcal{L}_{p}(\mbf) = 
\bigsqcup_{p=1}^z
\bigsqcup_{\Omega' \in \mathcal{L}_p(\mbf')}
\lrb{\mathcal{X}(\Omega')
\bigsqcup \mathcal{Y}(\Omega')}.
\label{eq:induction_partition2}
\end{align}
Thus, all decompositions of $\mbf$ can be obtained by modifying a corresponding decomposition of $\mbf'$ in the manner described above.

Now we combine Eqs.~\eqref{eq:induction_assumption2},
\eqref{eq:induction_assumption3},
\eqref{eq:induction_relation3},
\eqref{eq:induction_relation4},
\eqref{eq:induction_relation5},
and \eqref{eq:induction_partition2} and obtain
\begin{align}
\sum_{p=1}^{z+1} (-1)^p
\sum_{\Omega \in \mathcal{L}_{p}(\mbf)} 
\bar{\beta}(\mbf, \Omega, h)   
&=
\frac{1}{\lambda_{h(w_1)}+\lambda_{h(w_2)}-\lambda_{h(v)}}
\sum_{p=1}^{z} (-1)^p
\sum_{\Omega' \in \mathcal{L}_{p}(\mbf')} 
\lrb{\bar{\beta}(\mbf', \Omega', h'')
-\bar{\beta}(\mbf', \Omega', h')}  \\
&=\frac{(-1)^{z}}{\lambda_{h(w_1)}+\lambda_{h(w_2)}-\lambda_{h(v)}} \lrb{\gamma(\mbf', h'') 
-\gamma(\mbf', h')}.
\end{align}
Next, we claim that 
\begin{align}
\frac{1}{\lambda_{h(w_1)}+\lambda_{h(w_2)}-\lambda_{h(v)}} \lrb{\gamma(\mbf', h'') -\gamma(\mbf', h')}
=-\gamma(\mbf, h).
\label{eq:gamma_difference}
\end{align}
This would imply
\begin{align}
\sum_{p=1}^{z+1} (-1)^p
\sum_{\Omega \in \mathcal{L}_{p}(\mbf)} 
\bar{\beta}(\mbf, \Omega, h)   
=(-1)^{z+1}\gamma(\mbf, h),
\end{align}
as desired.

It remain to prove Eq.~\eqref{eq:gamma_difference}. Suppose $\mbf=(\mbt_1, \mbt_2, \dots, \mbt_a)$ and $\mbf'=(\mbt_1', \mbt_2', \dots, \mbt_a')$. Recall that, by assumption, $u, v, w_1, w_2 \in V(\mbt_i)$. Then for each $j \in [a] \setminus \lrcb{i}$, we have $\mbt_j'=\mbt_j$, and the restrictions of $h$, $h'$ and $h''$ to this tree are identical. This implies that
\begin{align}
\prod_{j \in [a]\setminus\lrcb{i}}\gamma(\mbt_j, h)=
\prod_{j \in [a]\setminus\lrcb{i}}\gamma(\mbt_j', h') =
\prod_{j \in [a]\setminus\lrcb{i}}\gamma(\mbt_j', h'').
\end{align}
Therefore, it suffices to prove
\begin{align}
\frac{1}{\lambda_{h(w_1)}+\lambda_{h(w_2)}-\lambda_{h(v)}} \lrb{\gamma(\mbt_i', h'')-\gamma(\mbt_i', h')}
=-\gamma(\mbt_i, h).
\label{eq:gamma_difference2}
\end{align}
Note that $\mbt_i'$ is obtained from $\mbt_i$ by removing the nodes $w_1$ and $w_2$, along with the edges $(v, w_1)$ and $(v, w_2)$. Furthermore, the restrictions of $h$, $h'$ and $h''$ to $V(\mbt_i)\setminus\lrcb{w_1,w_2,v}$ are identical.

Now, for every $\vec u \defeq (u_1, u_2, \dots, u_n) \in O(\mbt_i')$, 
we map it to possibly multiple elements of $O(\mbt_i)$. This mapping induces a partition of $O(\mbt_i)$. Specifically, since $\lrcb{u_1, u_2, \dots, u_n}$ contains all internal nodes of $\mbt_i'$, which includes $u$, there must exist $m \in [n]$ such that $u=u_m$. Then for $j=m, m+1, \dots, n$, we define
\begin{align}
    \vec u^{(j)} \defeq (u_1, u_2, \dots, u_j, v, u_{j+1}, \dots, u_{n}).
\end{align}
That is, $\vec u^{(j)}$ is obtained from $\vec u$ by inserting $v$ right after $u_j$. Clearly, $\vec u^{(j)}$ is a valid topological order of $\mbt_i$. Furthermore, let
\begin{align}
    \mathcal{Z}(\vec u) = \lrcb{\vec u^{(j)}:~j\in \lrcb{m, m+1, \dots, n}}.
\end{align}
Then one can see that
\begin{align}
    O(\mbt_i) = \bigsqcup_{\vec u \in O(\mbt_i')} \mathcal{Z}(\vec u).
    \label{eq:Oti_partition}
\end{align}
In other words, every topological order of $\mbt_i$ can be generated by extending a corresponding topological order of $\mbt_i'$ in the above manner.

Now let $r$ denote the root of $\mbt_i$. Then $r$ is also the root of $\mbt_i'$. Moreover, let 
\begin{align}
\tilde{\Delta} \defeq \lambda_{h(w_1)}+\lambda_{h(w_2)}-\lambda_{h(v)}, 
\end{align}
and for any $c \in [n]$, let 
\begin{align}
\Delta'_{\vec u, c} &\defeq \sum_{w \in C(\lrcb{u_1,u_2,\dots,u_c})} \lambda_{h'(w)} -\lambda_{h'(r)},    \\
\Delta''_{\vec u, c} &\defeq \sum_{w \in C(\lrcb{u_1,u_2,\dots,u_c})} \lambda_{h''(w)} -\lambda_{h''(r)}, 
\end{align}
where $C(\cdot)$ is defined with respect to $\mbt_i'$\footnote{In fact, it does not matter if $C(\cdot)$ is defined with respect to $\mbt_i$ or $\mbt_i'$ here, as they yield the same results.}. In addition, for any $j \in \lrcb{m, m+1, \dots, n}$ and $c \in [n+1]$, let
\begin{align}
\Delta_{\vec u^{(j)}, c} &\defeq \sum_{w \in C\lrb{\lrcb{u^{(j)}_1,u^{(j)}_2,\dots,u^{(j)}_c}}} \lambda_{h(w)} -\lambda_{h(r)},
\end{align}
where $C(\cdot)$ is defined with respect to $\mbt_i$. Using the definitions of $h$, $h'$ and $h''$, along with the fact that $v$ is a child of $u$ and the parent of two leaves $w_1, w_2$ in $\mbt_i$, we get
\begin{align}
    \Delta''_{\vec u, c}&=\Delta'_{\vec u, c}, \qquad 1 \le c <m; \\
    \Delta''_{\vec u, c}&=\Delta'_{\vec u, c}+\tilde{\Delta}, \qquad m \le c \le n, 
\end{align}
and for each $j \in \lrcb{m, m+1, \dots, n}$, 
\begin{align}
\Delta_{\vec u^{(j)}, c} &=\Delta'_{\vec u, c}, \qquad 1 \le c \le j;\\    
\Delta_{\vec u^{(j)}, c} &=\Delta'_{\vec u, c-1}+\tilde{\Delta}, \qquad j < c \le n+1.
\end{align}
Now we utilize the identity, valid for any $\delta \neq 0$ and $x_1, x_2, \dots, x_l \in \R \setminus \lrcb{0, -\delta}$, 
\begin{align}
\frac{1}{\delta}\lrb{
\frac{1}{x_1}    
\frac{1}{x_2}
\dots
\frac{1}{x_l}
-
\frac{1}{x_1+\delta}    
\frac{1}{x_2+\delta}
\dots
\frac{1}{x_l+\delta}
}
=\sum_{b=1}^{l}
\frac{1}{x_1}    
\frac{1}{x_2}
\dots
\frac{1}{x_b}
\frac{1}{x_b+\delta}
\frac{1}{x_{b+1}+\delta}
\dots
\frac{1}{x_{l}+\delta},
\end{align}
which can be proved by induction on $l$, and obtain
\begin{align}
\frac{1}{\tilde{\Delta}} \cdot \lrb{
\frac{1}{\Delta'_{\vec u, 1}}
\frac{1}{\Delta'_{\vec u, 2}} \dots
\frac{1}{\Delta'_{\vec u, n}} 
-
\frac{1}{\Delta''_{\vec u, 1}}
\frac{1}{\Delta''_{\vec u, 2}} \dots
\frac{1}{\Delta''_{\vec u, n}} }
=\sum_{j=m}^n 
\frac{1}{\Delta_{\vec u^{(j)}, 1}}
\frac{1}{\Delta_{\vec u^{(j)}, 2}}\dots 
\frac{1}{\Delta_{\vec u^{(j)}, n+1}}.
\end{align}

From now on, we will write $m(\vec u)$ to explicitly indicate the dependence of $m$ on $\vec u$. Let 
\begin{align}
\kappa(\vec u, h') &\defeq \frac{1}{\Delta'_{\vec u, 1}}
\frac{1}{\Delta'_{\vec u, 2}} \dots
\frac{1}{\Delta'_{\vec u, n}}, \\
\kappa(\vec u, h'') &\defeq \frac{1}{\Delta''_{\vec u, 1}}
\frac{1}{\Delta''_{\vec u, 2}} \dots
\frac{1}{\Delta''_{\vec u, n}},   
\end{align}
and for each $j \in \lrcb{m(\vec u), m(\vec u)+1, \dots, n}$, let 
\begin{align}
\kappa(\vec u^{(j)}, h) &\defeq \frac{1}{\Delta_{\vec u^{(j)}, 1}}
\frac{1}{\Delta_{\vec u^{(j)}, 2}} \dots
\frac{1}{\Delta_{\vec u^{(j)}, n+1}}.
\end{align}
Then we have
\begin{align}
\frac{1}{\tilde{\Delta}} \lrb{\kappa(\vec u, h'') - \kappa(\vec u, h')}
=-\sum_{j=m(\vec u)}^n \kappa(\vec u^{(j)}, h).
\end{align}
Summing both sides over $\vec u \in O(\mbf')$ and applying Eq.~\eqref{eq:gamma_t_h_def} yields
\begin{align}
\frac{1}{\tilde{\Delta}} \lrb{\gamma(\mbt_i', h'')-\gamma(\mbt_i', h')}
&=\frac{1}{\tilde{\Delta}} \lrb{\sum_{\vec u \in O(\mbt_i')}\kappa(\vec u, h'') - \sum_{\vec u \in O(\mbt_i')} \kappa(\vec u, h')}\\
&=-\sum_{\vec u \in O(\mbt_i')}\sum_{j=m(\vec u)}^n \kappa(\vec u^{(j)}, h) \\
&=-\sum_{\vec u \in O(\mbt_i')}\sum_{\vec v \in \mathcal{Z}(\vec u)}  \kappa(\vec v, h) \\
&=-\sum_{\vec v \in O(\mbt_i)}  \kappa(\vec v, h) \\
&=-\gamma(\mbt_i, h),
\end{align}
where the second last step follows from Eq.~\eqref{eq:Oti_partition}. This completes the proof of Eq.~\eqref{eq:gamma_difference2}.
\end{itemize}
To summarize, we have shown that Eq.~\eqref{eq:key_fact_beta_gamma}  holds for all $\mbf \in \mathcal{T}^{a}_{a+z+1}$ and $h \in L(\mbf)$. This completes the proof of the proposition.

To illustrate the above proof, we present an example for each of the two cases.

\begin{itemize}
    \item Example of Case 1: Consider the binary forest $\mbf=(\mbt_1,\mbt_2)$ illustrated in Figure~\ref{fig:case1_example}. It has vertex set $V(\mbf)=\lrcb{v, w_1, w_2, s, t, x, y, z}$. The first binary tree $\mbt_1$ has vertex set $V(\mbt_1)=\lrcb{v, w_1, w_2}$ and topological order set $O(\mbt_1)=\lrcb{(v)}$, while the second binary tree $\mbt_2$ has vertex set $V(\mbt_2)=\lrcb{s, t, x, y, z}$ and topological order set $O(\mbt_2)=\lrcb{(s, x)}$.
    
    Let $V'=V\setminus\lrcb{w_1,w_2}=\lrcb{v, s, t, x, y, z}$, and define the restricted forest $\mbf'=\mbf[V']=(\mbt_1', \mbt_2)$, where $\mbt_1'$ has vertex set $V(\mbt_1')=\lrcb{v}$ and topological order set $O(\mbt_1')=\emptyset$ (as it has no internal nodes). 

    Let $h$ be any labeling of $\mbf$ and let
    $h'=h[V']$ be its restriction to $\mbf'$.

\begin{figure}[h!]
\centering
\includegraphics[width=0.5\textwidth]{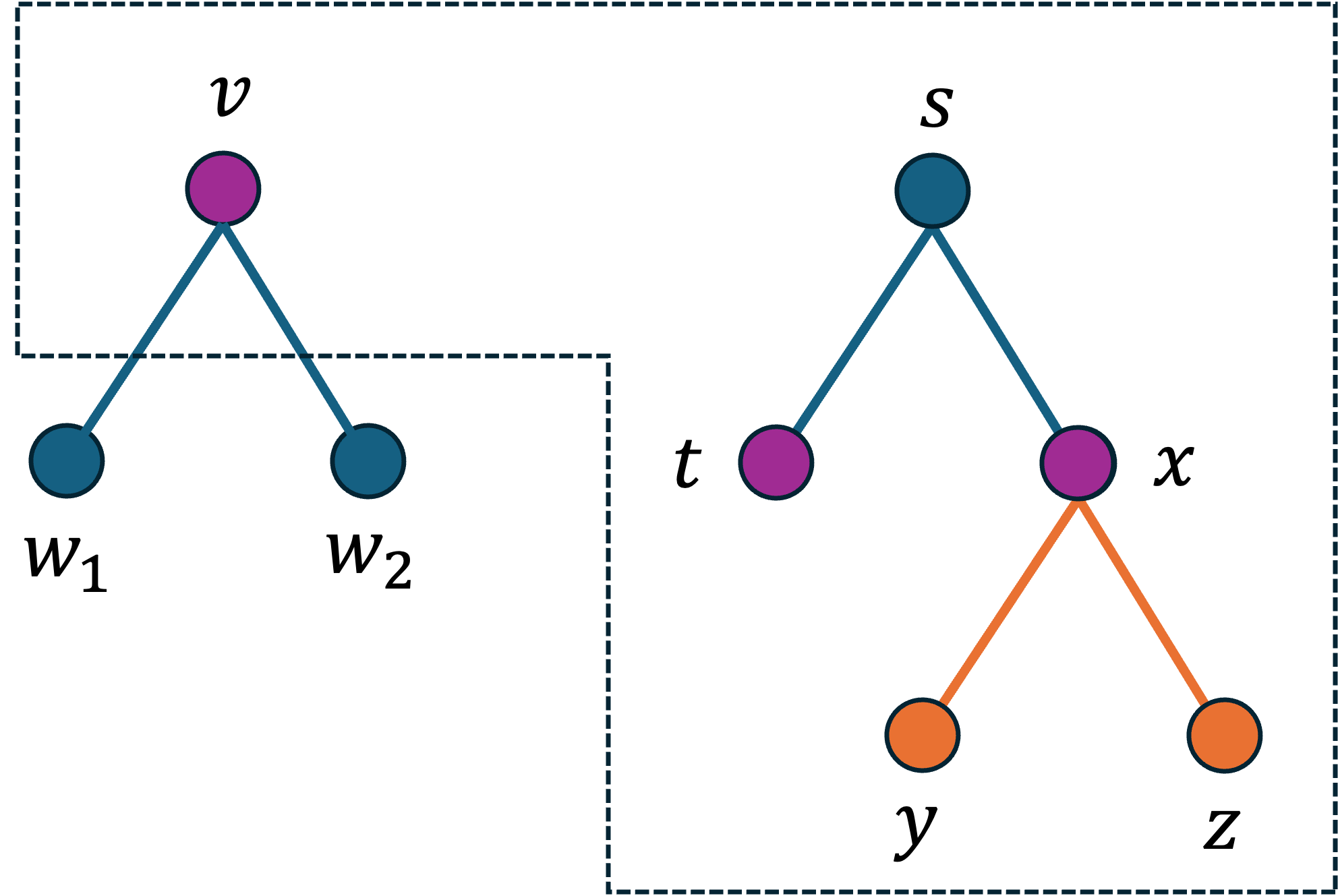}
\caption{An example of Case 1 in the proof of Proposition~\ref{prop:v_inv_ij_expression}. Here the binary forest $\mbf$ contains two binary trees, where the first tree includes three nodes $v, w_1, w_2$. A reduced forest $\mbf'$, enclosed by the dotted outline, is obtained by removing $w_1$ and $w_2$ from $\mbf$. A $2$-layer decomposition of $\mbf'$ is given by $(\lrcb{v, s, t, x}, \lrcb{v, t, x, y, z})$, which is highlighted in color in the figure. Note that this is just one of two possible decompositions of $\mbf'$.}
\label{fig:case1_example}
\end{figure}

$\mbf$ admits the following $1$-layer, $2$-layer and $3$-layer decompositions: 
\begin{align}
\mathcal{L}_1(\mbf)&=\lrcb{\Omega_1},    \\
\mathcal{L}_2(\mbf)&=\lrcb{\Omega_2, \Omega_3, \Omega_4, \Omega_5}, \\
\mathcal{L}_3(\mbf)&=\lrcb{\Omega_6, \Omega_7, \Omega_8},
\end{align}
where
\begin{align}
\Omega_1&=(\lrcb{v,w_1,w_2,s,t,x,y,z}), \\    
\Omega_2&=(\lrcb{v,w_1,w_2,s}, \lrcb{w_1,w_2,s,t,x,y,z}), \\
\Omega_3&=(\lrcb{v,s,t,x,y,z}, \lrcb{v,w_1,w_2,y,z}),\\
\Omega_4&=(\lrcb{v,s,t,x}, \lrcb{v,w_1,w_2,t,x,y,z}), \\
\Omega_5&=(\lrcb{v,w_1,w_2,s,t,x}, \lrcb{w_1,w_2,t,x,y,z}),\\
\Omega_6&=(\lrcb{v,w_1,w_2,s}, \lrcb{w_1,w_2,s,t,x}, \lrcb{w_1,w_2,t,x,y,z}),\\
\Omega_7&=(\lrcb{v,s,t,x}, \lrcb{v,w_1,w_2,t,x}, \lrcb{w_1,w_2,t,x,y,z}),\\
\Omega_8&=(\lrcb{v,s,t,x}, \lrcb{v,t,x,y,z}, \lrcb{v,w_1,w_2,y,z}).
\end{align}

Meanwhile, $\mbf'$ has the following $1$-layer and $2$-layer decompositions:
\begin{align}
\mathcal{L}_1(\mbf')&=\lrcb{\Omega_1'}, \\     
\mathcal{L}_2(\mbf')&=\lrcb{\Omega_2'},
\end{align}
where 
\begin{align}
\Omega_1'&=(\lrcb{v,s,t,x,y,z}), \\
\Omega_2'&=(\lrcb{v,s,t,x}, \lrcb{v,t,x,y,z}).
\end{align}

The decompositions of $\mbf$ and $\mbf'$ are related as follows:
\begin{align}
\mathcal{X}(\Omega_1')&=\lrcb{\Omega_1},    \\
\mathcal{Y}(\Omega_1')&=\lrcb{\Omega_2, \Omega_3}, \\
\mathcal{X}(\Omega_2')&=\lrcb{\Omega_4,\Omega_5},  \\
\mathcal{Y}(\Omega_2')&=\lrcb{\Omega_6,\Omega_7, \Omega_8}.
\end{align}

Note that
\begin{align}
\bigsqcup_{p=1}^3 \mathcal{L}_p(\mbf) = \bigsqcup_{p=1}^2 \bigsqcup_{\Omega' \in \mathcal{L}_p(\mbf')}\lrb{\mathcal{X}(\Omega') \bigsqcup \mathcal{Y}(\Omega')}.  
\end{align}

The $\bar{\beta}$ functions of $\mbf$ and $\mbf'$ with respect to their decompositions and labelings are given by:
\begin{align}
\bar{\beta}(\mbf,\Omega_i,h)&=
\frac{1}{\lambda_{h(w_1)}+\lambda_{h(w_2)}-\lambda_{h(v)}}
\cdot
\frac{1}{\lambda_{h(y)}+\lambda_{h(z)}-\lambda_{h(x)}} \cdot 
\frac{1}{\lambda_{h(t)}+\lambda_{h(y)}+\lambda_{h(z)}-\lambda_{h(s)}},~~\forall i \in \lrcb{1,2,3}, \\    
\bar{\beta}(\mbf,\Omega_j,h)&=
\frac{1}{\lambda_{h(w_1)}+\lambda_{h(w_2)}-\lambda_{h(v)}}
\cdot
\frac{1}{\lambda_{h(y)}+\lambda_{h(z)}-\lambda_{h(x)}} \cdot 
\frac{1}{\lambda_{h(t)}+\lambda_{h(x)}-\lambda_{h(s)}},~~\forall j\in \lrcb{4,5,6,7,8}, 
\end{align}
\begin{align}
\bar{\beta}(\mbf',\Omega_1',h')&=
\frac{1}{\lambda_{h(y)}+\lambda_{h(z)}-\lambda_{h(x)}} \cdot 
\frac{1}{\lambda_{h(t)}+\lambda_{h(y)}+\lambda_{h(z)}-\lambda_{h(s)}}, \\
\bar{\beta}(\mbf',\Omega_2',h')&=
\frac{1}{\lambda_{h(y)}+\lambda_{h(z)}-\lambda_{h(x)}} \cdot 
\frac{1}{\lambda_{h(t)}+\lambda_{h(x)}-\lambda_{h(s)}},    
\end{align}
and they satisfy the relationship
\begin{align}
\bar{\beta}(\mbf,\Omega_i,h)&=
\frac{1}{\lambda_{h(w_1)}+\lambda_{h(w_2)}-\lambda_{h(v)}} \cdot \bar{\beta}(\mbf',\Omega_1',h'),
~~\forall i \in \lrcb{1,2,3}, \\
\bar{\beta}(\mbf,\Omega_j,h)&=
\frac{1}{\lambda_{h(w_1)}+\lambda_{h(w_2)}-\lambda_{h(v)}} \cdot \bar{\beta}(\mbf',\Omega_2',h'),
~~\forall j \in \lrcb{4,5,6,7,8}.
\end{align}

Meanwhile, the $\gamma$ functions of $\mbf$ and $\mbf'$ with respect to their labelings are given by:
\begin{align}
\gamma(\mbf,h)&=\frac{1}{\lambda_{h(w_1)}+\lambda_{h(w_2)}-\lambda_{h(v)}}
\cdot \frac{1}{\lambda_{h(t)}+\lambda_{h(x)}-\lambda_{h(s)}} \cdot 
\frac{1}{\lambda_{h(t)}+\lambda_{h(y)}+\lambda_{h(z)}-\lambda_{h(s)}}, \\    
\gamma(\mbf',h')&=\frac{1}{\lambda_{h(t)}+\lambda_{h(x)}-\lambda_{h(s)}} \cdot 
\frac{1}{\lambda_{h(t)}+\lambda_{h(y)}+\lambda_{h(z)}-\lambda_{h(s)}},
\end{align}
and they satisfy the relationship
\begin{align}
\gamma(\mbf,h)=\frac{1}{\lambda_{h(w_1)}+\lambda_{h(w_2)}-\lambda_{h(v)}}
\cdot \gamma(\mbf',h').    
\end{align}

By direct calculation, one can verify that:
\begin{align}
\sum_{p=1}^{2} (-1)^p
\sum_{\Omega' \in \mathcal{L}_{p}(\mbf')} 
\bar{\beta}(\mbf', \Omega', h')   
=\gamma(\mbf', h'),
\end{align}

Then it follows that
\begin{align}
\sum_{p=1}^{3} (-1)^p
\sum_{\Omega \in \mathcal{L}_{p}(\mbf)} 
\bar{\beta}(\mbf, \Omega, h)  
&=-\frac{1}{\lambda_{h(w_1)}+\lambda_{h(w_2)}-\lambda_{h(v)}} \cdot
\sum_{p=1}^{2} (-1)^p
\sum_{\Omega' \in \mathcal{L}_{p}(\mbf')} 
\bar{\beta}(\mbf', \Omega', h')    \\
&=-\frac{1}{\lambda_{h(w_1)}+\lambda_{h(w_2)}-\lambda_{h(v)}} \cdot \gamma(\mbf', h') \\
&=-\gamma(\mbf, h),
\end{align}
as desired.

\item Example of case 2: Consider the binary forest $\mbf=(\mbt_1,\mbt_2)$ illustrated in Figure~\ref{fig:case2_example}. It has vertex set $V(\mbf)=\lrcb{u, v, k, w_1, w_2, s, t, x, y, z}$. The first binary tree $\mbt_1$ has vertex set $V(\mbt_1)=\lrcb{u, v, k, w_1, w_2}$, and topological order set $O(\mbt_1)=\lrcb{(u, v)}$, while the second binary tree $\mbt_2$ has vertex set $V(\mbt_2)=\lrcb{s, t, x, y, z}$ and topological order set $O(\mbt_2)=\lrcb{(s, x)}$.
    
Let $V'=V\setminus\lrcb{w_1,w_2}=\lrcb{u, v, k, s, t, x, y, z}$, and define the restricted forest $\mbf'=\mbf[V']=(\mbt_1', \mbt_2)$, where $\mbt_1'$ has vertex set $V(\mbt_1')=\lrcb{u,v,k}$ and topological order set $O(\mbt_1')=\lrcb{(u)}$. 

Let $h$ be any labeling of $\mbf$ and let $h'=h[V']$ be its restriction to $\mbf'$. Moreover, let $h''$ be a labeling of $\mbf'$ such that $h''(w)=h(w)$ for all $w \in V' \setminus \lrcb{v}$, and $\lambda_{h''(v)}=\lambda_{h(w_1)}+\lambda_{h(w_2)}$.

\begin{figure}[h!]
\centering
\includegraphics[width=0.5\textwidth]{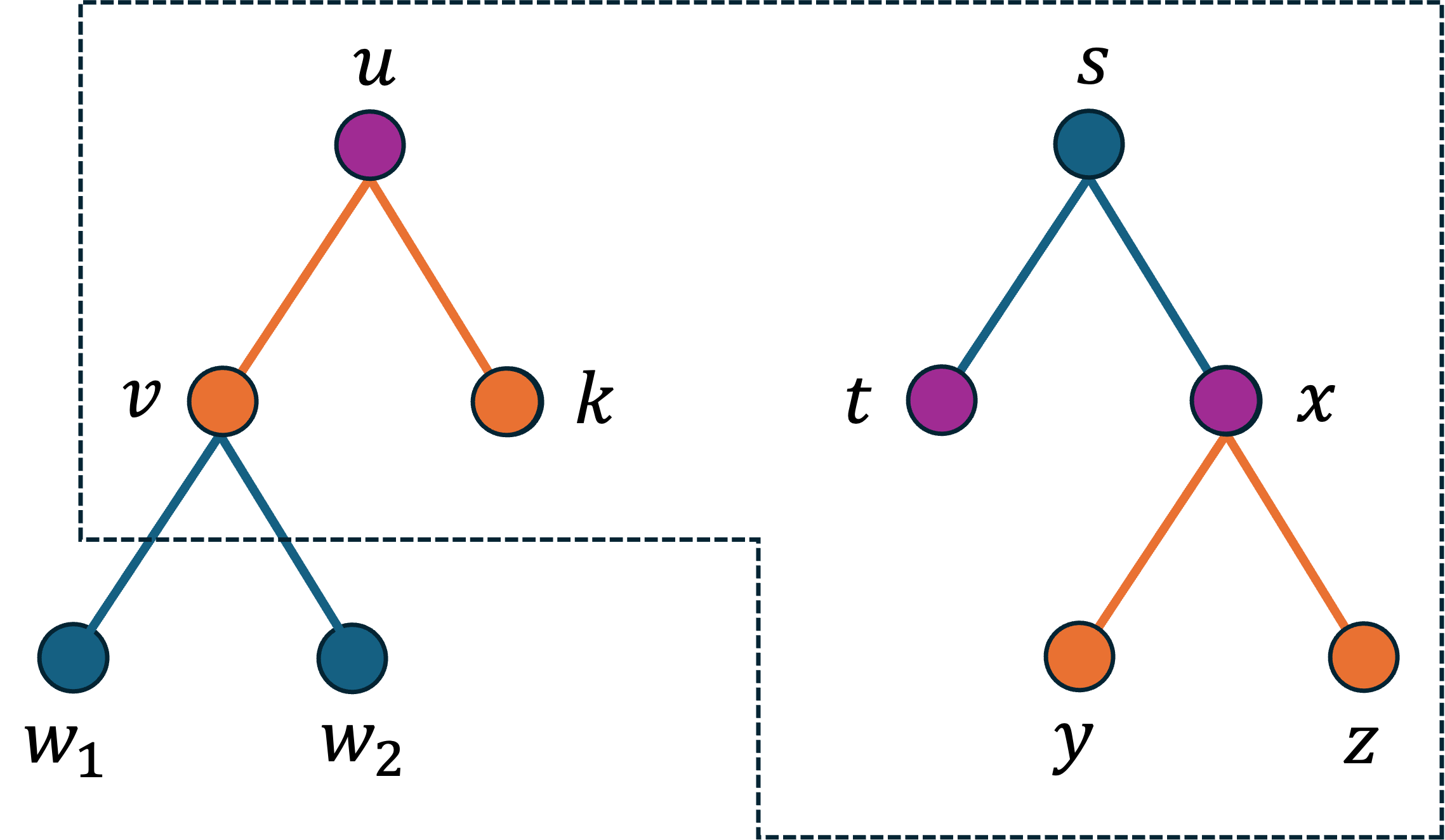}
\caption{An example of Case 2 in the proof of Proposition~\ref{prop:v_inv_ij_expression}. Here the binary forest $\mbf$ contains two binary trees, where the first tree has a node $v$ with parent $u$ and two children $w_1, w_2$. A reduced forest $\mbf'$, enclosed by the dotted outline, is obtained by removing $w_1$ and $w_2$ from $\mbf$. A $2$-layer decomposition of $\mbf'$ is given by $(\lrcb{u, s, t, x}, \lrcb{u, v, k, t, x, y, z})$, which is highlighted in color in the figure. Note that this is just one of eight possible decompositions of $\mbf'$.}
\label{fig:case2_example}
\end{figure}

$\mbf$ admits the following $1$-layer, $2$-layer, $3$-layer and $4$-layer decompositions: 

\begin{align}
\mathcal{L}_1(\mbf)&=\lrcb{\Omega_1},     \\
\mathcal{L}_2(\mbf)&=\lrcb{\Omega_2, \Omega_3, \Omega_4, \Omega_5, \Omega_6, \Omega_7, \Omega_8}, \\
\mathcal{L}_3(\mbf)&=\lrcb{\Omega_9, \Omega_{10}, \Omega_{11}, \Omega_{12}, \Omega_{13}, \Omega_{14}, \Omega_{15}, \Omega_{16}, \Omega_{17}, \Omega_{18}, \Omega_{19}, \Omega_{20}}, \\
\mathcal{L}_4(\mbf)&=\lrcb{ \Omega_{21}, \Omega_{22}, \Omega_{23},\Omega_{24}, \Omega_{25},\Omega_{26}},
\end{align}
where
\begin{align}
\Omega_1&=(\lrcb{u,v,k,w_1,w_2,s,t,x,y,z}), \\    
\Omega_2&=(\lrcb{u,v,k,w_1,w_2,s},\lrcb{w_1,w_2,k,s,t,x,y,z}), \\
\Omega_3&=(\lrcb{u,s,t,x,y,z},\lrcb{u,v,k,w_1,w_2,t,y,z}), \\
\Omega_4&=(\lrcb{u,v,k,s,t,x,y,z},\lrcb{v,w_1,w_2,k,t,y,z}), \\
\Omega_5&=(\lrcb{u,v,k,w_1,w_2,s,t,x}, \lrcb{w_1,w_2,k,t,x,y,z}), \\
\Omega_6&=(\lrcb{u,v,k,s,t,x}, \lrcb{v,w_1,w_2,k,t,x,y,z}), \\
\Omega_7&=(\lrcb{u,s,t,x}, \lrcb{u,v,k,w_1,w_2,t,x,y,z}), \\
\Omega_8&=(\lrcb{u,v,k,s}, \lrcb{v,k,w_1,w_2,s,t,x,y,z}), \\
\Omega_{9}&=(\lrcb{u,s,t,x}, \lrcb{u,t,x,y,z}, \lrcb{u,v,k,w_1,w_2,t,y,z}), \\
\Omega_{10}&=(\lrcb{u,v,k,s}, \lrcb{v,k,w_1,w_2,s}, \lrcb{w_1,w_2,k,s,t,x,y,z}), \\
\Omega_{11}&=(\lrcb{u,s,t,x,y,z}, \lrcb{u,v,k,t,y,z}, \lrcb{v,w_1,w_2,k,t,y,z}), \\
\Omega_{12}&=(\lrcb{u,v,k,s,t,x}, \lrcb{v,k,t,x,y,z,}, \lrcb{v,w_1,w_2,k,t,y,z}), \\
\Omega_{13}&=(\lrcb{u,s,t,x}, \lrcb{u,v,k,t,x,y,z}, \lrcb{v,w_1,w_2,k,t,y,z}), \\
\Omega_{14}&=(\lrcb{u,v,k,s}, \lrcb{v,k,s,t,x,y,z}, \lrcb{v,w_1,w_2,k,t,y,z}), \\
\Omega_{15}&=(\lrcb{u,v,k,w_1,w_2, s}, \lrcb{w_1,w_2,k,s,t,x}, \lrcb{w_1,w_2,k,t,x,y,z}), \\
\Omega_{16}&=(\lrcb{u,v,k,s,t,x}, \lrcb{v,w_1,w_2,k,t,x}, \lrcb{w_1,w_2,k,t,x, y,z}), \\
\Omega_{17}&=(\lrcb{u,v,k,s}, \lrcb{v,w_1,w_2,k,s,t,x}, \lrcb{w_1,w_2,k,t,x, y,z}), \\
\Omega_{18}&=(\lrcb{u,s,t,x}, \lrcb{u,v,k,w_1,w_2,t,x}, \lrcb{w_1,w_2,k,t,x, y,z}), \\
\Omega_{19}&=(\lrcb{u,s,t,x}, \lrcb{u,v,k,t,x}, \lrcb{v,w_1,w_2,k,t,x,y,z}), \\
\Omega_{20}&=(\lrcb{u,v,k,s}, \lrcb{v,k,s,t,x}, \lrcb{v,w_1,w_2,k,t,x,y,z}), \\
\Omega_{21}&=(\lrcb{u,s,t,x}, \lrcb{u,t,x,y,z}, \lrcb{u,v,k,t,y,z}, \lrcb{v,w_1,w_2,k,t,y,z}), \\
\Omega_{22}&=(\lrcb{u,s,t,x}, \lrcb{u,v,k,t,x}, \lrcb{v,k,t,x,y,z}, \lrcb{v,w_1,w_2,k,t,y,z}), \\
\Omega_{23}&=(\lrcb{u,v,k,s}, \lrcb{v,k,s,t,x}, \lrcb{v,k,t,x,y,z}, \lrcb{v,w_1,w_2,k,t,y,z}), \\
\Omega_{24}&=(\lrcb{u,v,k,s}, \lrcb{v,w_1,w_2,k,s}, \lrcb{w_1,w_2,k,s,t,x}, \lrcb{w_1,w_2,k,t,x,y,z}), \\
\Omega_{25}&=(\lrcb{u,v,k,s}, \lrcb{v,k,s,t,x}, \lrcb{v,w_1,w_2,k,t,x}, \lrcb{w_1,w_2,k,t,x,y,z}), \\
\Omega_{26}&=(\lrcb{u,s,t,x}, \lrcb{u,v,k,t,x}, \lrcb{v,w_1,w_2,k,t,x}, \lrcb{w_1,w_2,k,t,x,y,z}).
\end{align}

Meanwhile, $\mbf'$ has the following $1$-layer, $2$-layer and $3$-layer decompositions:
\begin{align}
\mathcal{L}_1(\mbf')&=\lrcb{\Omega_1'},     \\
\mathcal{L}_2(\mbf')&=\lrcb{\Omega_2', \Omega_3', \Omega_4', \Omega_5'}, \\
\mathcal{L}_3(\mbf')&=\lrcb{\Omega_6', \Omega_7', \Omega_8'},
\end{align}
where
\begin{align}
\Omega_1'&=(\lrcb{u,v,k,s,t,x,y,z}),    \\
\Omega_2'&=(\lrcb{u,s,t,x,y,z}, \lrcb{u,v,k,t,y,z}), \\
\Omega_3'&=(\lrcb{u,v,k,s,t,x}, \lrcb{v,k,t,x,y,z}), \\
\Omega_4'&=(\lrcb{u,s,t,x}, \lrcb{u,v,k,t,x,y,z}), \\
\Omega_5'&=(\lrcb{u,v,k,s}, \lrcb{v,k,s,t,x,y,z}), \\
\Omega_6'&=(\lrcb{u,s,t,x}, \lrcb{u,t,x,y,z}, \lrcb{u,v,k,t,y,z}), \\
\Omega_7'&=(\lrcb{u,s,t,x}, \lrcb{u,v,k,t,x}, \lrcb{v,k,t,x,y,z}), \\
\Omega_8'&=(\lrcb{u,v,k,s}, \lrcb{v,k,s,t,x}, \lrcb{v,k,t,x,y,z}).
\end{align}

The decompositions of $\mbf$ and $\mbf'$ are related as follows:
\begin{align}
\mathcal{X}(\Omega'_1)&=\lrcb{\Omega_1}, \\
\mathcal{Y}(\Omega'_1)&=\lrcb{\Omega_4}, \\
\mathcal{X}(\Omega'_2)&=\lrcb{\Omega_3}, \\
\mathcal{Y}(\Omega'_2)&=\lrcb{\Omega_{11}}, \\
\mathcal{X}(\Omega'_3)&=\lrcb{\Omega_5, \Omega_6}, \\
\mathcal{Y}(\Omega'_3)&=\lrcb{\Omega_{12}, \Omega_{16}}, \\
\mathcal{X}(\Omega'_4)&=\lrcb{\Omega_7}, \\
\mathcal{Y}(\Omega'_4)&=\lrcb{\Omega_{13}}, \\
\mathcal{X}(\Omega'_5)&=\lrcb{\Omega_2, \Omega_8}, \\
\mathcal{Y}(\Omega'_5)&=\lrcb{\Omega_{10}, \Omega_{14}}, \\
\mathcal{X}(\Omega'_6)&=\lrcb{\Omega_9}, \\
\mathcal{Y}(\Omega'_6)&=\lrcb{\Omega_{21}}, \\
\mathcal{X}(\Omega'_7)&=\lrcb{\Omega_{18}, \Omega_{19}}, \\
\mathcal{Y}(\Omega'_7)&=\lrcb{\Omega_{22}, \Omega_{26}}, \\
\mathcal{X}(\Omega'_8)&=\lrcb{\Omega_{15}, \Omega_{17}, \Omega_{20}}, \\
\mathcal{Y}(\Omega'_8)&=\lrcb{\Omega_{23}, \Omega_{24},\Omega_{25}}.
\end{align}

Note that
\begin{align}
\bigsqcup_{p=1}^4 \mathcal{L}_p(\mbf) = \bigsqcup_{p=1}^3 \bigsqcup_{\Omega' \in \mathcal{L}_p(\mbf')}\lrb{\mathcal{X}(\Omega') \bigsqcup \mathcal{Y}(\Omega')}.  
\end{align}

The $\bar{\beta}$ functions of $\mbf$ and $\mbf'$ with respect to their decompositions and labelings satisfy the following relationship:
\begin{align}
\bar{\beta}(\mbf, \Omega_1, h)
=\frac{1}{\lambda_{h(w_1)}+\lambda_{h(w_2)}-\lambda_{h(v)}} \cdot
\bar{\beta}(\mbf', \Omega'_1, h''),\\
\bar{\beta}(\mbf, \Omega_2, h)
=\frac{1}{\lambda_{h(w_1)}+\lambda_{h(w_2)}-\lambda_{h(v)}} \cdot
\bar{\beta}(\mbf', \Omega'_5, h''),\\
\bar{\beta}(\mbf, \Omega_3, h)
=\frac{1}{\lambda_{h(w_1)}+\lambda_{h(w_2)}-\lambda_{h(v)}} \cdot
\bar{\beta}(\mbf', \Omega'_2, h''),\\
\bar{\beta}(\mbf, \Omega_4, h)
=\frac{1}{\lambda_{h(w_1)}+\lambda_{h(w_2)}-\lambda_{h(v)}} \cdot
\bar{\beta}(\mbf', \Omega'_1, h'),\\
\bar{\beta}(\mbf, \Omega_5, h)
=\frac{1}{\lambda_{h(w_1)}+\lambda_{h(w_2)}-\lambda_{h(v)}} \cdot
\bar{\beta}(\mbf', \Omega'_3, h''),\\
\bar{\beta}(\mbf, \Omega_6, h)
=\frac{1}{\lambda_{h(w_1)}+\lambda_{h(w_2)}-\lambda_{h(v)}} \cdot
\bar{\beta}(\mbf', \Omega'_3, h'),\\
\bar{\beta}(\mbf, \Omega_7, h)
=\frac{1}{\lambda_{h(w_1)}+\lambda_{h(w_2)}-\lambda_{h(v)}} \cdot
\bar{\beta}(\mbf', \Omega'_4, h''),\\
\bar{\beta}(\mbf, \Omega_8, h)
=\frac{1}{\lambda_{h(w_1)}+\lambda_{h(w_2)}-\lambda_{h(v)}} \cdot
\bar{\beta}(\mbf', \Omega'_5, h'),\\
\bar{\beta}(\mbf, \Omega_9, h)
=\frac{1}{\lambda_{h(w_1)}+\lambda_{h(w_2)}-\lambda_{h(v)}} \cdot
\bar{\beta}(\mbf', \Omega'_6, h''),\\
\bar{\beta}(\mbf, \Omega_{10}, h)
=\frac{1}{\lambda_{h(w_1)}+\lambda_{h(w_2)}-\lambda_{h(v)}} \cdot
\bar{\beta}(\mbf', \Omega'_{5}, h'),\\
\bar{\beta}(\mbf, \Omega_{11}, h)
=\frac{1}{\lambda_{h(w_1)}+\lambda_{h(w_2)}-\lambda_{h(v)}} \cdot
\bar{\beta}(\mbf', \Omega'_2, h'),\\
\bar{\beta}(\mbf, \Omega_{12}, h)
=\frac{1}{\lambda_{h(w_1)}+\lambda_{h(w_2)}-\lambda_{h(v)}} \cdot
\bar{\beta}(\mbf', \Omega'_3, h'),\\
\bar{\beta}(\mbf, \Omega_{13}, h)
=\frac{1}{\lambda_{h(w_1)}+\lambda_{h(w_2)}-\lambda_{h(v)}} \cdot
\bar{\beta}(\mbf', \Omega'_4, h'),\\
\bar{\beta}(\mbf, \Omega_{14}, h)
=\frac{1}{\lambda_{h(w_1)}+\lambda_{h(w_2)}-\lambda_{h(v)}} \cdot
\bar{\beta}(\mbf', \Omega'_5, h'),\\
\bar{\beta}(\mbf, \Omega_{15}, h)
=\frac{1}{\lambda_{h(w_1)}+\lambda_{h(w_2)}-\lambda_{h(v)}} \cdot
\bar{\beta}(\mbf', \Omega'_8, h''),\\
\bar{\beta}(\mbf, \Omega_{16}, h)
=\frac{1}{\lambda_{h(w_1)}+\lambda_{h(w_2)}-\lambda_{h(v)}} \cdot
\bar{\beta}(\mbf', \Omega'_3, h'),\\
\bar{\beta}(\mbf, \Omega_{17}, h)
=\frac{1}{\lambda_{h(w_1)}+\lambda_{h(w_2)}-\lambda_{h(v)}} \cdot
\bar{\beta}(\mbf', \Omega'_8, h'),\\
\bar{\beta}(\mbf, \Omega_{18}, h)
=\frac{1}{\lambda_{h(w_1)}+\lambda_{h(w_2)}-\lambda_{h(v)}} \cdot
\bar{\beta}(\mbf', \Omega'_7, h''),\\
\bar{\beta}(\mbf, \Omega_{19}, h)
=\frac{1}{\lambda_{h(w_1)}+\lambda_{h(w_2)}-\lambda_{h(v)}} \cdot
\bar{\beta}(\mbf', \Omega'_7, h'),\\
\bar{\beta}(\mbf, \Omega_{20}, h)
=\frac{1}{\lambda_{h(w_1)}+\lambda_{h(w_2)}-\lambda_{h(v)}} \cdot
\bar{\beta}(\mbf', \Omega'_8, h'),\\
\bar{\beta}(\mbf, \Omega_{21}, h)
=\frac{1}{\lambda_{h(w_1)}+\lambda_{h(w_2)}-\lambda_{h(v)}} \cdot
\bar{\beta}(\mbf', \Omega'_6, h'), \\
\bar{\beta}(\mbf, \Omega_{22}, h)
=\frac{1}{\lambda_{h(w_1)}+\lambda_{h(w_2)}-\lambda_{h(v)}} \cdot
\bar{\beta}(\mbf', \Omega'_7, h'), \\
\bar{\beta}(\mbf, \Omega_{23}, h)
=\frac{1}{\lambda_{h(w_1)}+\lambda_{h(w_2)}-\lambda_{h(v)}} \cdot
\bar{\beta}(\mbf', \Omega'_8, h'), \\
\bar{\beta}(\mbf, \Omega_{24}, h)
=\frac{1}{\lambda_{h(w_1)}+\lambda_{h(w_2)}-\lambda_{h(v)}} \cdot
\bar{\beta}(\mbf', \Omega'_8, h'), \\
\bar{\beta}(\mbf, \Omega_{25}, h)
=\frac{1}{\lambda_{h(w_1)}+\lambda_{h(w_2)}-\lambda_{h(v)}} \cdot
\bar{\beta}(\mbf', \Omega'_8, h'), \\
\bar{\beta}(\mbf, \Omega_{26}, h)
=\frac{1}{\lambda_{h(w_1)}+\lambda_{h(w_2)}-\lambda_{h(v)}} \cdot
\bar{\beta}(\mbf', \Omega'_7, h').
\end{align}

Meanwhile, the $\gamma$ functions of $\mbf$ and $\mbf'$ with respect to their labelings are given by:
\begin{align}
\gamma(\mbf, h) &=
\frac{1}{\lambda_{h(v)}+\lambda_{h(k)}-\lambda_{h(u)}}
\cdot \frac{1}{\lambda_{h(w_1)}+\lambda_{h(w_2)}+\lambda_{h(k)}-\lambda_{h(u)}} \cdot 
\frac{1}{\lambda_{h(t)}+\lambda_{h(x)}-\lambda_{h(s)}} \nonumber\\ &\quad \cdot
\frac{1}{\lambda_{h(t)}+\lambda_{h(y)}+\lambda_{h(z)}-
\lambda_{h(s)}}, \\
    \gamma(\mbf', h')&= \frac{1}{\lambda_{h(v)}+\lambda_{h(k)}-\lambda_{h(u)}} \cdot 
\frac{1}{\lambda_{h(t)}+\lambda_{h(x)}-\lambda_{h(s)}} \cdot
\frac{1}{\lambda_{h(t)}+\lambda_{h(y)}+\lambda_{h(z)}-
\lambda_{h(s)}}, \\
\gamma(\mbf', h'')&= \frac{1}{\lambda_{h(w_1)}+\lambda_{h(w_2)}+\lambda_{h(k)}-\lambda_{h(u)}} \cdot 
\frac{1}{\lambda_{h(t)}+\lambda_{h(x)}-\lambda_{h(s)}} \cdot
\frac{1}{\lambda_{h(t)}+\lambda_{h(y)}+\lambda_{h(z)}-
\lambda_{h(s)}},
\end{align}
and they satisfy the relationship:
\begin{align}
\gamma(\mbf, h) = -\frac{1}{\lambda_{h(w_1)}+\lambda_{h(w_2)}-\lambda_{h(v)} } \cdot \lrb{\gamma(\mbf', h'')-\gamma(\mbf', h')}.
\end{align}

By direction calculation, one can verify that
\begin{align}
\sum_{p=1}^{3} (-1)^p
\sum_{\Omega' \in \mathcal{L}_{p}(\mbf')} 
\bar{\beta}(\mbf', \Omega', h')
= -\gamma(\mbf', h')
\end{align}
and
\begin{align}
\sum_{p=1}^{3} (-1)^p
\sum_{\Omega' \in \mathcal{L}_{p}(\mbf')} 
\bar{\beta}(\mbf', \Omega', h'')
= - \gamma(\mbf', h'').
\end{align}

Then it follows that
\begin{align}
\sum_{p=1}^{4} (-1)^p
\sum_{\Omega \in \mathcal{L}_{p}(\mbf)} 
\bar{\beta}(\mbf, \Omega, h)  
&=\frac{1}{\lambda_{h(w_1)}+\lambda_{h(w_2)}-\lambda_{h(v)}} \cdot
\sum_{p=1}^{3} (-1)^p
\sum_{\Omega' \in \mathcal{L}_{p}(\mbf')} 
\left (\bar{\beta}(\mbf', \Omega', h'') \right . \\
&\quad \left . -\bar{\beta}(\mbf', \Omega', h')\right )
\\
&=-\frac{1}{\lambda_{h(w_1)}+\lambda_{h(w_2)}-\lambda_{h(v)}} \cdot \lrb{\gamma(\mbf', h'')-\gamma(\mbf',h')} \\
&=\gamma(\mbf, h),
\end{align}
as desired.
\end{itemize}

\section{Proofs of Lemmas~\ref{lem:equivalence_linear_no_resonance_gap_poincare_domain},~\ref{lem:vij_norm_bound_poincare_domain} and \ref{lem:wij_norm_bound_poincare_domain}}
\label{app:proof_norm_bound_vij_v_inv_ij_poincare_domain}

The proof of the equivalence between \ref{assump:linear_no_resonance_gap} and the condition that $(\lambda_1, \lambda_2, \dots, \lambda_N)$ is nonresonant and lies in the Poincar\'{e} domain, is as follows.
\begin{proof}[Proof of Lemma \ref{lem:equivalence_linear_no_resonance_gap_poincare_domain}]
Suppose $\lambda_1, \lambda_2, \dots, \lambda_N$ satisfy Assumption~\ref{assump:linear_no_resonance_gap}. Then it is clear that $\vec\lambda=(\lambda_1, \lambda_2, \dots, \lambda_N)$ is nonresonant. Moreover, for any $\alpha_1, \alpha_3, \dots, \alpha_N \in \myN_0$ and $\alpha_2 \in \myN$,  Eq.~\eqref{eq:linear_no_resonance_gap} implies
\begin{equation}
    \abs{\frac{\lambda_1-\lambda_2}{\sum_{j=1}^N \alpha_j - 1} - \bar{\lambda}_{\vec \alpha}}  \ge \Delta >0,
    \label{eq:distance_to_point_in_convex_hull}
\end{equation}
where 
\begin{align}
\bar{\lambda}_{\vec \alpha} \defeq \frac{\sum_{j=1}^N (\alpha_j - \delta_{j,2}) \lambda_j}{\sum_{j=1}^N \alpha_j - 1}    
\end{align}
is a point in the convex hull of $\lambda_1,\lambda_2,\dots,\lambda_N$, denoted by $\operatorname{conv}(\vec \lambda)$. Moreover, by varying $\alpha_1$, $\alpha_2$, $\dots$, $\alpha_N$, we can generate a dense set of points in $\operatorname{conv}(\vec \lambda)$. Consequently, if the origin lies in $\operatorname{conv}(\vec \lambda)$, then
for arbitrary $K, \epsilon>0$, there exist  $\alpha^{(K,\epsilon)}_1, \alpha^{(K,\epsilon)}_3, \dots, \alpha^{(K,\epsilon)}_N \in \myN_0$ and $\alpha^{(K,\epsilon)}_2 \in \myN$ such that
\begin{align}
\sum_{j=1}^N \alpha^{(K,\epsilon)}_j - 1 \ge K \quad \mathrm{and}\quad \abs{\bar{\lambda}_{\vec \alpha^{(K,\epsilon)}}} \le \epsilon.   
\end{align}
But meanwhile, we have 
\begin{align}
\frac{\lambda_1-\lambda_2}{\sum_{j=1}^N \alpha_j^{(K,\epsilon)} - 1} \to 0, \quad \mathrm{as}\quad K \to +\infty.    
\end{align}
Therefore, we would have
\begin{align}
\abs{\frac{\lambda_1-\lambda_2}{\sum_{j=1}^N \alpha_j^{(K,\epsilon)} - 1} - \bar{\lambda}_{\vec \alpha^{(K,\epsilon)}}} \to 0, \quad \mathrm{as}\quad K \to +\infty, \quad \epsilon \to 0,     
\end{align}
which directly contradicts Eq.~\eqref{eq:distance_to_point_in_convex_hull}. Therefore, the origin must lie outside $\operatorname{conv}(\vec \lambda)$.

Conversely, suppose $\vec\lambda=(\lambda_1, \lambda_2, \dots, \lambda_N)$ is nonresonant and belongs to the Poincar\'{e} domain, i.e. 
$\vec{0} \not\in \operatorname{conv}(\vec \lambda)$. Then there exists a complex number $\omega$ with $\abs{\omega}=1$ such that 
$$f(\lambda_i)\defeq \realpart{\lambda_i^*\omega} > 0, $$
for all $i \in [N]$. Let $c_1=\mymin{f(\lambda_1),f(\lambda_2),\dots,f(\lambda_N)}$
and
$c_2=\mymax{f(\lambda_1),f(\lambda_2),\dots,f(\lambda_N)}$. Then we have $c_2 \ge c_1>0$. Let $k_0=\lceil c_2/c_1 \rceil + 1 \ge 2$, and define 
\begin{align}
\Delta_0=\frac{c_1k_0-c_2}{k_0-1}>0.    
\end{align}
Then for any $i\in [N]$ and  $\alpha_1,\alpha_2,\dots,\alpha_N \in \myN_0$ such that $\sum_{j=1}^N \alpha_N \ge k_0$, we have
\begin{align}
    f \lrb{\sum_{j=1}^N \alpha_j \lambda_{j} - \lambda_i}
    = \sum_{j=1}^N \alpha_j f(\lambda_{j}) - f(\lambda_i)
    \ge c_1 \sum_{j=1}^N \alpha_j - c_2 
    \ge \Delta_0\lrb{\sum_{j=1}^N \alpha_j-1}.
\end{align}
This implies that
\begin{align}
    \abs{\lambda_i-\sum_{j=1}^N \alpha_j\lambda_j} \ge \Delta_0\lrb{\sum_{j=1}^N \alpha_j-1}.    
\end{align}
Meanwhile, since $(\lambda_1, \lambda_2, \dots, \lambda_N)$ is nonresonant, we have 
\begin{align}
    \Delta_1 \defeq  \min_{i \in [N]} \inf_{\substack{\alpha_j \in \myN_0 \\ 2 \le \sum_{j=1}^N \alpha_j < k_0}} \frac{\abs{\lambda_i-\sum_{j=1}^N \alpha_j\lambda_j}}{\sum_{j=1}^N \alpha_j-1} > 0.
\end{align}
Combining these facts, we obtain
\begin{align}
    \Delta=\min_{i \in [N]} \inf_{\substack{\alpha_j \in \myN_0 \\ \sum_{j=1}^N \alpha_j \ge 2}} \frac{\abs{\lambda_i-\sum_{j=1}^N \alpha_j\lambda_j}}{\sum_{j=1}^N \alpha_j-1} \ge \mymin{\Delta_0, \Delta_1} > 0,    
\end{align}
as desired. 
\end{proof}

Before proceeding, we first establish two technical lemmas. The first provides an upper bound on the spectral norm of the Hadamard product between a sparse matrix and a matrix with small entries. The second lemma provides an upper bound on the number of binary forests with a specified number of roots and leaves.

\begin{lemma}
    Let $A,B \in \myC^{m \times n}$ be such that $A$ is $s$-column-sparse (or $s$-row-sparse) and $\norm{B}_{\mathrm{max}} =\max_{i,j} \abs{B_{i,j}} \le \delta$. Then we have $\norm{A \odot B} \le s\delta \norm{A}.$
\label{lem:hadamard_prod_norm_bound}
\end{lemma}
\begin{proof}
We will focus on the case where \( A \) is \( s \)-column-sparse. The case where \( A \) is \( s \)-row-sparse can be treated in a similar manner.

We can decompose $A$ as the sum of $s$ matrices $A_1, A_2, \dots, A_s$ such that each column of $A_j$ contains at most $1$ nonzero entry. This decomposition is achieved by allocating the $j$-th nonzero entry of each column of $A$ to $A_j$, for $j \in [s]$. 

Fix any $j \in [s]$. Since the rows of $A_j$ are orthogonal to each other, $\norm{A_j}$ equals the maximum norm of the rows of $A_j$. The same property  holds for $A_j \odot B$, as it is also $1$-column-sparse. Meanwhile, each nonzero entry of $A_j \odot B$ has magnitude at most $\delta$ times the magnitude of the corresponding entry of $A_j$. This implies that the norm of each row of $A_j \odot B$ is at most $\delta$ times the norm of the corresponding row of $A_j$. Consequently, we obtain 
\begin{align}
    \norm{A_j \odot B} \le \delta \norm{A_j} \le \delta \norm{A},    
\end{align}
where the last step follows from 
\begin{align}
    \norm{A_j}=\mathrm{maximum~norm~of~}A_j\mathrm{'s~rows} \le \mathrm{maximum~norm~of}A\mathrm{'s~rows} \le \norm{A}.    
\end{align}
It follows that
\begin{align}
    \norm{A \odot B} = \norm{\sum_{j=1}^s A_j \odot B}
    \le \sum_{j=1}^s \norm{A_j \odot B} \le s\delta \norm{A},
\end{align}
as claimed.
\end{proof}

\begin{lemma}
For any $i<j$, the number of binary forests consisting of $i$ binary trees with a total of $j$ leaves is at most $4^{j-i}\binom{j-1}{i-1}$.
\label{lem:count_binary_forests}
\end{lemma}
\begin{proof}
Let us begin by counting the number of binary trees with \( n \) leaves, denoted by \( D_n \). Consider the number of leaves in the subtrees rooted at the left and right children of the root. There are \( n-1 \) possible pairings: \( (1, n-1), (2, n-2), \dots, (n-1, 1) \), where \( (a, b) \) means the left and right subtrees have \( a \) and \( b \) leaves, respectively. This leads to the following recurrence relation for \( D_n \):
\begin{align}
    D_n = \sum_{k=1}^{n-1} D_k D_{n-k},    
\end{align}
with initial condition \( D_1 = 1 \). It turns out that \( D_n \) is exactly the \((n-1)\)-th Catalan number, denoted \( C_n \), which satisfies the recurrence:
\begin{align}
    C_n = \sum_{k=1}^{n} C_{k-1} C_{n-k},    
\end{align}
with initial condition \( C_0 = 1 \). Using known results for $C_n$ \cite{stanley1999book}, we obtain
\begin{align}
    D_n=C_{n-1}=\frac{1}{2n-1}\binom{2n-1}{n-1} \le \frac{4^{n-1}}{\sqrt{\pi} (n-1)^{3/2}+1}.
\end{align}

Next, we count the number of forests consisting of \( i \) binary trees with a total of \( j \) leaves. Given a binary forest \( \mbf = (\mathbf{t}_1, \mathbf{t}_2, \dots, \mathbf{t}_i)\) of this form, let \( n_l \) denote the number of leaves in $l$-th tree $\mbt_l$, for $l=1,2,\dots,i$. These values must satisfy the constraints \( n_1, n_2, \dots, n_i \in \mathbb{Z}^+ \) and \( n_1 + n_2 + \dots + n_i = j \). The number of such sequence of numbers is given by $\binom{j-1}{i-1}$. Furthermore, for fixed \( n_1, n_2, \dots, n_i \), there are \( D_{n_l} \) possible structures for the \( l \)-th binary tree, for \( l = 1, 2, \dots, i \), and hence the number of binary  forests consisting of such trees is given by
\begin{align}
\prod_{l=1}^i D_{n_l} \le \prod_{l=1}^i \frac{4^{n_l-1}}{\sqrt{\pi} (n_l-1)^{3/2}+1} 
\le  4^{j-i}.   
\end{align}
The upper bound $4^{j-i}$ here is independent of $n_1, n_2, \dots, n_i$. As a consequence, the number of binary forests consisting of $i$ trees with a total of $j$ leaves is at most
$4^{j-i}\binom{j-1}{i-1}$, as claimed.
\end{proof}

\begin{proof}[Proof of Lemma~\ref{lem:vij_norm_bound_poincare_domain}]
If $i=j$, then $\tilde{V}_{j,j}=I^{\otimes j}$ has norm $1$, and the claim holds trivially. So we assume $i<j$ from now on.

To prove the desired bound on $\norm{\tilde{V}_{i,j}}$, it suffices to show that for any binary tree $\mbt$ with $m$ leaves, 
\begin{align}
 \norm{\tilde{f}(\mbt)} \le \lrb{\frac{s\norm{\tilde{F}_2}}{\Delta}}^{m-1}.   
 \label{eq:bound_on_f_tree_poincare_domain}
\end{align}
To see this, note that for any binary forest $(\mbt_1, \mbt_2, \dots, \mbt_i) \in \mathcal{T}^i_j$, let $m_1,m_2,\dots,m_i$ be the number of leaves in $\mbt_1, \mbt_2,\dots, \mbt_i$ respectively, so that $m_1+m_2+\dots+m_i=j$. Then,
Eq.~\eqref{eq:bound_on_f_tree_poincare_domain} implies
\begin{align}
\norm{\tilde{f}(\mbt_1)\otimes \tilde{f}(\mbt_2)\otimes \dots \otimes \tilde{f}(\mbt_i)}
&= \norm{\tilde{f}(\mbt_1)}\cdot
        \norm{\tilde{f}(\mbt_2)} \cdot
        \dots
         \cdot   \norm{\tilde{f}(\mbt_i)}  \\
&\le 
 \lrb{\frac{s\norm{\tilde{F}_2}}{\Delta}}^{\sum_{l=1}^i m_l -1} \\
 &= 
 \lrb{\frac{s\norm{\tilde{F}_2}}{\Delta}}^{j-i}.
\end{align}
Then, it follows that
\begin{align}
\norm{\tilde{V}_{i,j}} &= \left\|\sum_{(\mbt_1, \mbt_2, \dots, \mbt_i) \in \mathcal{T}^i_j}
    \tilde{f}(\mbt_1)\otimes \tilde{f}(\mbt_2)\otimes \dots \otimes \tilde{f}(\mbt_i)\right\| \\ 
    &\le \sum_{(\mbt_1, \mbt_2, \dots, \mbt_i) \in \mathcal{T}^i_j}
    \norm{\tilde{f}(\mbt_1)\otimes \tilde{f}(\mbt_2)\otimes \dots \otimes \tilde{f}(\mbt_i)} \\ 
    &\le |\mathcal{T}^i_j| \cdot \lrb{\frac{s\norm{\tilde{F}_2}}{\Delta}}^{j-i} \le \binom{j-1}{i-1}\lrb{\frac{4s\norm{\tilde{F}_2}}{\Delta}}^{j-i},
\end{align}
where the last step follows from Lemma~\ref{lem:count_binary_forests}.

Let us first prove Eq.~\eqref{eq:bound_on_f_tree_poincare_domain} in the special case where $\tilde{F}_2$ is $1$-column-sparse. Recall that \( \tilde{f}(\mbt) \) is generated by the following process. We start with \( m \) independent subsystems, each assigned an \( I \) matrix, and then perform \( m-1 \) iterations of the following procedure: 
\begin{itemize} 
\item Apply \( \tilde{F}_2 \) to the tensor product of the matrices assigned to two adjacent subsystems, and combine these subsystems; 
\item Take the Hadamard product of the resulting matrix with \( N_l \), where \( l \) is the number of leaves in the combined subsystem. 
\end{itemize} 

Note that the standard and tensor products of two \(1\)-column-sparse matrices are both \(1\)-column-sparse. Moreover, the Hadamard product of a \(1\)-column-sparse matrix with any matrix is still \(1\)-column-sparse. Then, since \( \tilde{F}_2 \) is \(1\)-column-sparse, the tensor product of the matrices for the subsystems has column-sparsity $1$ throughout the whole procedure.

Furthermore, Assumption~\ref{assump:linear_no_resonance_gap} implies that for any $i, j_1,j_2,\dots,j_l \in [N]$, where $l \ge 2$, we have 
\begin{align}
|\lambda_i - \lambda_{j_1}-\lambda_{j_2}-\dots-\lambda_{j_l}|\ge \Delta.    
\end{align}
Thus, every entry of \( N_l \) has magnitude at most \( 1/\Delta \), i.e., 
\begin{align}
\norm{N_l}_{\mathrm{max}} \le \frac{1}{\Delta}, \quad \forall l \ge 2.
\end{align}
Meanwhile, multiplying a matrix by $\tilde{F}_2$ can increase the spectral norm of the matrix by at most $\norm{\tilde{F}_2}$. Then, by Lemma~\ref{lem:hadamard_prod_norm_bound}, we have that at each iteration of the procedure, the product of the spectral norms of the matrices for the subsystems increases by at most \( \frac{ \| \tilde{F}_2 \|}{\Delta} \). Consequently, after $m-1$ iterations of the procedure, the final operator $\tilde{f}(\mbt)$ has spectral norm at most \(\left( \frac{\| \tilde{F}_2 \|}{\Delta} \right)^{m-1} \), as desired.

Now, we consider the general case where $\tilde{F}_2$ is $s$-column-sparse for $s \ge 1$. We can write $\tilde{F}_2$ as the sum of $s$ matrices $\tilde{F}_2^{(1)}, \tilde{F}_2^{(2)}, \dots, \tilde{F}_2^{(s)}$,
where each $\tilde{F}_2^{(l)}$ is $1$-column-sparse. More precisely, we construct $\tilde{F}_2^{(l)}$ by allocating the $l$-th nonzero entry of each column of $\tilde{F}_2$ to this matrix, for $l \in [s]$. Note that the rows of $\tilde{F}_2^{(l)}$ are orthogonal to each other, and hence $\norm{\tilde{F}_2^{(l)}}$ equals the maximum norm of $\tilde{F}_2^{(l)}$'s rows, which is no larger than the maximum norm of $\tilde{F}_2$'s rows, which in turn is upper bounded by $\norm{\tilde{F}_2}$. Thus, we have $\norm{\tilde{F}_2^{(l)}} \le \norm{\tilde{F}_2}$ for all $l \in [s]$. 

Meanwhile, note that in the expression for $\tilde{f}(\mbt)$, $\tilde{F}_2$ occurs $m-1$ times. Since
\begin{align}
\tilde{F}_2= \tilde{F}_2^{(1)} + \tilde{F}_2^{(2)} + \dots + \tilde{F}_2^{(s)},
\end{align}
we can rewrite $\tilde{f}(\mbf)$ as the sum of $s^{m-1}$ terms, each of which is labeled by an integer tuple $(l_1,l_2,\dots,l_{m-1}) \in [s]^{m-1}$, i.e., the $p$-th $\tilde{F}_2$ in the expression for $\tilde{f}(\mbf)$ is replaced with $\tilde{F}_2^{l_p}$, for $p \in [s]$. Now, since each $\tilde{F}_2^{l_p}$ is $1$-column-sparse and has spectral norm at most $\norm{\tilde{F}_2}$, the spectral norm of each term in the expression for $\tilde{f}(\mbt)$ is at most \(\left( \frac{\| \tilde{F}_2 \|}{\Delta} \right)^{m-1} \). Then, since $\tilde{f}(\mbt)$ is the sum of $s^{m-1}$ such terms, we get that 
\begin{align}
\norm{\tilde{f}(\mbt)} \le \left( \frac{s\| \tilde{F}_2 \|}{\Delta} \right)^{m-1},    
\end{align}
as claimed.

This lemma is thus proved.
\end{proof}

To illustrate the idea behind the proof of Lemma~\ref{lem:vij_norm_bound_poincare_domain}, we consider the following example.
Suppose $\mbt$ is a binary tree with $m=3$ leaves, and 
\begin{align}
    \tilde{f}(\mbt)=N_3 \odot \lrcb{\tilde{F}_2 \lrsb{\lrb{N_2 \odot \tilde{F}_2} \otimes I}},    
\end{align}
where $\tilde{F}_2$ is $2$-column-sparse. Then we can construct $1$-column-sparse matrices $\tilde{F}_2^{(1)}$ and $\tilde{F}_2^{(2)}$ such that $\tilde{F}_2=\tilde{F}_2^{(1)}+\tilde{F}_2^{(2)}$, and  
$\norm{\tilde{F}_2^{(1)}}, \norm{\tilde{F}_2^{(2)}} \le \norm{F_2}$. It follows that 
\begin{align}
\tilde{f}(\mbt)&=N_3 \odot \lrcb{\tilde{F}_2^{(1)} \lrsb{\lrb{N_2 \odot \tilde{F}_2^{(1)}} \otimes I}}   \nonumber \\
&\quad +N_3 \odot \lrcb{\tilde{F}_2^{(1)} \lrsb{\lrb{N_2 \odot \tilde{F}_2^{(2)}} \otimes I}} \nonumber \\
&\quad +N_3 \odot \lrcb{\tilde{F}_2^{(2)} \lrsb{\lrb{N_2 \odot \tilde{F}_2^{(1)}} \otimes I}} \nonumber \\
&\quad +N_3 \odot \lrcb{\tilde{F}_2^{(2)} \lrsb{\lrb{N_2 \odot \tilde{F}_2^{(2)}} \otimes I}}.
\end{align}
Each term on the right-hand side of this equation is $1$-column-sparse and has spectral norm at most $\left( \frac{\| \tilde{F}_2 \|}{\Delta} \right)^{2}$. Consequently, we get $\norm{\tilde{f}(\mbt)} \le 4\left( \frac{\| \tilde{F}_2 \|}{\Delta} \right)^{2}$, in agreement with our predicted bound.

\begin{proof} [Proof of Lemma~\ref{lem:wij_norm_bound_poincare_domain}]
To establish this lemma, it suffices to prove that for any binary tree $\mbt$ with $m$ leaves, the operator $\tilde{g}(\mbt)$ satisfies
\begin{align}
\norm{\tilde{g}(\mbt)} \le \lrb{\frac{s\norm{\tilde{F}_2}}{\Delta}}^{m-1}.
\label{eq:bound_norm_gt_poincare_domain}
\end{align}
To see this, recall that by Proposition~\ref{prop:v_inv_ij_expression}, 
\begin{align}
    \lrb{\tilde{V}^{-1}}_{i,j} = (-1)^{j-i} \sum_{(\mbt_1,\mbt_2,\dots,\mbt_i) \in \mathcal{T}^i_j} \tilde{g}(\mbt_1) \otimes \tilde{g}(\mbt_2) \otimes \dots \otimes \tilde{g}(\mbt_i).   
\end{align}
For any $(\mbt_1,\mbt_2,\dots,\mbt_i) \in \mathcal{T}^i_j$, let $m_l$ be the number of leaves in $\mbt_l$, for $l \in [i]$. Then we have $m_1+m_2+\dots+m_i=j$. Consequently,  Eq.~\eqref{eq:bound_norm_gt_poincare_domain} implies
\begin{align}
\norm{\tilde{g}(\mbt_1) \otimes \tilde{g}(\mbt_2) \otimes \dots \otimes \tilde{g}(\mbt_i)} &=\norm{\tilde{g}(\mbt_1)}\norm{\tilde{g}(\mbt_2)} \dots \norm{\tilde{g}(\mbt_i)} \\
&\le \lrb{\frac{s\norm{\tilde{F}_2}}{\Delta}}^{\sum_{l=1}^i (m_l-1)} \\
&=\lrb{\frac{s\norm{\tilde{F}_2}}{\Delta}}^{j-i}.
\end{align}
Meanwhile, by Lemma~\ref{lem:count_binary_forests}, there are at most $4^{j-i}\binom{j-1}{i-1}$ binary forests consisting of $i$ binary trees with a total of $j$ leaves, i.e. $|\mathcal{T}^i_j|\le 4^{j-i}\binom{j-1}{i-1}$. Then, it follows that
\begin{align}
\norm{\lrb{\tilde{V}^{-1}}_{i,j}} \le \binom{j-1}{i-1}\lrb{\frac{4s\norm{\tilde{F}_2}}{\Delta}}^{j-i}, 
\end{align}    
as desired.

It remains to prove Eq.~\eqref{eq:bound_norm_gt_poincare_domain}. To this end, we need to introduce the following notation. Let $r$ be the root of $\mbt$, $l_1, l_2, \dots, l_m$ the leaves of $\mbt$, and 
$v_1,v_2,\dots,v_{m-1}$ the internal nodes of $\mbt$. Since $\tilde{F}_2$ is $s$-column-sparse, we can rewrite $\tilde{F}_2$ as the sum of $s$ matrices $\tilde{F}_2^{(1)}, \tilde{F}_2^{(2)}, \dots, \tilde{F}_2^{(s)}$, where each $\tilde{F}_2^{(l)}$ is $1$-column-sparse and has spectral norm no larger than that of  $\tilde{F}_2$, i.e. $\norm{\tilde{F}_2^{(l)}} \le \norm{\tilde{F}_2}$. Then for any $\vec{q}=(q_1,q_2,\dots,q_{m-1}) \in [s]^{m-1}$ and any labeling $h$ of $\mbt$, define
\begin{align}
\alpha_{\vec q}(\mbt, h) \defeq \prod_{i=1}^{m-1} \bra{h(v_i)} \tilde{F}_2^{(q_i)} \ket{h(c_1(v_i), h(c_2(v_i)))}.
\end{align}
Then define
\begin{align}
\tilde{g}_{\vec q}(\mbt) \defeq \sum_{h \in L(\mbt)} \alpha_{\vec q}(\mbt, h) \gamma(\mbt, h) 
\ket{h(r)}\bra{h(l_1), h(l_2),\dots, h(l_m)}.
\end{align}
Clearly, we have
\begin{align}
    \tilde{g}(\mbt)=\sum_{\vec q \in [s]^{m-1}} \tilde{g}_{\vec q}(\mbt).
\end{align}

Next, we will prove that 
\begin{align}
\tilde{g}_{\vec q}(\mbt) \le \lrb{\frac{\norm{\tilde{F}_2}}{\Delta}}^{m-1},~&~\forall \vec q \in [s]^{m-1}.    
\end{align}
Assuming this is true, we will get 
\begin{align}
    \norm{\tilde{g}(\mbt)}\le \sum_{\vec q \in [s]^{m-1}} \norm{\tilde{g}_{\vec q}(\mbt)} 
    \le \lrb{\frac{s\norm{\tilde{F}_2}}{\Delta}}^{m-1},
\end{align}
as claimed.

Fix arbitrary $\vec q \in [s]^{m-1}$. Since $\tilde{F}_2^{(1)}, \tilde{F}_2^{(2)}, \dots, \tilde{F}_2^{(s)}$ are $1$-column-sparse, we have that for each $\vec c= (c_1,c_2,\dots,c_m) \in [N]^m$, there exists at most one labeling $h$ of $\mbt$ that satisfies $h(l_1)=c_1$, $h(l_2)=c_2$, $\dots$, $h(l_m)=c_m$, and 
$\alpha_{\vec q}(\mbt, h) \neq 0$. In other words, for any labeling of the leaves of $\mbt$, there exists at most one corresponding labeling of its internal nodes that contributes a nonzero term to $\tilde{g}_{\vec q}(\mbt)$. Let $h_{\vec q, \vec c}$ denote this labeling of $\mbt$ \footnote{In fact, $h_{\vec q, \vec c}$ also depends on $\mbt$, but here we omit the subscript $\mbt$ for better readability.}. Then we can rewrite $\tilde{g}_{\vec q}(\mbt)$ as 
\begin{align}
    \tilde{g}_{\vec q}(\mbt) = \sum_{\vec c \in [N]^m} \alpha_{\vec q}(\mbt, h_{\vec q, \vec c}) \gamma(\mbt, h_{\vec q, \vec c}) \ket{h_{\vec q, \vec c}(r)}\bra{c_1,c_2,\dots,c_m}.
\end{align}
Note that this equation implies that $\tilde{g}_{\vec q}(\mbt)$ is $1$-column-sparse. 

In fact, we can derive an alternative expression for $\tilde{g}_{\vec q}(\mbt)$ as follows.
Define an operator $\omega_{\vec q}(\mbt)$ by traversing the nodes of $\mbt$ in a reverse topological order and assigning an operator to each node:
\begin{itemize}
    \item Start with the leaves of $\mbt$. Each leaf is assigned the operator $I$.
    \item For each internal node $v_l$ of $\mbt$, if its children $c_1(v_l)$ and $c_2(v_l)$ have been assigned operators $\omega_{\vec q}(c_1(v_l))$ and 
    $\omega_{\vec q}(c_2(v_l))$, respectively, then we assign the operator $$\tilde{F}_2^{(q_l)} \lrsb{\omega_{\vec q}(c_1(v_l)) \otimes \omega_{\vec q}(c_2(v_l))}$$ 
    to $v_l$.
    \item Finally, $\omega_{\vec q}(\mbt)$ equals the operator assigned to the root $r$ of $\mbt$.
\end{itemize}
Meanwhile, define  
\begin{align}
    M_{\vec q}(\mbt) \defeq  \sum_{\vec c \in [N]^m} \gamma(\mbt, h_{\vec q, \vec c}) \ket{h_{\vec q, \vec c}(r)}\bra{c_1,c_2,\dots,c_m}.
\end{align}
Then one can verify that
\begin{align}
    \tilde{g}_{\vec q}(\mbt) = \omega_{\vec q}(\mbt) \odot M_{\vec q}(\mbt).
\end{align}
Since $\tilde{F}_2^{(1)}, \tilde{F}_2^{(2)}, \dots, \tilde{F}_2^{(s)}$ are $1$-column-sparse and have spectral norms at most $\norm{\tilde{F}_2}$, we know that $\omega_{\vec q}(\mbt)$ is 
$1$-column-sparse and has spectral norm
\begin{align}
    \norm{\omega_{\vec q}(\mbt)} \le \norm{\tilde{F}_2^{(q_1)}}\norm{\tilde{F}_2^{(q_2)}}\dots \norm{\tilde{F}_2^{(q_{m-1})}} \le \norm{\tilde{F}_2}^{m-1}.
\end{align}
On the other hand, for any labeling $h$ of $\mbt$, we can bound the magnitude of $\gamma(\mbt, h)$ as follows. If $m=1$, then $\gamma(\mbt, h)=1$. Otherwise, we have
\begin{align}
\abs{\gamma(\mbt, h)} &\le  \sum_{(v_1,v_2,\dots,v_{m-1}) \in O(\mbt)} \prod_{l=1}^{m-1}\frac{1}{\abs{\sum_{w \in C(\lrcb{v_1,v_2,\dots,v_l})} \lambda_{h(w)} -\lambda_{h(r)}}} \\
&\le (m-2)! \cdot \frac{1}{\Delta} \cdot \frac{1}{2\Delta}  \cdot \dots \cdot 
\frac{1}{(m-1)\Delta} \\
&=\frac{1}{(m-1)\Delta^{m-1}},
\end{align}
where the second step follows from the following observations:
\begin{itemize}
    \item There are at most $(m-2)!$ valid topological orders of the $m-1$ internal nodes of $\mbt$ (since they always start with the root $r$), i.e. $|O(\mbt)| \le (m-2)!$;
    \item For any $(v_1,v_2,\dots,v_{m-1}) \in O(\mbt)$ and any $l \in [m-1]$, the induced subgraph of $\mbt$ on $\lrcb{v_1,v_2,\dots,v_l}$ is connected. Thus, we have
    $|C(\lrcb{v_1,v_2,\dots,v_l})|=l+1$;
    \item Then by Assumption~\ref{assump:linear_no_resonance_gap}, we have 
    \begin{align}
        \abs{\sum_{w \in C(\lrcb{v_1,v_2,\dots,v_l})} \lambda_{h(w)} -\lambda_{h(r)}} \ge l \Delta.
    \end{align}
\end{itemize}
Overall, we always have
\begin{align}
    \abs{\gamma(\mbt, h)} \le \frac{1}{\Delta^{m-1}}
\end{align}
regardless of the value of $m$. This implies that each entry of $M_{\vec q}(\mbt)$ has magnitude at most $\Delta^{-(m-1)}$.

Now since $\omega_{\vec q}(\mbt)$ is $1$-column-sparse and has spectral norm at most $\norm{\tilde{F}_2}^{m-1}$, and the entries of $M_{\vec q}(\mbt)$ have magnitudes at most $\Delta^{-(m-1)}$, we invoke Lemma~\ref{lem:hadamard_prod_norm_bound} to conclude that 
\begin{align}
\tilde{g}_{\vec q}(\mbt) = \omega_{\vec q}(\mbt) \odot M_{\vec q}(\mbt)    
\end{align}
has spectral norm at most $\lrb{\frac{\norm{\tilde{F}_2}}{\Delta}}^{m-1}$, as desired. This completes the proof of the lemma.
\end{proof}


\section{Details on the error analysis for systems in the Siegel domain}
\label{app:error_analysis_siegel_domain}

The analysis of the truncation error of the Carleman embedding for nonresonant systems in the Siegel domain is more challenging than in the Poincar\'{e} domain, primarily due to the different behavior of the integer combinations of the eigenvalues of $F_1$. Specifically, note that if the origin lies in the convex hull of $\lambda_1, \lambda_2, \dots, \lambda_N$, then we have
\begin{align}
    \lim_{m \to +\infty} \min_{i \in [N]} \min_{\substack{\alpha_j \in \myN_0 \\ 2 \le \sum_{j=1}^N  \alpha_j \le m}} \abs{\lambda_i-\sum_{j=1}^N \alpha_j\lambda_j} = 0.
\end{align}
In other words, the quantity $\abs{\lambda_i-\sum_{j=1}^N \alpha_j\lambda_j}$ does not necessarily increase as the $\alpha_j$'s become large. In fact, it can approach zero for suitably chosen $i$ and $\alpha_j$'s. To characterize the rate of this convergence, we introduce the following notion:

\begin{definition}
    A point $\vec \lambda=(\lambda_1, \lambda_2, \dots, \lambda_N) \in \myC^N$ is said to be of type $(C, \nu)$, where $C>0$ and $\nu \in \R$, if for all $i \in [N]$ and $\alpha_1, \alpha_2, \dots, \alpha_N \in \myN_0$ such that $\sum_{j=1}^N \alpha_j \ge 2$,
    we have
    \begin{align}
        \abs{\lambda_i - \sum_{j=1}^N \alpha_j \lambda_j} \ge C \lrb{\sum_{j=1}^N \alpha_j - 1}^{-\nu}.
    \end{align}    
\end{definition}

Note that Assumption~\ref{assump:linear_no_resonance_gap} in Section~\ref{subsubsec:error_analysis_poincare_domain} can be rephrased as stating that $(\lambda_1, \lambda_2, \dots, \lambda_N)$ is of type $(\Delta, -1)$. Furthermore, Lemma~\ref{lem:equivalence_linear_no_resonance_gap_poincare_domain} in Section~\ref{subsubsec:error_analysis_poincare_domain} implies that if $(\lambda_1, \lambda_2, \dots, \lambda_N)$ is nonresonant and belongs to the Poincar\'{e} domain, then it is of type $(\Delta, -1)$ for some $\Delta > 0$. In contrast, if 
$(\lambda_1, \lambda_2, \dots, \lambda_N)$ is nonresonant and lies in the Siegel domain, then they can only be of type $(\Delta, \nu)$ for some
$\Delta, \nu > 0$.

In what follows, we will assume that the following condition holds:
\begin{assumption}
   $(\lambda_1, \lambda_2, \dots, \lambda_N)$ is of type $(\Delta, \nu)$ for some $\Delta, \nu > 0$. 
   \label{assump:gap_condition_siegel_domain}
\end{assumption}

\subsection{Norm bounds for the blocks of \texorpdfstring{$\tilde{V}$ and $\tilde{V}^{-1}$}{diagonalizing matrix and its inverse}}

We now derive upper bounds on the spectral norms of the individual blocks of $\tilde{V}$ and $\tilde{V}^{-1}$, which will play a crucial role in the subsequent error analysis. The proofs of the following lemmas closely follow the arguments used in the proofs of Lemmas~\ref{lem:vij_norm_bound_poincare_domain} and \ref{lem:wij_norm_bound_poincare_domain}, respectively. The key distinction arises from the different behaviors of the no-resonance gaps associated with eigenvalues in the Poincar\'{e} and Siegel domains.

\begin{lemma}
    Under Assumptions~\ref{assump:gap_condition_siegel_domain}, supposing $\tilde{F}_2 = Q^{-1}F_2 Q^{\otimes 2}$ is $s$-column-sparse, we have 
    \begin{align}
        \norm{\tilde{V}_{i,j}} \le \lrsb{(j-i)!}^\nu\binom{j-1}{i-1} \lrb{\frac{4s\norm{\tilde{F}_2}}{\Delta}}^{j-i},~&~\forall~i \le j.
    \end{align}
    \label{lem:vij_norm_bound_siegel_domain}
\end{lemma}

\begin{proof}
If $i=j$, then $\tilde{V}_{j,j}=I^{\otimes j}$ and the claim  holds trivially. So we assume $i<j$ from now on.

To prove the desired bound on $\norm{\tilde{V}_{i,j}}$, it suffices to show that for any binary tree $\mbt$ with $m$ leaves, 
\begin{align}
 \norm{\tilde{f}(\mbt)} \le [(m-1)!]^{\nu} \lrb{\frac{s\norm{\tilde{F}_2}}{\Delta}}^{m-1}.   
 \label{eq:bound_on_f_tree_siegel_domain}
\end{align}
To see this, note that for any binary forest $(\mbt_1, \mbt_2, \dots, \mbt_i) \in \mathcal{T}^i_j$, let $m_1,m_2,\dots,m_i$ be the number of leaves in $\mbt_1, \mbt_2,\dots, \mbt_i$ respectively, so that $m_1+m_2+\dots+m_i=j$. Then
Eq.~\eqref{eq:bound_on_f_tree_siegel_domain} implies 
\begin{align}
\norm{\tilde{f}(\mbt_1)\otimes \tilde{f}(\mbt_2)\otimes \dots \otimes \tilde{f}(\mbt_i)}
&= \norm{\tilde{f}(\mbt_2)}\cdot
        \norm{\tilde{f}(\mbt_2)} \cdot
        \dots
            \norm{\tilde{f}(\mbt_i)}  \\
&\le 
[(m_1-1)!(m_2-1)!\dots (m_i-1)!]^{\nu} \lrb{\frac{s\norm{\tilde{F}_2}}{\Delta}}^{j-1} \\
&\le 
[(j-1)!]^\nu
\lrb{\frac{s\norm{\tilde{F}_2}}{\Delta}}^{j-1}.
\end{align}
It follows that
\begin{align}
\norm{\tilde{V}_{i,j}} &\le \sum_{(\mbt_1, \mbt_2, \dots, \mbt_i) \in \mathcal{T}^i_j}
    \norm{\tilde{f}(\mbt_1)\otimes \tilde{f}(\mbt_2)\otimes \dots \otimes \tilde{f}(\mbt_i)} \\ 
    &\le |\mathcal{T}^i_j| \cdot [(j-1)!]^\nu \lrb{\frac{s\norm{\tilde{F}_2}}{\Delta}}^{j-1} \\
    & \le \lrsb{(j-i)!}^\nu\binom{j-1}{i-1} \lrb{\frac{4s\norm{\tilde{F}_2}}{\Delta}}^{j-i},
\end{align}
where the last step follows from Lemma~\ref{lem:count_binary_forests}.

Let us first prove Eq.~\eqref{eq:bound_on_f_tree_poincare_domain} in the special case where $\tilde{F}_2$ is $1$-column-sparse. Recall that \( \tilde{f}(\mbt) \) is generated by the following process. We start with \( m \) independent subsystems, each assigned an \( I \) matrix, and then perform \( m-1 \) iterations of the following procedure: 
\begin{itemize} 
\item Apply \( \tilde{F}_2 \) to the tensor product of the matrices assigned to two adjacent subsystems, and combine these subsystems; 
\item Take the Hadamard product of the resulting matrix with \( N_l \), where \( l \) is the number of leaves in the combined subsystem. 
\end{itemize} 

Note that the standard and tensor products of two \(1\)-column-sparse matrices are both \(1\)-column-sparse. Moreover, the Hadamard product of a \(1\)-column-sparse matrix with any matrix is still \(1\)-column-sparse. Then since \( \tilde{F}_2 \) is \(1\)-column-sparse, the tensor product of the matrices for the subsystems has column-sparsity $1$ throughout the whole procedure.

Furthermore, Assumption~\ref{assump:linear_no_resonance_gap} implies that for any $i, j_1,j_2,\dots,j_l \in [N]$, where $l \ge 2$, we have 
\begin{align}
|\lambda_i - \lambda_{j_1}-\lambda_{j_2}-\dots-\lambda_{j_l}|\ge  \frac{\Delta}{(l-1)^{\nu}}.    
\end{align}
Thus, every entry of \( N_l \) has magnitude at most \( (l-1)^\nu/\Delta \), i.e., 
\begin{align}
\norm{N_l}_{\mathrm{max}} \le \frac{(l-1)^\nu}{\Delta}, \quad \forall l \ge 2 .     
\end{align}
Meanwhile, multiplying a matrix by $\tilde{F}_2$ can increase the spectral norm of the matrix by at most $\norm{\tilde{F}_2}$. Then using Lemma~\ref{lem:hadamard_prod_norm_bound}, we obtain 
\begin{align}
\norm{ \tilde{f}(\mbt)} \le    \prod_{u \in V_I(\mbt)} (|B(u)|-1)^\nu \cdot \lrb{ \frac{\| \tilde{F}_2 \|}{\Delta} }^{m-1},
\end{align}
where $V_I(\mbt)$ is the set of internal nodes of $\mbt$, and $B(u)$ is the set of $u$'s descendants which are also leaves of $\mbt$. 
One can prove that every binary tree $\mbt$ with $m$ leaves satisfies
\begin{align}
\prod_{u \in V_I(\mbt)} (|B(u)|-1)  \le (m-1)!, 
\end{align}
where the equality holds if and only if $\mbt$ has height $m-1$. Therefore, we have
\begin{align}
\norm{ \tilde{f}(\mbt)} \le  [(m-1)!]^\nu \cdot \lrb{ \frac{\| \tilde{F}_2 \|}{\Delta} }^{m-1},
\end{align}
as desired.

Now we consider the general case where $\tilde{F}_2$ is $s$-column-sparse for $s \ge 1$. We can write $\tilde{F}_2$ as the sum of $s$ matrices $\tilde{F}_2^{(1)}, \tilde{F}_2^{(2)}, \dots, \tilde{F}_2^{(s)}$,
where each $\tilde{F}_2^{(l)}$ is $1$-column-sparse. More precisely, we construct $\tilde{F}_2^{(l)}$ by allocating the $l$-th nonzero entry of each column of $\tilde{F}_2$ to this matrix, for $l \in [s]$. Note that the rows of $\tilde{F}_2^{(l)}$ are orthogonal to each other, and hence $\norm{\tilde{F}_2^{(l)}}$ equals the maximum norm of $\tilde{F}_2^{(l)}$'s rows, which is no larger than the maximum norm of $\tilde{F}_2$'s rows, which in turn is upper bounded by $\norm{\tilde{F}_2}$. Thus, we have $\norm{\tilde{F}_2^{(l)}} \le \norm{\tilde{F}_2}$ for all $l \in [s]$. 

Meanwhile, note that in the expression for $\tilde{f}(\mbt)$, $\tilde{F}_2$ occurs $c-1$ times. Since
\begin{align}
\tilde{F}_2= \tilde{F}_2^{(1)} + \tilde{F}_2^{(2)} + \dots + \tilde{F}_2^{(s)},
\end{align}
we can rewrite $\tilde{f}(\mbf)$ as the sum of $s^{m-1}$ terms, each of which is labeled by an integer tuple $(l_1,l_2,\dots,l_{m-1}) \in [s]^{m-1}$, i.e. the $p$-th $\tilde{F}_2$ in the expression for $\tilde{f}(\mbf)$ is replaced with $\tilde{F}_2^{l_p}$, for $p \in [s]$. Now since each $\tilde{F}_2^{l_p}$ is $1$-column-sparse and has spectral norm at most $\norm{\tilde{F}_2}$, the spectral norm of each term in the expression for $\tilde{f}(\mbt)$ is at most \([(m-1)!]^\nu \left( \frac{\| \tilde{F}_2 \|}{\Delta} \right)^{m-1} \). Then since $\tilde{f}(\mbt)$ is the sum of $s^{m-1}$ such terms, we get that 
$\norm{\tilde{f}(\mbt)} \le [(m-1)!]^\nu  \left( \frac{s\| \tilde{F}_2 \|}{\Delta} \right)^{m-1}$, as claimed.

This lemma is thus proved.
\end{proof}

\begin{lemma}
Under Assumptions~\ref{assump:gap_condition_siegel_domain},
supposing $\tilde{F}_2 = Q^{-1}F_2 Q^{\otimes 2}$ is $s$-column-sparse, we have
\begin{align}
\norm{\lrb{\tilde{V}^{-1}}_{a,b}} \le \lrsb{(b-a)!}^{1+\nu}\binom{b-1}{a-1}\lrb{\frac{4s\norm{\tilde{F}_2}}{\Delta}}^{b-a},~&~\forall a\le b.
\end{align}
\label{lem:wij_norm_bound_siegel_domain}    
\end{lemma}

\begin{proof}
To establish this lemma, it suffices to prove that for any binary tree $\mbt$ with $m$ leaves, the operator $\tilde{g}(\mbt)$ satisfies
\begin{align}
\norm{\tilde{g}(\mbt)} \le [(m-1)!]^{1+\nu}\lrb{\frac{s\norm{\tilde{F}_2}}{\Delta}}^{m-1}.
\label{eq:bound_norm_gt_siegel_domain}
\end{align}
To see this, recall that by Proposition~\ref{prop:v_inv_ij_expression}, 
\begin{align}
    \lrb{\tilde{V}^{-1}}_{a,b} = (-1)^{b-a} \sum_{(\mbt_1,\mbt_2,\dots,\mbt_a) \in \mathcal{T}^a_b} \tilde{g}(\mbt_1) \otimes \tilde{g}(\mbt_2) \otimes \dots \otimes \tilde{g}(\mbt_a).        
\end{align}
For any $(\mbt_1,\mbt_2,\dots,\mbt_a) \in \mathcal{T}^a_b$, let $m_l$ be the number of leaves in $\mbt_l$, for $l \in [a]$. Then we have $m_1+m_2+\dots+m_a=b$. Consequently,  Eq.~\eqref{eq:bound_norm_gt_siegel_domain} implies
\begin{align}
\norm{\tilde{g}(\mbt_1) \otimes \tilde{g}(\mbt_2) \otimes \dots \otimes \tilde{g}(\mbt_a)}
&=\norm{\tilde{g}(\mbt_1)}\norm{\tilde{g}(\mbt_2)} \dots \norm{\tilde{g}(\mbt_a)} \\
&\le [(m_1-1)!(m_2-1)!\dots(m_a-1)!]^{1+\nu} \cdot
\lrb{\frac{s\norm{\tilde{F}_2}}{\Delta}}^{\sum_{l=1}^a (m_l-1)} \\
&=[(m_1-1)!(m_2-1)!\dots(m_a-1)!]^{1+\nu}\cdot \lrb{\frac{s\norm{\tilde{F}_2}}{\Delta}}^{b-a} \\
&\le [(b-a)!]^{1+\nu}\cdot \lrb{\frac{s\norm{\tilde{F}_2}}{\Delta}}^{b-a}.
\end{align}
Meanwhile, by Lemma~\ref{lem:count_binary_forests}, there are at most $4^{b-a}\binom{b-1}{a-1}$ binary forests consisting of $a$ binary trees with a total of $b$ leaves, i.e. $|\mathcal{T}^a_b|\le 4^{b-a}\binom{b-1}{a-1}$. Then it follows that
\begin{align}
\norm{\lrb{\tilde{V}^{-1}}_{a,b}} \le [(b-a)!]^{1+\nu}\binom{b-1}{a-1}\lrb{\frac{4s\norm{\tilde{F}_2}}{\Delta}}^{b-a}, 
\end{align}    
as desired.

It remains to prove Eq.~\eqref{eq:bound_norm_gt_siegel_domain}. To this end, we reuse the notation introduced in the proof of Lemma~\ref{lem:wij_norm_bound_poincare_domain}. Let $r$ be the root of $\mbt$, $l_1, l_2, \dots, l_m$ the leaves of $\mbt$, and 
$v_1,v_2,\dots,v_{m-1}$ the internal nodes of $\mbt$. Since $\tilde{F}_2$ is $s$-column-sparse, we can rewrite $\tilde{F}_2$ as the sum of $s$ matrices $\tilde{F}_2^{(1)}, \tilde{F}_2^{(2)}, \dots, \tilde{F}_2^{(s)}$, where each $\tilde{F}_2^{(l)}$ is $1$-column-sparse and has spectral norm no larger than that of  $\tilde{F}_2$, i.e. $\norm{\tilde{F}_2^{(l)}} \le \norm{\tilde{F}_2}$. Then for any $\vec{j}=(j_1,j_2,\dots,j_{m-1}) \in [s]^{m-1}$ and any labeling $h$ of $\mbt$, let
\begin{align}
\alpha_{\vec j}(\mbt, h) = \prod_{i=1}^{m-1} \bra{h(v_i)} \tilde{F}_2^{(j_i)} \ket{h(c_1(v_i), h(c_2(v_i)))}
\end{align}
and
\begin{align}
\tilde{g}_{\vec j}(\mbt) = \sum_{h \in L(\mbt)} \alpha_{\vec j}(\mbt, h) \gamma(\mbt, h) 
\ket{h(r)}\bra{h(l_1), h(l_2),\dots, h(l_m)}.
\end{align}
Then we have
\begin{align}
    \tilde{g}(\mbt)=\sum_{\vec j \in [s]^{m-1}} \tilde{g}_{\vec j}(\mbt).
\end{align}

Next, we will prove that 
\begin{align}
\tilde{g}_{\vec j}(\mbt) \le [(b-a)!]^{1+\nu}\lrb{\frac{\norm{\tilde{F}_2}}{\Delta}}^{m-1},~&~\forall \vec j \in [s]^{m-1}.    
\end{align}
Assuming this is true, we will get 
\begin{align}
    \norm{\tilde{g}(\mbt)}\le \sum_{\vec j \in [s]^{m-1}} \norm{\tilde{g}_{\vec j}(\mbt)} 
    \le [(b-a)!]^{1+\nu}\lrb{\frac{s\norm{\tilde{F}_2}}{\Delta}}^{m-1},
\end{align}
as claimed.

Fix arbitrary $\vec j \in [s]^{m-1}$. Since $\tilde{F}_2^{(1)}, \tilde{F}_2^{(2)}, \dots, \tilde{F}_2^{(s)}$ are $1$-column-sparse, we have that for each $\vec c= (c_1,c_2,\dots,c_m) \in [N]^m$, there exists at most one labeling $h$ of $\mbt$ that satisfies $h(l_1)=c_1$, $h(l_2)=c_2$, $\dots$, $h(l_m)=c_m$, and 
$\alpha_{\vec j}(\mbt, h) \neq 0$. In other words, for any labeling of the leaves of $\mbt$, there exists at most one corresponding labeling of its internal nodes that contributes a nonzero term to $\tilde{g}_{\vec j}(\mbt)$. Let $h_{\vec j, \vec c}$ denote this labeling of $\mbt$ \footnote{Although $h_{\vec j, \vec c}$ also depends on $\mbt$, we omit the subscript $\mbt$ for better readability.}. Then we can rewrite $\tilde{g}_{\vec j}(\mbt)$ as 
\begin{align}
    \tilde{g}_{\vec j}(\mbt) = \sum_{\vec c \in [N]^m} \alpha_{\vec j}(\mbt, h_{\vec j, \vec c}) \gamma(\mbt, h_{\vec j, \vec c}) \ket{h_{\vec j, \vec c}(r)}\bra{c_1,c_2,\dots,c_m}.
\end{align}
Note that this equation implies that $\tilde{g}_{\vec j}(\mbt)$ is $1$-column-sparse. 

In fact, we can derive an alternative expression for $\tilde{g}_{\vec j}(\mbt)$ as follows.
Define an operator $\omega_{\vec j}(\mbt)$ by traversing the nodes of $\mbt$ in a reverse topological order and assigning an operator to each node:
\begin{itemize}
    \item Start with the leaves of $\mbt$. Each leaf is assigned the operator $I$.
    \item For each internal node $v_i$ of $\mbt$, if its children $c_1(v_i)$ and $c_2(v_i)$ have been assigned operators $\omega_{\vec j}(c_1(v_i))$ and 
    $\omega_{\vec j}(c_2(v_i))$, respectively, then we assign the operator $$\tilde{F}_2^{(j_i)} \lrsb{\omega_{\vec j}(c_1(v_i)) \otimes \omega_{\vec j}(c_2(v_i))}$$ 
    to $v_i$.
    \item Finally, $\omega_{\vec j}(\mbt)$ equals the operator assigned to the root $r$ of $\mbt$.
\end{itemize}
Meanwhile, define  
\begin{align}
    M_{\vec j}(\mbt) = \sum_{\vec c \in [N]^m} \gamma(\mbt, h_{\vec j, \vec c}) \ket{h_{\vec j, \vec c}(r)}\bra{c_1,c_2,\dots,c_m}.
\end{align}
Then one can verify that
\begin{align}
    \tilde{g}_{\vec j}(\mbt) = \omega_{\vec j}(\mbt) \odot M_{\vec j}(\mbt).
\end{align}
Since $\tilde{F}_2^{(1)}, \tilde{F}_2^{(2)}, \dots, \tilde{F}_2^{(s)}$ are $1$-column-sparse and have spectral norms at most $\norm{\tilde{F}_2}$, we know that $\omega_{\vec j}(\mbt)$ is 
$1$-column-sparse and has spectral norm
\begin{align}
    \norm{\omega_{\vec j}(\mbt)} \le \norm{\tilde{F}_2^{(j_1)}}\norm{\tilde{F}_2^{(j_2)}}\dots \norm{\tilde{F}_2^{(j_{m-1})}} \le \norm{\tilde{F}_2}^{m-1}.
\end{align}
On the other hand, for any labeling $h$ of $\mbt$, we can bound the magnitude of $\gamma(\mbt, h)$ as follows. If $m=1$, then $\gamma(\mbt, h)=1$. Otherwise, we have
\begin{align}
\abs{\gamma(\mbt, h)} &\le  \sum_{(v_1,v_2,\dots,v_{m-1}) \in O(\mbt)} \prod_{l=1}^{m-1}\frac{1}{\abs{\sum_{w \in C(\lrcb{v_1,v_2,\dots,v_l})} \lambda_{h(w)} -\lambda_{h(r)}}} \\
&\le (m-2)! \cdot \frac{1}{\Delta} \cdot \frac{2^{\nu}}{\Delta}  \cdot \dots \cdot 
\frac{(m-1)^{\nu}}{\Delta} \\
&\le \frac{[(m-1)!]^{1+\nu}}{\Delta^{m-1}},
\end{align}
where the second step follows from the following observations:
\begin{itemize}
    \item There are at most $(m-2)!$ valid topological orders of the $m-1$ internal nodes of $\mbt$ (since they always start with the root $r$), i.e. $|O(\mbt)| \le (m-2)!$;
    \item For any $(v_1,v_2,\dots,v_{m-1}) \in O(\mbt)$ and any $l \in [m-1]$, the induced subgraph of $\mbt$ on $\lrcb{v_1,v_2,\dots,v_l}$ is connected. Thus, we have
    $|C(\lrcb{v_1,v_2,\dots,v_l})|=l+1$;
    \item Then by Assumption~\ref{assump:gap_condition_siegel_domain}, we have 
    \begin{align}
        \abs{\sum_{w \in C(\lrcb{v_1,v_2,\dots,v_l})} \lambda_{h(w)} -\lambda_{h(r)}} \ge \Delta l^{-\nu}.
    \end{align}
\end{itemize}
Overall, we always have
\begin{align}
    \abs{\gamma(\mbt, h)} \le \frac{[(m-1)!]^{1+\nu}}{\Delta^{m-1}}
\end{align}
regardless of the value of $m$. This implies that each entry of $M_{\vec j}(\mbt)$ has magnitude at most $[(m-1)!]^{1+\nu}\Delta^{-(m-1)}$.

Now since $\omega_{\vec j}(\mbt)$ is $1$-column-sparse and has spectral norm at most $\norm{\tilde{F}_2}^{m-1}$, and the entries of $M_{\vec j}(\mbt)$ have magnitudes at most $[(m-1)!]^{1+\nu}\Delta^{-(m-1)}$, we invoke Lemma~\ref{lem:hadamard_prod_norm_bound} to conclude that 
\begin{align}
\tilde{g}_{\vec j}(\mbt) = \omega_{\vec j}(\mbt) \odot M_{\vec j}(\mbt)    
\end{align}
has spectral norm at most $[(m-1)!]^{1+\nu}\lrb{\frac{\norm{\tilde{F}_2}}{\Delta}}^{m-1}$, as desired. This completes the proof of the lemma.    
\end{proof}

\subsection{Carleman error bound}

With the established bounds on the norms of the blocks of $\tilde{V}$ and $\tilde{V}^{-1}$, we are now prepared to derive an upper bound on the error of truncated Carleman linearization for the system under consideration. 

\begin{thm}
[Carleman error bound for nonresonant systems in the Siegel domain]
\label{thm:non-resonant_Siegel}
Let $x(t)$ be the solution to the $N$-dimensional quadratic ODE system
\begin{align}
    \dot{x}(t) = F_1 x(t) + F_2 x(t)^{\ot 2},
\end{align}
and
\begin{align}
    \dot{y}(t) = A y(t),\qquad y(0)=[x(0),x(0)^{\ot 2},\dots,x(0)^{\ot k}]^T,
\end{align}
be the Carleman linearization of the system truncated at level $k$, see Eq.~\eqref{eq:CarlemanODE}. Assume $F_1=Q\Lambda Q^{-1}$, where $\Lambda=\diag{\lambda_1, \lambda_2, \dots, \lambda_N}$ with  $\realpart{\lambda_i} \le 0$ for all $i \in [N]$.
Suppose  $(\lambda_1, \lambda_2, \dots, \lambda_N)$ is of type $(\Delta, \nu)$ for some $\Delta, \nu > 0$, and $\tilde{F}_2 = Q^{-1}F_2 Q^{\otimes 2}$ is $s$-column-sparse. 
Let $\Xi_{\nu}(n)= \sum_{l=0}^n [(n-l)!]^{\nu}(l!)^{\nu-1}$, for $n \in \myN_0$, and $\norm{\tilde{x}_{\rm max}}=\max_{t \in [0,T]} \norm{Q^{-1}x(t)}$. Then the Carleman error vectors $\eta_i(t) = x(t)^{\ot i}-y^{[i]}(t)$ satisfy 
\begin{align}
\norm{\eta_i(t)} \le \frac{k! \cdot \Xi_{\nu}(k-i) }{(i-1)!} \cdot t  \norm{Q}^i \norm{\tilde{F}_2} \norm{x_{\rm max}}^{i+1}
\lrb{\frac{4 s \norm{\tilde{F}_2} \norm{\tilde{x}_{\rm max}}}{\Delta}}^{k-i},
\label{eq:error_bound_ith_block_siegel_domain}
\end{align}
for all $i\in [k]$ and $t \in [0, T]$. In particular, for $i=1$, we have
\begin{align}
\norm{\eta_1(t)} 
\le k! \cdot \Xi_{\nu}(k-1)  \cdot  t
\norm{Q} \norm{\tilde{F}_2} \norm{\tilde{x}_{\rm max}}^2 \lrb{\frac{4 s  \norm{\tilde{F}_2} \norm{\tilde{x}_{\rm max}}}{\Delta}}^{k-1},    
\label{eq:error_bound_1st_block_siegel_domain}
\end{align}
for all $t \in [0, T]$.
\label{thm:error_bound_nonresonant_siegel_domain}
\end{thm}

\begin{proof}
    The proof of this theorem parallels that of  Theorem~\ref{thm:error_bound_nonresonant_poincare_domain}, but it relies on the bounds for $\norm{\tilde{V}_{i,j}}$ and $\norm{\lrb{\tilde{V}^{-1}}_{i,j}}$ provided by Lemmas \ref{lem:vij_norm_bound_siegel_domain} and \ref{lem:wij_norm_bound_siegel_domain}, respectively, which differ from those used in the Poincar\'{e} domain case. 
    
    Recall that $\tilde{x}(t)=Q^{-1}x(t)$ is governed by the ODE system:
    \begin{align}
    \dot{\tilde{x}}(t)  = \Lambda \tilde{x}(t) + \tilde{F}_2 \tilde{x}(t)^{\otimes 2},   
    \end{align}
    where $\tilde{F}_2 = Q^{-1}F_2 Q^{\otimes 2}$. The order-$k$ Carleman linearization for this system is  
    \begin{align}
    \dot{\tilde{y}}(t) =\tilde{A}\tilde{y}(t)    
    \end{align}
     where $\tilde{y}(t)=[\tilde{y}^{[1]}(t), \tilde{y}^{[2]}(t),\dots,\tilde{y}^{[k]}(t)]^T$ with $\tilde{y}^{[i]}(t) \in \myC^{N^i}$ for $i \in [k]$. We have found $\tilde{V}$ such that $\tilde{A}=\tilde{V}D\tilde{V}^{-1}$, where $D=\diag{\Lambda_1, \Lambda_2, \dots, \Lambda_k}$, 
    with
    \begin{align}
    \Lambda_i=\diag{\lambda_{j_1}+\lambda_{j_2}+\dots+\lambda_{j_i}:~j_1,j_2,\dots,j_i \in [N]}, ~~~\forall i \in [k].
    \end{align}
    We will first establish an upper bound on $\norm{\tilde{x}(t)^{\otimes i}-\tilde{y}^{[i]}(t)}$, and then convert it into an upper bound on $\norm{x(t)^{\otimes i}-y^{[i]}(t)}$, for each $i \in [k]$.

    Let $\tilde{b}(t)=[\tilde{b}_1(t),\tilde{b}_2(t),\dots,\tilde{b}_k(t)]^T$, where $\tilde{b}_i(t)=\vec 0 \in \myC^{N^i}$ for $i \in [k-1]$, and 
    \begin{align}
    \tilde{b}_k(t)=\tilde{A}_{k,k+1}\tilde{x}(t)^{\otimes (k+1)}.    
    \end{align}
    Moreover, let $\tilde{\eta}_i(t)=\tilde{x}(t)^{\otimes i}-\tilde{y}^{[i]}(t)$ for $i \in [k]$, and let $\tilde{\eta}(t)=[\tilde{\eta}_1(t), \tilde{\eta}_2(t), \dots, \tilde{\eta}_k(t)]^T$. Then $\tilde{\eta}(t)$ is governed by the ODE system:
    \begin{align}
        \dot{\tilde{\eta}}(t) = {\tilde{A} \tilde{\eta}(t)+ \tilde{b}(t)}.
    \end{align}
    As a consequence, we have
    \begin{align}
    \tilde{\eta}(t) =\int_0^t \tilde{V}e^{D(t-\tau)}\tilde{V}^{-1}\tilde{b}(\tau)d\tau.
    \end{align}
    In other words, for each $i \in [k]$, we have
    \begin{align}
    \tilde{\eta}_i(t) = \sum_{j=i}^k \int_0^t \tilde{V}_{i,j} e^{\Lambda_j(t-\tau)} \lrb{\tilde{V}^{-1}}_{j,k} \tilde{b}_k(\tau) 
    d\tau.
    \end{align}
    It follows that
    \begin{align}
    \norm{\tilde{\eta}_i(t)} \le \sum_{j=i}^k \int_0^t \norm{\tilde{V}_{i,j}} \norm{e^{\Lambda_j(t-\tau)}} \norm{\lrb{\tilde{V}^{-1}}_{j,k}} \norm{\tilde{b}_k(\tau)} 
    d\tau.
    \end{align}
    
    Note that $\Lambda_j$ is a diagonal matrix whose diagonal entries have non-positive real parts. Thus, we have $\norm{e^{\Lambda_j r}} \le 1$ for all $r \in \R$. Then using the definition of $\tilde{A}_{k, k+1}$,  $\tilde{x}(t)$ and $\norm{\tilde{x}_{\rm max}}$, we obtain
    \begin{align}
        \norm{\tilde{b}_k(\tau)} \le \norm{\tilde{A}_{k, k+1}} \norm{\tilde{x}(\tau)}^{k+1}
        \le k \norm{\tilde{F}_2} \norm{\tilde{x}_{\rm max}}^{k+1}, ~~~\forall \tau \in [0, T]. 
    \end{align}
    Finally, since the given system satisfies Assumptions~\ref{assump:gap_condition_siegel_domain} and $\tilde{F}_2$ is $s$-column-sparse, Lemma~\ref{lem:vij_norm_bound_siegel_domain} provides an upper bound on $\norm{\tilde{V}_{i,j}}$:
    \begin{align}
       \norm{\tilde{V}_{i,j}} \le [(j-i)!]^\nu \binom{j-1}{i-1} \lrb{\frac{4s\norm{\tilde{F}_2}}{\Delta}}^{j-i},
    \end{align}
    while Lemma~\ref{lem:wij_norm_bound_siegel_domain} provides an upper bound on $\norm{\lrb{\tilde{V}^{-1}}_{j,k}}$:
    \begin{align}
        \norm{\lrb{\tilde{V}^{-1}}_{j,k}} \le [(k-j)!]^{1+\nu}\binom{k-1}{j-1}\lrb{\frac{4s\norm{\tilde{F}_2}}{\Delta}}^{k-j}.
    \end{align}
    Combining the above facts yields
    \begin{align}
    \norm{\tilde{\eta}_i(t)} &\le \sum_{j=i}^k \int_0^t \norm{\tilde{V}_{i,j}} \norm{\lrb{\tilde{V}^{-1}}_{j,k}} \norm{\tilde{b}_k(\tau)} 
    d\tau \\
    & \le \sum_{j=i}^k [(j-i)!]^\nu [(k-j)!]^{1+\nu} \cdot
    k t \norm{\tilde{F}_2} \norm{\tilde{x}_{\rm max}}^{k+1}
        \binom{k-1}{j-1}\binom{j-1}{i-1}\lrb{\frac{4s\norm{\tilde{F}_2}}{\Delta}}^{k-i} \\
    & =  \frac{k! \cdot \Xi_{\nu}(k-i)}{(i-1)!} 
        \cdot t \norm{\tilde{F}_2} \norm{\tilde{x}_{\rm max}}^{k+1} \lrb{\frac{4 s\norm{\tilde{F}_2}}{\Delta}}^{k-i},
    \label{eq:tilde_eta_i_bound_siegel_domain}
    \end{align}
    where 
    \begin{align}
    \Xi_{\nu}(n)=\sum_{l=0}^n [(n-l)!]^{\nu}(l!)^{\nu-1},       
    \end{align}
    for all $i \in [k]$ and $t \in [0, T]$.
    
    Now since $x(t)=Q\tilde{x}(t)$ and $y^{[i]}(t)=Q^{\otimes i} \tilde{y}^{[i]}(t)$, we have ${\eta}_i(t)=x(t)^{\otimes i}-{y}^{[i]}(t)=
    Q^{\otimes i} \tilde{\eta}_i(t)$, for all $i \in [k]$ and $t \ge 0$. Then, by Eq.~\eqref{eq:tilde_eta_i_bound_siegel_domain}, we obtain
    \begin{align}
    \norm{\eta_i(t)}
    &\le \norm{Q}^{i} \norm{\tilde{\eta}_i(t)} \\
    &\le 
    \frac{k! \cdot \Xi_{\nu}(k-i) }{(i-1)!} 
    \cdot t \norm{Q}^i \norm{\tilde{F}_2}  \norm{\tilde{x}_{\rm max}}^{k+1} \lrb{\frac{4 s\norm{\tilde{F}_2}}{\Delta}}^{k-i} \\
    &= \frac{k! \cdot \Xi_{\nu}(k-i) }{(i-1)!}  
    \cdot t \norm{Q}^i \norm{\tilde{F}_2} \norm{\tilde{x}_{\rm max}}^{i+1}\lrb{\frac{4 s\norm{\tilde{F}_2} \norm{\tilde{x}_{\rm max}}}{\Delta}}^{k-i},
    \end{align}
    for all $i \in [k]$ and $t \in [0, T]$. This completes the proof of the theorem.   

\end{proof}

\subsection{Technical obstacles for Carleman convergence in the Siegel domain}

In contrast to Theorem~\ref{thm:error_bound_nonresonant_poincare_domain} for the case where the spectrum of $F_1$ lies in the Poincar\'{e} domain, Theorem~\ref{thm:error_bound_nonresonant_siegel_domain} does not yield a convergence result for the case where the spectrum of $F_1$ belongs to the Siegel domain. Specifically, it does not imply that $\norm{x(t)-y^{[1]}(t)} \to 0$ as $k \to +\infty$, even when $\norm{\tilde{F}_2}$ is arbitrarily small. The underlying issue is that the error bound in  Eq.~\eqref{eq:error_bound_1st_block_siegel_domain} contains the factor $k! \cdot \Xi_{\nu}(k-1)$, which grows faster than $C^k$ for any constant $C>0$. As a result, this term eventually dominates the bound on $\norm{x(t) - y^{[1]}(t)}$, which cannot be made arbitrarily small by simply increasing $k$. It remains unclear whether the error bound in Theorem~\ref{thm:error_bound_nonresonant_siegel_domain} can be significantly improved and whether a convergence result can be established for all nonresonant systems in the Siegel domain.

Let us investigate more closely why the error bound in Theorem~\ref{thm:error_bound_nonresonant_siegel_domain} for the Siegel domain case is worse than the corresponding bound in Theorem~\ref{thm:error_bound_nonresonant_poincare_domain} for the Poincar\'{e} domain case. Although the proofs of both theorems follow the same strategy and rely critically on upper bounds for $\norm{\tilde{V}_{i,j}}$ and $\norm{(\tilde{V}^{-1})_{a,b}}$, the resulting error  estimates differ substantially in magnitude.

A comparison of Lemma~\ref{lem:vij_norm_bound_poincare_domain} and Lemma~\ref{lem:vij_norm_bound_siegel_domain} reveals that the upper bound in the Siegel case is worse by a factor of $[(j - i)!]^{\nu}$. Similarly, the upper bound for $\norm{(\tilde{V}^{-1})_{a,b}}$ in Lemma~\ref{lem:wij_norm_bound_siegel_domain} exceeds its counterpart in Lemma~\ref{lem:wij_norm_bound_poincare_domain} by a factor of $[(b - a)!]^{1 + \nu}$.

The key reason for these differences lies in the distinct behavior of nonresonance gaps in the two domains. If $(\lambda_1, \lambda_2, \dots, \lambda_N)$ lies in the Poincar\'{e} domain, then
\begin{align}
\left| \lambda_i - \lambda_{j_1} - \lambda_{j_2} - \dots - \lambda_{j_l} \right| \ge \Delta(l - 1),
\end{align}
which implies
\begin{align}
\left| \frac{1}{\lambda_i - \lambda_{j_1} - \lambda_{j_2} - \dots - \lambda_{j_l}} \right| \le \frac{1}{\Delta(l - 1)},
\end{align}
a quantity that decreases as $l$ increases. In contrast, if $(\lambda_1, \lambda_2, \dots, \lambda_N)$ belongs to the Siegel domain, then, in the worst case, we might have
\begin{align}
\left| \lambda_i - \lambda_{j_1} - \lambda_{j_2} - \dots - \lambda_{j_l} \right| \approx \frac{\Delta}{(l - 1)^{\nu}},
\end{align}
and hence
\begin{align}
\left| \frac{1}{\lambda_i - \lambda_{j_1} - \lambda_{j_2} - \dots - \lambda_{j_l}} \right| \approx \frac{(l - 1)^{\nu}}{\Delta},
\label{eq:worst_scenario_nonresonance_gap}
\end{align}
a quantity that increases as $l$ grows. This behavior leads to rapidly growing contributions in the Siegel case. However, it is important to emphasize that such worst-case denominators arise only for specific index choices.

Meanwhile, every entry of $\tilde{V}_{i,j}$ and $(\tilde{V}^{-1})_{a,b}$ can be expressed as a sum of multiple terms, where each term is a product of certain entries of $\tilde{F}_2$ and factors of the form $\frac{1}{\lambda_i - \lambda_{j_1} - \lambda_{j_2} - \dots - \lambda_{j_l}}$, for suitable index tuples $(i, j_1, j_2, \dots, j_l)$. More precisely, each term in these expressions takes the form
\[
\pm \prod_{l=1}^m \bra{a_l}\tilde{F}_2\ket{b_l, c_l} \cdot 
\prod_{l=1}^m \frac{1}{\lambda_{i_l} - \sum_{j=1}^N \alpha_j^{(l)} \lambda_j},
\]
for appropriate choices of indices $a_l$, $b_l$, $c_l$, $i_l$, and coefficients $\alpha_j^{(l)}$.

If $(\lambda_1, \lambda_2, \dots, \lambda_N)$ lies in the Siegel domain, then the denominator $\left|\lambda_{i_l} - \sum_{j=1}^N \alpha_j^{(l)} \lambda_j\right|$ can become very small for certain choices of $i_l$ and $\alpha_j^{(l)}$, leading to large individual terms in the sum. However, such problematic contributions can be effectively eliminated by setting the corresponding entries of $\tilde{F}_2$ to zero. That is, if the product $\prod_{l=1}^m \bra{a_l}\tilde{F}_2\ket{b_l, c_l}$ vanishes, then the entire term is removed from the expansion of the corresponding matrix entry.

Consequently, by imposing appropriate structural conditions on the system’s nonlinearity -- specifically, on the sparsity pattern of $\tilde{F}_2$ -- we can eliminate the terms responsible for poor scaling and thereby obtain significantly tighter bounds on the truncation error than those provided in Theorem~\ref{thm:error_bound_nonresonant_siegel_domain}. Following this strategy, we establish an improved error bound for certain nonresonant systems in the Siegel domain in Theorem~\ref{thm:error_bound_nonresonant_siegel_domain_decompose_f2}. We note that other subsets of systems in the Siegel domain may likewise admit strengthened bounds through similar techniques.

\section{Proof of Theorem~\ref{thm:error_bound_nonresonant_siegel_domain_decompose_f2}}
\label{app:error_bound_nonresonant_siegel_domain_decompose_f2}

Recall that $\tilde{x}(t)=Q^{-1}x(t)$ is governed by the ODE system:
\begin{align}
\dot{\tilde{x}}(t) = \Lambda \tilde{x}(t) + \tilde{F}_2 \tilde{x}(t)^{\otimes 2},   
\end{align}
where $\tilde{F}_2 = Q^{-1}F_2 Q^{\otimes 2}$. The order-$k$ Carleman linearization for this system is  
\begin{align}
\dot{\tilde{y}}(t) =\tilde{A}\tilde{y}(t)    
\end{align}
 where $\tilde{y}(t)=[\tilde{y}^{[1]}(t), \tilde{y}^{[2]}(t),\dots,\tilde{y}^{[k]}(t)]^T$ with $\tilde{y}^{[i]}(t) \in \myC^{N^i}$ for $i \in [k]$. We have found $\tilde{V}$ such that $\tilde{A}=\tilde{V}D\tilde{V}^{-1}$, where $D=\diag{\Lambda_1, \Lambda_2, \dots, \Lambda_k}$, 
with
\begin{align}
\Lambda_i=\diag{\lambda_{j_1}+\lambda_{j_2}+\dots+\lambda_{j_i}:~j_1,j_2,\dots,j_i \in [N]}, ~~~\forall i \in [k].
\end{align}
We will first establish an upper bound on $\norm{\tilde{x}(t)^{\otimes i}-\tilde{y}^{[i]}(t)}$, and then convert it into an upper bound on $\norm{x(t)^{\otimes i}-y^{[i]}(t)}$, for each $i \in [k]$.

Let $P_{\pm} = \operatorname{span}\{\ket{i} : i \in S_{\pm}\}$. Then $P_+ + P_- = I$. In addition, define $\tilde{x}_{\pm}(t) = P_{\pm} \tilde{x}(t)$, $\Lambda_{\pm} = P_{\pm} \Lambda P_{\pm}$, and $\tilde{F}_{2,\pm} = P_{\pm} \tilde{F}_2 (P_{\pm} \otimes P_{\pm})$. With these definitions, we have $\tilde{x}(t) = \tilde{x}_+(t) \oplus \tilde{x}_-(t)$, and $\tilde{F}_2=\tilde{F}_{2,+} \oplus \tilde{F}_{2,-}$. That is,  
$\tilde{x}_+(t)$ and $\tilde{x}_-(t)$ reside in orthogonal subspaces, while
$\tilde{F}_{2,+}$ and $\tilde{F}_{2,-}$ act only on their respective subspaces. Consequently, under Assumption~\ref{assump:block_diagonal_f2}, the dynamics satisfy
\begin{align}
    \frac{d \tilde{x}_\pm(t)}{dt} = \Lambda_{\pm} \tilde{x}_{\pm}(t) + \tilde{F}_{2,\pm}  \tilde{x}_{\pm}(t)^{\otimes 2}.
    \label{eq:decoupled_subsystems}
\end{align}
In other words, the original system is decoupled into two subsystems on the subspaces $\operatorname{span}\{\ket{i} : i \in S_+\}$ and $\operatorname{span}\{\ket{i} : i \in S_-\}$, whose eigenvalues are given by $\lrcb{\lambda_i: ~i \in S_+}$ and $\lrcb{\lambda_i:~i \in S_-}$, respectively. 

Moreover, this decomposition naturally extends to the Carleman linearization of the system. That is, the order-$k$ Carleman linearized system can itself be decomposed into two subsystems:
\begin{align}
\frac{d\tilde{y}_{\pm}(t)}{dt} = \tilde{A}_{\pm} \tilde{y}_{\pm}(t),
\end{align}
where $\tilde{y}_{\pm}(t)=[\tilde{y}^{[1]}_{\pm}(t), \tilde{y}^{[2]}_{\pm}(t), \dots, \tilde{y}^{[k]}_{\pm}(t)]^T$ resides in the tensor algebra generated by $\operatorname{span}\{\ket{i} : i \in S_\pm\}$, and $\tilde{A}_{\pm}$ acts only on this subspace.

Let us now consider each subsystem individually, beginning with the system for $\tilde{x}_+(t)$. By Lemma~\ref{lem:delta_constant_diagonal_f2}, $\vec\lambda_+=(\lambda_i : i \in S_+)$ is $\Delta$-nonresonant. Moreover, since $\tilde{F}_2$ is $s$-column-sparse, we know that $\tilde{F}_{2,+}$ is also $s$-column-sparse. Finally, note that 
\begin{align}
\norm{\tilde{x}_+(t)} \le \norm{\tilde{x}(t)} \le \norm{\tilde{x}_{\rm max}}, \qquad \forall t \in [0, T].
\end{align}
Combining these facts, we can adapt the proof of Theorem~\ref{thm:error_bound_nonresonant_poincare_domain} to obtain
\begin{align}
\norm{\tilde{x}_+(t)^{\otimes i}-\tilde{y}^{[i]}_+(t)} &\le  k t  \norm{\tilde{F}_{2,+}} \norm{\tilde{x}_{\rm max}}^{k+1} \binom{k-1}{i-1} \lrb{\frac{8 s\norm{\tilde{F}_{2,+}} }{\Delta}}^{k-i} \\
&\le  k t  \norm{\tilde{F}_{2}} \norm{\tilde{x}_{\rm max}}^{k+1} \binom{k-1}{i-1} \lrb{\frac{8 s\norm{\tilde{F}_{2}} }{\Delta}}^{k-i},
\end{align}
for all $i \in [k]$ and $t \in [0, T]$. Here the second step follows from $\norm{\tilde{F}_{2,+}} \le \norm{\tilde{F}_{2}}$, as $\tilde{F}_2=\tilde{F}_{2,+} \oplus \tilde{F}_{2,-}$.

Similarly, we can also show that 
\begin{align}
\norm{\tilde{x}_-(t)^{\otimes i}-\tilde{y}^{[i]}_-(t)}\le  k t \norm{\tilde{F}_{2}} \norm{\tilde{x}_{\rm max}}^{k+1} \binom{k-1}{i-1} \lrb{\frac{8 s\norm{\tilde{F}_{2}} }{\Delta}}^{k-i},
\end{align}
for all $i \in [k]$ and $t \in [0, T]$.

Then since $\tilde{x}(t)=\tilde{x}_+(t) \oplus \tilde{x}_-(t)$
and $\tilde{y}^{[i]}(t)=\tilde{y}^{[i]}_+(t) \oplus \tilde{y}^{[i]}_-(t)$, we get
\begin{align}
\norm{\tilde{x}(t)^{\otimes i}-\tilde{y}^{[i]}(t)}\le \sqrt{2} k t  \norm{\tilde{F}_{2}} \norm{\tilde{x}_{\rm max}}^{k+1} \binom{k-1}{i-1} \lrb{\frac{8 s\norm{\tilde{F}_{2}} }{\Delta}}^{k-i},
\end{align}
for all $i \in [k]$ and $t \in [0, T]$.

Finally, since $x(t)=Q \tilde{x}(t)$ and $y^{[i]}(t) =Q^{\otimes i} \tilde{y}^{[i]}(t)$ for $i \in [k]$, it follows that
\begin{align}
\norm{{x}(t)^{\otimes i}-{y}^{[i]}(t)}& \le \sqrt{2} k t  \norm{Q}^i \norm{\tilde{F}_{2}} \norm{\tilde{x}_{\rm max}}^{k+1} \binom{k-1}{i-1} \lrb{\frac{8 s\norm{\tilde{F}_{2}} }{\Delta}}^{k-i} \\ 
&= \sqrt{2} k t  \norm{Q}^i \norm{\tilde{F}_{2}} \norm{\tilde{x}_{\rm max}}^{i+1} \binom{k-1}{i-1} \lrb{\frac{8 s\norm{\tilde{F}_{2}} \norm{\tilde{x}_{\rm max}}}{\Delta}}^{k-i}.
\end{align}
This bound holds for all $i \in [k]$ and $t \in [0, T]$. This completes the proof of Theorem \ref{thm:error_bound_nonresonant_siegel_domain_decompose_f2}.

\section{Proof of Theorem~\ref{thm:Poincarealgorithm}}
\label{app:poincare_domain_algorithm}

\noindent\textbf{Algorithm description.}
To prepare a quantum state close to $\ket{x(T)}$, we first apply the Carleman linearization scheme to the nonlinear ODE system, transforming it into a linear ODE system. We then utilize an existing quantum algorithm to solve the resulting linear system. The output quantum state consists of multiple blocks, with the first block approximately proportional to $x(T)$. To extract this block, we perform a projective measurement on the output state and succeed if the state collapses onto the first subspace. To ensure a sufficiently high success probability at this stage, we rescale the original system so that its solution remains bounded in norm.

Specifically, our algorithm for solving the given problem proceeds as follows:

\begin{itemize}
\item First, we rescale the ODE system so that the solution norm does not exceed $c/\kappa_Q$. Precisely, let $\bar{x}(t)= \frac{c}{\rho\kappa_Q} x(t)$. Then \( \bar{x}(t) \) satisfies $\norm{\bar{x}(t)} \le c /\kappa_Q$ for all $t \in [0,T]$, and evolves according to the equation:
\begin{align}
\dot{\bar{x}}(t) = \bar{F}_1 \bar{x}(t) + \bar{F}_2  \bar{x}(t)^{\otimes 2},
\label{eq:ode_system_with_rescaled_solution}
\end{align}
where $\bar{F}_1=F_1$ and 
\( \bar{F}_2 = \rho \kappa_Q  F_2 / c \). Note that $F_1$ remains unchanged under this rescaling. 

\item Next, we choose an integer $k$ such that
\begin{align}
     k T g (2c)^{k-1} \kappa_Q  \rho  \norm{Q^{-1}}\norm{\tilde{F}_2} \le \frac{\epsilon}{4}.
\end{align}
Precisely, we pick
\begin{align}
k = \left\lceil \frac{W_{-1}\lrb{\zeta}}{\myln{2c}} \right\rceil,
\label{eq:def_truncation_order_k}
\end{align}
where $W_{-1}$ is the lower branch of the Lambert W function, and
\begin{align}
\zeta =\mymax{\frac{c\epsilon\myln{2c}}{2gT\kappa_Q\rho\norm{Q^{-1}}\norm{\tilde{F}_2}}, -\frac{1}{e}}.
\end{align}
Using the assumption $8\rho s \norm{Q^{-1}} \norm{\tilde{F}_2} \le c\Delta$, where $c \in (0, 1/2)$ is a constant, one can prove that
\begin{align}
k=\myO{\max(\mylog{ g \kappa_Q T \Delta/(s\epsilon)}, 1)}.
\end{align}
\item Now let 
\begin{align}
\dot{\bar{y}}(t) = \bar{A} \bar{y}(t)
\label{eq:carleman_linearized_system_rescaled}
\end{align}
be the Carleman linearization of the ODE system~\eqref{eq:ode_system_with_rescaled_solution} truncated at level $k$, where $\bar{y}(t)=[\bar{y}^{[1]}(t)$, $\bar{y}^{[2]}(t)$, $\dots$, $\bar{y}^{[k]}(t)]^T$ with $\bar{y}^{[i]}(t) \in \myC^{N^i}$ 
for $i \in [k]$. We apply the algorithm from Ref.~\cite{jennings2023cost} to solve this linear ODE system with initial condition 
\begin{align}
\bar{y}(0)=[\bar{x}(0), \bar{x}(0)^{\otimes 2}, \dots, \bar{x}(0)^{\otimes k}]^T,    
\end{align}
and obtain a quantum state $\ket{\psi(T)}$ that is $\epsilon_1$-close to $\ket{\bar{y}(T)}$ in Euclidean norm, where 
\begin{align}
\epsilon_1=\frac{3\epsilon}{20\sqrt{k}}.    
\end{align}
This algorithm requires a block-encoding of $\bar{A}$ and a unitary procedure for preparing $\ket{\bar{y}(0)}$. The block-encoding of $\bar{A}$ can be constructed using queries to (controlled-) $U_1$ and $U_2$, while the preparation of $\ket{\bar{y}(0)}$ can be accomplished using queries to (controlled-)~$V_0$, supplemented by additional elementary quantum gates. Further details on the implementation of these operations are provided in Appendix~\ref{app:implemenation_algorithm_poincare_domain_case}.
\item 
Let $M=N+N^2+\dots+N^k=\frac{N^{k+1}-N}{N-1}$, so that $\bar{y}(t) \in \myC^M$. For each $i \in [k]$, define the index set 
\begin{align}
\mathcal{M}_i \defeq \lrcb{\frac{N^i-N}{N-1}+1, \frac{N^i-N}{N-1}+2, \dots
\frac{N^{i+1}-N}{N-1}}.   
\end{align}
Then the entries of $\bar{y}^{[i]}(t)$ are embedded within $\bar{y}(t)$ at the positions indexed by $\mathcal{M}_i$. Define
\begin{align}
\Pi_i \defeq \sum_{j \in \mathcal{M}_i} \ket{j}\bra{j},    
\end{align}
for each $i \in [k]$. Then $\lrcb{\Pi_1, \Pi_2, \dots, \Pi_k}$ forms a projective measurement on $\myC^M$. 
We perform this measurement -- or alternatively, $\lrcb{\Pi_1, I- \Pi_1}$ -- on $\ket{\psi(T)}$, and consider it successful if the  outcome corresponds to $\Pi_1$, in which case we output the resulting post-measurement state. We will later show that the success probability is $\myOmega{1/k}$. Thus, by repeating the  procedure $\myO{k}$ times, we can obtain the desired state with  probability $1-\myO{1}$. Alternatively, using amplitude amplification, we can achieve this success probability with only $\myO{\sqrt{k}}$ repetitions.
\end{itemize}

\noindent \textbf{Correctness proof.}
Let $g_1=\frac{\norm{x(0)}}{\min_{t\in [0,T]} \norm{x(t)}}$ and
$g_2=\frac{\max_{t\in [0, T]}\norm{x(t)}}{\norm{x(0)}}$, so that $g=g_1g_2$. Moreover, define $\bar{\rho}=c/\kappa_Q$ and let
\begin{align}
\tilde{\bar{F}}_2=Q^{-1}\bar{F}_2Q^{\otimes 2}=\frac{\rho}{\bar{\rho}}\tilde{F}_2.    
\end{align}
Note that the ODE system \eqref{eq:ode_system_with_rescaled_solution} satisfies all the conditions of Theorem~\ref{thm:error_bound_nonresonant_poincare_domain}. In particular, $(\lambda_1, \lambda_2, \dots, \lambda_N)$ is $\Delta$-nonresonant. Furthermore, its solution $\bar{x}(t)$  satisfies 
\begin{align}
g_2 \norm{\bar{x}(0)} = \max_{t \in [0, T]}\norm{\bar{x}(t)} \le \bar{\rho}=\frac{c}{\kappa_Q}<\frac{1}{2}.
\end{align}
It follows that
\begin{align}
\norm{\tilde{\bar{x}}_{\rm max}} \defeq \max_{t \in [0, T]}\norm{Q^{-1}\bar{x}(t)} \le  g_2 \norm{\bar{x}(0)} \norm{Q^{-1}}
\le \bar{\rho} \norm{Q^{-1}} < \frac{\norm{Q^{-1}}}{2}.    
\label{eq:bound_tilde_bar_x_max}
\end{align}
Then by Theorem~\ref{thm:error_bound_nonresonant_poincare_domain}, we obtain
\begin{align}
\norm{\bar{x}(t)^{\otimes i}-\bar{y}^{[i]}(t)} 
&\le k t \norm{Q}^{i} \norm{\tilde{\bar{F}}_2}
\norm{\tilde{\bar{x}}_{\rm max}}^{k+1} 
\binom{k-1}{i-1} \lrb{\frac{8 s \norm{\tilde{\bar{F}}_2}}{\Delta}}^{k-i}
\\
&\le k t \norm{Q}^{i} \norm{\tilde{\bar{F}}_2}
\cdot \lrb{\bar{\rho}\norm{Q^{-1}}}^k 
\cdot g_2\norm{\bar{x}(0)} \norm{Q^{-1}}
\cdot \binom{k-1}{i-1} \lrb{\frac{8 s \norm{\tilde{\bar{F}}_2}}{\Delta}}^{k-i} \\
&= k t g_2 \kappa_Q^i \bar{\rho}^{k}\norm{\bar{x}(0)}\norm{Q^{-1}}\norm{\tilde{\bar{F}}_2} \binom{k-1}{i-1} \lrb{\frac{8 s\norm{Q^{-1}} \norm{\tilde{\bar{F}}_2}}{\Delta}}^{k-i} \\
& = 
k t g_2 \kappa_Q \rho c^{i-1} 
  \norm{\bar{x}(0)} \norm{Q^{-1}}\norm{\tilde{F}_2} \binom{k-1}{i-1} \lrb{\frac{8 \rho s\norm{Q^{-1}} \norm{\tilde{F}_2}}{\Delta}}^{k-i}  \\
&\le 
k t g_2 \kappa_Q \rho c^{k-1}  \norm{\bar{x}(0)} \norm{Q^{-1}}\norm{\tilde{F}_2} \binom{k-1}{i-1} \\
& \le 
k t g_2 \kappa_Q \rho (2c)^{k-1}  \norm{\bar{x}(0)} \norm{Q^{-1}}\norm{\tilde{F}_2},  
\label{eq:better_error_bounds_i}
\end{align}
for all $i\in [k]$ and $t \in [0,T]$. Here the second step follows from Eq.~\eqref{eq:bound_tilde_bar_x_max}, the fourth step follows from $\bar{\rho}=c/\kappa_Q$ and $\tilde{\bar{F}}_2=\rho \tilde{F}_2 / \bar{\rho}$, and the fifth step follows from the assumption $8\rho s \norm{Q^{-1}}\norm{\tilde{F}_2}\le c\Delta$.

By our choice of $k$, we guarantee that
\begin{align}
\norm{\bar{x}(t)^{\otimes i}-\bar{y}^{[i]}(t)} \le \frac{\epsilon}{4} \cdot \frac{\norm{\bar{x}(0)}}{g_1}
\le \frac{\epsilon}{4} \cdot\norm{\bar{x}(t)},
\label{eq:dist_barx_bary}
\end{align}
for all $i \in [k]$ and $t \in [0,T]$. Then since $\norm{\bar{x}(t)}\le \bar{\rho} < 1/2$, we get
\begin{align}
    \norm{\bar{y}^{[i]}(t)} \le \norm{\bar{x}(t)}^{i}+\frac{\epsilon}{4}\cdot \norm{\bar{x}(t)} \le 
    \lrb{1+\frac{\epsilon}{4}} \norm{\bar{x}(t)}
    \label{eq:upper_bound_norm_baryi}
\end{align}
for all $i \in [k]$ and $t \in [0, T]$.

Moreover, taking $i=1$ in Eq.~\eqref{eq:dist_barx_bary} yields
\begin{align}
    \norm{\bar{x}(t)-\bar{y}^{[1]}(t)} \le \frac{\epsilon}{4} \norm{\bar{x}(t)},
    \label{eq:dist_barx_bary1}
\end{align}
which implies 
\begin{align}
\lrb{1-\frac{\epsilon}{4}} \norm{\bar{x}(t)}
\le \norm{\bar{y}^{[1]}(t)}
\le \lrb{1+\frac{\epsilon}{4}} \norm{\bar{x}(t)}.
\label{eq:lower_bound_norm_bary1}
\end{align}
Meanwhile, one can prove that if two vectors $u, v$ of the same dimension satisfy $\norm{u-v}\le \delta \norm{u}$ for some  $\delta>0$, then
\begin{align}
    \norm{\frac{u}{\norm{u}}-\frac{v}{\norm{v}}} \le 2\delta.
\end{align}
Using this fact and Eq.~\eqref{eq:dist_barx_bary1}, we know that
\begin{align}
    \norm{\frac{\bar{x}(t)}{\norm{\bar{x}(t)}} - \frac{\bar{y}^{[1]}(t)}{\norm{\bar{y}^{[1]}(t)}}} \le \frac{\epsilon}{2}.
    \label{eq:dist_normalized_xt_y1t}
\end{align}
That is, the normalized version of $\bar{y}^{[1]}(t)$ is $\frac{\epsilon}{2}$-close to that of $\bar{x}(t)$, which coincides with that of $x(t)$.

Now if we perform the projective measurement $\lrcb{\Pi_1, \Pi_2, \dots, \Pi_k}$ (or alternatively $\lrcb{\Pi_1, I-\Pi_1}$) on the state $\ket{\bar{y}(T)}$, the probability of obtaining an outcome corresponding to $\Pi_1$ is
\begin{align}
    \bra{\bar{y}(T)} \Pi_1 \ket{\bar{y}(T)} = 
    \frac{\norm{\bar{y}^{[1]}(T)}^2}{\sum_{l=1}^k\norm{\bar{y}^{[l]}(T)}^2}
    \ge \frac{(1-\epsilon/4)^2}{(1-\epsilon/4)^2+ (k-1) (1+\epsilon/4)^2}
    \ge \frac{9}{25k},
\end{align}
where the second step follows from Eqs.~\eqref{eq:upper_bound_norm_baryi} and \eqref{eq:lower_bound_norm_bary1}, and the last step follows from $0<\epsilon<1$. Thus, we have
\begin{align}
\norm{\Pi_1 \ket{\bar{y}(T)}} \ge \frac{3}{5\sqrt{k}}.  
\label{eq:lower_bound_pi1_bary}
\end{align}
Then since we prepare $\ket{\bar{y}(T)}$ within $\epsilon_1=\frac{3\epsilon}{20\sqrt{k}}$ accuracy, i.e.
\begin{align}
\norm{\ket{\psi(T)} - \ket{\bar{y}(T)}} \le \frac{3\epsilon}{20\sqrt{k}},
\end{align}
it follows that
\begin{align}
\norm{\Pi_1 \ket{\psi(T)} - \Pi_1 \ket{\bar{y}(T)}} \le    \frac{3\epsilon}{20\sqrt{k}} \le \frac{\epsilon}{4} \norm{\Pi_1 \ket{\bar{y}(T)}},
\label{eq:dist_pi1psi_pi1bary}
\end{align}
and hence
\begin{align}
\norm{\frac{\Pi_1 \ket{\psi(T)}}{\norm{\Pi_1 \ket{\psi(T)}}} - \frac{\Pi_1 \ket{\bar{y}(T)}}{\norm{\Pi_1 \ket{\bar{y}(T)}}}} \le \frac{\epsilon}{2}. 
\label{eq:dist_pi1psit_pi1yt}
\end{align}
Meanwhile, by Eq.~\eqref{eq:dist_normalized_xt_y1t}, we have
\begin{align}
\norm{
\frac{\Pi_1 \ket{\bar{y}(T)}}{\norm{\Pi_1 \ket{\bar{y}(T)}}} -  \ket{x(T)}} 
= 
\norm{ \ket{\bar{y}^{[1]}(T)}
- \ket{\bar{x}(T)}}
\le \frac{\epsilon}{2}.
\label{eq:dist_pi1yt_xt}
\end{align} 
Then it follows from Eqs.~\eqref{eq:dist_pi1psi_pi1bary} and \eqref{eq:dist_pi1psit_pi1yt} that
\begin{align}
\norm{\frac{\Pi_1 \ket{\psi(T)}}{\norm{\Pi_1 \ket{\psi(T)}}} -  \ket{x(T)}} \le \epsilon. 
\end{align}
Thus, conditioned on successful post-selection, the output state produced by our algorithm is $\epsilon$-close to the target state, as required.

Furthermore, by Eqs.~\eqref{eq:lower_bound_pi1_bary} and \eqref{eq:dist_pi1psi_pi1bary}, we have
\begin{align}
    \norm{\Pi_1 \ket{\psi(T)} } \ge \lrb{1-\frac{\epsilon}{4}} \cdot \frac{3}{5\sqrt{k}}
\end{align}
and hence
\begin{align}
    \bra{\psi(T)}\Pi_1 \ket{\psi(T)} \ge \lrb{1-\frac{\epsilon}{4}}^2 \cdot \frac{9}{25k} \ge \frac{81}{400k}.
\end{align}
That is, the success probability of post-selection in each round is $\myOmega{1/k}$. Consequently, it suffices to repeat the procedure $\myO{k}$ times to boost the overall success probability to $1-\myO{1}$. Alternatively, by employing amplitude amplification, we can reduce the number of repetitions to $\myO{\sqrt{k}}$.\\

\noindent\textbf{Complexity analysis.}
We utilize the algorithm from Ref.~\cite{jennings2023cost} to solve the Carleman linearized system~\eqref{eq:carleman_linearized_system_rescaled}. The cost of this algorithm can be characterized as follows. Consider a linear ODE system
\begin{align}
    \dot{z}(t)  = B z(t).
\end{align}
Let $U_B$ be an $(\omega, a, 0)$-block-encoding of the matrix $B$, and let $U_0$ be a unitary that prepares the normalized initial state $\ket{z(0)}$. Then the algorithm of Ref.~\cite{jennings2023cost} prepares a quantum state that is $\epsilon$-close to $\ket{z(T)}$ using
$$\mytO{\omega C_{\rm max}\hat{g}T\mylog{1/\epsilon}}$$
queries to (controlled-) $U_B$, $U_0$, and their inverses, where 
\begin{align}
C_{\rm max} = \max_{t \in [0,T]} \norm{e^{Bt}},~~~\hat{g}=\frac{\max_{t \in [0,T]} \norm{z(t)}}{\norm{z(T)}},
\end{align}
and additional
$$\mytO{\omega C_{\rm max}\hat{g}T\mylog{1/\epsilon}\cdot \mypoly{\mylog{N}}}$$
elementary quantum gates.

In our case, $B$ is the order-$k$ Carleman matrix $\bar{A}$ associated with the rescaled ODE system \eqref{eq:ode_system_with_rescaled_solution}. That is,
\begin{align}
    B = \bar{A} = \begin{pmatrix}
    \bar{A}_{1,1} & \bar{A}_{1,2} & & & &\\
     & \bar{A}_{2,2} & \bar{A}_{2,3} & & &\\
            &  & \bar{A}_{3,3} & \bar{A}_{3,4} &       & \\
            &         & \ddots   & \ddots   & \ddots & \\
            &         &         &  & \bar{A}_{k-1,k-1}  & \bar{A}_{k-1,k}\\      
            &         &         &         &   & \bar{A}_{k,k}
    \end{pmatrix}
\end{align}
where
\begin{align}
\bar{A}_{j,j}&=\sum_{l=0}^{j-1} I^{\otimes l} \otimes \bar{F}_1 \otimes I^{\otimes j-l-1} \in \myC^{N^j\times N^j},
\label{eq:carleman_linearization_def_ajj_bar}\\ 
\bar{A}_{j,j+1}&=\sum_{l=0}^{j-1} I^{\otimes l} \otimes \bar{F}_2 \otimes I^{\otimes j-l-1} \in \myC^{N^{j}\times N^{j+1}}.
\end{align}
Meanwhile, $z(t)$ is given by $\bar{y}(t)=[\bar{y}^{[1]}(t)$, $\bar{y}^{[2]}(t)$, $\dots$, $\bar{y}^{[k]}(t)]^T$, with initial condition $z(0)=\bar{y}(0)=[\bar{x}(0), \bar{x}(0)^{\otimes 2}, \dots, \bar{x}(0)^{\otimes k}]^T$. 

In Appendix~\ref{app:implemenation_algorithm_poincare_domain_case}, we show that one can construct a
$(k\alpha_1+(k-1)\alpha_2', q+a+4, 0)$-block-encoding of $\bar{A}$ using $\myO{k}$ queries to controlled-$U_1$ and controlled-$U_2$, along with $\mytO{k^2 \mylog{N}}$ additional elementary gates, where $\alpha_2'=\rho \kappa_Q\alpha_2/c$, $a=\mymax{a_1,a_2}$, and $q=\myO{\mylog{k}}$. Moreover, we also demonstrate a unitary procedure that prepares a state proportional to $\bar{y}(0)$ 
using $k$ queries to controlled-$V_0$ and $\mytO{k}$ extra elementary gates.

In addition, using Eq.~\eqref{eq:dist_barx_bary} and the fact that
\begin{align}
\frac{\norm{\bar{x}(0)}}{g_1} \le \norm{\bar{x}(t)} \le g_2\norm{\bar{x}(0)} < \frac{1}{2}, \qquad \forall t \in [0, T], 
\end{align}
we obtain
\begin{align}
\hat{g} &=\frac{\max_{t \in [0,T]} \norm{\bar{y}(t)}}{\norm{\bar{y}(T)}}  \\
& 
\le \frac{\max_{t \in [0,T]} \sqrt{\sum_{i=1}^k \norm{\bar{y}^{[i]}(t)}^2}}{\norm{\bar{y}^{[1]}(T)}} \\
&\le \frac{\sqrt{k}g_2\norm{x(0)} (1+\epsilon/4)}{g_1^{-1}\norm{x(0)}(1-\epsilon/4)} \\
&= \sqrt{k}g_1g_2 \cdot \frac{1+\epsilon/4}{1-\epsilon/4} \\
&= \myO{\sqrt{k} g}.    
\end{align}

Finally, we prove:
\begin{lemma}
\begin{align}
\norm{e^{\bar{A}t}}\le \frac{2^k \kappa_Q^{k}}{\sqrt{3}},\quad \forall~t \in \R.     
\label{eq:bound_exp_At}
\end{align}    
\label{lem:bound_exp_At}
\end{lemma}
\begin{proof}
    Let $\tilde{\bar{A}}=W^{-1}\bar{A}W$, 
where
\begin{align}
W = \diag{Q, Q^{\otimes 2}, \dots, Q^{\otimes k}}.
\end{align}
Then we have
\begin{align}
    \tilde{\bar{A}}=
    \begin{pmatrix}
    \tilde{\bar{A}}_{1,1} & \tilde{\bar{A}}_{1,2} & & & &\\
     & \tilde{\bar{A}}_{2,2} & \tilde{\bar{A}}_{2,3} & & &\\
            &  & \tilde{\bar{A}}_{3,3} & \tilde{\bar{A}}_{3,4} &       & \\
            &         & \ddots   & \ddots   & \ddots & \\
            &         &         &  & \tilde{\bar{A}}_{k-1,k-1}  & \tilde{\bar{A}}_{k-1,k}\\      
            &         &         &         &   & \tilde{\bar{A}}_{k,k}
    \end{pmatrix},
\end{align}
where
\begin{align}
\tilde{\bar{A}}_{j,j}&=\sum_{l=0}^{j-1} I^{\otimes l} \otimes \Lambda \otimes I^{\otimes j-l-1} \in \myC^{N^j\times N^j}, \\    
\tilde{\bar{A}}_{j,j+1}&=\sum_{l=0}^{j-1} I^{\otimes l} \otimes \tilde{\bar{F}}_2 \otimes I^{\otimes j-l-1} \in \myC^{N^{j}\times N^{j+1}}.
\end{align}
Using the method described in Section~\ref{subsec:carleman_matrix_diagonalization}, we can diagonalize $\tilde{\bar{A}}$ as 
\begin{align}
    \tilde{\bar{A}}=\tilde{\bar{V}} D \tilde{\bar{V}}^{-1},
\end{align}
where
\begin{align}
D=\begin{pmatrix}
        \Lambda_1  &           &        & \\
                 & \Lambda_2 &        & \\
                 &           & \ddots & \\
                 &           &      &\Lambda_k \end{pmatrix},~~~
        \tilde{\bar{V}}=\begin{pmatrix}
        \tilde{\bar{V}}_{1,1}  & \tilde{\bar{V}}_{1,2}   &  \tilde{\bar{V}}_{1,3} & \dots  & \tilde{\bar{V}}_{1,k} \\
                 & \tilde{\bar{V}}_{2,2}   &  \tilde{\bar{V}}_{2,3} & \dots & \tilde{\bar{V}}_{2,k} \\
                 &           & \ddots &   \ddots       & \vdots \\
                 &           &          & \tilde{\bar{V}}_{k-1,k-1} & \tilde{\bar{V}}_{k-1, k} \\
                &           &          &  & \tilde{\bar{V}}_{k, k} \\\end{pmatrix},     
\end{align}
with
\begin{align}
\Lambda_j=\diag{\sum_{l=1}^j \lambda_{i_l}:~i_1, i_2,\dots,i_j\in [N]}, ~~~\forall j \in [k],    
\end{align}
and $\tilde{\bar{V}}_{i,j} \in \myC^{N^i \times N^j}$ for $1 \le i \le j \le k$ (in particular, 
$\tilde{\bar{V}}_{j,j}=I^{\otimes j}$ for $j \in [k]$). Note that $\tilde{\bar{V}}$ is upper-triangular, and Lemma~\ref{lem:vij_norm_bound_poincare_domain} gives an upper bound on the norm of each block of this matrix:
\begin{align}
    \norm{\tilde{\bar{V}}_{i,j}} 
    \le \binom{j-1}{i-1} \lrb{\frac{4s\norm{\tilde{\bar{F}}_2}}{\Delta}}^{j-i},~&~\forall~i \le j.
    \label{eq:bound_on_vij_norm}
\end{align}
Furthermore, the inverse of $\tilde{\bar{V}}$ is also upper-triangular, i.e.
\begin{align}
 \tilde{\bar{V}}^{-1}=\begin{pmatrix}
        (\tilde{\bar{V}}^{-1})_{1,1}  & (\tilde{\bar{V}}^{-1})_{1,2}   &  (\tilde{\bar{V}}^{-1})_{1,3} & \dots  & (\tilde{\bar{V}}^{-1})_{1,k} \\
                 & (\tilde{\bar{V}}^{-1})_{2,2}   &  (\tilde{\bar{V}}^{-1})_{2,3} & \dots & (\tilde{\bar{V}}^{-1})_{2,k} \\
                 &           & \ddots &   \ddots       & \vdots \\
                 &           &          & (\tilde{\bar{V}}^{-1})_{k-1,k-1} & (\tilde{\bar{V}}^{-1})_{k-1, k} \\
                &           &          &  & (\tilde{\bar{V}}^{-1})_{k, k} \\\end{pmatrix},   
\end{align}
where $(\tilde{\bar{V}}^{-1})_{i,j} \in \myC^{N^i \times N^j}$ for $1 \le i \le j \le k$ 
(in particular, $(\tilde{\bar{V}}^{-1})_{j,j}=I^{\otimes j}$ for $j \in [k]$),  and Lemma~\ref{lem:wij_norm_bound_poincare_domain} provides an upper bound on the norm of each block of this matrix:
\begin{align}
\norm{\lrb{\tilde{\bar{V}}^{-1}}_{a,b}} \le \binom{b-1}{a-1}\lrb{\frac{4s\norm{\tilde{\bar{F}}_2}}{\Delta}}^{b-a},~&~\forall a\le b.
\label{eq:bound_on_wab_norm}
\end{align}

Now we bound the norm of $e^{\bar{A}t}$ as follows. Note that
\begin{align}
e^{\bar{A}t}=W e^{\tilde{\bar{A}}t} W^{-1}  
=W \tilde{\bar{V}} e^{Dt} \tilde{\bar{V}}^{-1}  W^{-1},
\end{align}
where $D$ is diagonal, $W$ and $W^{-1}$ are block-diagonal, $\tilde{\bar{V}}$ and $\tilde{\bar{V}}^{-1}$ are upper-triangular.

Let $x=[x^{[1]}, x^{[2]}, \dots, x^{[k]}]^T \in \myC^M$, where $x^{[j]} \in \myC^{N^j}$, be arbitrary. We will prove an upper bound on the norm of $e^{\bar{A}t} x$. Using the block structures of $D$, $W$, $W^{-1}$, $\tilde{\bar{V}}$ and $\tilde{\bar{V}}^{-1}$, we get
\begin{align}
    e^{\bar{A}t} x = [y^{[1]}, y^{[2]}, \dots, y^{[k]}]^T,
\end{align}
where
\begin{align}
    y^{[i]}=Q^{\otimes i} \sum_{j=i}^{k} \sum_{l=j}^k \tilde{\bar{V}}_{i,j} e^{\Lambda_j}t \tilde{\bar{V}}_{j,l} (Q^{-1})^{\otimes l} x^{[l]},~~~\forall i \in [k].
\end{align}
Since $\Lambda_j$ is diagonal, and
its diagonal entries have non-positive real parts, we have $\norm{e^{\Lambda_j t}} \le 1$ for all $t \in \R$. Consequently, we get
\begin{align}
    \norm{y^{[i]}}\le \norm{Q}^i \sum_{j=i}^{k} \sum_{l=j}^k \norm{\tilde{\bar{V}}_{i,j}}  \norm{\tilde{\bar{V}}_{j,l}} \norm{Q^{-1}}^l \norm{x^{[l]}},~~~\forall i \in [k].
\end{align}
It follows that
\begin{align}
\norm{e^{\bar{A}t} x}^2 
&=\sum_{i=1}^k \norm{y^{[i]}}^2 \\
&\le 
\sum_{i=1}^k 
\norm{Q}^{2i} \lrb{\sum_{j=i}^{k} \sum_{l=j}^k \norm{\tilde{\bar{V}}_{i,j}}  \norm{\tilde{\bar{V}}_{j,l}} \norm{Q^{-1}}^l \norm{x^{[l]}}}^2 \\
&\le 
\sum_{i=1}^k 
\norm{Q}^{2i} \lrsb{\sum_{l=i}^k \norm{Q^{-1}}^{2l}\lrb{\sum_{j=i}^{l}   \norm{\tilde{\bar{V}}_{i,j}}  \norm{\tilde{\bar{V}}_{j,l}} }^2}
\lrb{\sum_{l=i}^k \norm{x^{[l]}}^2} \\
&\le 
\lrcb{\sum_{i=1}^k 
\norm{Q}^{2i} \lrsb{\sum_{l=i}^k \norm{Q^{-1}}^{2l}\lrb{\sum_{j=i}^{l}   \norm{\tilde{\bar{V}}_{i,j}}  \norm{\tilde{\bar{V}}_{j,l}} }^2}} \norm{x}^2,
\label{eq:bound_exp_at_x}
\end{align}
where we use the triangle inequality in 
the third step. Then using Eqs.~\eqref{eq:bound_on_vij_norm},  \eqref{eq:bound_on_wab_norm} and $\tilde{\bar{F}}_2=\rho \kappa_Q \tilde{F}_2/c$, we obtain that for any $i \le l$,
\begin{align}
    \sum_{j=i}^{l} \norm{\tilde{\bar{V}}_{i,j}}  \norm{\tilde{\bar{V}}_{j,l}}
    &\le \sum_{j=i}^{l} \binom{j-1}{i-1} \binom{l-1}{j-1} \lrb{\frac{4s \norm{\tilde{\bar{F}}_2}}{\Delta}}^{l-i} \\
    &=\binom{l-1}{i-1}\lrb{\frac{8s \norm{\tilde{\bar{F}}_2}}{\Delta}}^{l-i} \\
    &=
    \sigma^{l-i} \norm{Q}^{l-i}  \binom{l-1}{i-1},
    \label{eq:bound_sum_vij_vjl}
\end{align}
where
\begin{align}
    \sigma = \frac{8\rho  s \norm{Q^{-1}} \norm{\tilde{F}_2}}{c\Delta} \in [0, 1].
\end{align}
Then it follows from Eqs.~\eqref{eq:bound_exp_at_x} and \eqref{eq:bound_sum_vij_vjl} that
\begin{align}
    \norm{e^{\bar{A}t} x}^2 
    &\le \lrcb{\sum_{i=1}^k \norm{Q}^{2i} \lrsb{\sum_{l=i}^k \norm{Q^{-1}}^{2l} \sigma^{2(l-i)} \norm{Q}^{2(l-i)} \binom{l-1}{i-1}^2}} \cdot \norm{x}^2 \\
    &= \lrcb{\sum_{l=1}^k \kappa_{Q}^{2l} \lrsb{\sum_{i=1}^l \sigma^{2(l-i)} \binom{l-1}{i-1}^2}} \cdot \norm{x}^2 \\
    & \le 
    \lrcb{\sum_{l=1}^k \kappa_{Q}^{2l} \lrsb{\sum_{i=1}^l \sigma^{l-i} \binom{l-1}{i-1}}^2} \cdot \norm{x}^2 \\
    &=
     \lrsb{\sum_{l=1}^k \kappa_{Q}^{2l} \lrb{1+\sigma}^{2(l-1)}} \cdot \norm{x}^2.
\end{align}
Now using $\sigma \in [0, 1]$, we obtain
\begin{align}
    \norm{e^{\bar{A}t} x}^2 
    \le
     \lrb{\sum_{l=1}^k \kappa_{Q}^{2l} 4^{l-1}} \cdot \norm{x}^2
    = \frac{\kappa_Q^2\lrsb{(4\kappa_Q^2)^k-1}}{4\kappa_Q^2-1}
    \cdot \norm{x}^2.
\end{align}
This inequality holds for all $x \in \myC^M$ and $t \in \R$. Therefore, we have 
\begin{align}
\norm{e^{\bar{A}t}} \le  \kappa_Q \sqrt{\frac{(4\kappa_Q^2)^k-1}{4\kappa_Q^2-1}}
\le \kappa_Q \sqrt{\frac{(4\kappa_Q^2)^k}{3\kappa_Q^2}}
=\frac{2^k \kappa_Q^{k}}{\sqrt{3}}, 
\end{align}
for all $t \in \R$, as claimed. 
\end{proof}

Now we combine the above results and conclude that each run of the algorithm from Ref.~\cite{jennings2023cost} costs $$\mytO{\alpha g T k^2\sqrt{k} 2^k \kappa_Q^k \mylog{k/\epsilon}}
$$ 
queries to 
controlled-$U_1$,
controlled-$U_2$,
controlled-$V_0$ and their inverses,
where $\alpha=\alpha_1+{\alpha_2\rho \kappa_Q}$, 
along with
$$\mytO{\alpha g T k^3 \sqrt{k} 2^k \kappa_Q^k \mylog{k/\epsilon} \cdot \mypoly{\mylog{N}}}$$ 
additional elementary gates. Using amplitude amplification,  this procedure only needs to be repeated $\myO{\sqrt{k}}$ times. Note that all $\mypoly{k}$ factors are negligible compared to the dominate cost factor $2^k \kappa_Q^k$. Thus, our algorithm requires a total of
$$\mytO{\alpha g T  2^k \kappa_Q^k\mylog{1/\epsilon}}$$
queries to 
controlled-$U_1$,
controlled-$U_2$, 
controlled-$V_0$ and their inverses,
and
$$\mytO{\alpha g T 2^k \kappa_Q^k \mylog{1/\epsilon} \cdot \mypoly{\mylog{N}}}$$
extra elementary gates. 

This completes the proof of the theorem.

\subsection{Implementation details}
\label{app:implemenation_algorithm_poincare_domain_case}

The algorithm in Theorem~\ref{thm:algorithm_complexity_poincare_domain_case} requires (i) a block-encoding of the Carleman matrix $\bar{A}$ and (ii) a unitary that prepares the normalized initial state $\ket{\bar{y}(0)}$. In this appendix, we describe  how both components can be implemented using the provided unitaries $U_1$, $U_2$ and $V_0$, together with elementary quantum gates.

For simplicity of presentation, we assume that $N=2^n$ and $k=2^q$ are powers of two. In cases where they are not, the system can be padded and embedded into a larger system of appropriate size, which does not alter the overall complexity of the algorithm. Each $N$-dimensional system is encoded with $n$ qubits, and each $k$-dimensional system is encoded with $q$ qubits. We will use $I_m$ to denote the identity operator acting on $m$ qubits for any $m \ge 1$.

Let $r=q+kn$. To facilitate the implementation of our algorithm, we map the $M$-dimensional matrix
$\bar{A}$ to an $r$-qubit linear operator
\begin{align}
    \mathcal{A} \defeq \sum_{j=0}^{k-1} \ket{j}\bra{j} \otimes \mathcal{A}_{j+1,j+1} +
    \sum_{j=0}^{k-2} \ket{j}\bra{j+1} \otimes \mathcal{A}_{j+1,j+2},
\end{align}
where
\begin{align}
\mathcal{A}_{j,j} &\defeq \sum_{l=0}^{j-1} I_n^{\otimes l} \otimes \bar{F}_1 \otimes I_n^{\otimes (k-l-1)},
\label{eq:carleman_linearization_def_ajj_cal}\\ 
\mathcal{A}_{j,j+1}&\defeq \sum_{l=0}^{j-1} I_n^{\otimes l} \otimes \bar{F}_2 \otimes I_n^{\otimes (j-l-1)} \otimes \ket{0^n} \otimes I_n^{\otimes (k-j-1)}
\end{align}
are both linear operators acting on $kn$ qubits. Accordingly, we map the $M$-dimensional vector $\bar{y}(t)$ to an unnormalized $r$-qubit state 
\begin{align}
    \ket{\tilde{\mathcal{Y}}(t)} \defeq \sum_{j=0}^{k-1} \norm{\bar{y}^{[j+1]}(t)}\ket{j} \otimes \ket{\bar{y}^{[j+1]}(t)} \otimes \ket{0^n}^{\otimes (k-j-1)}.
\end{align}
The interaction between $\mathcal{A}$ and $\ket{\tilde{\mathcal{Y}}(t)}$ is identical to that between $\bar{A}$ and $\bar{y}(t)$. Thus, this mapping preserves the essential behavior of the algorithm. \\

\noindent\textbf{Block-encoding of the Carleman matrix}. Recall that $U_1$ is an $(\alpha_1, a_1, 0)$-block-encoding of $F_1$, i.e.
\begin{align}
\lrb{\bra{0^{a_1}} \otimes I_n}
U_1
\lrb{\ket{0^{a_1}} \otimes I_n}
=\frac{F_1}{\alpha_1}.
\end{align}
Then since $\bar{F}_1=F_1$, $U_1$ is also an $(\alpha_1, a_1, 0)$-block-encoding of $\bar{F}_1$. Meanwhile, 
$U_2$ is an $(\alpha_2, a_2, 0)$-block-encoding of $F_2$, i.e.
\begin{align}
\lrb{\bra{0^{a_2}} \otimes I_n \otimes I_n}
U_2
\lrb{\ket{0^{a_2}} \otimes I_n \otimes I_n}
=\frac{F_2 \otimes \ket{0^n}}{\alpha_2}.
\end{align}  
Note that $F_2 \otimes \ket{0^n}$ is a linear operator acting on $2n$ qubits. Then by $\bar{F}_2=\rho \kappa_Q F_2 / c$, we know that $U_2$ is an $(\alpha_2', a_2, 0)$-block-encoding of $\bar{F}_2$,
where 
$\alpha_2'=\rho\kappa_Q\alpha_2/c$.

Now we use (controlled-) $U_1$, $U_2$ and additional quantum gates to implement a block-encoding of $\mathcal{A}$. To this end, we decompose $\mathcal{A}$ as the sum of $2k-1$ terms:
\begin{align}
    \mathcal{A} = \sum_{l=0}^{k-1} \mathcal{B}_l + 
    \sum_{l=0}^{k-2} \mathcal{C}_l, 
\end{align}
where
\begin{align}
    \mathcal{B}_l &= \sum_{j=l}^{k-1} \ket{j}\bra{j} \otimes I_n^{\otimes l} \otimes \bar{F}_1 \otimes I_n^{\otimes (k-l-1)}, \\
    \mathcal{C}_l &=\sum_{j=l}^{k-2} \ket{j}\bra{j+1}\otimes I_n^{\otimes l} \otimes \bar{F}_2 \otimes I_n^{\otimes (j-l)} \otimes \ket{0^n} \otimes I_n^{\otimes (k-j-2)}.
\end{align}
Note that $\mathcal{B}_l$ applies $\bar{F}_1$ to the $(l+1)$-th level of the Carleman system, and 
$\mathcal{C}_l$ applies $\bar{F}_2$ to the $(l+1)$-th and $(l+2)$-th levels of the Carleman system.

Next, for each $l \in \lrcb{0,1,\dots,k-1}$, 
we construct a block-encoding of $\mathcal{B}_l$ using controlled-$U_1$ and the following operations. For any $m \in \lrcb{0,1,\dots,k-1}$, let $\mathcal{R}_m$ be an $(q+1)$-qubit unitary operator such that
\begin{align}
    \mathcal{R}_m \ket{j}\ket{0} =
    \begin{cases}
    \ket{j}\ket{0}, ~&~\mathrm{if}~0\le j < m;\\        
    \ket{j}\ket{1}, ~&~\mathrm{if}~m\le j\le k-1.
    \end{cases}
\label{eq:def_rm}
\end{align}
Note that $\mathcal{R}_m$ can be implemented using $\myO{\mylog{k}}$ elementary gates, with the help of one ancilla qubits.

Consider the following quantum circuit on $r +a_1+1$ qubits, partitioned into $k+3$ registers containing $q$, $1$, $a_1$, $n$, $n$, $\dots$, $n$ qubits, respectively. First, we apply $\mathcal{R}_l$ to the first and second registers. Then, we apply controlled-$U_1$ to the second, third and $(l+4)$-th registers, using the second register as the control qubit. Finally, we apply an $X$ gate to the second register. That is, this circuit implements the $(r+a_1+1)$-qubit unitary operator
\begin{align}
\mathcal{U}_l = X^{(2)}\text{c-}U_1^{(2, 3,\,l+4)}\mathcal{R}_l^{(1,2)},
\end{align}
where the superscripts indicate the subsystems on which each operator acts. One can verify that
\begin{align}
    \lrb{I_q \otimes \bra{0} \otimes \bra{0^{a_1}} \otimes I_n^{\otimes k}}\mathcal{U}_l
    \lrb{I_q \otimes \ket{0} \otimes \ket{0^{a_1}} \otimes I_n^{\otimes k}}    
    =\frac{\mathcal{B}_l}{\alpha_1}.    
\end{align}
Thus, $\mathcal{U}_l$ is an $(\alpha_1, a_1+1, 0)$-block-encoding of $\mathcal{B}_l$ -- or  
an $(\alpha_1, a_1+2, 0)$-block-encoding of $\mathcal{B}_l$ if we account for the ancilla qubit used for implementing $\mathcal{R}_l$. The above implementation of $\mathcal{U}_l$ requires one query to controlled-$U_1$ and $\myO{\mylog{k}}$ additional elementary gates.

Similarly, for each $l \in \lrcb{0,1,\dots,k-2}$, we can construct a block-encoding of $\mathcal{C}_l$ using controlled-$U_2$ and the following operations. First, let $\mathcal{R}_m$ be defined as in Eq.~\eqref{eq:def_rm} for $m \in \lrcb{0,1,\dots,k-1}$. Second, let $\mathrm{SWAP}_n$ denote the unitary that swaps two $n$-qubit systems, i.e.
\begin{align}
\mathrm{SWAP}_n \ket{i_1,i_2,\dots,i_n}\ket{j_1,j_2,\dots,j_n} = \ket{j_1,j_2,\dots,j_n}\ket{i_1,i_2,\dots,i_n},
\end{align}
for all $i_1,i_2,\dots,i_n, j_1,j_2,\dots,j_n\in \lrcb{0,1}$. Clearly, $\mathrm{SWAP}_n$ can be implemented with $n$ individual $\mathrm{SWAP}$ gates. Finally, let 
\begin{align}
X_k=\sum_{j=0}^{k-1} \ket{(j-1)~\mathrm{mod}~k}\bra{j}    
\end{align}
be the $k$-dimensional cyclic shift operator, which can be implemented using $\myO{\mypoly{\mylog{k}}}$ elementary gates.

Consider the following quantum circuit on $r +a_2+2$ qubits, partitioned into $k+4$ registers containing $q$, $1$, $1$, $a_2$, $n$, $n$, $\dots$, $n$ qubits, respectively. The circuit proceeds as follows:
\begin{enumerate}
    \item Apply $\mathcal{R}_{l+1}$ to the first and second register.
    \item Apply controlled-$U_2$ to the second, fourth, $(l+5)$-th and $(l+6)$-th registers, using the second register as the control qubit.   \item Apply an $X$ gate to the second register.
    \item For $m=l+2, l+3, \dots, k-1$, do the following:
    \begin{enumerate}
        \item Apply $\mathcal{R}_m$ to the first and third register.
        \item Apply controlled-$\mathrm{SWAP}_n$ to the third, $(m+4)$-th and $(m+5)$-th registers, using the third register as the control qubit.
        \item Apply $\mathcal{R}_m^{\dagger}$ to the first and third register.
    \end{enumerate}
    \item Apply $X_k$ to the first register. 
\end{enumerate}
In other words, this circuit implements the $(r+a_2+2)$-qubit unitary operator
\begin{align}
    \mathcal{V}_l=X_k^{(1)} 
    \mathcal{W}_{k-1} 
    \mathcal{W}_{k-2}\dots \mathcal{W}_{l+3}
    \mathcal{W}_{l+2}
    X^{(2)}\text{c-}U_2^{(2, 4,l+5, l+6)}\mathcal{R}_{l+1}^{(1,2)} ,
\end{align}
where
\begin{align}
    \mathcal{W}_m = \lrb{(\mathcal{R}_m^\dagger)^{(1,3)} \otimes I_n^{(m+4)} \otimes I_n^{(m+5)}}
    \lrb{I_q^{(1)} \otimes \text{c-}\mathrm{SWAP}_n^{(3,m+4,m+5)}}
    \lrb{\mathcal{R}_m^{(1,3)} \otimes I_n^{(m+4)} \otimes I_n^{(m+5)}}.
\end{align}
Here the superscripts indicate the subsystems on which each operator acts. The purpose of the $\mathcal{W}_m$ operations is to move the $\ket{0^n}$ state - introduced by the application of $U_2$ - to the correct position. One can verify by direct calculation that
\begin{align}
    \lrb{I_q \otimes \bra{0} \otimes \bra{0} \otimes \bra{0^{a_2}} \otimes I_n^{\otimes k}}\mathcal{V}_l
    \lrb{I_q \otimes \ket{0} \otimes \ket{0} \otimes \ket{0^{a_2}} \otimes I_n^{\otimes k}}    
    =\frac{\mathcal{C}_l}{\alpha_2'}.    
\end{align}
Thus, $\mathcal{V}_l$ is an $(\alpha_2', a_2+2, 0)$-block-encoding of $\mathcal{C}_l$ -- or  
an $(\alpha_2, a_2+3, 0)$-block-encoding of $\mathcal{C}_l$ if we account for the ancilla qubit used for implementing the $\mathcal{R}_m$ operations. 
The above implementation of $\mathcal{V}_l$ requires one controlled-$U_2$ and $\mytO{k\mylog{N}}$ additional elementary gates. 

Finally, we employ the standard linear combination of unitaries (LCU) technique to construct a block-encoding of $\mathcal{A}$ from the individual block-encodings of $\mathcal{B}_l$ and $\mathcal{C}_l$. Specifically, let $\mathrm{PREP}$ be an $(q+1)$-qubit unitary operator such that
\begin{align}
\mathrm{PREP}\ket{0^{q+1}}=\frac{1}{\sqrt{k\alpha_1+(k-1)\alpha_2'}} \lrb{\sqrt{\alpha_1}\sum_{j=0}^{k-1} \ket{j} + \sqrt{\alpha_2'}\sum_{j=k}^{2k-2} \ket{j}}.    
\end{align}
Moreover, let $a=\mymax{a_1, a_2}$, and let $\mathrm{SELECT}$ be an $(r+a+q+4)$-qubit unitary operator defined as
\begin{align}
\mathrm{SELECT} =\sum_{j=0}^{k-1} \ket{j}\bra{j} \otimes \mathcal{U}_j'  
+
\sum_{j=0}^{k-2} \ket{k+j}\bra{k+j} \otimes \mathcal{V}_j',
\end{align}
where $\mathcal{U}_j'$ and $\mathcal{V}_j'$ are essentially equivalent to $\mathcal{U}_j$ and $\mathcal{V}_j$, respectively, except that we have re-ordered the qubits, and added extra (unused) qubits if necessary, to make sure they conform to the standard definition of a block-encoding of $\mathcal{B}_j$, $\mathcal{C}_j$, respectively, and use the same number of ancilla qubits. Formally, we have
\begin{align}
    \lrb{\bra{0^{a+3}} \otimes I_r}\mathcal{U}_j'
    \lrb{\ket{0^{a+3}} \otimes I_r}    
    &=\frac{\mathcal{B}_j}{\alpha_1}, \\
    \lrb{\bra{0^{a+3}} \otimes I_r}\mathcal{V}_j'
    \lrb{\ket{0^{a+3}} \otimes I_r}    
    &=\frac{\mathcal{C}_j}{\alpha_2'}.    
\end{align}
Then by direct calculation, one can verify that
\begin{align}
\lrb{\bra{0^{q+a+4}} \otimes I_{r}} 
\lrb{\mathrm{PREP}^\dagger \otimes I_{r+a+3}}  \cdot 
\mathrm{SELECT} \cdot
\lrb{\mathrm{PREP}  \otimes I_{r+a+3}}  
\lrb{\ket{0^{q+a+4}} \otimes I_{r}}
=\frac{\mathcal{A}}{k\alpha_1+(k-1)\alpha_2'}.
\end{align}
In other words, we have obtained a 
$(k\alpha_1+(k-1)\alpha_2', q+a+4, 0)$-block-encoding of $\mathcal{A}$. Note that $\mathrm{PREP}$ can be implemented using $\myO{k}$ elementary gates, while $\mathrm{SELECT}$ can be implemented using at most 
one query to each of the controlled-$\mathcal{U}_j$ and controlled-$\mathcal{V}_j$ operations, along with $\myO{\mylog{k}}$ additional elementary gates. Therefore, this block-encoding of $\mathcal{A}$ can be implemented using $\myO{k}$ queries to controlled-$U_1$ and controlled-$U_2$, along with $\mytO{k^2 \mylog{N}}$ extra elementary gates.\\

\noindent\textbf{Initial state preparation}.
Recall that $V_0$ is a unitary that prepares a state proportional to $\bar{x}(0)$, i.e. 
$V_0\ket{0^n}=\ket{\bar{x}(0)}$. Now we show how to use $V_0$ and elementary gates to generate the initial state for the Carleman system:
\begin{align}
\ket{\mathcal{Y}(0)} 
=\frac{1}{\sqrt{Z}}\sum_{j=0}^{k-1}  \norm{\bar{x}(0)}^{j+1}\ket{j} \otimes \ket{\bar{x}(0)}^{\otimes (j+1)} \otimes \ket{0^n}^{\otimes (k-j-1)},
\end{align}
where the normalization factor $Z$ is given by
\begin{align}
Z = \sum_{j=1}^k \norm{\bar{x}(0)}^{2j}
=  \frac{\beta^2(\beta^{2k}-1)}{\beta^2-1}
\end{align}
with $\beta=\norm{\bar{x}(0)}=\frac{c\norm{x(0)}}{\rho \kappa_Q}$.  Let $\mathcal{W}$ be a $q$-qubit unitary operator such that
\begin{align}
    \mathcal{W}\ket{0^q}=\frac{1}{\sqrt{Z}}\sum_{j=0}^{k-1} \beta^{j+1}\ket{j},
\end{align}
and let $\mathcal{V}$ be the $r$-qubit unitary operator defined as
\begin{align}
    \mathcal{V}=\sum_{j=0}^{k-1} \ket{j}\bra{j}\otimes V_0^{\otimes (j+1)} \otimes I_n^{\otimes (k-j-1)}.
\end{align}
Then one can verify that
\begin{align}
    \mathcal{V}(\mathcal{W} \otimes I_n^{\otimes k})(\ket{0^q}\otimes \ket{0^n}^{\otimes k}) = \ket{\mathcal{Y}(0)}.
\end{align}
Note that $\mathcal{W}$ can be implemented using $\myO{k}$ elementary gates, and $\mathcal{V}$ can be implemented using $k$ queries to controlled-$V_0$ along with $\mytO{k}$ additional elementary gates. Thus, this procedure for preparing $\ket{\mathcal{Y}(0)}$ costs $k$ queries to controlled-$V_0$ and $\mytO{k}$ extra elementary gates.

\section{Proof of Theorem~\ref{thm:algorithmstable}}
\label{app:proof_thm_algorithmstable}

Here we sketch the proof. Similarly to the proof of Theorem~\ref{thm:Carlemanstable}, we perform the change of coordinates $\tilde{x} = Q x$, where $Q^\dag Q = P$. Setting
\begin{align}
\label{eq:Ftildetransformation1}
    \tilde{F}_2 = Q F_2 Q^{-1} \otimes Q^{-1}, \quad \tilde{F}_1 = Q F_1 Q^{-1}, \quad \tilde{F}_0 = Q F_0,
\end{align}
we get
\begin{align}
\label{eq:ysystem1}
    \dot{\tilde{x}} = \tilde{F}_2 \tilde{x} \otimes \tilde{x} + \tilde{F}_1 \tilde{x} + \tilde{F}_0.
\end{align}
Recall that since $\mu_P(F_1)  <0$ and $R_P <1$, we proved in Theorem~\ref{thm:Carleman_cons} that a rescaling exist for which $\mu(\tilde{F}_1)<0$ and $\| \tilde{x}(0)\|<1$. We shall choose such rescaling. This is the problem to which we apply the Carleman embedding, obtaining after truncation at level $k$, the following finite-dimensional linear ODE system:
\begin{align}
    \dot{\tilde{y}} = \tilde{A} \tilde{y} + \tilde{a}, \quad \tilde{a} = [\tilde{F}_0, 0, \dots, 0].
\end{align} 
This problem is tackled with quantum linear system solvers after time discretization via a truncated Taylor series.

We construct a $\alpha_{F_0} \alpha_Q$-block-encoding of $\tilde{F}_0$ with a $1$ call to the block-encoding of $F_0$ and $1$ call to the block-encoding of $Q$. Then consider
\begin{align}
    \tilde{F}_1 = \alpha_{F_1} \kappa_Q \frac{Q}{\|Q\|}  \frac{F_1}{\alpha_{F_1}} \frac{Q^{-1}}{\|Q^{-1}\|}.
\end{align}
A $\alpha_{F_1} \kappa_Q$-block-encoding can be obtained with $1$ call to the block-encoding of $F_1$ and $O(\kappa_Q \log(1/\epsilon))$ calls to the block-encoding of $Q$, where we used quantum singular value transform to achieve the matrix inversion~\cite{gilyen2019quantum}. Finally consider
\begin{align}
    \tilde{F}_2 = \frac{\alpha_{F_2} \kappa^2_Q}{\|Q\|} \frac{Q}{\|Q\|}  \frac{F_2}{\alpha_{F_2}} \frac{Q^{-1}}{\|Q^{-1}\|}\frac{Q^{-1}}{\|Q^{-1}\|}.
\end{align}
A $\frac{\alpha_{F_2} \kappa_Q^2}{\|Q\|}$-block-encoding of $\tilde{F}_2$ can be obtained with $O(\kappa_Q^2 \log(1/\epsilon)^2)$ calls to the block-encoding of $Q$ and $1$ call to the block-encoding of $F_2$.

From these building blocks, we can use LCU and techniques similar to those in Appendix \ref{app:implemenation_algorithm_poincare_domain_case} to block-encode the Carleman matrix~$\tilde{A}$ with rescaling factor 
\begin{align}
    k \left(\alpha_{\tilde{F}_0} + \alpha_{\tilde{F}_1} + \alpha_{\tilde{F}_2}\right) = k \left(\alpha_{F_0} \alpha_Q + \alpha_{F_1} \kappa_Q + \frac{\alpha_{F_2} \kappa_Q^2}{\alpha_Q}\right),
\end{align} 
with a single call to each of the $\tilde{F}_i$, which means a single call to each of the $F_i$, and $O(\kappa_Q^2 \log(1/\epsilon)^2)$ calls to the block-encoding of $Q$. Correspondingly, we have a block-encoding of $\tilde{A} h$ with rescaling $\alpha_{\tilde{A} h} = O(1)$. 

We shall then define a linear system of equations encoding the truncated Taylor series method as in Ref.~\cite{berry2015simulating}, i.e., $
    L \tilde{\psi} = \tilde{b}$
where setting $M = \lceil T/h \rceil$
\begin{align}
    L = I -\sum_{m=0}^{M-1} \ketbra{m+1}{m} \otimes \sum_{j=0}^K \frac{(\tilde{A}h)^j}{j!},
\end{align}
and
\begin{align}
\tilde{y} = [ \tilde{y}^{0}, \tilde{y}^{1}, \dots, \tilde{y}^{m}, \dots, \tilde{y}^{M}], \quad 
    \tilde{b} = [ \tilde{y}(0), \underbrace{v, \dots, v}_{M \textrm{ times}}], \quad v = \sum_{j=1}^K \frac{(\tilde{A}h)^{j-1}}{j!} \tilde{a}. 
\end{align}
From the block-encoding of $\tilde{A}h$, another LCU allows us to block-encode truncated Taylor series $\sum_{j=0}^K (\tilde{A}h)^j/j!$, with rescaling factor $O(1)$, using $O(K)$ calls to the block-encoding of $\tilde{A}h$, where we shall take $K= \log (T/\epsilon)$ to ensure the time-discretization error is $O(\epsilon)$. With a single call to this block-encoding we obtain, via LCU, a block-encoding of $L$, with rescaling factor $O(1)$. The cost of a call to the block-encoding of $L$ is  then
\begin{align}
O(\log(T/\epsilon)) & \quad \textrm{calls to each of the block-encodings of the } F_i, \\
O(\log(T/\epsilon) \kappa^2_Q \log(1/\epsilon)^2) & \quad  \textrm{calls to the block-encoding of } Q.
\end{align}

Similarly, we can block-encode $\sum_{j=1}^K (Ah)^{j-1}/j!$ with rescaling factor $O(1)$, with $O(K)$ queries to the block-encoding of $\tilde{A}h$. With this, and a single call to a state preparation for $\ket{a}$ (which means $1$ call to the block-encoding of $F_0$ and $1$ call to the block-encoding of $Q$), we can  prepare~$\ket{v}$ at an overall cost scaling in the same way as the cost the block-encoding of $L$. 

With $1$ call to the state preparation of~$\ket{v}$, and one call to the state preparation of $\ket{\tilde{y}(0)}$, we prepare $\ket{\tilde{b}}$. The call to $\ket{\tilde{y}(0)}$ can be achieved, following the strategy in Section~4.4 of Ref.~\cite{liu2021efficient}, with $O(k) = O(\log(T/\epsilon))$ calls to the state preparation of $\ket{\tilde{x}(0)}$. The latter can be finally obtained by applying the block-encoding of $Q$ to $\ket{x(0)}$ and the performing $O(\alpha_Q)$ rounds of amplitude amplification, resulting in a cost of $O(\alpha_Q)$ calls to a state preparation of $\ket{x(0)}$ and $\alpha_Q$ calls to $Q$.

Putting everything together, the cost of preparing $\ket{\tilde{b}}$ is 
\begin{align}
O(\log(T/\epsilon)) & \quad \textrm{calls to each of the block-encodings of the } F_i, \\
O(\log(T/\epsilon) \kappa^2_Q \log(1/\epsilon)^2) & \quad  \textrm{calls to the block-encoding of } Q, \\
O(\alpha_Q \log(T/\epsilon)) & \quad \textrm{call to the state preparation of $\ket{x(0)}$}. 
\end{align}

We then apply an optimal quantum linear solver algorithm the the linear system $L \ket{\tilde{\psi}} = \ket{\tilde{b}}$, which outputs a quantum state $\epsilon$-close to
\begin{align}
\ket{\tilde{\psi}(T)} = \frac{1}{\sqrt{\sum_{m'=1}^M \| \tilde{y}(m'h)\|^2 }} \sum_{m=1}^{\lceil T/h \rceil } \| \tilde{y}(m h)\| \ket{\tilde{y}(mh)} \otimes \ket{t}.
    \end{align}
The cost of applying an optimally scaling quantum linear solver algorithm~\cite{costa2022optimal, dalzell2024shortcut, jennings2023efficient}, using that $k = O(\log(T/\epsilon))$ due to the convergence results in Theorem~\ref{thm:Carlemanstable}, and the fact that we can bound $\kappa_L = O(1)$ using techniques similar to those discussed in Ref.~\cite{an2024fast}, is
\begin{align}
    O\left((\alpha_{F_0} \alpha_Q + \alpha_{F_1} \kappa_Q + \frac{\alpha_{F_2} \kappa_Q^2}{\alpha_Q})\log(1/\epsilon)\log(T/\epsilon) \right) 
\end{align}
calls to a block-encoding of $L$ and to the state preparation of $\ket{\tilde{b}}$. 

Recall that, due to $\mu(\tilde{F}_1)<0$ and $\| \tilde{x}(0)\|<1$, one has $\| \tilde{x}(t)\| <1$ at all times.  This means that the components in the $\tilde{y}(t)$ vector that correspond to $k =2, \dots, k$ Carleman blocks have exponentially smaller norm than the $k=1$ component, and can hence be removed via amplitude amplification of $k=1$ at a constant cost, reducing to a state $\epsilon$-close to
\begin{align}
\ket{\tilde{\psi}_1(T)} \propto \sum_{m=1}^{\lceil T/h \rceil} \| \tilde{x}(mh)\| \ket{\tilde{x}(mh)}.
\end{align}
We can now apply a block-encoding of $Q^{-1}/\|Q^{-1}\|$, using again $O(\kappa_Q \log(1/\epsilon))$ calls to the block-encoding of $Q$, giving a state $\epsilon$-close to
\begin{align}
    \ket{\psi_1} \propto \sum_{m=1}^{\lceil T/h \rceil} \| x(mh)\| \ket{x(mh)},
\end{align}
after $O(\|Q^{-1}\|)$ rounds of amplitude amplification. 

The overall cost is then 
\footnotesize
\begin{align}
O(\|Q^{-1}\|(\alpha_{F_0} \alpha_Q + \alpha_{F_1} \kappa_Q + \frac{\alpha_{F_2} \kappa_Q^2}{\alpha_Q})\log(1/\epsilon) \log^2(T/\epsilon)) & \quad \textrm{calls to each of the block-encodings of the } F_i, \nonumber \\
O(\|Q^{-1}\|(\alpha_{F_0} \alpha_Q + \alpha_{F_1} \kappa_Q + \frac{\alpha_{F_2} \kappa_Q^2}{\alpha_Q})\log^3(1/\epsilon) \log^2(T/\epsilon) \kappa^3_Q) & \quad  \textrm{calls to the block-encoding of } Q, \nonumber \\
 O(\kappa_Q (\alpha_{F_0} \alpha_Q + \alpha_{F_1} \kappa_Q + \frac{\alpha_{F_2} \kappa_Q^2}{\alpha_Q})\log(1/\epsilon)\log^2(T/\epsilon)) & \quad \textrm{calls to the state preparation of $\ket{x(0)}$}. 
\end{align}
\normalsize


\section{Proof of Theorem~\ref{thm:bqp}}
\label{sec:proofbqpcompleteness}


\subsection{Change of coordinates and parameter choices}

In the first part of the proof, we will perform a coordinate transformation that highlight the norm-preserving properties of the linear, dissipationless dynamics, together with a rescaling that will be key later on to argue that the kinetic energy can be efficiently extracted by the quantum algorithm. We shall also make specific choices concerning the strength the nonlinearity and dissipation. The choices may seem arbitrary at this stage, but they have been made so that later on we will be able to argue for the convergence of the Carleman embedding procedure, and derive bounds on the difference between the kinetic energy computed for the nonlinear oscillator vs the same quantity computed for the linear, dissipationless problem.

Let us start with Eq.~\eqref{eq:bqpproblem}. 
The vector $(q, \dot{q})$ satisfies the equations
\begin{align}
    \frac{d}{dt} \begin{bmatrix} q \\ \dot{q} \end{bmatrix}  =  \begin{bmatrix}
       0 &  I \\
       - D & - R
   \end{bmatrix}  \begin{bmatrix} q \\ \dot{q} \end{bmatrix} + F_2  \begin{bmatrix} q \\ \dot{q} \end{bmatrix}^{\otimes 2},
\end{align}
where 
\begin{align}
    F_2 = \begin{bmatrix}
        0 & 0 & 0 & 0 \\
        0 & 0 & 0 & -W
    \end{bmatrix}. 
\end{align}
Perform the coordinate transformation 
\begin{align}
 \begin{bmatrix} q \\ \dot{q} \end{bmatrix} \mapsto   x := \begin{bmatrix} y \\ \dot{q} \end{bmatrix} = \frac{1}{c} \begin{bmatrix}
        \sqrt{D} & 0 \\ 
        0 & I
    \end{bmatrix} \begin{bmatrix} q \\ \dot{q} \end{bmatrix}  = \frac{1}{c} \begin{bmatrix} \sqrt{D} q \\ \dot{q} \end{bmatrix}, 
\end{align}
as described in Ref.~\cite{krovi2024quantum}, but with an additional rescaling $$c = 16 \sqrt{2}/3.$$ The transformed equations read
\begin{align}
\label{eq:transformedoscillator}
    \frac{dx}{dt} = \begin{bmatrix}
       0 &  \sqrt{D} \\
       - \sqrt{D} & -R
   \end{bmatrix} x + \bar{F}_2 x^{\otimes 2} := F_1 x + \bar{F}_2 x^{\otimes 2},
\end{align}
where 
\begin{align}
    F_1 = \begin{bmatrix}
        0 & \sqrt{D} \\  - \sqrt{D} & -R
    \end{bmatrix},
\end{align}
and $\bar{F}_2 = c F_2$. Note that
\begin{align}
    \| x(0) \|^2 = \frac{\| \sqrt{D} q(0) \|^2}{c^2} +  \frac{\|\dot{q}(0)\|^2}{c^2} =  0 + \frac{2}{c^2} = \frac{9}{256}<1.
\end{align}

We set $\epsilon = \Theta(1/\mathrm{polylog}(N)) \in (0,1)$, and make the following choices:
\begin{itemize}
\item The matrix $D$ takes the form given in Eq. (C4) of Ref.~\cite{babbush2023exponential} and encodes a universal quantum circuit consisting of Hadamard, Pauli-X and Toffoli gates.
In particular, the eigenvalues of $D$ are nonnegative and of order $\Theta(1)$ -- precisely, 
\begin{equation}
    d_i:=\lambda_i(D)  \geq d_{\mathrm{min}} \ge 3-\sqrt{2}.
\end{equation}
This follows the Gershgorin's circle theorem, the construction of $D$ in Ref.~\cite{babbush2023exponential} and the fact that the encoded circuit does not contain two consecutive Hadamard gates.

\item $R= r I_{N'} \oplus 0$, where  $I_{N'}$ is the $N'$-dimensional identity matrix, $0$ is the matrix of all zeros on the complement, and 
    \begin{equation}
        r= \frac{\epsilon}{4 t_\star} \leq \sqrt{2}.
    \end{equation} 
    So  $\|R\| =  \Theta(1/\mathrm{polylog}(N))$, as required.
    
\item Let $q = (q',q'')$, where $q'$ denotes the first $N'$ components of $q$. On vectors $[\dot{q}'^{\otimes 2}, \dot{q}' \otimes \dot{q}'', \dot{q}'' \otimes \dot{q}', \dot{q}'^{\otimes 2}]^T$ the matrix $W$ acts as
    \begin{align}
    \label{eq:Wmatrix}
        W= \begin{bmatrix}
            W' & 0 & 0 & 0\\
            0 & 0 & 0 & 0
        \end{bmatrix}, 
    \end{align}
    where we take $W'_{k|ij} = w \delta_{ij} \delta_{ik}$ for all $i,j,k= 1, \dots, N'$. We will take $w = \Theta(1/\mathrm{polylog}(N))$.
     
    Note that if we apply $W'$ to a canonical basis state $\ket{k}$ we get either $0$ or $w \ket{k}$ depending on $k$ and so the norm of $W'$ is at least $|w|$. What is more, since the matrix elements of $W'$ are zero or $w$, and the matrix is $1$-sparse, the largest singular value is upper bounded by $w$. So the norm of $W'$ is exactly $|w|$. We then have that $W$ acts on $\mathrm{polylog}(N)$ oscillators and $\|W\| = \|W'\| = \Theta(1/\mathrm{polylog}(N))$, as required.   

\item We shall impose
    \begin{align}
    \label{eq:wcondition}
        |w| = \frac{9\sqrt{2} \epsilon}{2048 \, e \, t_\star} = \frac{9\sqrt{2}}{512 e}r,
    \end{align}  
    which is consistent with the previous requirement $|w| = 1/\mathrm{polylog}(N)$.
\end{itemize}


\subsection{The problem is in BQP}

In the second part of the proof, we shall argue that a quantum algorithm exists for extracting the kinetic energy of an oscillator in the nonlinear oscillator network. We will do so require in three steps. In the first step, we need to argue that transformed via Carleman embedding into a larger linear problem with controlled errors and overheads. This will require us to get estimates for the solution norm, the condition number of the diagonalizing matrix of the linear components which will feed into a $R$-number condition. In the second step, we will argue that the resulting problem can be  solved via quantum linear solver algorithms with $O(\mathrm{polylog}(N)$ queries to certain oracles. This will require some further estimates: the norm the Carleman matrix generating the embedded linear problem, and that of its matrix exponential. In the third step, finally, we will see that the oracles can be efficiently constructed with $O(\mathrm{polylog}(N)$ gates.

\bigskip

\textbf{Step 1: Carleman convergence.} If $U$ is the unitary that diagonalizes $D$, we have that $U \oplus U$ brings $F_1$ to the form,
\begin{align}
   (U^\dag \oplus U^\dag) F_1 (U \oplus U)= \begin{bmatrix}
        0 & \mathrm{diag}[\sqrt{d_1}, \dots, \sqrt{d_N}] \\  - \mathrm{diag}[\sqrt{d_1}, \dots, \sqrt{d_N}] & -r I_{N'} \oplus 0
    \end{bmatrix}.
\end{align}
Note that if we reorder the coordinates as 
\begin{align}
    \label{eq:bqp_order}
    [y_1, \dot{q}_1, \dots,  y_{N'}, \dot{q}_{N'}, y_{N'+1}, \dot{q}_{N'+1}, \dots, y_N, \dot{q}_N ],
\end{align}
the transformed $F_1$ matrix takes the form
\begin{align}
\label{eq:directproductstructureoscillators}
    \bigoplus_{i=1}^{N'} \begin{bmatrix}
        0 & \sqrt{d_i} \\ 
       - \sqrt{d_i} & -r
    \end{bmatrix}   \bigoplus_{i=N'+1}^{N} \begin{bmatrix}
        0 & \sqrt{d_i} \\ 
       - \sqrt{d_i} & 0
    \end{bmatrix} .
\end{align}
From this, it is clear that to complete the diagonalization of $F_1$ we need to apply a matrix of the form
\begin{align}
V' = \bigoplus_{i=1}^{N'} M_i \bigoplus_{i=N'+1}^{N} U_i,
\end{align}
where $M_i$ are non-unitary matrices diagonalizing the matrices with dissipation $r>0$ on the left-hand-side of Eq.~\eqref{eq:directproductstructureoscillators}, and $U_i$ are unitary matrices diagonalizing the matrices with no dissipation on the right-hand-side of Eq.~\eqref{eq:directproductstructureoscillators}. The $M_i$ matrices read
\begin{align}
M_i = \left[
\begin{array}{cc}
 \frac{\sqrt{r^2-4 d_i}-r}{2 \sqrt{d_i}} & -\frac{\sqrt{r^2-4 d_i}+r}{2 \sqrt{d_i}} \\
 1 & 1 \\
\end{array}
\right].
\end{align}
The condition number of $M_i$ is
\begin{align}
    \kappa_ i = \sqrt{1+\frac{2 r}{2 \sqrt{d_i}-r}} \leq \sqrt{1+\frac{ r}{ \sqrt{d_i}}} \leq 1 + \frac{r}{2\sqrt{d_i}} \leq 1+ \frac{r}{2{ \sqrt{d_{\mathrm{min}}}}}.
\end{align}
From these considerations, the condition number of the matrix $V'$ is
\begin{align}
    \kappa_{V'} = \max_{i=1}^{N'} \kappa_i \leq 1+ \frac{r}{2 {\sqrt{d_{\mathrm{min}}}}} = O(1), 
\end{align}
since $r = \Theta(1/\mathrm{polylog}(N))$. Consequently, denoting by $V$ the matrix diagonalizing $F_1$ (meaning $V^{-1} F_1 V$ is diagonal), we have $V= V' (U \oplus U)$ and 
\begin{align}
    \kappa_V = \kappa_{V'} = O(1).
\end{align}

Next, we will explain under what conditions the Carleman embedding error converges for the analysed system. The eigenvalues of $F_1$ are given by
\begin{align}
    \label{eq:bqp_convergence}
   \left\{ \frac{1}{2} (- r \pm \sqrt{r^2 - 4 |d_j|}) \right\}_{j=1}^{N'} \quad \cup \quad \{\pm i |d_j|\}_{i=N'}^N,
\end{align}
with the purely imaginary eigenvalues living in the subspace corresponding to subsystems $i\in\{N'+1, \dots, N\}$, and all the other eigenvalues having a negative real part. In other words, if we order the coordinates as in Eq.~\eqref{eq:bqp_order}, then the imaginary eigenvalues live in the subspace spanned by canonical basis states $\ket{k}$ for $k>2N'$. At the same time, from the structure of $F_2$ and that of $W$ in Eq.~\eqref{eq:Wmatrix}, we see that $F_2$ satisfies $\bra{k} F_2 = 0$ for $k>2N'$. This means that we can apply the oscillating and conservative Carleman analysis from Section~\ref{sec:conservative}, with the convergence of the Carleman error vector occurring under the following condition (see the proof of Theorem~\ref{thm:Carleman_cons} in Appendix~\ref{app:thm_cons} for details):
\begin{align}
R_\delta':=\max \left\{ 2 \|\tilde{x}_{\mathrm{max}}\| \| V\| , \frac{2e\|\tilde{x}_{\mathrm{max}}\|\|V^{-1}\bar{F}_2V^{\otimes 2}\|}{\delta(F_1)} \right\} < 1.
\end{align}
In our case $V = V' (U \oplus U)$, $\|\tilde{x}_{\mathrm{max}}\| = \|V^{-1} x_{\mathrm{max}}\|$, and $\delta(F_1) = r/2$, and the only remaining ingredient to specify the convergence regime is to bound the norm of the solution.

Start with
\begin{align}
\frac{d\| x\|^2}{dt} = x^\dag \bar{F}_2 (x \otimes x) + (x^\dag \otimes x^\dag) \bar{F}_2^\dag x + x^\dag (F_1 + F_1^\dag) x.
\end{align}
Furthermore,
\begin{align}
   x^\dag( F_1 + F_1^\dag ) x= x^\dag\begin{bmatrix}
        0 & 0 \\  0 & -2R
    \end{bmatrix} x = -2 \dot{q}^\dag R \dot{q} = -2 r \| \dot{q}'\|^2
\end{align}
and
\begin{align}
    x^\dag \bar{F}_2 (x \otimes x) & = [ y^\dag \quad \dot{q}^\dag] \begin{bmatrix}
        0 & 0 & 0 & 0 \\
        0 & 0 & 0 & -c W
    \end{bmatrix} \begin{bmatrix}
        y \otimes y\\ 
        y \otimes \dot{q} \\
        \dot{q} \otimes y \\ \dot{q} \otimes \dot{q} 
    \end{bmatrix} =  -c \dot{q}^\dag W \dot{q} \otimes \dot{q}   \\ & =-c \dot{q}^{\prime\dag} W' \dot{q}' \otimes \dot{q}' = -c w \dot{q}^{\prime\dag} (\dot{q}' \ast \dot{q}') = -c w \sum_{i=1}^{N'} \dot{q}_i^{\prime 3}. 
\end{align}
Hence,
\begin{align}
    {\frac{d \|x\|^2}{dt}} = -2\sum_{i=1}^{N'} (cw\dot{q}_i'+r) \dot{q}_i^{\prime 2} \leq   2\sum_{i=1}^{N'} (c|w|\|x\|-r) \dot{q}_i^{\prime 2},
\end{align}
which leads to the above upper bound
\begin{align}
    {\frac{d \|x\|}{dt}} \leq \left(c|w|-\frac{r}{\|x\|}\right)  \sum_{i=1}^{N'}  \dot{q}_i^{\prime 2}.
\end{align}
Note that at $t =0$ 
\begin{align}
    c|w|-\frac{r}{\|x(0)\|} = \left(|w| - \frac{r}{\sqrt{2}}\right)c.
\end{align}
Moreover, Eq.~\eqref{eq:wcondition} implies
\begin{align}
    |w|<r/\sqrt{2},
\end{align} 
which ensures the above derivative is non-positive at $t=0$ and, consequently, at all times. So $\|x(t)\|$ can be bounded by its value at the initial time:
\begin{equation}
    \label{eq:normUpperBound}
    \|x(t)\|\leq \|x(0)\|.
\end{equation}

As for the lower bound, we have
\begin{align}
    {\frac{d \|x\|}{dt}} = \frac{1}{\|x\|}\sum_{i=1}^{N'} (-cw\dot{q}_i'-r) \dot{q}_i^{\prime 2} \geq   -\frac{1}{\|x\|}\sum_{i=1}^{N'} (c|w|\|x\|+r) \dot{q}_i^{\prime 2} \geq -(c|w|\|x\|+r)\|x\|.
\end{align}
Combining the above with Eq.~\eqref{eq:normUpperBound}, we get
\begin{align}
    {\frac{d \|x\|}{dt}}  \geq -(c|w|\|x(0)\|+r)\|x(0)\|,
\end{align}
and so
\begin{equation}
    \|x(t)\| \geq \|x(0)\|(1 - (r+c|w|\|x(0)\|) t),
\end{equation}
By setting the RHS of the above equation to be larger or equal than $\|x(0)\|/2$ we get the condition
\begin{align}
    |w| \leq \frac{1}{2 \sqrt{2} t} - \frac{r}{\sqrt{2}}.
\end{align}
So, with the condition $r= \frac{1}{4t_\star} \leq \frac{1}{4t}$, it suffices
\begin{align}
    |w| \leq \frac{1}{4 \sqrt{2} t_\star},
\end{align}
which is satisfied via the assumption in Eq.~\eqref{eq:wcondition}. We conclude that $\|x(t)\| = \Theta(1)$, and so $\|x_{\mathrm{max}}\|= \Theta(1)$.

We can now go back to the convergence condition, Eq.~\eqref{eq:bqp_convergence} and use what we have derived together with the fact that for $N$ large enough, $r < \sqrt{2}/3$ and so $\kappa_V \leq 4/3$:
\begin{align}
    R_\delta' & \leq \max\left\{2 \|x_{\mathrm{max}}\| \kappa_V, \frac{2e\|x_{\mathrm{max}}\|\|\bar{F}_2\|\kappa_V^2}{\delta(F_1)}\right\}
    \\ & \leq \max\left\{\frac{8 \|x_{\mathrm{max}}\|}{3}, \frac{32e\|x_{\mathrm{max}}\|\|\bar{F}_2\|}{9\delta(F_1)}\right\}
        \\ & = \max\left\{\frac{8 \sqrt{2}}{3c}, \frac{128\sqrt{2} e|w|}{9 r}\right\}.
\end{align}
Since $c = 16 \sqrt{2}/3$ and $|w|= \frac{9r}{256 \sqrt{2} e}$, the max equals $1/2$.

We conclude that, for all $N$ larger than some constant, one has
\begin{align}
       R'_\delta  \leq 1/2.
\end{align}
Thus, we ensured $R'_{\delta}<1$, so the convergence condition for the Carleman system is satisfied, and to get linear embedding errors smaller than $\epsilon$ it then suffices to take the Carleman truncation order
\begin{equation}
    k = \Theta \left(\log\frac{1}{\epsilon}\right)= O (\mathrm{polyloglog}(N)).
\end{equation}

\bigskip

\textbf{Step 2: Efficiency of the quantum algorithm.}  After truncation to level $k$, we get a linear system of equations for the Carleman vector $z$,
\begin{align}
    \frac{dz}{dt} = A z.
\end{align}
We know that $\alpha(F_1) = 0$. Furthermore, since $F_0 = 0$, the matrix
\begin{align}
   V_A:= V \oplus V^{\otimes 2} \oplus \dots \oplus V^{\otimes k},
\end{align}
brings the Carleman matrix to upper triangular form:
\begin{align}
    V_A A V_A^{-1} = A',
\end{align}
where
\begin{align}
    A' = B_D + B.
\end{align}
Here, $B_D$ is a diagonal matrix with largest real parts of elements equal to zero, and $B$ is the matrix
\begin{align}
    B = \begin{bmatrix}
        0 & \tilde{A}_{0,1} & 0 & 0 & \cdots \\
        0 & 0 & \tilde{A}_{1,2} & 0 & \cdots \\
        \vdots & \vdots & \ddots & \ddots & \ddots \\
        0 & 0 & 0 & 0 & \tilde{A}_{k, k+1}
    \end{bmatrix}\end{align}
with 
$$\tilde{A}_{j,j+1} = V \bar{F}_2 (V^{-1})^{\otimes 2} \otimes \underbrace{I \otimes \cdots \otimes I}_{j-1} + \mathrm{shifts}.$$
Recall that for any two matrices $M_1$ and $M_2$, if we know that $\|e^{M_1 t}\| \leq e^\beta$, then
\begin{equation}
    \|e^{(M_1 + M_2)t} \| \leq e^{(\beta + \| M_2\|)t},
\end{equation}
as shown in Refs.~\cite{kato2013perturbation,krovi2024quantum}.
Since $\|e^{B_D t}\| \leq 1$ we can apply the bound with $M_1 = B_D$, $M_2 = B$ and $\beta =0$:
\begin{align}
    \| e^{A' t} \| \leq e^{\|B\| t}. 
\end{align}
Now note that, since $\kappa_V = O(1)$, $k = O(\mathrm{polyloglog}(N))$, and $\|\bar{F_2}\| = O(1/\mathrm{polylog}(N))$, we have that 
\begin{align}
    \|B\| = O(k \| \bar{F}_2\|)  = O(1/ \mathrm{polylog}(N)).
\end{align}
Since $t_\star = \Theta(\mathrm{polylog}(N))$, the overall scaling of $\|B\| t = O(1)$ and so for all $t \in [0, t_\star]$
\begin{align}
    \|e^{A' t}\| = O(1).
\end{align}
We can then conclude that
\begin{align}
    \| e^{A t}\| \leq \kappa_{V_A} \|e^{A' t}\| = \kappa_V^{k} \|e^{A' t}\| = O(\mathrm{polylog}(N)).
\end{align}
Finally note that $\|A\| \leq k (\| F_1\| + \| F_2\|) = O(\mathrm{polyloglog}(N))$.

We can then apply the results from Ref.~\cite{jennings2023cost} [Table 2, history state] with
\begin{itemize}
    \item $C_{\mathrm{max}} = O(\mathrm{polylog}(N))$,
    \item $g = \Theta(1)$,
    \item $\|A\| = O(\mathrm{polyloglog}(N))$
    \item $t_\star = \mathrm{polylog}(N)$,
    \item $\max_{t \in [0,t_\star]}\|x(t)\| = \Theta(1)$,
    \item $\epsilon = 1/\mathrm{poly log}(N)$,
    \item $\kappa_V=O(1)$,
\end{itemize}
to obtain a quantum linear solver algorithm extracting a state $\epsilon$-close to a history state $\ket{z_H(t_\star)}$ for $\ket{z(t_\star)}$ encoded as a quantum state and using $\mathrm{polylog}(N)$ queries to the $A$ matrix. This is also $\epsilon$-close to an idealized Carleman state which, via a constant number of rounds of amplitude amplification, gives a state $\epsilon$-close to a history state $\ket{x_H(t_\star)}$ for $\ket{x(t_\star)}$ with any chosen $\Theta(1)$ failure probability. Using the results of Ref.~\cite{dalzell2024shortcut}, we also get a $\delta$ multiplicative approximation of $\|x_H(t_\star)\|$ with a cost $O(1/\delta)$ and any chosen $\Theta(1)$ failure probability.

Note that 
\begin{align}
    \frac{1}{2} \dot{q}_1(t_\star)^\dag \dot{q}_1(t_\star) & =    \frac{1}{2} x(t_\star) \begin{bmatrix} 0 & 0 \\ 0 & \Lambda
    \end{bmatrix}
   x(t_\star) =   \frac{1}{2} x_H(t_\star) \left(\ketbra{t_\star}{t_\star} \otimes \begin{bmatrix} 0 & 0 \\ 0 & \Lambda 
    \end{bmatrix} \right)
   x_H(t_\star)  \\ & =
   \frac{1}{2}\| x_H(t_\star)\|^2 \left(\bra{x_H(t_\star)} \left[\ketbra{t_\star}{t_\star} \otimes  \begin{bmatrix} 0 & 0 \\ 0 & \Lambda 
    \end{bmatrix} \right] \ket{x_H(t_\star)}  \right),
\end{align}
where $\Lambda$ is the diagonal matrix with a single $1$ in the top-left position and the rest zeroes; $x_H(t_\star)$ denotes a vector encoding $\{x(t)\}_{t \in [0,t_\star]}$ with timestep $\Delta t = 1/\|A\| = O(\mathrm{polyloglog}(N))$, and $\ket{x_H(t_\star)}$ is the corresponding normalized quantum state. 

Since we have an algorithm to prepare $\ket{x_H(t_\star)}$ with $O(\mathrm{polylog}(N))$ queries to $A$, and the matrix $\ketbra{t_\star}{t_\star} \otimes  \begin{bmatrix} 0 & 0 \\ 0 & \Lambda \end{bmatrix}$ can be trivially block-encoded, a Hadamard test gives an $\epsilon = \Theta(1/\mathrm{polylog}(N))$-accurate estimate of the term in parenthesis with the same query complexity scaling. Also, we mentioned that at the same query cost we can estimate $\|x_H(t)\|$ up to an error $\delta \| x_H(t_\star)\|$. Since $\|x_H(t_\star)\| = O( t_\star \max_{t \in [0,t_\star]} \|x(t)\|) = O(\mathrm{polylog}(N))$, we can take $\delta = 1/\Theta(\mathrm{polylog(N)})$ to get an estimate of $\frac{1}{2} \dot{q}(t_\star)^\dag \dot{q}(t_\star)$ with overall query cost $\Theta(\mathrm{polylog}(N))$ in terms of number of calls to a block-encoding of $A$ and a state preparation of $\ket{z(0)}$. The extra gates we used are $O(\mathrm{polylog(N)})$. 

\bigskip

\textbf{Step 3: Unpacking the oracles.} Queries to the $A$ matrix can be reduced to queries to the matrices $R$, $\sqrt{D}$ and $W$. Let us see how. First, a query to the Carleman matrix can be reduced, via LCU, to querying $F_1$ and $F_2$, with a query overhead $\mathrm{poly}(k) = \Theta(\mathrm{polylog}(N))$ and $\mathrm{polylog}(N)$ extra gates. The block-encoding of $F_2$ immediately reduces to the block-encoding of $W'$ and $\Theta(\mathrm{polylog}(N))$ extra gates. As for $F_1$, it can be written as a LCU
\begin{align}
   F_1 = \begin{bmatrix}
        0 & 0 \\  0 & -R
    \end{bmatrix} + \begin{bmatrix}
        0 & \sqrt{D} \\  - \sqrt{D} & 0
    \end{bmatrix}
\end{align}
The first matrix can be block-encoded with a single call to a block-encoding of $R$. For the second matrix, it can be easily reduced to two calls to a block-encoding of $\sqrt{D}$, controlled on a single qubit. Hence, we need $O(\mathrm{polylog}(N))$ queries to $R$, $W$, $\sqrt{D}$ to run the algorithm.

To block-encode $\sqrt{D}$ using access to a block-encoding of $D$, we can use the result of Ref.~\cite{babbush2023exponential} as reported in Ref.~\cite{krovi2024quantum} [Theorem 4]. The theorem says that a block-encoding of $\sqrt{D}$, up to error $\tilde{\epsilon}$, can be constructed from the following number of calls to an $s$-sparse access oracle for $D$:
\begin{align}
O\left(\| D\|_{\mathrm{max}} s \log\left(\frac{1}{\tilde{\epsilon}}\right) \frac{1}{\tilde{\epsilon}} \min \left\{ \sqrt{\|D^{-1}\|}, \frac{1}{\tilde{\epsilon}} \right\} \right).
\end{align}
Since $\| D\|_{\mathrm{max}} = \Theta(1)$, $\| D^{-1}\| = \Theta(1)$, $s = \Theta(1)$, we have that the number of calls to $D$ to realize a single call to $\sqrt{D}$ is $O(1/\tilde{\epsilon})$. Since we query a $\tilde{\epsilon}$-block-encoding of $\sqrt{D}$ a number of times $O(\mathrm{polylog}(N))$, the following relation is required to ensure that the overall error accumulated during the algorithm is $O(\epsilon)$
\begin{align}
\tilde{\epsilon} \times O(\mathrm{polylog}(N)) \leq \epsilon \Rightarrow  \tilde{\epsilon}= \Theta(1/\mathrm{polylog}(N)).   
\end{align}
Hence $\tilde{\epsilon} = 1/ \mathrm{polylog}(N)$, and the overall number of queries to a sparse access oracle for $D$ is $O(\mathrm{polylog}(N))$. 

Similarly, the state preparation of $z(0)$ can be reduced to the state preparation requires $O(k) = O(\mathrm{polylog}(N))$ calls to the state preparation of $x(0)$, as well as $O(\mathrm{polylog}(N))$ extra gates.

Finally, we need to argue that the block-encodings of $R$, $D$, $W'$, and the state preparation of $x(0)$, are efficient. The only nontrivial claim is for $D$. For $D$ we note that the matrix has $\Theta(1)$ maximum element, efficiently specifiable sparsity structure, $\Theta(1)$ sparsity. The matrix elements are specified via an efficient quantum circuit, according to Eq.~(C4) in Ref.~\cite{babbush2023exponential}. Hence, all oracles can be realized at $\mathrm{polylog}(N)$ gate cost. Hence the problem described in the theorem statement is contained in BQP. 


\subsection{The problem is BQP-hard} 

In the last part of the proof, we show that the kinetic energy estimated for the nonlinear, dissipative problem is a $O(1/\mathrm{polylog}(N))$ approximation for the kinetic energy of the linear, dissipationless problem. Since the latter is known to be BQP-hard~\cite{babbush2023exponential}, it follows that so is the former.

Let $\mathcal{K}_{}(t_\star) = \frac{1}{2} \dot{q}_{1}^\dag(t_\star) \dot{q}_{1}(t_\star) =  \frac{1}{2}x_{}^\dag(t_\star) \Lambda x_{}(t_\star)$, where $\Lambda$ is a diagonal matrix with a single $1$ at position $N+1$ and $q_{}(t_\star)$ is the solution to \eqref{eq:bqpproblem} with initial conditions $(q(0), \dot{q}(0))$. We instead denote by $q_L(t_\star)$ the solution to the problem where we set $W=R=0$ (so the linear, dissipation-less version of the same problem),  and $\mathcal{K}_{\mathrm{L}}(t_\star) = \frac{1}{2} \dot{q}_{\mathrm{L},1}^\dag(t_\star) \dot{q}_{\mathrm{L},1}(t_\star) =  \frac{1}{2}x_{\mathrm{L}}^\dag(t_\star) \Lambda x_{\mathrm{L}}(t_\star)$. 

Let us compute
\begin{align}
  |\mathcal{K}_{}(t_\star) -   \mathcal{K}_{\mathrm{L}}(t_\star)| & = \left| \frac{1}{2}x_{}^\dag(t_\star) \Lambda x_{}(t_\star) - \frac{1}{2}x_{\mathrm{L}}^\dag(t_\star) \Lambda x_{\mathrm{L}}(t_\star) \right|\\ 
   & = \left| \frac{1}{2}x_{}^\dag(t_\star) \Lambda (x_{}(t_\star) - x_{\mathrm{L}}(t_\star) ) +  \frac{1}{2}(x_{}(t_\star)- x_{\mathrm{L}}(t_\star))^\dag \Lambda x_{\mathrm{L}}(t_\star) \right| \\
   & \leq  \frac{1}{2}(\|x_{}(t_\star)\|+\|x_{L}(t_\star)\|) \| x_{}(t_\star)- x_{\mathrm{L}}(t_\star)\|
  \\ & \leq (\|x_{}(t_\star)- x_{\mathrm{L}}(t_\star)\| + \|x_{\mathrm{L}}(t_\star) \|)  \|x_{}(t_\star)- x_{\mathrm{L}}(t_\star)\|
\end{align}
We need to bound $ \| x_{}(t_\star) -  x_{\mathrm{L}}(t_\star)\|$ and $\|x_L(t_\star)\|$. We have
\begin{align}
     \frac{d y_{}}{dt} &= \sqrt{D} \dot{q}_{}, \\
     \frac{d \dot{q}_{}}{dt} &= -\sqrt{D} y_{} - R \dot{q}_{} + c W q_{}^{\otimes 2}, \\
        \frac{d y_{\mathrm{L}}}{dt}  &= \sqrt{D} \dot{q}_{\mathrm{L}}, \\
     \frac{d \dot{q}_{\mathrm{L}}}{dt} &= -\sqrt{D} y_{\mathrm{L}}.
\end{align}
Hence,
\begin{align*}
  \frac{d}{dt} \| x_{} - x_{\mathrm{L}} \|^2 =&\frac{d}{dt} \| y_{} - y_{\mathrm{L}} \|^2 + \frac{d}{dt} \| \dot{q}_{} - \dot{q}_{\mathrm{L}} \|^2  \\
     =& \frac{d}{dt}( y_{} - y_{\mathrm{L}})^\dag (y_{} - y_{\mathrm{L}}) +  ( y_{} - y_{\mathrm{L}})^\dag\frac{d}{dt}(y_{} - y_{\mathrm{L}}) \\
   & + \frac{d}{dt}( \dot{q}_{} - \dot{q}_{\mathrm{L}})^\dag (\dot{q}_{} - \dot{q}_{\mathrm{L}}) +  ( \dot{q}_{} - \dot{q}_{\mathrm{L}})^\dag\frac{d}{dt}(\dot{q}_{} - \dot{q}_{\mathrm{L}})
   \\
     =& (\dot{q}_{} - \dot{q}_{\mathrm{L}})^\dag \sqrt{D}(y_{} - y_{\mathrm{L}}) +  ( y_{} - y_{\mathrm{L}})^\dag \sqrt{D}(\dot{q}_{} - \dot{q}_{\mathrm{L}}) \\
   &  - (y_{} - y_{\mathrm{L}})^\dag \sqrt{D} (\dot{q}_{} - \dot{q}_{\mathrm{L}}) -   \dot{q}_{} ^\dag R (\dot{q}_{} - \dot{q}_{\mathrm{L}}) + (\dot{q}^{\otimes 2})^\dag_{} c W^\dag (\dot{q}_{} - \dot{q}_{\mathrm{L}})   \\ &  - (\dot{q}_{} - \dot{q}_{\mathrm{L}})^\dag \sqrt{D} (y_{} - y_{\mathrm{L}})  -   \dot{q}_{}^\dag R (\dot{q}_{} - \dot{q}_{\mathrm{L}}) +  (\dot{q}_{} - \dot{q}_{\mathrm{L}})^\dag c W \dot{q}^{\otimes 2}_{}  \\ 
      =&  -2 \dot{q}^\dag R (\dot{q}_{} - \dot{q}_{\mathrm{L}}) + (\dot{q}^{\otimes 2})^\dag_{} c W^\dag (\dot{q}_{} - \dot{q}_{\mathrm{L}})    +  (\dot{q}_{} - \dot{q}_{\mathrm{L}})^\dag c W \dot{q}^{\otimes 2}_{}  
    \\
     =&  - 2r \sum_{i=1}^{N'} \dot{q}'_{i}( \dot{q}'_{i} - \dot{q}'_{\mathrm{L},i}) + 2c w \sum_{i=1}^{N'} ( \dot{q}'_{i} - \dot{q}'_{\mathrm{L}, i}) (\dot{q}'_{i})^2
     \\
     \leq &   2r \|\dot{q}'\|\| \dot{q}' - \dot{q}'_{\mathrm{L}}\| + 2c |w| \| \dot{q}'_{} - \dot{q}'_{L} \| \sqrt{\sum_{i=1}^{N'} \dot{q}_i^{'4}}
      \\
     \leq &   2r \|x(0)\| \|x - x_{\mathrm{L}}\| + 2c|w| \| \dot{q}'_{} - \dot{q}'_{\mathrm{L}} \| \|\dot{q}'_{}\|_4^2
     \\
       \leq &   2r \|x(0)\| \|x - x_{\mathrm{L}}\| + 2c|w| \| x_{} - x_{\mathrm{L}} \| \|x\|^2_4
        \\
       \leq &   2r \|x(0)\| \|x - x_{\mathrm{L}}\| + 2c|w| \| x_{} - x_{\mathrm{L}} \| \|x\|^2 \\
           \leq &   2r \|x(0)\| \|x - x_{\mathrm{L}}\| + 2c|w| \| x_{} - x_{\mathrm{L}} \| \|x(0)\|^2.
\end{align*}
And so, for $t\in[0,t_\star]$,
\begin{align}
    \frac{d}{dt} \| x_{}(t) - x_{\mathrm{L}}(t) \|   
     &\leq  r \|x(0)\|  + c|w| \|x(0)\|^2 \\
     & = \frac{\sqrt{2}r}{c}  + \frac{2|w|}{c}. 
\end{align}

 Hence, at time $t_\star$
 \begin{align}
     \|x(t_\star) - x_\mathrm{L}(t_\star)\|  \leq  \frac{\sqrt{2}r t_\star}{c}  + \frac{2|w| t_\star}{c} = \frac{3 r t_\star}{16} + \frac{3 t_\star |w| }{8 \sqrt{2}} = \left(\frac{3}{64}+ \frac{27}{16384 e}\right)\epsilon=: c_1 \epsilon.
 \end{align}

Furthermore, since $x_\mathrm{L}$ evolves under an antisymmetric generator, its norm is constant, in particular we have $\| x_L(t_\star)\| = \| x_{}(0)\|= \sqrt{2}/c = 3/16$. So,
\begin{align}
    |\mathcal{K}_{}(t_\star) -   \mathcal{K}_{\mathrm{L}}(t_\star)| \leq \left(c_1 \epsilon + \frac{3}{16}\right) c_1 \epsilon.
\end{align}
which is smaller than $\epsilon$ for every $\epsilon \in (0,1)$.

This allows us to use the results of the algorithm applied to the nonlinear problem to solve the decision problem stated in the theorem but for the linear problem. The latter problem was shown to be BQP-hard in Ref.~\cite{babbush2023exponential}.

\bibliographystyle{quantum}
\bibliography{references}

\end{document}